\documentclass[reqno,12pt]{amsart}
\usepackage{color} % black,white,red,green,blue,cyan,magenta,yellow
\usepackage[pdftex]{graphicx}
\usepackage[pdftex]{hyperref}
\hypersetup{
    unicode=false,          % non-Latin characters in Acrobat’s bookmarks
    pdftoolbar=true,        % show Acrobat’s toolbar?
    pdfmenubar=true,        % show Acrobat’s menu?
    pdffitwindow=false,     % window fit to page when opened
    pdfstartview={FitH},    % fits the width of the page to the window
    pdftitle={MCP Article},      % title
    pdfauthor={Michael Holst},   % author
    pdfsubject={Mathematics},    % subject of the document
    pdfcreator={Michael Holst},  % creator of the document
    pdfproducer={Michael Holst}, % producer of the document
    pdfkeywords={PDE, analysis, mathematical physics}, % list of keywords
    pdfnewwindow=true,      % links in new window
    colorlinks=true,        % false: boxed links; true: colored links
    linkcolor=red,          % color of internal links
    citecolor=blue,         % color of links to bibliography
    filecolor=magenta,      % color of file links
    urlcolor=cyan           % color of external links
}
\usepackage{times}
\usepackage{amsfonts}
\usepackage{amsmath}
\usepackage{amsthm}
\usepackage{amssymb}
\usepackage{amsbsy}
\usepackage{amscd}
\usepackage{enumerate}
\newtheorem{theorem}{Theorem}[section]
\newtheorem{corollary}[theorem]{Corollary}
\newtheorem{lemma}[theorem]{Lemma}
\newtheorem{proposition}[theorem]{Proposition}
\newtheorem{definition}[theorem]{Definition}
\newtheorem{remark}[theorem]{Remark}
\newtheorem{weakf}{Weak Formulation}
\newtheorem*{classicalf}{Classical Formulation}
\numberwithin{equation}{section}

\newcounter{mnote}
\let\oldmarginpar\marginpar
  \renewcommand\marginpar[1]{\-\oldmarginpar[\raggedleft\footnotesize #1]%
  {\raggedright\footnotesize #1}}

\newcommand{\grad}{\nabla}

\newcommand{\lab}{\label}
\newcommand{\tiwedge}{\mbox{{\tiny $\wedge$}}}
\newcommand{\tivee}{\mbox{{\tiny $\vee$}}}
\newcommand{\floor}[1]{\lfloor{#1}\rfloor}

\renewcommand{\div}{{\operatorname{div}}}
\DeclareMathOperator*{\esssup}{ess\,sup}
\DeclareMathOperator*{\essinf}{ess\,inf}
\DeclareMathOperator*{\supp}{supp}

\newenvironment{itemizeX}
{\begin{list}{\labelitemi}
 {\setlength{\leftmargin}{1.5em}
  \setlength{\topsep}{0.5em}
  \setlength{\itemsep}{0.5em}
  \setlength{\labelwidth}{50.0em}}}
 {\end{list}}
\newenvironment{itemizeXX}
{\begin{list}{$\circ$}
 {\setlength{\leftmargin}{1.5em}
  \setlength{\topsep}{0.5em}
  \setlength{\itemsep}{0.5em}
  \setlength{\labelwidth}{50.0em}}}
 {\end{list}}
\newenvironment{itemizeXM}
{\begin{list}{\labelitemi}
 {\setlength{\leftmargin}{1.5em}
  \setlength{\topsep}{0.0em}
  \setlength{\itemsep}{0.0em}
  \setlength{\labelwidth}{50.0em}}}
 {\end{list}}
\newenvironment{itemizeXXM}
{\begin{list}{$\circ$}
 {\setlength{\leftmargin}{1.5em}
  \setlength{\topsep}{0.0em}
  \setlength{\itemsep}{0.0em}
  \setlength{\labelwidth}{50.0em}}}
 {\end{list}}

\begin{document}

\title[Rough Solutions on Asymptotically Flat Manifolds]
      {Rough Solutions of the Einstein Constraint Equations \\
       on Asymptotically Flat Manifolds Without Near-CMC Conditions}

\author[A. Behzadan]{A. Behzadan}
\email{abehzada@math.ucsd.edu}

\author[M. Holst]{M. Holst}
\email{mholst@math.ucsd.edu}

\address{Department of Mathematics\\
         University of California San Diego\\ 
         La Jolla CA 92093}

\thanks{AB was supported by NSF Award~1262982.}
\thanks{MH was supported in part by 
        NSF Awards~1262982, 1318480, and 1620366.}

\date{\today}

\begin{abstract}
In this article we consider the conformal decomposition of the
Einstein constraint equations introduced by Lichnerowicz,
Choquet-Bruhat, and York, on asymptotically flat (AF) manifolds.
Using the non-CMC fixed-point framework developed in 2009 by
Holst, Nagy, and Tsogtgerel and by Maxwell, we combine {\em a
priori} estimates for the individual Hamiltonian and momentum
constraints, barrier constructions for the Hamiltonian
constraint, Fredholm-Riesz-Schauder theory for the momentum
constraint, together with a topological fixed-point argument for
the coupled system, to establish existence of coupled non-CMC
weak solutions for AF manifolds. As was the case with the 2009
rough solution results for closed manifolds, and for the more
recent 2014 results of Holst, Meier, and Tsogtgerel for rough
solutions on compact manifolds with boundary, our results here
avoid the near-CMC assumption by assuming that the freely
specifiable part of the data given by the traceless-transverse
part of the rescaled extrinsic curvature and the matter fields are
sufficiently small. Using a coupled topological fixed-point
argument that avoids near-CMC conditions, we establish existence
of coupled non-CMC weak solutions for AF manifolds of class
$W^{s,p}_{\delta}$ (or $H^{s,p}_{\delta}$) where $p\in
(1,\infty)$, $s\in(1+\frac{3}{p},\infty)$, $-1<\delta<0$, with
metric in the positive Yamabe class. The non-CMC rough solutions
results here for AF manifolds may be viewed as an extension of
the 2009 and 2014 results on rough far-from-CMC positive Yamabe
solutions for closed and compact manifolds with boundary to the
case of AF manifolds. Similarly, our results may be viewed as
extending the recent 2014 results for AF manifolds of Dilts,
Isenberg, Mazzeo and Meier; while their results are restricted to
smoother background metrics and data, the results here allow the
regularity to be extended down to the minimum regularity allowed
by the background metric and the matter, further completing the
rough solution program initiated by Maxwell and Choquet-Bruhat in
2004.
\end{abstract}

\maketitle
\vspace*{-1.0cm}
{\tiny
\tableofcontents
}

%%%%%%%%%%%%%%%%%%%%%%%%%%%%%%%%%%%%%%%%%%%%%%%%%%%%%%%%%%%%%%%%%%%%%%%%%%%%%%
\section{Introduction}
   \label{sec:intro}

In this article, we give an analysis of the coupled Hamiltonian and
momentum constraints in the Einstein equations on a 3-dimensional
asymptotically flat (AF) manifold.
%We consider the equations with matter
%sources satisfying an energy condition implied by the dominant energy
%condition in the 4-dimensional spacetime;
 The unknowns are a Riemannian three-metric and a two-index
symmetric tensor. The equations form an under-determined system;
therefore, we focus entirely on a standard reformulation used in
both mathematical and numerical general relativity, called the
conformal method, introduced by Lichnerowicz~\cite{aL44},
Choquet-Bruhat~\cite{yC-B58}, and York~\cite{jY71,jY72}. The
(standard) conformal method assumes that the unknown metric is
known up to a scalar field called a conformal factor, and also
assumes that the trace and a term proportional to the trace-free
divergence-free part of the two-index symmetric tensor is known,
leaving as unknown a term proportional to the traceless
symmetrized derivative of a vector. Therefore, the new unknowns
are a scalar and a vector field, transforming the original
under-determined system for a metric and a symmetric tensor into
a (potentially) well-posed elliptic system for a scalar and a
vector field, which we will refer to as the
Lichnerowicz-Choquet-Bruhat-York (LCBY) system.

%See~\cite{rBjI04} for a recent review article.
%The question of existence of solutions to the Lichnerowicz-York
%conformally rescaled Einstein's constraint equations, for an
%arbitrarily prescribed mean extrinsic curvature, remained an open
%problem for more than thirty years, until it was partially resolved
%in 2008~\cite{HNT07a} and 2009~\cite{HNT07b,dM09},
%which opened the door for a number of new results in this direction
%appearing over the last five years.

The LCBY equations, which are a coupled nonlinear elliptic system
consisting of the scalar Hamiltonian constraint coupled to the
vector momentum constraint, had been studied through 2008 almost
exclusively in the setting of constant mean extrinsic curvature,
known as the \emph{CMC case}.  In the CMC case the equations
decouple, and it has long been known how to establish existence
of solutions. The case of CMC data on closed (compact without
boundary) manifolds was completely resolved by several authors
over the last thirty years, with the last remaining sub-cases
resolved and all the CMC sub-cases on closed manifolds summarized
by Isenberg in~\cite{jI95}. Over the ten years that followed,
other CMC cases on different types of manifolds containing
various kinds of matter fields were studied and partially or
completely resolved; the survey~\cite{rBjI04} gives a thorough
summary of the state of the theory through about 2004. New
results through 2008 included extensions of the CMC theory to AF
manifolds~\cite{yCBjIjY00}, including the first results for rough
solutions~\cite{2,dM05,dM06,yCB04}. The CMC case with interior
black hole boundaries is of particular interest in numerical
general relativity; solution theory for this case involves the
careful mathematical treatment of trapped surface boundary
conditions that model apparent horizons; this was completed by
2005~\cite{sD04,dM05b}. Although it is the primary formulation of
the constraint equations actually used in numerical relativity,
the complete CMC solution theory for compact manifolds with an
exterior boundary that models AF behavior, and interior trapped
surface boundaries that model apparent horizons, was developed
only recently~\cite{HMT13a}. Results for existence of solutions
for non-constant mean extrinsic curvature, but under the
assumption that the mean extrinsic curvature was nearly constant
(the \emph{near-CMC} case), began to appear in
1996~\cite{jIvM96,jIjP97,yCBjIjY00,pAaCjI07}; these were
essentially the only non-CMC results through 2008.

The first true \emph{non-CMC} existence results, without any
smallness requirement on $\tau$, began to appear in
2008~\cite{HNT07a,HNT07b,dM09}. The analysis techniques first
developed and refined in~\cite{HNT07a,HNT07b,dM09} for closed
manifolds were intensively studied and extended to a number of
other cases over the last five years, including compact manifolds
with boundary~\cite{HMT13a,51}, AF manifolds without interior
boundaries~\cite{24}, and AF manifolds with inner trapped surface
boundary conditions that model apparent horizons~\cite{HoMe14a}.
A variation of the fixed-point analysis
from~\cite{HNT07a,HNT07b,dM09} was developed in~\cite{52}, which
builds on the framework to construct an associated \emph{limit
equation} and has led to a different class of non-CMC-type
results~\cite{54,57,55,53}. One of the initially alarming
implications of the topological fixed-point arguments introduced
in~\cite{HNT07a,HNT07b,dM09} was the lack of uniqueness results,
which had always been available in the CMC case. Rather than
being a limitation in the techniques, this now appears to be
\emph{generic} when far-from-CMC data is encountered, and has
even been explicitly demonstrated~\cite{61}. Moreover, analytic
bifurcation analysis has now also been done for some versions of
the LCBY system, and the existence of a quadratic fold with
respect to certain parameterizations has now been established
using those techniques~\cite{HoMe12a} (see also the recent
related work~\cite{65}). 
An important recent development is the new method introduced in~\cite{56,58},
which makes use of the Implicit Function Theorem to prove
existence of non-CMC solutions to the LCBY system.
This approach allows for the use of classical techniques in bifurcation
theory for analyzing multiplicity of solutions, similar to the
approach taken in~\cite{HoMe12a}.
A second important development in the non-CMC theory of the LCBY system has
been the analysis~\cite{30} of the somewhat hidden underlying
structure that is common to the primary variations of the
conformal method, including the original CMC
formulation~\cite{aL44,yC-B58,jY71,jY72}, the LCBY
formulation~\cite{nOMjY74,nOMjY74a}, and the conformal thin
sandwich formulation~\cite{jY99,PY03}. The analysis in~\cite{30}
shows that the standard conformal method and the conformal thin
sandwich method are in fact the same; in addition to allowing for
the immediate transfer of known results for one method to the
other method, further analysis of the structure has led to a much
deeper understanding of the shortcomings of the conformal method
as a parameterization of initial data~\cite{64}.

%%%%%%%%%%%%%%%
%For completeness, we mention a few of the substantial works in this area,
%including the original work on the Lichnerowicz equation
%\cite{aL44};
%the development of the conformal method
%\cite{jY71,jY72,jY73,jY74};
%the initial solution theory for the Hamiltonian constraint
%\cite{nOMjY73,nOMjY74,nOMjY74a};
%the thin sandwich alternative to the conformal method
%\cite{rBgF93,cMkTjW70};
%the complete classification of CMC initial data
%\cite{jI95}
%and the few non-CMC results prior to 2008
%\cite{jIvM96,jIjP97,yCBjIjY00,pAaCjI07};
%non-existence results~\cite{jInOM04};
%various technical results on transverse-traceless tensors
%and the conformal Killing operator
%\cite{rB96,rBnOM96};
%the more recent development of the conformal thin sandwich formulation
%\cite{jY99};
%initial data for black holes
%\cite{rB00,jBjY80};
%initial data for Kerr-like black holes
%\cite{sD99,sD00b};
%initial data with trapped surface boundaries
%\cite{sD04,dM05b};
%rough solution theory for CMC initial data
%\cite{dM05,dM06,yCB04};
%and the gluing approach to generating initial data
%\cite{jC00}.
%A survey of many of these results appears in~\cite{rBjI04}.
%%%%%%%%%%%%%%%

In this article, our goal is to tackle one of the remaining open
questions with the LCBY system: The existence of rough non-CMC
solutions to the LCBY problem on AF manifolds without near-CMC
assumptions. Using the overall non-CMC fixed-point framework
developed for the closed case in~\cite{HNT07a,HNT07b,dM09}, but
now developed in the setting of the function spaces that are
relevant in the AF case, we combine {\em a priori} estimates for
the individual Hamiltonian and momentum constraints, barrier
constructions for the Hamiltonian constraint,
Fredholm-Riesz-Schauder theory for the momentum constraint,
together with a topological fixed-point argument for the coupled
system, to establish existence of coupled non-CMC weak solutions
for AF manifolds. As was the case for the earlier 2009 rough
solutions results for closed manifolds~\cite{HNT07b,dM09}, and
for the more recent 2014 rough solutions results of Holst, Meier,
and Tsogtgerel for compact manifolds with boundary~\cite{HMT13a},
our results here avoid the near-CMC assumption by assuming that
the freely specifiable part of the data given by the
traceless-transverse part of the rescaled extrinsic curvature and
the matter fields are sufficiently small. Using a coupled
topological fixed-point argument that avoids near-CMC conditions,
we establish existence of coupled non-CMC weak solutions for AF
manifolds of class $W^{s,p}_{\delta}$ (or $H^{s,p}_{\delta}$)
where $p\in (1,\infty)$, $s\in(1+\frac{3}{p},\infty)$,
$-1<\delta<0$, with metric in the positive Yamabe class. The
non-CMC rough solution results here for AF manifolds may be
viewed as an extension of the 2009 and 2014 results on rough
far-from-CMC positive Yamabe solutions for closed and compact
manifolds with boundary to the case of AF manifolds. Similarly,
our results may be viewed as extending the recent 2014 results
for AF manifolds of Dilts, Isenberg, Mazzeo and Meier~\cite{24};
while their results are restricted to smoother background metrics
and data, the results here allow the regularity to be extended
down to the maximum allowed by the background metric and the
matter, further completing the rough solution program initiated
initiated by Maxwell in~\cite{dM05,dM06} (see also~\cite{yCB04}),
and thus further extending the known solution theory for the
Einstein constraint equations.

{\bf\em A Brief Remark Concerning the Results Contained in the Paper.}
Along the way to proving the main existence result in the paper,
we will need to assemble a number of new supporting technical results;
we include some of these results in the main body of the paper when
needed to maintain the flow of an argument, whereas it was possible to
place other supporting results into the included appendices without
damaging the flow of the main arguments.
One of the technical results we need, which is not available in the
literature, concerns cases of multiplication properties of functions
in weighted spaces.
While the limited version of the result needed for this paper is included,
this has led to a second project on establishing some multiplication lemmas
that are not yet in the literature.
These results will appear in~\cite{35}.
Lastly, note that we have included the complete bootstrapping
argument that has been only outlined in prior articles (including
some of our own) for obtaining the higher-smoothness results from
the rough results.
This argument is in fact somewhat non-trivial, and we felt that it should
be included somewhere in the literature on the conformal method.

{\bf\em Outline of the Paper.}
An extended outline of the remainder of the paper is as follows.

In Section~\ref{sec:prelim}, we give a brief overview of the
conformal method, and introduce notation that we use throughout
the paper. We summarize the conformal decomposition of Einstein
constraint equations introduced by Lichnerowicz and York, on an
AF manifold, and describe the classical strong formulation of the
resulting coupled elliptic system.

In Section~\ref{sec:weak}, we define weak formulations of the
constraint equations that will allow us to develop solution
theories for the constraints in the spaces with the weakest
possible regularity. In particular, we focus on one of two
possible weak formulations of the LCBY equations; a second
alternative, which has some advantages but which we do not use in
the main body of the paper, is described in
Appendix~\ref{app:weak2}.

In Section~\ref{sec:mom}, we study the momentum constraint in isolation
from the Hamiltonian constraint.
We develop some basic technical results for the momentum constraint
operator under the weakest possible assumptions on the problem data, including
existence of weak solutions to the momentum constraint, given the
conformal factor as data.

In Section \ref{sec:ham}, we study the individual Hamiltonian constraint.
We assume the existence of barriers (weak sub- and super-solutions) to the
Hamiltonian constraint equation forming a nonempty positive bounded interval,
and then derive several properties of the Hamiltonian constraint that are
needed in the analysis of the coupled system.
The results are established under the weakest possible assumptions on the
problem data.

In Section~\ref{sec:subsuper}, we develop a new approach for the
construction of global sub- and supersolutions for the Hamiltonian constraint
on AF manifolds.
In particular, we give constructions for both sub- and supersolutions in the
positive Yamabe case that have several key features, including:
(1) they are near-CMC free;
(2) they require minimal assumptions on the data in order to be
       used for developing rough solutions;
and
(3) they have appropriate asymptotic behavior to be compatible
       with an overall fixed-point argument for the coupled system.

Finally, in Section~\ref{sec:main} we develop our main results for the
coupled system.
In particular, we clearly state and then prove the main existence result
(Theorem~\ref{thm:main})
for rough positive Yamabe solutions to the constraint equations on
AF manifolds without near-CMC assumptions.

For ease of exposition, various supporting technical results are
given in several appendices as follows:
Appendix~\S\ref{app:spaces} -- construction of fractional order
Sobolev spaces on AF manifolds; Appendix~\S\ref{app:operators} --
{\em a priori} estimates and related results for elliptic
operators on AF manifolds; Appendix~\S\ref{app:covariance} --
artificial conformal covariance of the Hamiltonian constraint on
AF manifolds; Appendix~\S\ref{app:posyam} -- results on Yamabe
positive metrics on AF manifolds; Appendix~\S\ref{app:bessel} --
some remarks on the alternative use of Bessel Potential spaces;
and Appendix~\S\ref{app:weak2} -- an alternative weak formulation
of the LCBY system on AF manifolds that makes possible additional
results that are not developed in the paper.

%%%%%%%%%%%%%%%%%%%%%%%%%%%%%%%%%%%%%%%%%%%%%%%%%%%%%%%%%%%%%%%%%%%%%%%%%%%%%%
\section{Preliminary Material}
   \label{sec:prelim}

We give a brief overview of the Einstein constraint equations and the
conformal method. A more detailed overview can be
found in \cite{HNT07b,rBjI04}.
Let $(\mathcal{M},g_{\mu\nu})$ be a 4-dimensional
globally hyperbolic spacetime, that is, $\mathcal{M}$ is a
4-dimensional smooth manifold, $g_{\mu\nu}$ is smooth, Lorentzian
metric on $\mathcal{M}$ with signature $(-,+,+,+)$ and
$\mathcal{M}$ admits a Cauchy surface (so it can be foliated by a
family of spacelike hypersurfaces). Let $\nabla_\mu$ be the
Levi-Civita connection associated with the metric $g_{\mu\nu}$.
The Einstein field equation is
\begin{equation*}
R_{\mu\nu}-\dfrac{1}{2}R g_{\mu\nu}=\kappa T_{\mu\nu},
\end{equation*}
where $T_{\mu\nu}$ is the stress-energy tensor, and $\kappa=8\pi
G/c^4$, with $G$ the gravitation constant and $c$ the speed of
light. The Ricci tensor is
$R_{\mu\nu}=R_{\mu\sigma\nu}{}^{\sigma}$ and
$R=R_{\mu\nu}g^{\mu\nu}$ is the Ricci scalar, where $g^{\mu\nu}$
is the inverse of $g_{\mu\nu}$, that is
$g_{\mu\sigma}g^{\sigma\nu}=\delta_{\mu}{}^{\nu}$.  The Riemann
tensor is defined by $R_{\mu\nu\sigma}{}^{\rho} w_{\rho}
=\big(\nabla_{\mu}\nabla_{\nu} -\nabla_{\nu}\nabla_{\mu}\bigr)
w_{\sigma}$, where $w_{\mu}$ is any 1-form on $M$.
%The stress
%energy tensor $T_{\mu\nu}$ is assumed to be symmetric and to
%satisfy the condition $\nabla_{\mu}T^{\mu\nu} = 0$ and the {\bf
%dominant energy condition}, that is, the vector
%$-T^{\mu\nu}v_{\nu}$ is time-like and future-directed, where
%$v^{\mu}$ is any time-like and future-directed vector field.

The Einstein field equation allows a formulation as an initial
value problem. The metric is the fundamental variable and the
equations involve the second derivatives of the metric. Roughly
speaking, since the equation is of order two in time, in order to
solve the problem, we need initial data on the metric and on a
first order time derivative of the metric. In the case of a
globally hyperbolic spacetime, which supports a complete
foliation with space-like hypersurfaces parameterized by a scalar
time function, one can pick a constant time hypersurface of the
spacetime $\Sigma$ and then specify the initial data
$(g|_\Sigma=\hat{h}, \dfrac{\partial g}{\partial t}|_\Sigma\sim
\hat{k})$ on that hypersurface \cite{19}. The problem then
becomes that one is not allowed to freely specify the initial
conditions in that hypersurface; rather, the
Gauss-Codazzi-Menardi equations imply that the initial data
satisfy certain conditions which are known as \emph{constraint
equations}~\cite{20}. More precisely, we have the following
definition:
\begin{definition}
A triple $(M,\hat{h},\hat{k})$ is said to be an initial data set
for the Cauchy formulation of the Einstein field equations iff
$(M,\hat{h})$ is a $3$-dimensional smooth Riemannian manifold and $\hat{k}$
is a symmetric covariant tensor of order $2$ on $M$ such that
\begin{align*}
\hat{R}-|\hat{k}|_{\hat{h}}^2+(tr_{\hat{h}} \hat{k})^2&=2\kappa\hat{\rho},
     \qquad \textrm{(Hamiltonian constraint)} \\
\div_{\hat{h}} \hat{k}-d(tr_{\hat{h}}\hat{k})&=\kappa \hat{J},
     \qquad \ \ \textrm{(Momentum constraint)}
\end{align*}
where $\hat{R}$ is the scalar curvature of $\hat{h}$, and where
$\hat{\rho}$ is a non-negative scalar field and $\hat{J}$ is a 1
form on $M$, representing the energy and momentum densities of the
matter and non-gravitational fields, respectively. $\kappa$ is a
constant.
\end{definition}
The above equations are called the \textbf{Einstein constraint
equations}. Using any local frame we may write the above equations
as follows:
\begin{align*}
\hat{R}+(\hat{h}^{ab}\hat{k}_{ab})^2-\hat{k}_{ab}\hat{k}^{ab}&=2\kappa\hat{\rho},\\
%\quad (\hat{k}=tr_{\hat{h}}{\hat{k}}=\hat{h}^{ab}\hat{k}_{ab})
\hat{\grad}^b (\hat{h}^{ac}\hat{k}_{ac})-\hat{\grad}_a
\hat{k}^{ab}&=-\kappa \hat{J}^b, \quad 1\leq b \leq 3.
\end{align*}
When the above equations hold, the manifold $M$ can be embedded
as a hypersurface in a 4-dimensional manifold corresponding to a
solution of the Einstein field equations, and the push forward of
$\hat{h}$ and $\hat{k}$ represent the first and second
fundamental forms of the embedded hypersurface. This leads to the
terminology \textbf{extrinsic curvature} for $\hat{k}$,
and \textbf{mean extrinsic curvature} for its trace
$tr_{\hat{h}}{\hat{k}}$.
%\begin{remark}
%If the source terms are zero ($\hat{\rho}=0$, $\hat{J}=0$), the
%constraint equations are called the
%\textbf{vacuum} constraint equations.
%In the case where the source terms are not zero, we assume the following consequence of the
%dominant energy condition holds true \cite{HNT07b}:\\
%The matter fields satisfy the energy condition
%$-\hat{\rho}^2+\hat{h}_{ab}\hat{j}^a\hat{j}^b \leq 0$ (with strict
%inequality at points on $M$ where $\hat{\rho}\neq 0$), which is
%implied by the dominant energy condition on the stress-energy
%tensor $T^{\mu\nu}$ in spacetime \cite{HNT07b}.
If the source terms
are zero ($\hat{\rho}=0$, $\hat{J}=0$), the constraint equations
are called the
\textbf{vacuum} constraint equations.
%\end{remark}

A general statement of the problem we are interested in is as follows.

\textit{\textbf{The Initial Data Problem in GR:}}
\textit{Given a $3$-dimensional smooth manifold $M$, a scalar function
$\hat{\rho}$ and a vector valued function (or 1 form) $\hat{J}$,
find a Riemannian metric (a symmetric, positive definite
covariant tensor $\hat{h}$ of order $2$), and a symmetric
covariant tensor $\hat{k}$ of order 2, such that the triple
$(M,\hat{h},\hat{k})$ forms an initial data set for the Einstein
constraint equations (i.e., such that $(\hat{h},\hat{k})$
satisfies the constraint equations).}

\indent The constraint equations constitute an under-determined
system of equations (the number of unknowns is twelve, whereas the
number of equations is four).
In order to produce a unique solution we must specify certain unknowns
and then solve the constraint equations for the remaining unknowns.
To this end, we employ a standard reformulation
%used in both mathematical and numerical general relativity,
known as the \textbf{conformal transverse-traceless method},
introduced by Lichnerowicz, York, and O Murchadha \cite{aL44,jY73,nOMjY74}.
In this method, the initial data on $M$ is divided into two
sets: the \emph{Free (Conformal) Data}, and the
\emph{Determined Data},
such that given a choice of free data, the constraint
equations become a \textbf{determined} system to be solved for
the determined data \cite{rBjI04}.
There are several ways to do this; here we focus on the
``semi-decoupling split'', and examine briefly how the method works.

\begin{itemizeX}
\item \textbf{Step 1:} The original unknowns, $\hat{h}$ and $\hat{k}$,
each has six distinct components, therefore we have twelve unknowns.
We can decompose $\hat{k}_{ab}$ into the trace-free and the
pure trace parts:\\
\begin{align*}
\hat{k}^{ab}=\hat{s}^{ab}+\dfrac{1}{3}(tr_{\hat{h}}{\hat{k}})\hat{h}^{ab}.
\end{align*}
Clearly $tr_{\hat{h}}{\hat{s}}=0$.

\item \textbf{Step 2:} \textbf{Conformal rescaling.} Let
\begin{equation*}
\hat{h}_{ab}=\phi^{r} h_{ab}, \quad \hat{s}^{ab}=\phi^{s} s^{ab},
\quad tr_{\hat{h}}\hat{k}=\phi^{t} \tau,
\end{equation*}
where $r$, $s$, and $t$ are fixed but arbitrary integers. Note
that if $t=0$ then $\tau$ is the mean extrinsic curvature. We
denote the Levi-Civita connection for $h_{ab}$ by $\grad_{a}$. We
will assume $h_{ab}$ and $\tau$ are given (i.e we consider them
as free data) so we are left with $7$ unknowns (components of
$s_{ab}$ and $\phi$).

\item \textbf{Step 3:} \textbf{York decomposition.}
We begin by first defining the \emph{conformal
Killing operator}
$\mathcal{L}_h : \chi(M) \rightarrow
\tau^{0}_2(M)$ as follows:
\begin{equation*}
\mathcal{L}_h(W)= \grad^b W^a+\grad^a W^b-\dfrac{2}{3}(\div_h
W)h^{ab}\quad (\div_h W=\grad_c W^c).
\end{equation*}
Here $\chi(M)$ denotes the collection of vector fields on $M$ and
$\tau^{0}_2(M)$ is the collection of contravariant tensors of
order 2. The elements in the kernel of $\mathcal{L}_h$ are called
\emph{conformal Killing fields}. In the case where the background
metric is clear from the context we may denote the conformal
Killing operator by $\mathcal{L}$ instead of $\mathcal{L}_h$. In
particular, in what follows $\grad$, $\mathcal{L}$ and $\div$ are
all taken with respect to the metric $h$. For closed manifolds
and AF manifolds, under mild conditions on the regularity
%and rate of
%decay
of $h$, one can show that if $\psi$ is a symmetric traceless
contravariant tensor of order $2$, then there exists $W\in
\chi(M)$, uniquely determined up to conformal Killing fields, such
that $\div(\mathcal{L} W)=\div \psi$ \cite{jI95,dM06}.
%\begin{theorem}\cite{jI95}
%Let $\psi$ be a symmetric traceless contravariant tensor of order
%2, then there exists $W\in \chi(M)$, uniquely up to conformal
%Killing fields, such that $\div(\mathcal{L} W)=\div \psi$.
%\end{theorem}
$\div \mathcal{L}$ is sometimes called \emph{vector Laplacian} and
is denoted by $\Delta_L$.
Therefore, there exists $W \in \chi(M)$ such that
\begin{equation*}
\Delta_L W :=\div(\mathcal{L} W)=\div s \quad
(\grad_c(\mathcal{L} W)^{ac}=\grad_c s^{ac}).
\end{equation*}
Now define $\sigma^{ab}:= s^{ab}-(\mathcal{L} W)^{ab}$.
Clearly, $\div \sigma=0$. It is easy to
check that $\sigma$ is trace-free as well. So in fact $\sigma$ is
a \emph{transverse-traceless} tensor.

\item \textbf{Step 4:} We assume $\sigma^{ab}$ is given, i.e, we will
consider it as part of the free data; now we are left with four
unknowns (components of the vector field $W^a$ and the scalar
function $\phi$).

\end{itemizeX}

Therefore, the set of free (conformal) data consists of a
background Riemannian metric $h$, a transverse-traceless
symmetric tensor $\sigma$, and a function $\tau$. The set of
determined data consists of a positive function $\phi$ and a
vector field $W$. The transformed system consists of the
\emph{Lichnerowicz-Choquet-Bruhat-York (LCBY) equations}. For the
semi-decoupling split we set $r=4$, $s=-10$, $t=0$. When energy
and momentum densities of matter and non-gravitational fields are
present, one also takes $\rho=\phi^8 \hat{\rho}$ and
$J^b=\phi^{10} \hat{J}^b$.

\medskip
\noindent
\textbf{The conformal formulation of the Einstein constraint equations.}
Applying the conformal method by following \textbf{Steps~1--4} above,
one produces a coupled nonlinear elliptic system for the unknown
conformal factor $\phi \in C^{\infty}(M)$ and $W \in \chi(M)$:
\begin{align}
-8\Delta\phi+R\phi+\dfrac{2}{3}\tau^2\phi^5-[\sigma_{ab}+(\mathcal{L}W)_{ab}][\sigma^{ab}+(\mathcal{L}W)^{ab}]\phi^{-7}&=2\kappa
\rho \phi^{-3},
\label{eqn:ham} \\
-\grad_a (\mathcal{L}W)^{ab}+\dfrac{2}{3}\phi^6\grad^b\tau
&=-\kappa J^b,
\label{eqn:mom}
\end{align}
where the first equation~\eqref{eqn:ham} is referred to as the
\textbf{conformal formulation of the Hamiltonian constraint},
and the second equation~\eqref{eqn:mom} is referred to as the
\textbf{conformal formulation of the momentum constraint}.
In the vacuum case, the right-hands sides of both equations vanish.

In order to give a complete and well-defined mathematical
formulation of the problem we study here, we begin by setting
\begin{equation*}
F(\phi,W)=a_R\phi+a_\tau\phi^5-a_W\phi^{-7}-a_\rho\phi^{-3},\quad
\mathbb{F}(\phi)=b_\tau \phi^6+b_J,
\end{equation*}
where
\begin{align*}
&b_\tau^b=(2/3)\grad^b \tau ,\quad b_J^b=\kappa J^b, \quad a_R=R/8,\quad a_\tau=\tau^2/12,\\
&a_\rho=\kappa\rho/4, \quad
a_W=[\sigma_{ab}+(\mathcal{L}W)_{ab}][\sigma^{ab}+(\mathcal{L}W)^{ab}]/8.
%&F(\phi,W)=a_R\phi+a_\tau\phi^5-a_W\phi^{-7}-a_\rho\phi^{-3},\quad
%\mathbb{F}(\phi)=b_\tau \phi^6+b_J.
\end{align*}
The \textbf{classical formulation} of the LCBY equations can be
stated as follows.
\begin{classicalf}
Given smooth functions $\tau$ and $\rho$, rank $2$
transverse-traceless tensor field $\sigma$, and vector field $J$
on the smooth $3$-dimensional Riemannian manifold $(M,h)$, find a
scalar field $\phi>0$ in $C^{\infty}(M)$ and a vector field $W$
in $\chi(M)$ such that
\begin{align*}
-\Delta \phi+F(\phi,W)&=0,\\
-\Delta_L W+\mathbb{F}(\phi)&=0.
\end{align*}
\end{classicalf}
As motivated clearly in the introduction, our goal in this
article is to provide an answer to the question of existence of
non-CMC solutions in the case of AF manifolds with very low
regularity assumptions on the data. Our approach follows closely
that taken in~\cite{HNT07b}, and is based on the following
fundamental ideas:
\begin{itemizeX}
\item \textbf{Abstract interpretation of the differential equation}: We
interpret any PDE as an equation of the form $Au=f$ where $A$ is
an operator between suitable function spaces. In this view, the
existence of a unique solution for all $f$ is equivalent to $A$
being bijective. This abstract interpretation allows one to employ
a number of general results from linear and nonlinear analysis.
\item \textbf{Conformal covariance of the Hamiltonian constraint}: The basic idea is that in the study of existence of solutions to the
Hamiltonian constraint, we have some sort of freedom in the choice
of the background metric $h$. Note that the coefficient $a_W$ is
the only part of the Hamiltonian constraint that depends on the
solution of the momentum constraint. Let's consider the
individual Hamiltonian constraint by assuming that $a_W$ is given
as data. Clearly, the Hamiltonian constraint depends on the
background metric $h$: the differential operator in the
Hamiltonian constraint is the Laplacian which is defined using
$h$, and also the scalar curvature $R$ is with respect to $h$.
An important question that one may ask is ``does the existence of
solution depend on the background metric $h$''?
More specifically,
if $h$ and $\tilde{h}$ are two conformally equivalent metrics,
%that is $\tilde{h}=\theta^4 h$ where $\theta$ is a positive
%function.
does the existence of solution for $\tilde{h}$ imply the existence
of solution for $h$? The answer for the general ``non-CMC'' case is,
unfortunately, a resounding ``NO'' \cite{30}.
However, the situation is not completely hopeless.
We examine this at length in Appendix~\ref{app:covariance}, and we show that
one can \emph{artificially} define $a_W$, $a_\rho$ and other
coefficients in the Hamiltonian constraint with respect to the new
conformally equivalent metric in such a way that some sort of
connection is made between the two equations.
This generalized type of conformal covariance is enough for our purposes here.
This should not be confused with the genuine (geometric) conformal covariance
that is true for the CMC case, and is discussed in \cite{30}.
For both CMC case and non-CMC case (as discussed in
Appendix~\ref{app:covariance}),
in the study of existence of solutions to the Hamiltonian constraint,
one may perform a conformal transformation and use a metric in the conformal
class whose scalar curvature has ``nice'' properties.
This is exactly why the
\textbf{Yamabe classes} play an important role in the study of
constraint equations.
\item \textbf{Fredholm alternative:} If $A$ is a ``nice'' linear operator
(in this context, meaning Fredholm of index zero),
then uniqueness implies existence.
%\item a priori estimates:
\item \textbf{Maximum Principle:} A linear operator $A$ satisfies the
maximum principle if $Au\leq 0$ implies $u\leq0$ in some suitable pointwise
sense.
If $A$ satisfies the
maximum principle then the solution of $Au=f$ (if it exists) is
unique.
\item \textbf{Sub- and Supersolutions and A Priori Estimates:}
Consider the equation $-\Delta
\phi+G(\phi)=0$ where $G$ is a given function. Functions
$\phi_{+}$ and $\phi_{-}$ satisfying
\begin{equation*}
-\Delta \phi_{+}+G(\phi_{+})\geq 0,\quad -\Delta
\phi_{-}+G(\phi_{-})\leq 0
\end{equation*}
%along with a boundedness condition $-K\leq \varphi_{-}\leq
%\varphi_{+}\leq K$ for some $K$,
are called a \emph{supersolution} and \emph{subsolution}, respectively.
One can show that under certain conditions the existence of super- and
subsolutions implies the existence of a solution $\phi$ to the PDE.
%One can show that under certain conditions the existence of
%super- and sub-solutions implies the existence of a solution to
%the PDE.
\item \textbf{Fixed Point Theorems:}
(in particular the contraction mapping and Schauder theorems)
We may reduce the problem of existence of solutions to the problem of
existence of fixed points of suitably defined operators.
\item \textbf{The Implicit Function Theorem:}
Although we do not use the implicit function theorem in this paper, it is important
to know that the implicit function theorem can be used in several
different ways to prove existence of solutions. For instance in
\cite{yCBjIjY00} this theorem has been used to prove the existence of
solutions of the coupled constraint equations near a given one.
Also the ``Continuity Method'', which for instance is used in
\cite{6} to study the constraint equations, usually makes use of
the implicit function theorem. The basic idea of the continuity
method is as follows: Let $\Phi (u)= 0$ be the equation to solve.
The Continuity Method consists of the following three steps
\cite{Aubin82,Gilbarg-Trudinger}:
\begin{itemizeXX}
\item \textbf{Step 1:} Find a continuous family of functions
$\Phi_\tau$ with $\tau \in [0,1]$, such that $\Phi_1 (u)=\Phi
(u)$ and $\Phi_0 (u)= 0$ is a known equation which has a solution
$u_0$.
\item \textbf{Step 2:} Prove that the set $J= \{\tau\in [0,1] : \Phi_\tau (u)=0 \, \textrm{has a
solution}\}$ is open. To show this, the Implicit Function Theorem is
typically used.
\item \textbf{Step 3:} Prove that the set $J$ is closed.
\end{itemizeXX}
Therefore $J$ is a nonempty subset of $[0,1]$ that is both open
and closed. This means $J=[0,1]$ and in particular $1\in J$.
\end{itemizeX}
The main difficulty is in finding the appropriate function spaces
as the domain and codomain of the differential operator $A$,
and ensuring that by using those function spaces we are allowed to
apply the maximum principle, Fredholm theory, fixed point
theorems, and so forth.
For elliptic equations on the whole space
 $\mathbb{R}^n$ (and also for AF manifolds), the appropriate
spaces are weighted Sobolev spaces.
To make for a reasonably self-contained article, a summary of the main
properties of weighted Sobolev spaces, and differential operators between
such spaces, has been included in Appendices A and B.

We note that although the situation that we study in this article
is more complicated, the main ideas which are employed to prove
the theorems, mostly follow those which have been used in
\cite{HNT07b} (non-CMC case on closed manifolds) and
\cite{2,dM06} (CMC case on AF manifolds).\\\\
\textbf{Notation.} Throughout this paper we use the standard
notations for Sobolev spaces. See Appendix~\S\ref{app:spaces} for
a summary of the standard notation we use here for norms. We use
the notation $A \preceq B$ to mean $A\leq cB$ where $c$ is a
positive constant that does not depend on the non-fixed parameters
appearing in $A$ and $B$.

%%%%%%%%%%%%%%%%%%%%%%%%%%%%%%%%%%%%%%%%%%%%%%%%%%%%%%%%%%%%%%%%%%%%%%%%%%%%%%
\section{Weak Formulation on Asymptotically Flat Manifolds}
   \label{sec:weak}

First let us precisely define what we mean by an asymptotically
flat manifold.
\begin{definition}\lab{defAE}
Let $M$ be an $n$-dimensional smooth connected oriented manifold
and let $h$ be a metric on $M$ for which $(M,h)$ is complete. Let
$E_r=\{x\in \mathbb{R}^n : |x|> r\}$. We say $(M,h)$ is
asymptotically flat (AF) of class $W^{s,p}_{\delta}$ (where
$s\geq 0$, $p\in(1,\infty)$, and $\delta<0$) if
\begin{enumerate}
\item $h\in W^{s,p}_{loc}$.
\item There is a finite collection $\{U_i\}_{i=1}^m$ of open sets
of $M$ and diffeomorphisms $\phi_i: U_i \rightarrow E_1$ such that
$M\setminus (\cup_{i=1}^m U_i)$ is compact.
\item There exists a constant $\vartheta\geq 1$ such that
for each $i$
\begin{equation*}
\forall\, x\in E_1 \,\, \forall y\in \mathbb{R}^n \quad
\vartheta^{-1}|y|^2\leq ((\phi_i^{-1})^{*}h)_{rs}(x)y^r y^s\leq
\vartheta |y|^2. \quad (\textrm{see Remark
\ref{remextraconditionaf}})
\end{equation*}
\item There exists a positive constant $\omega$ such that for each $i$, $(\phi_i^{-1})^{*}h-\omega\bar{h}\in
W^{s,p}_{\delta}(E_1)$, where $\bar{h}$ is the Euclidean metric.
\end{enumerate}
The charts $(U_i,\phi_i)$ are called \emph{end charts}, and the
corresponding coordinates are called \emph{end coordinates}.
\end{definition}
Our goal is to come up with a weak formulation of LCBY
conformally rescaled Einstein constraint equations in order to
accommodate nonsmooth data on a $3$-dimensional AF manifold
$(M,h)$ of class $W^{s,p}_{\delta}$. There are at least two
different general settings where the LCBY equations are
well-defined with rough data; one of them is described in this
section and the other is discussed in Appendix~\ref{app:weak2}.
In both settings it is assumed that the AF manifold is of class
$W^{s,p}_{\delta}$ where $s>\frac{3}{p}$ (and of course
$p\in(1,\infty)$, $\delta<0$). So by Corollary
 \ref{coroA1},  $W^{s,p}_{\delta}$ is a Banach algebra and
$W^{s,p}_{\delta}\hookrightarrow C_{\delta}^{0}\hookrightarrow
L^{\infty}_{\delta}$. The framework that is described in
Appendix~\ref{app:weak2}
(which we refer to as \textbf{Weak Formulation~2}) only works for
$s\leq 2$, but the framework that is described in this section
(which we refer to as \textbf{Weak Formulation~1}) works for all
$s > 3/p$ with $p \in (1,\infty)$, even when $s>2$.
(However, as we explain at some length in Appendix~\ref{app:weak2},
\textbf{Weak Formulation~2} is not simply a special case of
\textbf{Weak Formulation 1}.)

Note that if $(M,h)$ is a $3$-dimensional AF manifold of class
$W^{s,p}_{\delta}$ and if $u\in W^{s,p}_{\delta}$ is a positive
function, then $(M, u^{4}h)$ is not asymptotically flat of class
$W^{s,p}_{\delta}$ (item (4) in Definition~\ref{defAE} is not satisfied).
However, if $u$ is a
positive function such that $u-\mu\in W^{s,p}_{\delta}(M)$ for
some positive constant $\mu$, then $(M, u^{4}h)$ is also AF of
class $W^{s,p}_{\delta}$. Indeed,
\begin{equation*}
u-\mu\in W^{s,p}_{\delta} \Rightarrow u-\mu\in W^{s,p}_{loc}
\Rightarrow u\in W^{s,p}_{loc} \Rightarrow u^{4}h\in
W^{s,p}_{loc} \quad (\textrm{$W^{s,p}_{loc}$ is an algebra}).
\end{equation*}
In addition, $u-\mu\in W^{s,p}_{\delta}$ implies that $u$ is bounded and
$\inf u>0 $ (see Remark \ref{rem1}; note that $u$ is a positive function).
Thus, there exists a positive number $\zeta$ such that
$\zeta^{-1}<u^4<\zeta$. Consequently for each $i$,
$(\phi_i^{-1})^{*}u^4=u^4\circ \phi_i^{-1}$ is between
$\zeta^{-1}$ and $\zeta$ which subsequently implies that
\begin{equation*}
\forall\, x\in E_1 \,\, \forall y\in \mathbb{R}^n \quad (\zeta
\vartheta)^{-1}|y|^2\leq ((\phi_i^{-1})^{*}(u^4h))_{rs}(x)y^r
y^s\leq (\zeta \vartheta) |y|^2.
\end{equation*}
Finally, since $(M,h)$ is AF of class $W^{s,p}_{\delta}$, there
exists a constant $\omega$ such that
$(\phi_i^{-1})^{*}h-\omega^{4}\bar{h}\in W^{s,p}_{\delta}(E_1)$;
if we let $v=u-\mu$ and $f(x)=(\mu+x)^4$, then for each $1\leq i
\leq m$ (end coordinates) we have
\begin{align*}
(\phi_i^{-1})^{*}(u^4
h)-(\mu\omega)^{4}\bar{h}
&=(u^4\circ\phi_i^{-1})(\phi_i^{-1})^{*} h-(\mu\omega)^{4}\bar{h}
\\
&=(\mu+v\circ\phi_i^{-1})^4
(\phi_i^{-1})^{*} h-(\mu\omega)^{4}\bar{h}
\\
&=f(v\circ\phi_i^{-1})
(\phi_i^{-1})^{*}h-(\mu\omega)^{4}\bar{h}
\\
&=f(v\circ\phi_i^{-1})((\phi_i^{-1})^{*}h-\omega^{4}\bar{h})
     +(\omega^{4}f(v\circ\phi_i^{-1})-(\mu\omega)^{4})\bar{h}.
\end{align*}
%Note that $v\in W^{s,p}_{\delta}(M)$ and so $v\circ\phi_i^{-1}\in
%W^{s,p}_{\delta}(E_1) $.
Since
$(\phi_i^{-1})^{*}h-\omega^{4}\bar{h}\in W^{s,p}_{\delta}(E_1)$,
by Lemma \ref{lempA1} the first term on the right is in
$W^{s,p}_{\delta}(E_1)$. Also as a direct consequence of Corollary
\ref{corpA1}, the second term on the right is in
$W^{s,p}_{\delta}(E_1)$.

In the LCBY equations, $\phi>0$ is the conformal factor, so
assuming $(M,h)$ is a $3$-dimensional AF manifold of class
$W^{s,p}_{\delta}$, by what was mentioned above, it seems
reasonable to let $\phi=\psi+\mu$ (so $\psi>-\mu$) where $\psi \in
W^{s,p}_{\delta}$ and $\mu$ is an arbitrary but fixed positive
constant (we have freedom in choosing the constant
$\mu$).

We can write the Hamiltonian constraint in terms of
$\psi$ as:
$$
-\Delta \psi+f(\psi,W)=0,
$$
where
\begin{align*}
f(\psi,W)&=F(\phi,W) \\
   &=a_R\phi+a_\tau\phi^5-a_W\phi^{-7}-a_\rho\phi^{-3}\\
   &=a_R(\psi+\mu)+a_\tau(\psi+\mu)^5-a_W(\psi+\mu)^{-7}-a_\rho(\psi+\mu)^{-3}.
\end{align*}
Since $\psi \in W^{s,p}_{\delta}$, we want to be able to extend
$-\Delta:C^{\infty}\rightarrow C^{\infty}$ to an operator $A_L:
W^{s,p}_{\delta}\rightarrow W^{s-2,p}_{\delta-2}$. As discussed in
Appendix~\ref{app:operators}, since $-\Delta\in
D^{s,p}_{2,\delta}$ (See Definition \ref{defellipticoperator}),
by the extension theorem (Theorem \ref{thmB1}), the only extra
assumption needed to ensure the above extension is possible is
$s\geq 1$. Indeed, according to Theorem \ref{thmB1}, we must
check the following conditions (following the numbering used in
Theorem \ref{thmB1}):
\begin{center}
\begin{tabular}{rll}
   (i) & $p\in (1, \infty)$,  & (true by assumption) \\
  (ii) & $s\geq 2-s$,         & (so need to assume $s\geq 1$) \\
 (iii) & $s-2\leq s-2$,       & (trivially true) \\
  (iv) & $s-2 < s-2 + s-\frac{3}{p}$,
                              & (since $s> \frac{3}{p}$) \\
   (v) & $s-2-\frac{3}{p}\leq s-\frac{3}{p}-2$,
                              & (trivially true) \\
  (vi) & $s-\frac{3}{p}> 2-3-s+\frac{3}{p}$.
                              & (since $s> \frac{3}{p}$)
\end{tabular}
\end{center}

\noindent
{\bf Framework 1:}\\
In this framework we look for $W$ in $\textbf{W}^{e,q}_{\beta}$
where $\beta<0$. For the momentum constraint to be well-defined,
we need to ensure that
\begin{align}
& \mbox{The operator: } -\Delta_L \colon
  \mathbf{C}^{\infty}\to \mathbf{C}^{\infty}
  \mbox{ can be extended to }
  \mathcal{A}_L\colon \mathbf{W}^{e,q}_{\beta}\to
  \mathbf{W}^{e-2,q}_{\beta-2},
\label{eqn:condition_1}
\\
& \mbox{It holds that: }
  b_\tau (\psi+\mu)^6+b_J \in \mathbf{W}^{e-2,q}_{\beta-2}.
\label{eqn:condition_2}
\end{align}
The vector Laplacian belongs to the class $D^{s,p}_{2,\delta}$
(See Definition \ref{defellipticoperator}).
Therefore, by Theorem \ref{thmB1},
in order to ensure that condition~\eqref{eqn:condition_1} holds true,
it is enough to require $e$ and $q$ satisfy the following conditions
(again, the numbering below corresponds to numbering in Theorem \ref{thmB1}):
\begin{center}
\begin{tabular}{rll}
   (i) &  $q\in (1,\infty)$, & \\
  (ii) &  $e> 2-s$,       & \\
 (iii) &  $e-2\leq \min\{e,s\}-2,
             \ p\leq q \mbox{ if } e=s\not\in\mathbb{N}_0$,
           & (in particular, need $e\leq s$) \\
  (iv) &  $e-2< e-2+s-\frac{3}{p}$,
           & (holds by assumption $s>\frac{3}{p}$) \\
   (v) &  $e-2-\frac{3}{q}\leq s-\frac{3}{p}-2$,
           & (must assume $e\leq s+\frac{3}{q}-\frac{3}{p}$) \\
  (vi) &  $e-\frac{3}{q}> 2-3-s+\frac{3}{p}$.
           & (must assume $e> -s+\frac{3}{p}-1+\frac{3}{q}$)
\end{tabular}
\end{center}
Combining these constraints, we see it is enough to have
\begin{align*}
&q\in (1,\infty),\\
&e\in (2-s,s]\cap (-s+\frac{3}{p}-1+\frac{3}{q},
s+\frac{3}{q}-\frac{3}{p}].
    \quad (p=q \mbox{ if } e=s\not \in \mathbb{N}_0)
\end{align*}
Note that in case $e=s \not \in \mathbb{N}_0$ we need to assume
$p\leq q$, which together with the inequality $s=e\leq
s+\frac{3}{q}-\frac{3}{p}$ justifies the assumption $p=q$ in this
case.

In order to ensure that condition~\eqref{eqn:condition_2} holds true,
it is enough to make the extra assumptions that $\tau$ is given in
$W^{e-1,q}_{\beta-1}$ and $J$ is given in
$\textbf{W}^{e-2,q}_{\beta-2}$. Indeed, note that $\tau \in
W^{e-1,q}_{\beta-1}$ implies $b_{\tau}\in
\textbf{W}^{e-2,q}_{\beta-2}$. Since $\psi \in W^{s,p}_{\delta}$,
it follows from Lemma \ref{lempA1} that $b_\tau (\psi+\mu)^6 \in
\textbf{W}^{e-2,q}_{\beta-2}$; Lemma \ref{lempA1} can be applied
since
(numbering corresponds to the numbering of conditions in Lemma \ref{lempA1}):
\begin{center}
\begin{tabular}{rll}
  (i) & $e-2 \in (-s,s)$,
             & (since $e\in (2-s,s]$) \\
%& q\in (1,\infty)\quad (\textrm{checked above})\\
 (ii) & $e-2-\frac{3}{q}\leq s-\frac{3}{p}$,
             & (since $e\leq s+\frac{3}{q}-\frac{3}{p}$) \\
      & $-3-s+\frac{3}{p}\leq e-2-\frac{3}{q}$.
             & (since $e> -s+\frac{3}{p}-1+\frac{3}{q}$)
\end{tabular}
\end{center}
\emph{In summary, for the momentum constraint to be well-defined,
it is enough to make the following additional assumptions}:
\begin{equation*}
q\in (1,\infty),\quad e\in (2-s,s]\cap
(-s+\frac{3}{p}-1+\frac{3}{q}, s+\frac{3}{q}-\frac{3}{p}],\quad
\tau \in W^{e-1,q}_{\beta-1},\quad J\in
\textbf{W}^{e-2,q}_{\beta-2}.
\end{equation*}
Of course, we let $p=q$ if $e=s \not \in \mathbb{N}_0$,
and the base assumptions hold as well
($s\geq 1,\, p\in(1,\infty),\delta,\beta<0,\, s>\frac{3}{p}$).
Note that for $(2-s,s]$ to be nonempty, in fact we need $s>1$.

Finally, we now consider the Hamiltonian constraint.
Note that $W \in
\textbf{W}^{e,q}_{\beta}$ and so that $\mathcal{L}W\in W^{e-1,q}_{\beta-1}$.
For $a_W=\frac{1}{8}|\sigma+\mathcal{L}W|^2$ to be well-defined,
it is enough to assume $\sigma \in W^{e-1,q}_{\beta-1}$.
Recall that $A_L$ is a well-defined operator from $W^{s,p}_{\delta}$ to
$W^{s-2,p}_{\delta-2}$.
If we set $\eta=\max\{\beta, \delta\}$,
then $W^{s-2,p}_{\delta-2}\hookrightarrow W^{s-2,p}_{\eta-2}$.
In fact, $A_L$ can be considered as an operator from
$W^{s,p}_{\delta}$ to $W^{s-2,p}_{\eta-2}$ where $\eta=\max\{\beta, \delta\}$.
Consequently, for the Hamiltonian constraint to be well-defined,
we need to have
\begin{equation*}
f(\psi,W)=a_R(\psi+\mu)+a_\tau(\psi+\mu)^5-a_W(\psi+\mu)^{-7}-a_\rho(\psi+\mu)^{-3}\in
W^{s-2,p}_{\eta-2}.
\end{equation*}
One way to guarantee that the above statement holds true is to
ensure that
\begin{equation*}
a_\tau, a_\rho, a_W\in W^{s-2,p}_{\beta-2}, \quad a_R \in
W^{s-2,p}_{\delta-2},
\end{equation*}
and then show that if $f$ is smooth on $(-\mu,\infty)$, $u\in
W^{s,p}_{\delta}$, and $v\in W^{s-2,p}_{\eta-2}$ then $f(u)v \in
W^{s-2,p}_{\eta-2}$.
We claim that for above statements to be
true it is enough to make the following extra assumptions:
\begin{equation*}
e>1+\frac{3}{q},\quad e\geq s-1, \quad e\geq
\frac{3}{q}+s-\frac{3}{p}-1, \quad \rho\in W^{s-2,p}_{\beta-2}.
\end{equation*}
The details are as follows:
\begin{enumerate}
\item  If $f$ is smooth and $u\in W^{s,p}_{\delta}$,
$v\in W^{s-2,p}_{\eta-2}$ then $f(u)v \in W^{s-2,p}_{\eta-2}$.\\
By Lemma \ref{lempA1}, we just need to check the following
(the numbering matches that of the conditions in Lemma~\ref{lempA1}):
\begin{center}
\begin{tabular}{lll}
  (i) & $s-2\in (-s,s)$, & (since $s> 1$) \\ %\quad (\delta<0)
 (ii) & $s-2-\frac{3}{p}
   \in [-3-s+\frac{3}{p},s-\frac{3}{p}]$. & (since $s>\frac{3}{p}$)
\end{tabular}
\end{center}
This shows that $f(u)v \in W^{s-2,p}_{\eta-2}$.
\item $a_\tau=\frac{1}{12} \tau^2$.\\
We want to ensure $a_\tau \in W^{s-2,p}_{\beta-2}$. Note that
$\tau \in W^{e-1,q}_{\beta-1}$; since $e-1>\frac{3}{q}$,
$W^{e-1,q}_{\beta-1}\times W^{e-1,q}_{\beta-1}\hookrightarrow
W^{e-1,q}_{2\beta-2}$ (see Corollary \ref{coroA2}). Therefore
$\tau^2\in W^{e-1,q}_{2\beta-2}$.
Thus we want to have
$W^{e-1,q}_{2\beta-2}\hookrightarrow W^{s-2,p}_{\beta-2}$. We
will see that because of the assumptions $e\geq s-1$ and $e\geq
\frac{3}{q}+s-\frac{3}{p}-1$ this embedding holds true.
We just
need to check that the assumptions of Theorem  \ref{thmA4} are
satisfied (numbering follows assumptions of Theorem \ref{thmA4}):
\begin{center}
\begin{tabular}{lll}
%&2\beta-2<\beta-2 \quad (\beta<0)\\
%& p\leq r \,\,(\textrm{because}\quad s\leq 2)\\
   (ii) & $e-1\geq s-2$, & (since $e\geq s-1$) \\
  (iii) & $e-1-\frac{3}{q}\geq s-2-\frac{3}{p}$,
          & (since $e\geq \frac{3}{q}+s-\frac{3}{p}-1$) \\
   (iv) & $2\beta-2<\beta-2$. & (since $\beta<0$)
\end{tabular}
\end{center}
\item $a_R=\frac{R}{8}$.\\
We want to ensure $a_R\in W^{s-2,p}_{\delta-2}$. Note that $h$
is an AF metric of class $W^{s,p}_{\delta}$ and $R$ involves the
second derivatives of $h$, so $R\in W^{s-2,p}_{\delta-2}$. We do
not need to impose any extra restrictions for this one.
\item $a_\rho=\kappa\rho/4$.\\
Clearly $a_\rho\in W^{s-2,p}_{\beta-2}$ iff $\rho\in
W^{s-2,p}_{\beta-2}$.
\item
$a_W=[\sigma_{ab}+(\mathcal{L}W)_{ab}][\sigma^{ab}+(\mathcal{L}W)^{ab}]/8$.\\
We want to ensure that $a_W\in W^{s-2,p}_{\beta-2}$.
Note that
$\mathcal{L}W, \sigma \in W^{e-1,q}_{\beta-1}$ and by our
restrictions on $e$, $W^{e-1,q}_{\beta-1}\times
W^{e-1,q}_{\beta-1}\hookrightarrow W^{e-1,q}_{2\beta-2}$ and
$W^{e-1,q}_{2\beta-2}\hookrightarrow W^{s-2,p}_{\beta-2}$.
Thus, we have
$a_W=\frac{1}{8}|\sigma+ \mathcal{L}W|^2\in
W^{e-1,q}_{2\beta-2}\hookrightarrow W^{s-2,p}_{\beta-2}$.
\end{enumerate}

We are finally in a position to give a well-defined weak formulation of the
Einstein constraint equations on AF manifolds with rough data, through the
use of {\bf Framework~1}.
(In Appendix~\ref{app:weak2},
we show how \textbf{Framework~2} leads to an alternative
weak formulation, leading to slightly different existence results.)
\begin{weakf}\lab{weakf2}
Let $(M,h)$ be a $3$-dimensional AF Riemannian manifold of class
$W^{s,p}_{\delta}$ where $p\in (1,\infty)$, $\beta, \delta<0$ and
$s\in (1+\frac{3}{p},\infty)$. Select $q$ and $e$ to satisfy
\begin{align*}
&\frac{1}{q}\in (0,1)\cap (0,\frac{s-1}{3})\cap [\frac{3-p}{3p},\frac{3+p}{3p}],\\
& e\in
(1+\frac{3}{q},\infty)\cap[s-1,s]\cap[\frac{3}{q}+s-\frac{3}{p}-1,\frac{3}{q}+s-\frac{3}{p}].
\end{align*}
 Let $q=p$ if $e=s \not \in \mathbb{N}_0$. Fix source functions:
\begin{equation*}
\tau \in W^{e-1,q}_{\beta-1},\quad \sigma \in
W^{e-1,q}_{\beta-1},\quad \rho \in W^{s-2,p}_{\beta-2} (\rho\geq
0),\quad J\in \textbf{W}^{e-2,q}_{\beta-2}.
\end{equation*}
Let $\eta=\max\{\beta, \delta\}$. Define
$f:W^{s,p}_{\delta}\times \textbf{W}^{e,q}_{\beta}\rightarrow W^{s-2,p}_{\eta-2}$
and
$\textbf{f}:W^{s,p}_{\delta}\rightarrow \textbf{W}^{e-2,q}_{\beta-2}$
by
\begin{align*}
&\quad f(\psi,W)=a_R(\psi+\mu)+a_\tau(\psi+\mu)^5-a_W(\psi+\mu)^{-7}-a_\rho(\psi+\mu)^{-3},\\
&\quad \textbf{f}(\psi)=b_\tau
(\psi+\mu)^6+b_J.
\end{align*}
Find $(\psi, W)\in W^{s,p}_{\delta}\times
\textbf{W}^{e,q}_{\beta}$ such that
\begin{align}\lab{eqweak1}
A_L \psi+f(\psi,W)=0, \\
\mathcal{A}_L W+\textbf{f}(\psi)=0.\lab{eqweak2}
\end{align}
\end{weakf}
\begin{remark}\lab{rem2.2}
We make the following observations regarding \textbf{Weak Formulation 1}.
\begin{itemizeXXM}
\item Since $s\geq 1$, the condition $e>1$ implies $e> 2-s$.
Therefore, we did not explicitly state the condition $e> 2-s$ in the
above formulation.
\item The condition $e>\frac{3}{q}+1$ together with $s>\frac{3}{p}$
imply that $e>-s+\frac{3}{p}-1+\frac{3}{q}$.
Therefore, we did not explicitly state the condition
$e>-s+\frac{3}{p}-1+\frac{3}{q}$
in the above formulation.
\item For $(1+\frac{3}{q},\infty)\cap[s-1,s]$ to be nonempty we
need to have $1+\frac{3}{q}<s$. This is why we have
$\frac{1}{q}\in(0,\frac{s-1}{3})$ in the weak formulation.
\item For $(1+\frac{3}{q},\infty)\cap
[\frac{3}{q}+s-\frac{3}{p}-1,\frac{3}{q}+s-\frac{3}{p}]$ to be
nonempty we need to have
$1+\frac{3}{q}<\frac{3}{q}+s-\frac{3}{p}$. That is,
$s>1+\frac{3}{p}$ (therefore, we did not need to explicitly state $s\geq 1$).
\item For
$[s-1,s]\cap[\frac{3}{q}+s-\frac{3}{p}-1,\frac{3}{q}+s-\frac{3}{p}]$
to be nonempty we need to have $\frac{1}{p}-\frac{1}{3}\leq
\frac{1}{q}\leq \frac{1}{p}+\frac{1}{3}$. That is,
$\frac{1}{q}\in[\frac{3-p}{3p},\frac{3+p}{3p}]$.
\end{itemizeXXM}
\end{remark}
\begin{remark}
\textbf{Our analysis in this paper is based on the weak
formulation described above}. In some of the theorems that
follow, for the claimed estimates to be true in the case $e\leq
2$ we will need to restrict the admissible space of $\tau$. In
those cases we will assume $\tau \in W^{1,z}_{\beta-1}$ where
$z=\frac{3q}{3+(2-e)q}$.
We note that $z$ has been chosen in this form to
ensure that $W^{1,z}_{\beta-1}\hookrightarrow W^{e-1,q}_{\beta-1}$
(and so $L^z_{\beta-2}\hookrightarrow W^{e-2,q}_{\beta-2}$).
Indeed, by Theorem \ref{thmA3}, for
$W^{1,z}_{\beta-1}\hookrightarrow W^{e-1,q}_{\beta-1}$ to hold
true we need to have (the numbering follows the assumptions in
Theorem \ref{thmA3}):
\begin{center}
\begin{tabular}{rll}
   (i) & $z\leq q$,   & \\
  (ii) & $1\geq e-1$, & (true for $e\leq 2$)\\
 (iii) & $1-\frac{3}{z}\geq e-1-\frac{3}{q}$. &
%\Leftrightarrow
%\frac{3}{z}\leq (2-e)+\frac{3}{q}\Leftrightarrow z\geq
%\frac{3q}{(2-e)q+3}.\\
\end{tabular}
\end{center}
Now note that if we set $z=\frac{3q}{(2-e)q+3}$, then the first
condition and the third condition are both satisfied (for $e\leq
2$):
\begin{align*}
z &\leq q \Leftrightarrow \frac{3q}{(2-e)q+3}\leq q
\Leftrightarrow \frac{3}{(2-e)q+3}\leq 1 \Leftrightarrow 2-e\geq
0,\\
1-\frac{3}{z} &\geq e-1-\frac{3}{q}\Leftrightarrow
\frac{3}{z}\leq (2-e)+\frac{3}{q}\Leftrightarrow z\geq
\frac{3q}{(2-e)q+3}.
\end{align*}
\end{remark}

%%%%%%%%%%%%%%%%%%%%%%%%%%%%%%%%%%%%%%%%%%%%%%%%%%%%%%%%%%%%%%%%%%%%%%%%%%%%%%
\section{Results for the Momentum Constraint}
   \label{sec:mom}

We now develop the main results will need for the momentum
constraint operator on AF manifolds with rough data.
%\begin{equation*}
%\Delta_L W^b=b_\tau^b\phi^6+b_J^b.
%\end{equation*}
%Note that in the CMC case $b_\tau=0$. For non-CMC case assume
%$\phi\in L^{\infty}$ is given.
\begin{theorem}\lab{thm1}
Let $(M,h)$ be a 3-dimensional AF Riemannian manifold of class
$W^{s,p}_{\delta}$ with $p\in (1,\infty)$, $\delta<0$ and $s\in
(\frac{3}{p},\infty)\cap (1,\infty)$. Select $q, e$ to satisfy:
\begin{equation}\lab{eqnice}
q\in (1,\infty),\quad e\in(2-s,s]\cap
(-s+\dfrac{3}{p}-1+\dfrac{3}{q},s-\dfrac{3}{p}+\dfrac{3}{q}].
\end{equation}
In case $e=s \not \in \mathbb{N}_0$, assume $q=p$. In case $e=s
\in \mathbb{N}_0, p>2, q<2$, assume $e>\frac{3}{q}-\frac{1}{2}$.
In case $e=s-\frac{3}{p}+\frac{3}{q}, p<2, q>2$ assume
$e>\frac{1}{2}$. Suppose $\beta\in (-1,0)$ and $b_J$ and $b_\tau$
and $\psi$ are such that $\textbf{f}(\psi)\in
W^{e-2,q}_{\beta-2}$ (in particular, we know that if we fix the
source terms $b_J$ and $b_\tau$ in $\textbf{W}^{e-2,q}_{\beta-2}$
and $\psi\in W^{s,p}_{\delta}$ then $\textbf{f}(\psi)\in
\textbf{W}^{e-2,q}_{\beta-2}$). Then $\mathcal{A}_L:
\textbf{W}^{e,q}_{\beta}\rightarrow \textbf{W}^{e-2,q}_{\beta-2}$
is Fredholm of index zero. Moreover if $h$ has no nontrivial
conformal Killing fields, then the momentum constraint
$\mathcal{A}_L W+\textbf{f}(\psi)=0$ has a unique solution $W\in
\textbf{W}^{e,q}_{\beta}$ with
\begin{equation*}
\|W\|_{\textbf{W}^{e,q}_{\beta}}\leq
C\|\textbf{f}(\psi)\|_{\textbf{W}^{e-2,q}_{\beta-2}},
\end{equation*}
where $C>0$ is a constant.
\end{theorem}
\begin{remark}
In the above theorem the ranges for $e$ and $q$ are chosen so
that the momentum constraint is well-defined. Also note that for
$(2-s,s]$ to be a nonempty interval we had to assume that s is
\textbf{strictly} larger than $1$.
\end{remark}
\begin{remark}\lab{remconformalkilling}
There are important cases where the assumption that  ``$h$ has no
nontrivial conformal Killing fields'' is automatically satisfied.
For instance in \cite{dM06} it is proved that if $(M,h)$ is AF of
class $W^{s,2}_{\delta}$ with $s>\frac{3}{2}$ (and of course
$\delta<0$) and if $X\in W^{s,2}_{\rho}$ with $\rho<0$ is a
conformal Killing field, then $X$ vanishes identically. We do not
pursue this issue here, but interested readers may find more
information in \cite{dM06} and \cite{dM05b}.
\end{remark}
%%%%%%%%%%%%%%%%%%%%%%%%%%%%%%%%%%%%%%%
\begin{proof}{\bf (Theorem~\ref{thm1})}
The proof will involve three main steps.
\begin{itemizeX}
\item \textbf{Step 1: Establish that $\mathcal{A}_L$ is Fredholm of index zero.}

$\mathcal{A}_L$ is of class $D^{s,p}_{2,\delta}$. Therefore by
Proposition \ref{propb1}, $\mathcal{A}_L:
\textbf{W}^{e,q}_{\beta}\rightarrow \textbf{W}^{e-2,q}_{\beta-2}$
is semi-Fredholm (this is exactly why
%$e=2-s$ and
%$e=-s+\dfrac{3}{p}-1+\dfrac{3}{q}$ are not included in the
%admissible range of $e$ and why
it is assumed $\beta\in (-1,0)$).
On the other hand, vector Laplacian of the rough metric can be
approximated by the vector Laplacian of smooth metrics and it is
well known that vector Laplacian of a smooth metric is Fredholm of
index zero. Therefore since the index of a semi-Fredholm map is
locally constant, it follows that $\mathcal{A}_L:
\textbf{W}^{e,q}_{\beta}\rightarrow \textbf{W}^{e-2,q}_{\beta-2}$
is Fredholm with index 0.
%By Lemma (\ref{lemb7}), we know that $\mathcal{A}_L$ can be
%approximated by vector Laplacian of smooth metrics and also
%according to Lemma (\ref{lemb6}) vector Laplacian of a smooth
%metric is Fredholm of index 0. Therefore since the index of a
%semi-Fredholm map is locally constant, it follows that
%$\mathcal{A}_L: \textbf{W}^{e,q}_{\beta}\rightarrow
%\textbf{W}^{e-2,q}_{\beta-2}$ is Fredholm with index 0.
\medskip
\item \textbf{Step 2: Show that if $Ker \mathcal{L}=\{0\}$, then $Ker\mathcal{A}_L=\{0\}$.}

The proof of this step involves considering six distinct cases.
In each case, we denote the operator $\mathcal{A}_L$ acting on
$\textbf{W}^{e,q}_{\beta}$ by $(\mathcal{A}_L)_{e,q,\beta}$. In
order to best organize the arguments for these six cases, we make
the following definitions:
\begin{center}
\begin{tabular}{ll}
{\bf\em nice triple:}
   & A triple $(e,q,\beta)$ where $-1<\beta<0$ and $e, q$
     satisfy (\ref{eqnice}).
\\
{\bf\em super nice triple:}
   & A {\bf\em nice} triple $(e,q,\beta)$
     where $e\neq s$ and $e\neq s-\frac{3}{p}+\frac{3}{q}$.
\end{tabular}
\end{center}
We now make three observations about relationships
between these definitions.
\begin{itemizeXX}
\item \textbf{Observation 1}: For any $-1<\beta<0$, $(e=1, q=2,
\beta)$ is {\bf\em super nice} and $(e=s,q=p, \beta)$ is {\bf\em nice}.
Indeed,
\begin{center}
\begin{tabular}{lll}
  $1 \in (2-s,s)$, & (since $s>1$) \\
  $1>-s+\frac{3}{p}-1+\frac{3}{2}$,
     & (since $s>\frac{3}{p}$) \\
  $1<s-\frac{3}{p}+\frac{3}{2}$,
     & (since $s>\frac{3}{p}$) \\
  $s\in (2-s,s]$, & (trivially true; note $s>1$) \\
  $s>-s+\frac{3}{p}-1+\frac{3}{p}$,
     & (since $s>\frac{3}{p}$) \\
  $s\leq s-\frac{3}{p}+\frac{3}{p}$. & (trivially true)
\end{tabular}
\end{center}
\item \textbf{Observation 2}: If $(e,q,\beta)$ is {\bf\em super nice}, then
$(2-e,q',-1-\beta)$ is also {\bf\em super nice}. Indeed,
\begin{center}
\begin{tabular}{lll}
 $q\in (1,\infty)\Rightarrow q' \in (1,\infty)$,     & \\
 $\beta \in (-1,0) \Rightarrow -1-\beta \in (-1,0)$, & \\
 $e\in (2-s,s) \Rightarrow 2-e \in (2-s,s)$,         & \\
 $e< s-\frac{3}{p}+\frac{3}{q}\Rightarrow
   2-e> 2-s+\frac{3}{p}-\frac{3}{q}=-s+\frac{3}{p}-1+\frac{3}{q'}$, & \\
 $e>-s+\frac{3}{p}-1+\frac{3}{q}\Rightarrow
   2-e<2+s-\frac{3}{p}+1-\frac{3}{q}=s-\frac{3}{p}+\frac{3}{q'}$. & \\
\end{tabular}
\end{center}
\item \textbf{Observation 3}: Suppose $(e_1,q_1,\beta_1)$ and
$(e_2,q_2,\beta_2)$ are {\bf\em nice} triples.
If we have
$\textbf{W}^{e_2,q_2}_{\beta_2}\hookrightarrow
\textbf{W}^{e_1,q_1}_{\beta_1}$, then
$(\mathcal{A}_L)_{e_2,q_2,\beta_2}$ is the restriction of
$(\mathcal{A}_L)_{e_1,q_1,\beta_1}$ to
$\textbf{W}^{e_2,q_2}_{\beta_2}$ and so $Ker
(\mathcal{A}_L)_{e_2,q_2,\beta_2}\subseteq Ker
(\mathcal{A}_L)_{e_1,q_1,\beta_1}$. In particular, if
$Ker (\mathcal{A}_L)_{e_1,q_1,\beta_1}=\{0\}$ holds, then $Ker
(\mathcal{A}_L)_{e_2,q_2,\beta_2}=\{0\}$.
\end{itemizeXX}

\smallskip
\noindent Now let $(e,q,\beta)$ be a {\bf\em nice} triple. We
consider the following six cases:
\begin{itemizeXX}
\medskip
\item \textbf{Case 1: $e=1, q=2$}

\smallskip
In order to prove the claim first we show
that if $\beta'\in (-1,\frac{-1}{2})$ then
\begin{equation*}
\forall\, X, Y \in \textbf{W}^{1,2}_{\beta'}\quad
\langle\mathcal{A}_L X, Y \rangle_{(M,h)}
=\frac{1}{2}\langle\mathcal{L}X,\mathcal{L}Y\rangle_{L^2}.
\end{equation*}
First let us ensure that both sides are well-defined. Note that
$\mathcal{A}_L: \textbf{W}^{1,2}_{\beta'}\rightarrow
\textbf{W}^{-1,2}_{\beta'-2}$ and so $\mathcal{A}_L X\in
\textbf{W}^{-1,2}_{\beta'-2}$. According to our discussion on
duality pairing in Appendix~\ref{app:operators}, we know that the
duality pairing of $\textbf{W}^{-1,2}_{\beta'-2}$ and
$\textbf{W}^{1,2}_{-1-\beta'}$ is well-defined. So for the LHS to
be well-defined, we just need to ensure that $Y\in
\textbf{W}^{1,2}_{-1-\beta'}$, that is we need to have
$\textbf{W}^{1,2}_{\beta'}\hookrightarrow
\textbf{W}^{1,2}_{-1-\beta'}$. But clearly this is true because
by assumption $\beta'<\frac{-1}{2}$. Also note that
$\mathcal{L}X,\mathcal{L}Y \in L^2_{\beta'-1}$; since
$\beta'-1<\frac{-3}{2}$, by Remark \ref{remunweighted} we have
$L^2_{\beta'-1}\hookrightarrow L^2$ and so the RHS makes sense.
Now, it is well known that the claimed equality holds true for
$X,Y \in C_c^{\infty}$ and so by density it holds true for $X, Y
\in \textbf{W}^{1,2}_{\beta'}$.

Let $X\in Ker (\mathcal{A}_L)_{e=1,q=2,\beta}$. Since
$(e=1,q=2,\beta)$ is a {\bf\em nice} triple, by Lemma \ref{lemb3} there
exists $\beta'\in(-1,\frac{-1}{2})$ such that $X\in
\textbf{W}^{1,2}_{\beta'} $. So by what was proved above we can
conclude that $\langle\mathcal{L}X,\mathcal{L}X\rangle_{L^2}=0$ which implies
that $X$ is a conformal Killing field and so $X=0$.

\medskip
\item \textbf{Case 2: $e\neq 1, q=2$}

\smallskip
If $e > 1$, then $\textbf{W}^{e,q}_{\beta}\hookrightarrow
\textbf{W}^{1,q}_{\beta}$ and hence the claim follows from
Observation 3. Suppose $e<1$. So in particular $e\neq s$ and
$e\neq s-\frac{3}{p}+\frac{3}{2}$ (because both $s$ and
$s-\frac{3}{p}+\frac{3}{2}$ are larger than $1$) and therefore
$(e,q=2,\beta)$ is {\bf\em super nice}. Consequently $(2-e,q'=2,-1-\beta)$
is also {\bf\em super nice}. Since $2-e>1$ we know that $Ker
(\mathcal{A}_L)_{2-e,q'=2,-1-\beta}=\{0\}$. But $\mathcal{A}_L$ is
formally self adjoint and so $Ker
((\mathcal{A}_L)_{e,q=2,\beta})^{*}=Ker
(\mathcal{A}_L)_{2-e,q'=2,-1-\beta}=\{0\}$. Finally
$(\mathcal{A}_L)_{e,q=2,\beta}$ is Fredholm of index zero, so $Ker
(\mathcal{A}_L)_{e,q=2,\beta}=\{0\}$.

\medskip
\item \textbf{Case 3: ($p\leq 2, q<2$) or ($e>\frac{3}{q}-\frac{1}{2}, q<2$)}

\smallskip
It is enough to show that there exists $\tilde{e}$ such that
$\textbf{W}^{e,q}_{\beta}\hookrightarrow
\textbf{W}^{\tilde{e},2}_{\beta}$ where $(\tilde{e},2,\beta)$ is
{\bf\em super nice}. That is, we need to find $\tilde{e}$ that
satisfies
\begin{align*}
&e\geq \tilde{e},\\
&e-\frac{3}{q}\geq\tilde{e}-\frac{3}{2}\quad (\Leftrightarrow \tilde{e}\leq e+\frac{3}{2}-\frac{3}{q}),\\
&\tilde{e}\in (2-s,s),\\
&\tilde{e}\in
(-s+\frac{3}{p}-1+\frac{3}{2},s-\frac{3}{p}+\frac{3}{2}).
\end{align*}
Since $\frac{3}{q}>\frac{3}{2}$, the second condition is stronger
than the first condition. Also $s> \frac{3}{p}$ so
$(-s+\frac{3}{p}+\frac{1}{2},s-\frac{3}{p}+\frac{3}{2})$ is
nonempty. So such an $\tilde{e}$ exists if
\begin{equation*}
(-\infty,e+\frac{3}{2}-\frac{3}{q})\cap (2-s,s)\cap
(-s+\frac{3}{p}-1+\frac{3}{2},s-\frac{3}{p}+\frac{3}{2})\neq
\emptyset\,.
\end{equation*}
Now note that
\begin{align*}
& \textrm{If $p\leq 2$ then}\quad
(-s+\frac{3}{p}-1+\frac{3}{2},s-\frac{3}{p}+\frac{3}{2})\subseteq
(2-s,s),\\
&\textrm{If $p> 2$ then}\quad (2-s,s)\subseteq
(-s+\frac{3}{p}-1+\frac{3}{2},s-\frac{3}{p}+\frac{3}{2}).
\end{align*}
Therefore in order to ensure such an $\tilde{e}$ exists it is
enough to have
\begin{align*}
& 2-s<e+\frac{3}{2}-\frac{3}{q}\quad \textrm{if $p>2$},\\
& -s+\frac{3}{p}+\frac{1}{2}<e+\frac{3}{2}-\frac{3}{q}\quad
\textrm{if $p\leq 2$}.
\end{align*}
The second inequality is true because $(e,q,\beta)$ is a {\bf\em
nice} triple.
%Now if $p\leq 2$, then
%\begin{equation*}
%e>-s+\frac{3}{p}-1+\frac{3}{q}\geq -s+\frac{3}{q}+\frac{1}{2},
%\end{equation*}
%so the first inequality also holds true.
Moreover, for all values of $p$, if $e>\frac{3}{q}-\frac{1}{2}$,
then the first inequality
holds true (note that $s>1$).

\medskip
\item \textbf{Case 4: ($p\geq 2, q>2$) or ($e>\frac{1}{2}, q>2$)}

\smallskip
First we consider the case where $e\neq s$ or $e=s\in
\mathbb{N}_0$. Let $\beta'\in (\beta,0)$. By Theorem \ref{thmA4},
$\textbf{W}^{e,q}_{\beta}\hookrightarrow
\textbf{W}^{e,2}_{\beta'}$. So it is enough to show that under
the assumption of this case, $(e,2,\beta')$ is a {\bf\em nice}
triple. Note that since we have assumed $e\neq s$ or $e=s\in
\mathbb{N}_0$, we do not require $p$ to be equal to $2$. Since
$(e,q,\beta)$ is {\bf\em nice}, we know $e\in (2-s,s]$. Therefore
we just need to check that $e\in
(-s+\frac{3}{p}-1+\frac{3}{2},s-\frac{3}{p}+\frac{3}{2}]$.
\begin{equation*}
q>2 \Rightarrow \frac{3}{q}<\frac{3}{2}\Rightarrow e\leq
s-\frac{3}{p}+\frac{3}{q}<s-\frac{3}{p}+\frac{3}{2}.
\end{equation*}
Also if $p\geq2$, then $-s+\frac{3}{p}+\frac{1}{2}\leq -s+2<e$.
Moreover, for all values of $p$, if $e>\frac{1}{2}$, then
$e>-s+\frac{3}{p}+\frac{1}{2}$.\\
Now let's consider the case where $e=s\not \in \mathbb{N}_0$. By
the statement of the theorem and the assumptions of this case, we
must have $p=q>2$. It is enough to show that there exists
$\tilde{e}$ and $\tilde{\beta}$ such that
$\textbf{W}^{e,q}_{\beta}\hookrightarrow
\textbf{W}^{\tilde{e},2}_{\tilde{\beta}}$ where
$(\tilde{e},2,\tilde{\beta})$ is {\bf\em super nice}. Let
$\tilde{\beta}\in (\beta,0)$. We need to find $\tilde{e}$ that
satisfies
\begin{align*}
&\tilde{e}\leq e=s,\\
%&e-\frac{3}{q}\geq\tilde{e}-\frac{3}{2}\quad (\Leftrightarrow \tilde{e}\leq e+\frac{3}{2}-\frac{3}{q}),\\
&\tilde{e}\in (2-s,s),\\
&\tilde{e}\in
(-s+\frac{3}{p}-1+\frac{3}{2},s-\frac{3}{p}+\frac{3}{2}).
\end{align*}
Such an $\tilde{e}$ exists if
\begin{equation*}
(2-s,s)\cap
(-s+\frac{3}{p}-1+\frac{3}{2},s-\frac{3}{p}+\frac{3}{2})\neq
\emptyset\,.
\end{equation*}
Since $p>2$, the above intersection is equal to $(2-s,s)$ which
is clearly nonempty (since $s>1$).

\medskip
\item \textbf{Case 5: $p<2, q>2$}

\smallskip
First note that $e$ cannot be equal to $s$. Otherwise we would
have $s=e\leq s-\frac{3}{p}+\frac{3}{q}$ and so $p\geq q$ which
contradicts the assumption of this case.

If $e=s-\frac{3}{p}+\frac{3}{q}$, then the claim follows from case
4 (because by assumption $e>\frac{1}{2}$). So WLOG we can assume
that $(e,q>2,\beta)$ is a {\bf\em super nice} triple. This implies that
$(2-e,q'<2,-1-\beta)$ is also a {\bf\em super nice} triple. Since $q'<2$
by what was proved in case 3, $Ker
(\mathcal{A}_L)_{2-e,q',-1-\beta}=\{0\}$. So by an argument
exactly the same as the one given in case 2 we can conclude that
$Ker (\mathcal{A}_L)_{e,q,\beta}=\{0\}$.

\medskip
\item \textbf{Case 6: $p>2, q<2$}

\smallskip
First note that $e$ cannot be
equal to $s-\frac{3}{p}+\frac{3}{q}$. Otherwise we would have
$s-\frac{3}{p}+\frac{3}{q}=e\leq s$ and so $p\leq q$ which
contradicts the assumption of this case.

Since $p\neq q$, if $e=s$, then we must have $e=s\in
\mathbb{N}_0$. If $e=s\in \mathbb{N}_0$, then the claim follows
from case 3 (because by assumption $e>\frac{3}{q}-\frac{1}{2}$).
So WLOG we can assume $(e,q<2,\beta)$ is a {\bf\em super nice}
triple. Therefore $(2-e,q'>2,-1-\beta)$ is also a {\bf\em super
nice} triple. Since $q'>2$ by what was proved in case 4, $Ker
(\mathcal{A}_L)_{2-e,q',-1-\beta}=\{0\}$. So by an argument
exactly the same as the one given in case 2 we can conclude that
$Ker (\mathcal{A}_L)_{e,q,\beta}=\{0\}$.
\end{itemizeXX}
%\item \textbf{Step 3: If $Ker \mathcal{L}=\{0\}$, then $\mathcal{A}_L$ is an isomorphism}
\medskip
\item \textbf{Step 3: Show that if $Ker \mathcal{A}_L=\{0\}$, then $\mathcal{A}_L$ is an isomorphism.}

By the previous steps we know that $\mathcal{A}_L$ is Fredholm of
index zero and also it is injective. It follows that
$\mathcal{A}_L$ is a bijective continuous operator and so
according to the open mapping theorem it is an isomorphism. In
particular $(\mathcal{A}_L)^{-1}$ is continuous and so
$\|W\|_{\textbf{W}^{e,q}_{\beta}}\leq
C\|\textbf{f}(\psi)\|_{\textbf{W}^{e-2,q}_{\beta-2}}$.
\end{itemizeX}
\end{proof}
\begin{corollary}\lab{coro2.8}
Let the following assumptions hold:
\begin{itemize}
\item $(M,h)$ is a 3-dimensional AF Riemannian manifold of
class $W^{s,p}_{\delta}$.
\item $p\in (1,\infty)$, $s\in (1+\frac{3}{p},\infty)$, $\delta<0$.
\item $q\in (3,\infty)$, $e\in(1,s]\cap
(1+\frac{3}{q},s-\frac{3}{p}+\frac{3}{q}]\cap(1,2]$. ($q=p$ if $e=s\not \in \mathbb{N}_0$) %\textbf{Mike does not say $e\leq 2$ because he assumes $s\leq 2$}
\item $-1<\beta <0$, $z=\frac{3q}{3+(2-e)q}$, $b_\tau \in
L^z_{\beta-2}$.
\item $h$ has no conformal Killing fields.
\item $W\in \textbf{W}^{e,q}_{\beta}$
uniquely solves the momentum constraint with source
$\psi\in W^{s,p}_{\delta}$.
\end{itemize}
Then:
\begin{equation*}
\|\mathcal{L}W\|_{L^{\infty}_{\beta-1}}\preceq
\| b_\tau
(\mu+\psi)^6\|_{L^z_{\beta-2}}+\|b_J\|_{\textbf{W}^{e-2,q}_{\beta-2}}\preceq
(\mu+\|\psi\|_{L^{\infty}_{\delta}})^6\|b_\tau\|_{L^z_{\beta-2}}+\|b_J\|_{\textbf{W}^{e-2,q}_{\beta-2}}.
\end{equation*}
Moreover, $\|W\|_{\textbf{W}^{e,q}_{\beta}}$ can be bounded by the same
expressions.
The implicit constants in the above inequalities do
not depend on $\mu$, $W$, or $\psi$.
\end{corollary}
\begin{remark}
%Similar to the argument that we had for the previous theorem, we
%do not need to make the assumption $s\leq 2$ here.
In this theorem, the restrictions on $e$ and $q$ serve the
following purposes:
\begin{itemizeXXM}
\item $L^{z}_{\beta-2}\hookrightarrow W^{e-2,q}_{\beta-2}$. (note
that $e\leq 2$)
\item $\mathcal{A}_L: \textbf{W}^{e,q}_{\beta}\rightarrow
\textbf{W}^{e-2,q}_{\beta-2}$ is well-defined.
\item $e>1+\frac{3}{q}$ and so $W^{e,q}_{\beta}\hookrightarrow L^{\infty}_{\beta}$ and also
$W^{e-1,q}_{\beta-1}\hookrightarrow L^{\infty}_{\beta-1}$.
\end{itemizeXXM}
Also note that
\begin{itemizeXXM}
\item If $e>1+\frac{3}{q}$ then $e>-s+\frac{3}{p}-1+\frac{3}{q}$
is automatically satisfied.
\item For $(1+\frac{3}{q},s-\frac{3}{p}+\frac{3}{q}]$ to be
nonempty we must have $s>1+\frac{3}{p}$.
\item If $s>1$ then $e>2-s$ follows from $e>1$.
\item $2\geq e >1+\frac{3}{q}$ and so we must have $q>3$.
\end{itemizeXXM}
%and Since $s$ is not necessarily less than $2$ we had to
%explicitly state that $e\leq 2$.
\end{remark}
\begin{proof}{\bf (Corollary~\ref{coro2.8})}
First note that $e>1+\frac{3}{q}$ and so
$W^{e,q}_{\beta}\hookrightarrow L^{\infty}_{\beta}$ and also
$W^{e-1,q}_{\beta-1}\hookrightarrow L^{\infty}_{\beta-1}$. That
is, $W^{e,q}_{\beta}\hookrightarrow W^{1,\infty}_{\beta}$. Also
%$\mathcal{L}$ is a differential operator of order $1$, so
$\mathcal {L}: W^{1,\infty}_{\beta}\rightarrow
L^{\infty}_{\beta-1}$ is continuous ($\mathcal{L}$ is a
differential operator of order $1$) and so we have
 %\textbf{ WARNING: TRY TO EXPLAIN THIS MORE RIGOROUSLY}.
 %Therefore we have:
\begin{align*}
\|\mathcal{L}W\|_{L^{\infty}_{\beta-1}}&\preceq
\|W\|_{\textbf{W}^{1,\infty}_{\beta}}\preceq
\|W\|_{\textbf{W}^{e,q}_{\beta}}\preceq
\|\textbf{f}(\psi)\|_{\textbf{W}^{e-2,q}_{\beta-2}}\\
&=\|b_\tau (\mu+\psi)^6+b_J\|_{\textbf{W}^{e-2,q}_{\beta-2}} \leq
\|b_\tau
(\mu+\psi)^6\|_{\textbf{W}^{e-2,q}_{\beta-2}}+\|b_J\|_{\textbf{W}^{e-2,q}_{\beta-2}}\\
&\preceq \|b_\tau
(\mu+\psi)^6\|_{L^z_{\beta-2}}+\|b_J\|_{\textbf{W}^{e-2,q}_{\beta-2}}
\quad (\textrm{note that} \quad L^z_{\beta-2}\hookrightarrow
W^{e-2,q}_{\beta-2})
\end{align*}
Now note that
\begin{align*}
\|b_\tau (\mu+\psi)^6\|_{L^z_{\beta-2}}&=\|b_\tau \sum_{k=0}^6
{6\choose k} \mu^{6-k}\psi^k\|_{L^z_{\beta-2}}\leq \sum_{k=0}^6
{6\choose
k}\mu^{6-k}\|b_\tau \psi^k\|_{L^z_{\beta-2}}\\
&\preceq \sum_{k=0}^6 {6\choose k}\mu^{6-k}\|b_\tau
\psi^k\|_{L^z_{\beta+\delta-2}}\preceq \sum_{k=0}^6 {6\choose
k}\mu^{6-k}\|\psi^k\|_{L^{\infty}_{\delta}}\|b_\tau\|_{L^z_{\beta-2}}
\\
& \quad \quad (\textrm{note that} \quad L^{\infty}_{\delta}\times
L^z_{\beta-2}\hookrightarrow L^z_{\beta+\delta-2})\\
&\preceq \sum_{k=0}^6 {6\choose
k}\mu^{6-k}\|\psi\|^k_{L^{\infty}_{\delta}}\|b_\tau\|_{L^z_{\beta-2}}
\\
& \quad \quad (\textrm{note that} \quad L^{\infty}_{\delta}\times
L^{\infty}_{\delta}\hookrightarrow
L^{\infty}_{2\delta}\hookrightarrow L^{\infty}_{\delta})\\
&=(\mu+\|\psi\|_{L^{\infty}_{\delta}})^6\|b_\tau\|_{L^z_{\beta-2}}.
\end{align*}
Hence
\begin{align*}
\|\mathcal{L}W\|_{L^{\infty}_{\beta-1}}
&\preceq \|W\|_{\textbf{W}^{e,q}_{\beta}}
\\
&\preceq
  \| b_\tau (\mu+\psi)^6\|_{L^z_{\beta-2}}+\|b_J\|_{\textbf{W}^{e-2,q}_{\beta-2}}
\\
&\preceq
  (\mu+\|\psi\|_{L^{\infty}_{\delta}})^6\|b_\tau\|_{L^z_{\beta-2}}+\|b_J\|_{\textbf{W}^{e-2,q}_{\beta-2}}.
\end{align*}
\end{proof}
\begin{lemma}\lab{lemglobal1}
All the assumptions in corollary \ref{coro2.8} hold. In
particular, $W$ is the solution to the momentum constraint with
source $\psi$. Then
\begin{equation*}
a_W\preceq r^{2\beta-2}(k_1\|\mu+\psi\|_{\infty}^{12}+k_2)
\end{equation*}
where
\begin{enumerate}
\item $r=(1+|x|^2)^{\frac{1}{2}}$ and $|x|$ is the geodesic distance from a fixed
point $O$ in the compact core (see Remark \ref{remgeodesic}),
\item $k_1=\| b_\tau\|^2_{L^z_{\beta-2}}$, $\quad k_2=\|
\sigma\|^2_{L^{\infty}_{\beta-1}}+\|b_J\|^2_{\textbf{W}^{e-2,q}_{\beta-2}}$.
\end{enumerate}
The implicit constant in the above inequality does not depend on
$\mu$, $W$ or $\psi$.
\end{lemma}
\begin{proof}{\bf (Lemma~\ref{lemglobal1})}
By Corollary \ref{coro2.8} we have
\begin{align*}
\| \mathcal{L}W\|_{L^{\infty}_{\beta-1}}&\preceq \|b_\tau
(\mu+\psi)^6\|_{L^z_{\beta-2}}+\|b_J\|_{\textbf{W}^{e-2,q}_{\beta-2}}\\
& \leq \|b_\tau \|_{L^z_{\beta-2}}\| \mu+\psi
\|^6_{\infty}+\|b_J\|_{\textbf{W}^{e-2,q}_{\beta-2}},
\quad \textrm{(here we used Remark \ref{remunweighted})}
\end{align*}
and considering Remark \ref{reminfinitybound} we get the
following pointwise bound for $\mathcal{L}W$:
\begin{equation*}
|\mathcal{L}W|\preceq r^{\beta-1}(\|b_\tau
\|_{L^z_{\beta-2}}\| \mu+\psi
\|^6_{\infty}+\|b_J\|_{\textbf{W}^{e-2,q}_{\beta-2}}).
\end{equation*}
Note that $\mathcal{L}W$ has a continuous version and so the above
inequality holds everywhere (not just ``almost everywhere''). Now we
can write
\begin{align*}
a_W=\frac{1}{8}|\sigma+\mathcal{L}W|^2 &\preceq
|\sigma|^2+|\mathcal{L}W|^2 \\
& \preceq r^{2\beta-2}\|
\sigma\|^2_{L^{\infty}_{\beta-1}}+|\mathcal{L}W|^2 \\
& \quad
\textrm{(here we used Remark \ref{reminfinitybound}; note $\sigma \in W^{e-1,q}_{\beta-1}\hookrightarrow C^{0}_{\beta-1}\hookrightarrow L^{\infty}_{\beta-1}$)}\\
&\preceq r^{2\beta-2}\|
\sigma\|^2_{L^{\infty}_{\beta-1}}+
r^{2\beta-2}\big(\|b_\tau \|_{L^z_{\beta-2}}\| \mu+\psi
\|^6_{\infty}+\|b_J\|_{\textbf{W}^{e-2,q}_{\beta-2}}\big)^2\\
&\preceq r^{2\beta-2}(k_1\|\mu+\psi\|_{\infty}^{12}+k_2).
\end{align*}
\end{proof}
\begin{remark}
We make the following important remark concerning notation.
Consider the space
$W^{\alpha,\gamma}_{\delta}(M)$ where $\alpha \gamma
>3$. An order on $W^{\alpha,\gamma}_{\delta}(M)$ can be defined as follows: the functions $\chi_1, \chi_2 \in W^{\alpha,\gamma}_{\delta}(M)$
 satisfy $\chi_2\geq \chi_1$ if and only if  the continuous versions of $\chi_1, \chi_2$ satisfy $\chi_2(x) \geq
 \chi_1(x)$ for all $x\in M$ (clearly this definition agrees with the one that is described in Remark \ref{remorder}). Equipped with this order, $W^{\alpha,\gamma}_{\delta}(M)$
 becomes an ordered Banach space. By the interval $[\chi_1,\chi_2]_{\alpha,\gamma,\delta}$ we mean
the set of all functions $\chi \in W^{\alpha,\gamma}_{\delta}(M)$
such that $\chi_1\leq \chi \leq \chi_2$.
\end{remark}
\begin{lemma}\lab{lem1}
Let the following assumptions hold:
\begin{itemize}
\item All the assumptions in corollary \ref{coro2.8} hold.
\item $\tilde{s}\in (\frac{3}{p},s]$ and $\tilde{\delta}\in [\delta,0)$ are such that $W^{e-2,q}_{\beta-2}\times W^{\tilde{s},p}_{\tilde{\delta}}\hookrightarrow
W^{e-2,q}_{\beta+\tilde{\delta}-2}\hookrightarrow
W^{e-2,q}_{\beta-2}$. For example, using multiplication lemma, one
can easily check that for $\tilde{s}=s$ and
$\tilde{\delta}=\delta$ these inclusions hold true.
\item $\psi_-, \psi_+\in W^{s,p}_{\delta}$, $\psi_+\geq \psi_->-\mu$, \,\,$\psi_1, \psi_2 \in
[\psi_-,\psi_+]_{\tilde{s},p,\tilde{\delta}}$.
\item $W_1$ and $W_2$ are solutions to the momentum constraint
corresponding to $\psi_1$ and $\psi_2$, respectively.
\end{itemize}
Then:
\begin{equation*}
 \|W_1-W_2\|_{e,q,\beta}\preceq
 \big(1+\max\{\|\psi_-\|_{L^{\infty}_{\tilde{\delta}}},\|\psi_+\|_{L^{\infty}_{\tilde{\delta}}}\}\big)^5\|b_\tau\|_{L^z_{\beta-2}}\|\psi_2-\psi_1\|_{\tilde{s},p,\tilde{\delta}}.
\end{equation*}
The implicit constant in the above inequality depends on $\mu$ but
it is independent of $\psi_1, \psi_2, W_1, W_2$.
\end{lemma}
\begin{proof}{\bf (Lemma~\ref{lem1})}
The momentum equation is linear and so $W_1-W_2$ is
the solution to the momentum constraint with right hand side
$\textbf{f}(\psi_1)-\textbf{f}(\psi_2)$.
\begin{align*}
\|W_1-W_2\|_{e,q,\beta}&\preceq
\|\textbf{f}(\psi_1)-\textbf{f}(\psi_2)\|_{e-2,q,\beta-2}=\|b_\tau
[(\mu+\psi_1)^6-(\mu+\psi_2)^6]\|_{e-2,q,\beta-2}\\
&=\|b_\tau
\sum_{j=0}^5(\mu+\psi_2)^j(\mu+\psi_1)^{5-j}(\psi_2-\psi_1)\|_{e-2,q,\beta-2}\\
&\leq \sum_{j=0}^5\|b_\tau
(\mu+\psi_2)^j(\mu+\psi_1)^{5-j}(\psi_2-\psi_1)\|_{e-2,q,\beta-2}\\
&\preceq \sum_{j=0}^5\|b_\tau
(\mu+\psi_2)^j(\mu+\psi_1)^{5-j}\|_{e-2,q,\beta-2}\|(\psi_2-\psi_1)\|_{\tilde{s},p,\tilde{\delta}}
\\
& \quad \quad (\textrm{by assumption } W^{e-2,q}_{\beta-2}\times
W^{\tilde{s},p}_{\tilde{\delta}}\hookrightarrow
W^{e-2,q}_{\beta-2})\\
&\preceq \sum_{j=0}^5\|b_\tau
(\mu+\psi_2)^j(\mu+\psi_1)^{5-j}\|_{L^z_{\beta-2}}\|(\psi_2-\psi_1)\|_{\tilde{s},p,\tilde{\delta}}
\\
& \quad \quad (\textrm{using } L^z_{\beta-2}\hookrightarrow W^{e-2,q}_{\beta-2})\\
&\preceq \sum_{j=0}^5\|b_\tau \sum_{m=0}^j{j \choose m}\psi_2^m
\sum_{l=0}^{5-j}{5-j \choose
l}\psi_1^l\|_{L^z_{\beta-2}}\|(\psi_2-\psi_1)\|_{\tilde{s},p,\tilde{\delta}}\\
&\leq \sum_{j=0}^5 \sum_{m=0}^j \sum_{l=0}^{5-j}{j \choose m}{5-j
\choose l}\|b_\tau \psi_2^m \psi_1^l\|_{L^z_{\beta-2}}\|(\psi_2-\psi_1)\|_{\tilde{s},p,\tilde{\delta}}\\
&\leq \sum_{j=0}^5 \sum_{m=0}^j \sum_{l=0}^{5-j}{j \choose m}{5-j
\choose l}\|b_\tau \psi_2^m \psi_1^l\|_{L^z_{\beta+2\tilde{\delta}-2}}\|(\psi_2-\psi_1)\|_{\tilde{s},p,\tilde{\delta}} \\
& \quad \quad (\textrm{since } \tilde{\delta}<0)\\
&\preceq \sum_{j=0}^5 \sum_{m=0}^j \sum_{l=0}^{5-j}{j \choose
m}{5-j \choose l}\|b_\tau\|_{L^z_{\beta-2}}
\|\psi_2^m\|_{L^{\infty}_{\tilde{\delta}}}
\|\psi_1^l\|_{L^{\infty}_{\tilde{\delta}}}\|(\psi_2-\psi_1)\|_{\tilde{s},p,\tilde{\delta}}
\\
& \quad \quad (\textrm{using } W^{\tilde{s},p}_{\tilde{\delta}}\hookrightarrow
L^{\infty}_{\tilde{\delta}},
\quad
L^{\infty}_{\tilde{\delta}}\times L^z_{\beta-2}\hookrightarrow
L^z_{\beta+\tilde{\delta}-2})\\
&\leq \sum_{j=0}^5 \sum_{m=0}^j \sum_{l=0}^{5-j}{j \choose m}{5-j
\choose l}\|b_\tau\|_{L^z_{\beta-2}}
\|\psi_2\|^m_{L^{\infty}_{\tilde{\delta}}}
\|\psi_1\|^l_{L^{\infty}_{\tilde{\delta}}}\|(\psi_2-\psi_1)\|_{\tilde{s},p,\tilde{\delta}}
\\
& \quad \quad (\textrm{due to } L^{\infty}_{\tilde{\delta}}\times
L^{\infty}_{\tilde{\delta}}\hookrightarrow
L^{\infty}_{2\tilde{\delta}}\hookrightarrow L^{\infty}_{\tilde{\delta}})\\
&=\sum_{j=0}^5 \|b_\tau\|_{L^z_{\beta-2}}
(1+\|\psi_2\|_{L^{\infty}_{\tilde{\delta}}})^j(1+\|\psi_1\|_{L^{\infty}_{\tilde{\delta}}})^{5-j}
\|(\psi_2-\psi_1)\|_{\tilde{s},p,\tilde{\delta}}\\
&\leq \sum_{j=0}^5
\|b_\tau\|_{L^z_{\beta-2}}(1+\max\{\|\psi_-\|_{L^{\infty}_{\tilde{\delta}}},\|\psi_+\|_{L^{\infty}_{\tilde{\delta}}}\})^5\|(\psi_2-\psi_1)\|_{\tilde{s},p,\tilde{\delta}}\\
&=6(1+\max\{\|\psi_-\|_{L^{\infty}_{\tilde{\delta}}},\|\psi_+\|_{L^{\infty}_{\tilde{\delta}}}\})^5\|b_\tau\|_{L^z_{\beta-2}}\|(\psi_2-\psi_1)\|_{\tilde{s},p,\tilde{\delta}},
\end{align*}
where we have used $(|\psi_i|\leq \max\{|\psi_+|,|\psi_-|\}$,
so that $\|\psi_i\|_{L^{\infty}_{\tilde{\delta}}}\leq \max\{\|\psi_-\|_{L^{\infty}_{\tilde{\delta}}},\|\psi_+\|_{L^{\infty}_{\tilde{\delta}}}\})$.
\end{proof}

%%%%%%%%%%%%%%%%%%%%%%%%%%%%%%%%%%%%%%%%%%%%%%%%%%%%%%%%%%%%%%%%%%%%%%%%%%%%%%
\section{Results for the Hamiltonian Constraint}
   \label{sec:ham}

We now develop the main results will need for the Hamiltonian
constraint on AF manifolds with rough data.
We study primarily the ``shifted'' Hamiltonian constraint;
the reason for introducing a shift (the function $a_s$ in the
following lemma) is briefly discussed in Remark
\ref{remreasonshift}.
\begin{lemma}\lab{lem2}
Let the following assumptions hold:
\begin{itemize}
\item $(M,h)$ is a 3-dimensional AF Riemannian manifold of
class $W^{s,p}_{\delta}$.
\item $p\in (\frac{3}{2},\infty)$, $s\in (\frac{3}{p},\infty)\cap
[1,3]$.
\item $\beta<0$, $-1<\delta<0$, and $\eta=\max\{\delta, \beta\}$.
\item $a_\tau, a_\rho, a_W\in W^{s-2,p}_{\beta-2}$, $a_R\in
W^{s-2,p}_{\delta-2}$.
\item $a_0\in W^{s-2,p}_{\eta-2}$, $a_0\neq 0$, and $a_0\geq 0$ (see Remark \ref{remorder}).
\item $\tilde{s}\in (\frac{3}{p},s]\cap [1,1+\frac{3}{p})$, $\delta
\leq\tilde{\delta}<0$.
\item $t\in (\frac{3}{p},\tilde{s}]\cap [1,1+\frac{3}{p})$,
$\tilde{\delta} \leq \gamma<0$.
%\item $V\in W^{\tilde{s},p}_{\tilde{\delta}}$, $V>0$ and $a_s=a_0+a_W
%V \in W^{s-2,p}_{\eta-2}$.
\item $\psi_-, \psi_+ \in W^{\tilde{s},p}_{\tilde{\delta}}$ and $-\mu<\psi_-\leq
\psi_+$.
\item $V\in W^{\tilde{s},p}_{loc}$, $V>0$ is such that $a_W V\in W^{s-2,p}_{\eta-2}$, and \\
$\|a_W V\|_{s-2,p,\eta-2}\preceq
C(\psi_+,\psi_-)\|a_W\|_{s-2,p,\eta-2}$ where $C(\psi_+,\psi_-)$
is a constant independent of $V$.
\item $a_s=a_0+a_W V \in W^{s-2,p}_{\eta-2}$.
\item $A_L^{shifted}: W^{s,p}_{\delta}\rightarrow
W^{s-2,p}_{\eta-2}$ is defined by $A_L^{shifted}\psi=A_L \psi+a_s
\psi$.
\item $\textbf{f}_W^{shifted}:
[\psi_-,\psi_+]_{\tilde{s},p,\tilde{\delta}}\rightarrow
W^{s-2,p}_{\eta-2}$ is defined by
$\textbf{f}_W^{shifted}(\psi)=\textbf{f}_W(\psi)-a_s \psi$ where
\begin{equation*}
\textbf{f}_W(\psi)=a_{\tau}(\mu+\psi)^5+a_R(\mu+\psi)-a_{\rho}(\mu+\psi)^{-3}-a_W(\mu+\psi)^{-7}.
\end{equation*}
\end{itemize}
Then:
\begin{enumerate}
\item Suppose $A_L^{shifted}: W^{s,p}_{\delta}\rightarrow
W^{s-2,p}_{\eta-2}$ is an isomorphism. If we define
$T^{shifted}:[\psi_-,\psi_+]_{\tilde{s},p,\tilde{\delta}}\times
W^{s-2,p}_{\beta-2}\rightarrow W^{s,p}_{\delta}$ by
$T^{shifted}(\psi,
a_W)=-(A_L^{shifted})^{-1}\textbf{f}_W^{shifted}(\psi)$, then
$T^{shifted}$ is continuous in both arguments and moreover
\begin{equation*}
\|T^{shifted}(\psi, a_W)\|_{s,p,\delta}\preceq
(1+\|a_W\|_{s-2,p,\eta-2})(1+\|\psi\|_{t,p,\gamma}).
\end{equation*}
The implicit constant in the above inequality depends on $\mu$ but
it is independent of $\psi$ and $a_W$.
\item If $\beta\leq\delta$, (that is if $\eta=\delta$), then $A_L^{shifted}: W^{s,p}_{\delta}\rightarrow
W^{s-2,p}_{\eta-2}$ is an isomorphism.
\end{enumerate}
\end{lemma}
\begin{proof}{\bf (Lemma~\ref{lem2})}
The proof will involve six main steps.
\begin{itemizeX}
\item \textbf{Step 1:} We first check that the assumptions
actually make sense. To this end, we need to check that both
$A_L^{shifted}$ and $\textbf{f}_W^{shifted}(\psi)$ are well-defined.
%$a_W V\in W^{s-2,p}_{\eta-2}$. To ensure this is true we
%just need to prove that $W^{s-2,p}_{\beta-2}\times
%W^{\tilde{s},p}_{\tilde{\delta}}\hookrightarrow
%W^{s-2,p}_{\eta-2}$. To this end we use the multiplication lemma.
%\begin{align*}
%&s-2\geq s-2,\quad \tilde{s}\geq 1 \geq s-2.\\
%&s-2+\tilde{s}\geq 0 \quad (\textrm{because}\,\, s-2\geq -1,\,\,
%\tilde{s}\geq1).\\
%& (s-2)-(s-2)\geq 3(\frac{1}{p}-\frac{1}{p})=0, \quad
%\tilde{s}-(s-2)\geq 3(\frac{1}{p}-\frac{1}{p})=0.\\
%&(s-2)+\tilde{s}-(s-2)=\tilde{s}>\frac{3}{p}=3(\frac{1}{p}+\frac{1}{p}-\frac{1}{p}).\\
%&\textrm{In the case $s-2<0$:}\quad
%(s-2)+\tilde{s}>\frac{3}{p}-2+\frac{3}{p}>\frac{3}{p}+\frac{3}{p}-3=3(\frac{1}{p}+\frac{1}{p}-1).
%\end{align*}
%So we can conclude that $W^{s-2,p}_{\beta-2}\times
%W^{\tilde{s},p}_{\tilde{\delta}}\hookrightarrow
%W^{s-2,p}_{\delta+\beta-2}\hookrightarrow W^{s-2,p}_{\eta-2}$.

\medskip
We first verify that $A_L^{shifted}$ is well-defined, that is it sends
elements of $W^{s,p}_{\delta}$ to elements in $W^{s-2,p}_{\eta-2}$.
Since we know this is true for $A_L$, we just need to show that if
$\psi\in W^{s,p}_{\delta}\hookrightarrow W^{t,p}_{\gamma}$ then
$a_s\psi \in W^{s-2,p}_{\eta-2}$ (note that $a_s\in
W^{s-2,p}_{\eta-2}$). To this end we use the multiplication lemma
(Lemma \ref{lemA1}) to prove that $W^{s-2,p}_{\eta-2}\times
W^{t,p}_{\gamma}\hookrightarrow W^{s-2,p}_{\eta-2}$.
To use the lemma, we need the following conditions
(the numbering follows the numbering in Lemma~\ref{lemA1}):
\begin{center}
\begin{tabular}{lll}
  (i) & $s-2\geq s-2$, & (trivially true) \\
      & $t\geq s-2$,   & (since $t\geq 1 \geq s-2$) \\
 (ii) & $s-2+t\geq 0$, & (since $s-2\geq -1,\,\, t\geq 1$) \\
      &(\textrm{note that $s-2+t=0$ if and only if $s=t=1\in \mathbb{N}_0$})\\
(iii) & $(s-2)-(s-2)\geq 3(\frac{1}{p}-\frac{1}{p})$,
                    & (trivially true) \\
      & $t-(s-2)\geq 3(\frac{1}{p}-\frac{1}{p})$,
                    & (since $t\geq 1 \geq s-2$) \\
(iv)  & $(s-2)+t-(s-2)>3(\frac{1}{p}+\frac{1}{p}-\frac{1}{p})$,
                    & (since $t>\frac{3}{p}$) \\
 (v)  & Case $s-2<0$:
        $(s-2)+t>3(\frac{1}{p}+\frac{1}{p}-1)$, & \\
\end{tabular}
\end{center}
where the last item holds since
$(s-2)+t>\frac{3}{p}-2+\frac{3}{p}>3(\frac{1}{p}+\frac{1}{p}-1)$.
Therefore, we can conclude that $W^{s-2,p}_{\eta-2}\times
W^{t,p}_{\gamma}\hookrightarrow
W^{s-2,p}_{\gamma+\eta-2}\hookrightarrow W^{s-2,p}_{\eta-2}$.
%\begin{align*}
%&*) s\geq 1 \Rightarrow s-2\in [-s,s],\\
%&*)s-\frac{3}{p}>0\geq 1-\frac{3}{2} \Rightarrow
%s+\frac{3}{2}>1+\frac{3}{p}\Rightarrow 2s+3>
%2+\frac{6}{p}\Rightarrow
%s-2-\frac{3}{p}>-3-s+\frac{3}{p}\\
%&\Rightarrow (s-2)-\frac{3}{p}\in [-3-s+\frac{3}{p},
%s-\frac{3}{p}].
%\end{align*}

\medskip
We now confirm that $\textbf{f}_W^{shifted}(\psi)$ is well-defined.
To this end, we just need to show $\textbf{f}_W$ sends
$W^{\tilde{s},p}_{\tilde{\delta}}$ to $W^{s-2,p}_{\eta-2}$. Note
that in previous sections by using Lemma \ref{lempA1} we showed
that $\textbf{f}_W$ sends $W^{s,p}_{\delta}$ to
$W^{s-2,p}_{\eta-2}$. By the same argument the above claim can be
proved.

\medskip
\item \textbf{Step 2:} As a direct consequence of Lemma
\ref{lempA1} and the multiplication lemma,
$\textbf{f}_W^{shifted}$is a continuous function from
$W^{\tilde{s},p}_{\tilde{\delta}}$ to $W^{s-2,p}_{\eta-2}$ (note
that $a_W, a_{\tau}, a_R, a_{\rho}$ are fixed). The continuity of
$a_W\rightarrow \textbf{f}_W(\psi)$ for a fixed $\psi\in
W^{\tilde{s},p}_{\tilde{\delta}}$ also follows from Lemma
 \ref{lempA1}.

\medskip
\item \textbf{Step 3:} According to Step 2 and the assumption that $A_L^{shifted}$ is an isomorphism, $T^{shifted}$ is a
composition of continuous maps with respect to each of its
arguments. Therefore $T^{shifted}$ is continuous in both
arguments.

\medskip
\item \textbf{Step 4:} Let
$\theta=\frac{1}{p}-\frac{t-1}{3}$; note that by assumption $t<
1+\frac{3}{p}$ and so $\theta >0$. We claim that $\frac{1}{p}\in
(\frac{s-1}{2}\theta,1-\frac{3-s}{2}\theta)$. Indeed,
$\frac{1}{p}< 1-\frac{3-s}{2}\theta$ because
\begin{align*}
& t>\frac{3}{p}\Rightarrow
\theta=\frac{1}{3}+\frac{1}{p}-\frac{t}{3}<\frac{1}{3},\\
& s\geq 1 \Rightarrow \frac{3-s}{2}\leq 1\Rightarrow
1-\frac{3-s}{2}\theta> 1-\theta> \frac{2}{3}.
\end{align*}
Consequently, since $p>\frac{3}{2}$, we have $\frac{1}{p}<
\frac{2}{3}< 1-\frac{3-s}{2}\theta$. It remains to show that
$\frac{1}{p}>\frac{s-1}{2}\theta$. Note that
\begin{equation*}
\frac{s-1}{2}\theta=\frac{s-1}{2}(\frac{1}{p}-\frac{t-1}{3})=\frac{s-1}{2p}-\frac{(s-1)(t-1)}{6},
\end{equation*}
and so
\begin{equation*}
\frac{1}{p}>\frac{s-1}{2}\theta \Leftrightarrow
\frac{(s-1)(t-1)}{6}>\frac{s-1}{2p}-\frac{1}{p}=\frac{s-3}{2p}.
\end{equation*}
The latter inequality is obviously true: if $s=1$ then LHS is zero
but RHS is negative. If $s>1$ then LHS is positive but RHS is
less than or equal to zero (recall that by assumption $s\leq 3$).

\medskip
\item \textbf{Step 5:} Since $s-2\in[-1,1]$ and $\frac{1}{p}\in
(\frac{s-1}{2}\theta,1-\frac{3-s}{2}\theta)$, we may use Lemma
 \ref{lemA8} to estimate
$\|\textbf{f}_W^{shifted}(\psi)\|_{s-2,p,\eta-2}$.
(Lemma~\ref{lemA8} is used for estimating similar quantities in
later arguments as well, so we give the justification for use of
Lemma~\ref{lemA8} as Remark~\ref{rem1} following this proof.)

For all $\psi \in [\psi_-,\psi_+]_{\tilde{s},p,\tilde{\delta}}$
we have (note that
$W^{\tilde{s},p}_{\tilde{\delta}}\hookrightarrow
W^{t,p}_{\gamma}$ so $\psi \in W^{t,p}_{\gamma}$)
\begin{align*}
\|\textbf{f}_W^{shifted}&(\psi)\|_{s-2,p,\eta-2}
\\
&=\|a_{\tau}(\mu+\psi)^5+a_R(\mu+\psi)-a_{\rho}(\mu+\psi)^{-3}-a_W(\mu+\psi)^{-7}-a_s \psi\|_{s-2,p,\eta-2}\\
&\leq
\|a_{\tau}(\mu+\psi)^5\|_{s-2,p,\eta-2}+\|a_R(\mu+\psi)\|_{s-2,p,\eta-2}+\|a_{\rho}(\mu+\psi)^{-3}\|_{s-2,p,\eta-2}\\
& \quad +\|a_W(\mu+\psi)^{-7}\|_{s-2,p,\eta-2}+\|a_s \psi\|_{s-2,p,\eta-2}\\
&\preceq
\|a_\tau\|_{s-2,p,\eta-2}(\|(\mu+\psi)^5\|_{L^{\infty}}+\|5(\mu+\psi)^4\|_{L^{\infty}}\|\psi\|_{t,p,\gamma})\\
&\quad +\|a_R\|_{s-2,p,\eta-2}(\|\mu+\psi\|_{L^{\infty}}+\|1\|_{L^{\infty}}\|\psi\|_{t,p,\gamma})\\
&\quad +\|a_0 \psi\|_{s-2,p,\eta-2}+\|a_W V \psi\|_{s-2,p,\eta-2}\\
&\quad
+\|a_\rho\|_{s-2,p,\eta-2}(\|(\mu+\psi)^{-3}\|_{L^{\infty}}+\|-3(\mu+\psi)^{-4}\|_{L^{\infty}}\|\psi\|_{t,p,\gamma})\\
&\quad +
\|a_W\|_{s-2,p,\eta-2}(\|(\mu+\psi)^{-7}\|_{L^{\infty}}+\|-7(\mu+\psi)^{-8}\|_{L^{\infty}}\|\psi\|_{t,p,\gamma})\\
&\preceq \|a_\tau\|_{s-2,p,\eta-2}(\|(\mu+\psi)^4\|_{L^{\infty}}\|\mu+\psi\|_{L^{\infty}}+\|5(\mu+\psi)^4\|_{L^{\infty}}\|\psi\|_{t,p,\gamma})\\
&\quad +\|a_R\|_{s-2,p,\eta-2}(\|\mu+\psi\|_{L^{\infty}}+\|\psi\|_{t,p,\gamma})\\
&\quad +\|a_0\|_{s-2,p,\eta-2}\|\psi\|_{t,p,\gamma}+\|a_W
V\|_{s-2,p,\eta-2}\|\psi\|_{t,p,\gamma}
\\
&\quad \quad (\textrm{note
that}\,\,a_W V\in W^{s-2,p}_{\eta-2},\,\,\,
W^{s-2,p}_{\eta-2}\times W^{t,p}_{\gamma}\hookrightarrow
W^{s-2,p}_{\eta-2})\\
&\quad +\|a_\rho\|_{s-2,p,\eta-2}(\|(\mu+\psi)^{-4}\|_{L^{\infty}}\|\mu+\psi\|_{L^{\infty}}+\|-3(\mu+\psi)^{-4}\|_{L^{\infty}}\|\psi\|_{t,p,\gamma})\\
&\quad
+\|a_W\|_{s-2,p,\eta-2}(\|(\mu+\psi)^{-8}\|_{L^{\infty}}\|\mu+\psi\|_{L^{\infty}}+\|-7(\mu+\psi)^{-8}\|_{L^{\infty}}\|\psi\|_{t,p,\gamma}).\\
\end{align*}
Now note that $W^{t,p}_{\gamma}\hookrightarrow
L^{\infty}_{\gamma}\hookrightarrow L^{\infty}$, so
\begin{equation*}
\|\mu+\psi\|_{L^{\infty}}\leq \mu+\|\psi\|_{L^{\infty}}\preceq
\mu+\|\psi\|_{t,p,\gamma}\preceq 1+\|\psi\|_{t,p,\gamma}.
\end{equation*}
Hence
\begin{align*}
\|\textbf{f}_W^{shifted}(\psi)\|_{s-2,p,\eta-2}\preceq
&\|a_\tau\|_{s-2,p,\eta-2}\|(\mu+\psi)^{4}\|_{L^{\infty}}(1+\|\psi\|_{t,p,\gamma})\\
&+\|a_R\|_{s-2,p,\eta-2}(1+\|\psi\|_{t,p,\gamma})\\
&+\|a_0\|_{s-2,p,\eta-2}(1+\|\psi\|_{t,p,\gamma})\\
&+\|a_W\|_{s-2,p,\eta-2}C(\psi_+,\psi_-)(1+\|\psi\|_{t,p,\gamma})\\
&+\|a_\rho\|_{s-2,p,\eta-2}\|(\mu+\psi)^{-4}\|_{L^{\infty}}(1+\|\psi\|_{t,p,\gamma})\\
&+\|a_W\|_{s-2,p,\eta-2}\|(\mu+\psi)^{-8}\|_{L^{\infty}}(1+\|\psi\|_{t,p,\gamma}).
\end{align*}
Consequently
\begin{align*}
\|\textbf{f}_W^{shifted}&(\psi)\|_{s-2,p,\eta-2}
\\
&\preceq \big[\|a_\tau\|_{s-2,p,\eta-2}\|(\mu+\psi_+)^{4}\|_{L^{\infty}}+\|a_R\|_{s-2,p,\eta-2}
\\
& \quad +\|a_\rho\|_{s-2,p,\eta-2}\|(\mu+\psi_-)^{-4}\|_{L^{\infty}} +\|a_0\|_{s-2,p,\eta-2}
\\
& \quad + \|a_W\|_{s-2,p,\eta-2}(\|(\mu+\psi_-)^{-8}\|_{L^{\infty}} +C(\psi_+,\psi_-))\big](1+\|\psi\|_{t,p,\gamma})\\
& \preceq [1+\|a_W\|_{s-2,p,\eta-2}](1+\|\psi\|_{t,p,\gamma}).
\end{align*}
Finally note that by assumption $(A_L^{shifted})^{-1}:
W^{s-2,p}_{\eta-2}\rightarrow W^{s,p}_{\delta}$ is continuous and
therefore
\begin{align*}
\|T^{shifted}(\psi, a_W)\|_{s,p,\delta}
&=\|-(A_L^{shifted})^{-1}\textbf{f}_W^{shifted}(\psi)\|_{s,p,\delta}
\\
&\preceq \|\textbf{f}_W^{shifted}(\psi)\|_{s-2,p,\eta-2}
\\
&\preceq [1+\|a_W\|_{s-2,p,\eta-2}](1+\|\psi\|_{t,p,\gamma}).
\end{align*}
\item \textbf{Step 6:} In this step we prove the second claim. By
the last item in Lemma \ref{lemb4},
$A_L:W^{s,p}_{\delta}\rightarrow W^{s-2,p}_{\delta-2}$ is Fredholm
of index zero. By Lemma \ref{lemb5},
$A_L^{shifted}:W^{s,p}_{\delta}\rightarrow W^{s-2,p}_{\delta-2} $
is a compact perturbation of $A_L$. Since $A_L$ is Fredholm of
index zero we can conclude that $A_L^{shifted}$ is also Fredholm
of index zero. Now maximum principle (Lemma \ref{lemb8}) implies
that the kernel of $A_L^{shifted}:W^{s,p}_{\delta}\rightarrow
W^{s-2,p}_{\delta-2}$ is trivial. An injective operator of index
zero is surjective as well. Consequently
$A_L^{shifted}:W^{s,p}_{\delta}\rightarrow W^{s-2,p}_{\delta-2}$
is a continuous bijective operator. Therefore by the open mapping
theorem, $A_L^{shifted}:W^{s,p}_{\delta}\rightarrow
W^{s-2,p}_{\delta-2}$ is an isomorphism of Banach spaces. In
particular the inverse is continuous and so
$\|u\|_{s,p,\delta}\preceq \|A_L^{shifted} u\|_{s-2,p,\delta-2}$.
\end{itemizeX}
\end{proof}
\begin{remark}\lab{rem1}
In the above proof we used Lemma \ref{lemA8} to estimate
$\|\textbf{f}_W^{shifted}(\psi)\|_{s-2,p,\eta-2}$. Note that
since $\psi \in W^{\tilde{s},p}_{\tilde{\delta}}\hookrightarrow
C^{0}_{\tilde{\delta}}$, and $\tilde{\delta}<0$ we can conclude
that $\psi\rightarrow 0$ as $|x|\rightarrow \infty$ (in the
asymptotic ends). Therefore there exists a compact set $B$ such
that outside of $B$, $|\psi|<\frac{\mu}{2}$. On the compact set
$B$, the continuous function $\psi$ attains its minimum which by
assumption must be larger than $-\mu$. Consequently $\inf \psi
> \min\{-\frac{\mu}{2}, \min_{x\in B}\psi(x)\}>-\mu $. Because of
this functions of the form $f(x)=(\mu+x)^{-m}$ where $m\in
\mathbb{N}$ are smooth on $[\inf \psi, \sup \psi]$ as it is
required by Lemma \ref{lemA8}.
\end{remark}

\begin{lemma}\lab{lem3}
In addition to the conditions of Lemma \ref{lem2} (including
$\beta \leq \delta$), assume $a_R\geq 0$ (see Remark
\ref{remorder}) and define the shift function $a_s$ by
\begin{equation*}
a_s=a_R+3\frac{(\mu+\psi_+)^2}{(\mu+\psi_-)^{6}}a_\rho+5(\mu+\psi_+)^4a_\tau+7\frac{(\mu+\psi_+)^6}{(\mu+\psi_-)^{14}}a_W.
\end{equation*}
Then for any fixed $a_W\in W^{s-2,p}_{\beta-2}$, the map
$T^{shifted}:
[\psi_-,\psi_+]_{\tilde{s},p,\tilde{\delta}}\rightarrow
W^{s,p}_{\delta}$ is monotone increasing.
\end{lemma}
\begin{proof}{\bf (Lemma~\ref{lem3})}
First note that the above definition of $a_s$
satisfies the assumptions that we had for $a_s$ in Lemma
 \ref{lem2}. Note that
\begin{align*}
&a_0=a_R+3\frac{(\mu+\psi_+)^2}{(\mu+\psi_-)^{6}}a_\rho+5(\mu+\psi_+)^4a_\tau,\\
& V=7\frac{(\mu+\psi_+)^6}{(\mu+\psi_-)^{14}}.
\end{align*}
We first must check $a_0\in W^{s-2,p}_{\eta-2}$ and $\|a_W
V\|_{s-2,p,\eta-2}\preceq C(\psi_+,\psi_-)\|a_W\|_{s-2,p,\eta-2}$.

We first check that $a_0\in W^{s-2,p}_{\eta-2}$.
By assumption $a_R \in W^{s-2,p}_{\delta-2}=W^{s-2,p}_{\eta-2}$. The fact that
$\frac{(\mu+\psi_+)^2}{(\mu+\psi_-)^{6}}a_\rho$ and
$(\mu+\psi_+)^4a_\tau$ are in $W^{s-2,p}_{\eta-2}$ follows
directly from Lemma \ref{lempA1}.
%(similar to the way that we
%proved $f_W$ and $f^{shifted}_W$ are well-defined; note that
%$g(x)=(\mu+x)^{-6}$ is smooth over $(-\mu,\infty)$ and by
%assumption $\psi_->-\mu$).
Therefore $a_0\in W^{s-2,p}_{\eta-2}$.

We now check that
$\|a_W V\|_{s-2,p,\eta-2}\preceq C(\psi_+,\psi_-)\|a_W\|_{s-2,p,\eta-2}$.
By Lemma \ref{lemA8} we have
\begin{align*}
\|a_W (\mu+\psi_+)^6\|_{s-2,p,\eta-2}
&\preceq
\|a_W\|_{s-2,p,\eta-2}(\|(\mu+\psi_+)^6\|_{L^{\infty}}+\|6(\mu+\psi_+)^5\|_{L^{\infty}}\|\psi_+\|_{\tilde{s},p,\tilde{\delta}})\\
&=C_1(\psi_+)\|a_W\|_{s-2,p,\eta-2},
\end{align*}
and so (recall Remark \ref{rem1})
\begin{align*}
\|a_W V\|_{s-2,p,\eta-2}&\preceq \|a_W
(\mu+\psi_+)^6(\mu+\psi_-)^{-14}\|_{s-2,p,\eta-2}\\
&\preceq \|a_W
(\mu+\psi_+)^6\|_{s-2,p,\eta-2}(\|(\mu+\psi_-)^{-14}\|_{L^{\infty}}
\\
& \quad +\|-14(\mu+\psi_-)^{-15}\|_{L^{\infty}}\|\psi_-\|_{\tilde{s},p,\tilde{\delta}})\\
&=C_2(\psi_-)\|a_W (\mu+\psi_+)^6\|_{s-2,p,\eta-2}
\\
&\preceq C_1(\psi_+)C_2(\psi_-)\|a_W\|_{s-2,p,\eta-2}
\\
&=C(\psi_+,\psi_-)\|a_W\|_{s-2,p,\eta-2}.
\end{align*}

Now that we have confirmed the two conditions we can proceed.
For all $\psi_1,
\psi_2\in[\psi_-,\psi_+]_{\tilde{s},p,\tilde{\delta}}$ with
$\psi_1 \leq \psi_2$ we have
\begin{align*}
\textbf{f}_W^{shifted}(\psi_2)-\textbf{f}_W^{shifted}(\psi_1)&=\textbf{f}_W(\psi_2)-\textbf{f}_W(\psi_1)-a_s(\psi_2-\psi_1)\\
&=a_\tau [(\mu+\psi_2)^5-(\mu+\psi_1)^5]+a_R
[\psi_2-\psi_1]
\\
& \quad -a_\rho[(\mu+\psi_2)^{-3}-(\mu+\psi_1)^{-3}]\\
&\quad
-a_W[(\mu+\psi_2)^{-7}-(\mu+\psi_1)^{-7}]-a_s(\psi_2-\psi_1).
\end{align*}
Note that for all $m\in \mathbb{N}$
\begin{align*}
(\mu+\psi_2)^m-(\mu+\psi_1)^m
&=\big(\sum_{j=0}^{m-1}(\mu+\psi_2)^j(\mu+\psi_1)^{m-1-j}\big)(\psi_2-\psi_1)
\\
&\leq
m(\mu+\psi_+)^{m-1}(\psi_2-\psi_1)
-[(\mu+\psi_2)^{-m}-(\mu+\psi_1)^{-m}]
\\
&=\frac{(\mu+\psi_2)^m-(\mu+\psi_1)^m}{[(\mu+\psi_2)
(\mu+\psi_1)]^m}
\\
&\leq m\frac{(\mu+\psi_+)^{m-1}}{(\mu+\psi_-)^{2m}}(\psi_2-\psi_1).
\end{align*}
Therefore
\begin{align*}
\textbf{f}_W^{shifted}(\psi_2)-\textbf{f}_W^{shifted}(\psi_1)
&\leq [5(\mu+\psi_+)^4a_\tau+a_R+3\frac{(\mu+\psi_+)^2}{(\mu+\psi_-)^{6}}a_\rho
\\
& \quad \quad +7\frac{(\mu+\psi_+)^6}{(\mu+\psi_-)^{14}}a_W-a_s](\psi_2-\psi_1)
\\
&=0.
\end{align*}
So $\textbf{f}_W^{shifted}$ is decreasing over
$[\psi_-,\psi_+]_{\tilde{s},p,\tilde{\delta}}$. Also
$A_L^{shifted}:W^{s,p}_{\delta}\rightarrow W^{s-2,p}_{\eta-2}$
satisfies the maximum principle, hence the inverse
$(A_L^{shifted})^{-1}$ is monotone increasing \cite{HNT07b}.
Consequently $T^{shifted}:
[\psi_-,\psi_+]_{\tilde{s},p,\tilde{\delta}}\rightarrow
W^{s,p}_{\delta}$ defined by
$-(A_L^{shifted})^{-1}\textbf{f}_W^{shifted}$ is monotone
increasing.
\end{proof}
\begin{lemma}\lab{lem4}
Let the conditions of Lemma \ref{lem3} hold, with $\psi_-$ and
$\psi_+$ sub- and supersolutions of the Hamiltonian constraint
(equation (\ref{eqweak1})), respectively (with $a_W$ as source).
Then, we have $T^{shifted}(\psi_+,a_W)\leq \psi_+$ and
$T^{shifted}(\psi_-,a_W)\geq \psi_-$. In particular, since
$T^{shifted}$ is monotone increasing in its first variable,
$T^{shifted}$ is invariant on
$U=[\psi_-,\psi_+]_{\tilde{s},p,\tilde{\delta}}$, that is, if
$\psi \in [\psi_-,\psi_+]_{\tilde{s},p,\tilde{\delta}}$, then
$T^{shifted}(\psi, a_W) \in
[\psi_-,\psi_+]_{\tilde{s},p,\tilde{\delta}}$.
\end{lemma}
\begin{proof}{\bf (Lemma~\ref{lem4})}
Since $\psi_+$ is a supersolution, by definition
(which can be found in the next section),
$A_L\psi_++\textbf{f}_W(\psi_+)\geq 0$ with respect to the order
of $W^{s-2,p}_{\delta-2}$ (see Remark \ref{remorder}). We have
\begin{align*}
\psi_+-T^{shifted}(\psi_+,a_W)
&=(A_L^{shifted})^{-1}[A_L^{shifted}\psi_++\textbf{f}_W^{shifted}(\psi_+)]
\\
&=(A_L^{shifted})^{-1}[A_L\psi_++\textbf{f}_W(\psi_+)],
\end{align*}
which is nonnegative since $\psi_+$ is supersolution and
$(A_L^{shifted})^{-1}$ is linear and monotone increasing. The
proof of the other inequality is completely analogous.
\end{proof}
\begin{remark}\lab{remreasonshift}
As seen in the proof of the above lemmas, the introduction of
the shift function $a_s$ into $\textbf{f}_W^{shifted}$
ensures it is a decreasing
function on $[\psi_-,\psi_+]_{\tilde{s},p,\tilde{\delta}}$,
which subsequently implies that $T^{shifted}$ is invariant on
$U=[\psi_-,\psi_+]_{\tilde{s},p,\tilde{\delta}}$.
This property of $T^{shifted}$ plays an important role in the
fixed point framework we use for our existence theorem for the
coupled system, following closely the approach taken in~\cite{HNT07b}.
\end{remark}

%%%%%%%%%%%%%%%%%%%%%%%%%%%%%%%%%%%%%%%%%%%%%%%%%%%%%%%%%%%%%%%%%%%%%%%%%%%%%%
\section{Global Sub- and Supersolution Constructions}
   \label{sec:subsuper}

In this section, based on a combination of ideas employed in
\cite{HNT07b,dM06,24}, we introduce a new method for constructing global
sub- and supersolutions for the Hamiltonian constraint on AF
manifolds. We begin with giving the precise definitions of local
and global sub- and
supersolutions.

 Consider the Hamiltonian constraint (equation \ref{eqweak1}):
\begin{equation*}
A_L \psi+f(\psi,W)=0.
\end{equation*}
\vspace*{-0.70cm}
\begin{itemizeX}
\item
A \textbf{local subsolution} of (\ref{eqweak1}) is a function $\psi_-\in
W^{s,p}_{\delta}$, $\psi_->-\mu$ such that
\begin{equation*}
A_L \psi_-+f(\psi_-,W)\leq 0
\end{equation*}
for at least one $W\in \textbf{W}^{e,q}_{\beta}$. Note that the
inequality is with respect to the order of $W^{s-2,p}_{\delta-2}$
(see Remark \ref{remorder}).

\item
A \textbf{local supersolution} of (\ref{eqweak1}) is a
function $\psi_+\in W^{s,p}_{\delta}$, $\psi_+>-\mu$ such that
\begin{equation*}
A_L \psi_++f(\psi_+,W)\geq 0
\end{equation*}
for at least one $W\in \textbf{W}^{e,q}_{\beta}$.

\item
A \textbf{global subsolution} of (\ref{eqweak1}) is a function $\psi_-\in
W^{s,p}_{\delta}$, $\psi_->-\mu$ such that
\begin{equation*}
A_L \psi_-+f(\psi_-,W_{\psi})\leq 0
\end{equation*}
for all vector fields $W_{\psi}$ solution of (\ref{eqweak2})
(momentum constraint) with source $\psi \in W^{s,p}_{\delta}$ and
$\psi \geq \psi_-$.

\item
A \textbf{global supersolution} of (\ref{eqweak1}) is a
function $\psi_+\in W^{s,p}_{\delta}$, $\psi_+>-\mu$ such that
\begin{equation*}
A_L \psi_++f(\psi_+,W_{\psi})\geq 0
\end{equation*}
for all vector fields $W_{\psi}$ solution of (\ref{eqweak2})
(momentum constraint) with source $\psi \in W^{s,p}_{\delta}$ and
$-\mu<\psi \leq \psi_+$.

\item
We say a pair of a subsolution and a supersolution, $\psi_-$ and $\psi_+$, is \textbf{compatible}
if $-\mu< \psi_-\leq \psi_+<\infty$ (so
$[\psi_-,\psi_+]_{s,p,\delta}$ is nonempty).

\end{itemizeX}

For our main existence theorem we need to have compatible global
subsolution and supersolution. The goal of this section is to
prove the existence of such compatible global barriers. In what
follows we use the following notation: Given any scalar function
$v\in L^{\infty}$, let $v^{\tiwedge}=\esssup_M v$, and
$v^{\tivee}=\essinf_M v$.
\begin{proposition}\lab{propglobal1}
Assume all the conditions of \textbf{Weak Formulation~\ref{weakf2}}
and Corollary \ref{coro2.8} hold true. Additionally assume that $h$
belongs to the positive Yamabe class, $-1<\beta\leq\delta<0$, and
$\|\sigma\|_{L^{\infty}_{\beta-1}}$,
$\|\rho\|_{L^{\infty}_{2\beta-2}}$,
$\|J\|_{\textbf{W}^{e-2,q}_{\beta-2}}$ are sufficiently small.
Moreover, suppose that there exists a positive continuous function
$\Lambda \in W^{s-2,p}_{\delta-2}$ and a number $\delta'\in
(2\beta,\delta)$ such that $\Lambda \sim r^{\delta'-2}$  (that is,
$r^{\delta'-2}\preceq \Lambda \preceq r^{\delta'-2}$) for
sufficiently large $r=(1+|x|^2)^{\frac{1}{2}}$ (see Remark
\ref{remlambdaexistence}). If $\mu>0$ is chosen to be sufficiently
small, then there exists a global supersolution $\psi_+\in
W^{s,p}_{\delta}$ to the Hamiltonian constraint.
\end{proposition}
\begin{proof}{\bf (Proposition~\ref{propglobal1})}
Since $h$ belongs to the positive Yamabe class,
there exists a function $\xi\in W^{s,p}_{\delta}$, $\xi>-1$ such
that if we set $\tilde{h}=(1+\xi)^4 h$, then $R_{\tilde{h}}=0$.
Let $H(\psi, a_W, a_\tau, a_\rho)$ and $\tilde{H}(\psi, a_W,
a_\tau, a_\rho)$ be as in Appendix~\ref{app:covariance}.
In what follows we will show that there exists $\tilde{\psi}_+>0$ such that
\begin{equation}\lab{eqnglobal1}
\forall \,  \varphi \in (-\mu,
(\xi+1)\tilde{\psi}_++\mu\,\xi]_{s,p,\delta}\quad
\tilde{H}(\tilde{\psi}_+, a_{W_{\varphi}}, a_\tau, a_\rho)\geq 0.
\end{equation}
Here $W_{\varphi}$ is the solution of the momentum constraint
with source $\varphi$. Let's assume we find such a function. Then
if we define $\psi_+=(\xi+1)\tilde{\psi}_++\mu\,\xi$, we have
$\psi_+\in W^{s,p}_{\delta}$,  $\psi_+>-\mu$ and it follows from
Corollary \ref{corocovariance} that
\begin{equation*}
\forall \, \varphi \in (-\mu, \psi_+]_{s,p,\delta}\quad H(\psi_+,
a_{W_{\varphi}}, a_\tau, a_\rho)\geq 0
\end{equation*}
which precisely means that $\psi_+$ is a global supersolution of
the Hamiltonian constraint. So it is enough to prove the
existence of $\tilde{\psi}_+$.

Let $\Lambda\in W^{s-2,p}_{\delta-2}$ be a positive continuous
function such that $\Lambda\sim r^{\delta'-2}$  for sufficiently
large $|x|$; here $\delta'$ is a fixed but arbitrary number in
the interval $(2\beta,\delta)$. By Lemma \ref{lemb4} there exists
a unique function $u\in W^{s,p}_{\delta}$ such that
$-\Delta_{\tilde{h}}u=\Lambda$. By the maximum principle (Lemma
\ref{lemb8}) u is positive ($u>0$). Recall that $\mu$ is a fixed
nonzero number but we have freedom in choosing $\mu$. We claim
that if $\mu>0$ is sufficiently small, then $\tilde{\psi}_+:= \mu
\, u$ satisfies (\ref{eqnglobal1}). Indeed, for all $\varphi \in
(-\mu, (\xi+1)\tilde{\psi}_++\mu\,\xi]_{s,p,\delta}$ we have
\begin{align*}
\tilde{H}(\tilde{\psi}_+, a_{W_{\varphi}}, a_\tau,
a_\rho)&=-\Delta_{\tilde{h}}\tilde{\psi}_++a_\tau(\tilde{\psi}_++\mu)^5
-(1+\xi)^{-12}a_{W_{\varphi}}(\tilde{\psi}_++\mu)^{-7}
\\
& \quad \quad -(1+\xi)^{-8}a_{\rho}(\tilde{\psi}_++\mu)^{-3}
\quad (R_{\tilde{h}}=0)\\
&=\mu\,\Lambda
+a_\tau(\tilde{\psi}_++\mu)^5-(1+\xi)^{-12}a_{W_{\varphi}}(\tilde{\psi}_++\mu)^{-7}
\\
& \quad \quad -(1+\xi)^{-8}a_{\rho}(\tilde{\psi}_++\mu)^{-3}
\\
&\geq \mu\,\Lambda
-(1+\xi)^{-12}a_{W_{\varphi}}(\tilde{\psi}_++\mu)^{-7}-(1+\xi)^{-8}a_{\rho}(\tilde{\psi}_++\mu)^{-3}.
\end{align*}
The argument in Remark \ref{rem1} shows that $(\inf \xi)>-1$ and
so $\inf(1+\xi)>0$. Therefore if we let
$\tilde{C}=\max\{((1+\xi)^{\tivee})^{-12},
((1+\xi)^{\tivee})^{-8}\}$, then
\begin{align*}
\tilde{H}(\tilde{\psi}_+, a_{W_{\varphi}}, a_\tau, a_\rho)& \geq
\mu\,\Lambda
-\tilde{C}a_{W_{\varphi}}(\tilde{\psi}_++\mu)^{-7}-\tilde{C}a_{\rho}(\tilde{\psi}_++\mu)^{-3}\\
&=\mu\,\Lambda
-\tilde{C}\mu^{-7}a_{W_{\varphi}}(u+1)^{-7}-\tilde{C}\mu^{-3}a_{\rho}(u+1)^{-3}\\
&\geq \mu\,\Lambda
-C\mu^{-7}r^{2\beta-2}(k_1\|\mu+\varphi\|_{\infty}^{12}+k_2)(u+1)^{-7}
\\
& \quad \quad -\tilde{C}\mu^{-3}a_{\rho}(u+1)^{-3},
\end{align*}
where we have used Lemma \ref{lemglobal1}.
Recall that $C$ (the implicit constant in Lemma \ref{lemglobal1})
does not depend on $\mu$. Now note that $\forall \,  \varphi \in
(-\mu, (\xi+1)\tilde{\psi}_++\mu\,\xi]_{s,p,\delta}$
%if we let $\psi_+=(\xi+1)\tilde{\psi}_++\mu\,\xi$, then for all $\varphi\in (-\mu,\psi_+]$,
we have $0\leq \mu + \varphi\leq (\xi+1)(\mu+\tilde{\psi}_+)$ and
so
\begin{equation*}
\|\mu+\varphi\|_{\infty}^{12}\leq[(\xi+1)^{\tiwedge}]^{12}
[(\mu+\tilde{\psi}_+)^{\tiwedge}]^{12}.
\end{equation*}
Let
$k_3=(\xi+1)^{\tiwedge}\frac{(1+u)^{\tiwedge}}{(1+u)^{\tivee}}$.
We can write
\begin{align*}
\|\mu+\varphi\|_{\infty}^{12} &
\leq [(\xi+1)^{\tiwedge}]^{12}[(\mu+\tilde{\psi}_+)^{\tiwedge}]^{12}
=[(\xi+1)^{\tiwedge}]^{12} \mu^{12}[(1+u)^{\tiwedge}]^{12}
\\
& = k_3^{12} \mu^{12}[(1+u)^{\tivee}]^{12}\leq k_3^{12} \mu^{12}(u+1)^{12}.
\end{align*}
Consequently
\begin{align*}
\tilde{H}(\tilde{\psi}_+, a_{W_{\varphi}}, a_\tau, a_\rho)
&\geq \mu\,\Lambda -C\mu^{-7}r^{2\beta-2}(k_1 k_3^{12}\mu^{12}(u+1)^{12}+k_2)(u+1)^{-7}
\\
& \quad \quad -\tilde{C}\mu^{-3}a_{\rho}(u+1)^{-3}\\
&=\mu\,\Lambda -C\mu^{5}r^{2\beta-2}k_1k_3^{12}(u+1)^5
-C\mu^{-7}r^{2\beta-2}k_2(u+1)^{-7}
\\
& \quad \quad -\tilde{C}\mu^{-3}a_{\rho}(u+1)^{-3}\\
&\geq \mu\,\Lambda
-C\mu^{5}r^{2\beta-2}k_1k_3^{12}((u+1)^{\tiwedge})^5
-C\mu^{-7}r^{2\beta-2}k_2((u+1)^{\tivee})^{-7}
\\
& \quad \quad -\tilde{C}\mu^{-3}a_{\rho}((u+1)^{\tivee})^{-3}.
\end{align*}
Note that $\Lambda \sim r^{\delta'-2}$ for sufficiently large
$r$ and $2\beta-2<\delta'-2<0$.
We claim that this allows one to choose $\mu$ small enough so that
\begin{equation}
\frac{\Lambda}{2}>C\mu^{4}r^{2\beta-2}k_1k_3^{12}((u+1)^{\tiwedge})^5.
\label{eqn:mu-scaling}
\end{equation}
The justification of this claim is as follows.
 There exists a constant $C_1$ and a
number $r_1>0$ such that for $r> r_1$, we have $\Lambda \geq C_1
r^{\delta'-2}$. Also since $2\beta-2<\delta'-2<0$, there exists
$r_2>0$ such that for all $r>r_2$
\begin{equation*}
\frac{C_1}{2} r^{\delta'-2}>r^{2\beta-2}[C
k_1k_3^{12}((u+1)^{\tiwedge})^5].
\end{equation*}
Consequently for all $0<\mu\leq 1$ and $r>\max\{r_1,r_2\}$
\begin{equation*}
\frac{\Lambda}{2}>r^{2\beta-2}[C
k_1k_3^{12}((u+1)^{\tiwedge})^5]\geq
C\mu^{4}r^{2\beta-2}k_1k_3^{12}((u+1)^{\tiwedge})^5.
\end{equation*}
Also the positive continuous function $\Lambda$ attains its
minimum $\Lambda^{\tivee}>0$ on the compact set $r \leq
\max\{r_1,r_2\}$. We choose $\mu\leq 1$ small enough such that
\begin{equation}\lab{lambdacompact}
\frac{\Lambda^{\tivee}}{2}>C\mu^{4}k_1k_3^{12}((u+1)^{\tiwedge})^5.
\end{equation}
Since $r^{2\beta-2}\leq 1$ the above inequality implies that
~\eqref{eqn:mu-scaling} holds even if $r \leq \max\{r_1,r_2\}$.
%\begin{equation*}
%\frac{\Lambda}{2}>C\mu^{4}r^{2\beta-2}k_1k_3^{12}((u+1)^{\tiwedge})^5.
%\end{equation*}
(Note that on the entire $M$, $\Lambda^{\tivee}=0$, so we could
not use (\ref{lambdacompact}) on whole $M$ to determine $\mu$;
this is exactly why first we needed to study what happens for
large $r$.)
 For such $\mu$ by requiring that
$\|\sigma\|_{L^{\infty}_{\beta-1}}$,
$\|\rho\|_{L^{\infty}_{2\beta-2}}$,
$\|J\|_{\textbf{W}^{e-2,q}_{\beta-2}}$ are sufficiently small
(note that according to Remark \ref{reminfinitybound}, $a_\rho\leq
r^{2\beta-2}\|a_\rho\|_{L^{\infty}_{2\beta-2}} \textrm{a.e.}$) we
can ensure that
\begin{equation*}
\frac{\Lambda}{2}\geq
C\mu^{-8}r^{2\beta-2}k_2((u+1)^{\tivee})^{-7}+\tilde{C}\mu^{-4}a_{\rho}((u+1)^{\tivee})^{-3}.
\end{equation*}
\end{proof}
\begin{remark}\lab{remlambdaexistence}
Pick an arbitrary number $\delta'\in (2\beta,\delta)$.  If $s\leq
2$, then $\Lambda=r^{\delta'-2}$ satisfies the desired conditions:
clearly $\Lambda$ is positive, continuous, and $\Lambda \sim
r^{\delta'-2}$. Also obviously $\Lambda \in
L^{\infty}_{\delta'-2}$ and
\begin{equation*}
L^{\infty}_{\delta'-2}\hookrightarrow
L^{p}_{\delta-2}\hookrightarrow W^{s-2,p}_{\delta-2}\quad
(\Longrightarrow \Lambda \in W^{s-2,p}_{\delta-2})
\end{equation*}
The first inclusion is true because $\delta'$is strictly less than
$\delta$; the second inclusion is true because $s-2\leq 0$.
\end{remark}
\begin{proposition}\lab{propglobal2}
Assume all the conditions of \textbf{Weak Formulation~\ref{weakf2}}.
Additionally assume that $h$ belongs to the positive Yamabe class
and $-1<\beta\leq\delta<0$. If $\mu>0$ is chosen to be
sufficiently small, then there exists a global subsolution
$\psi_-\in W^{s,p}_{\delta}$ to the Hamiltonian constraint which
is compatible with the global supersolution that was constructed
in Proposition \ref{propglobal1} (provided the extra assumptions
of that proposition hold true).
\end{proposition}
\begin{proof}{\bf (Proposition~\ref{propglobal2})}
Since $h$ belongs to the positive Yamabe class,
there exists a function $\xi\in W^{s,p}_{\delta}$, $\xi>-1$ such
that if we set $\tilde{h}=(1+\xi)^4 h$, then $R_{\tilde{h}}=0$.
Let $H(\psi, a_W, a_\tau, a_\rho)$ and $\tilde{H}(\psi, a_W,
a_\tau, a_\rho)$ be as in Appendix~\ref{app:covariance}.
In what follows we will show that there exists $-\mu<\tilde{\psi}_-<0$
such that
\begin{equation}\lab{eqnglobal2}
\forall \,  \varphi \in W^{s,p}_{\delta}, \quad
\tilde{H}(\tilde{\psi}_-, a_{W_{\varphi}}, a_\tau, a_\rho)\leq 0.
\end{equation}
Here $W_{\varphi}$ is the solution of the momentum constraint
with source $\varphi$. Note that since $\tilde{\psi}_+>0$
($\tilde{\psi}_+$ is the function that was introduced in the
proof of the previous proposition ), clearly $\tilde{\psi}_-\leq
0 < \tilde{\psi}_+ $. Let's assume we find such a function. Then
if we define $\psi_-=(\xi+1)\tilde{\psi}_-+\mu\,\xi$, we have
$\psi_-\in W^{s,p}_{\delta}$, and
\begin{align*}
& \tilde{\psi}_->-\mu \Longrightarrow
(\xi+1)(\tilde{\psi}_-+\mu)>0 \Longrightarrow
(\xi+1)\tilde{\psi}_-+\mu\,\xi> -\mu \Longrightarrow \psi_->-\mu\\
& \tilde{\psi}_-\leq \tilde{\psi}_+ \Longrightarrow \psi_- \leq
\psi_+
\end{align*}
Moreover, it follows from Corollary \ref{corocovariance} that
\begin{equation*}
\forall \, \varphi \in W^{s,p}_{\delta},
\quad H(\psi_-, a_{W_{\varphi}}, a_\tau, a_\rho)\leq 0,
\end{equation*}
which clearly implies that $\psi_-$ is a global subsolution of the
Hamiltonian constraint. So it is enough to prove the
existence of $\tilde{\psi}_-$.

We may consider two cases: \\
\textbf{Case 1: $a_\tau \equiv 0$}\\
In this case $\tilde{\psi}_- \equiv 0$ satisfies the desired
conditions; Indeed,
\begin{equation*}
\tilde{H}(\tilde{\psi}_-\equiv 0, a_{W_{\varphi}}, a_\tau,
a_\rho)=-(1+\xi)^{-12}a_{W_{\varphi}}\mu^{-7}-(1+\xi)^{-8}a_\rho\mu^{-3}\leq
0.
\end{equation*}
\textbf{Case 2: $a_\tau \not \equiv 0$}\\
By Lemma \ref{lemb4} there exists a unique function $u\in
W^{s,p}_{\delta}$ such that $-\Delta_{\tilde{h}}u=-a_\tau$. By the
maximum principle (Lemma \ref{lemb8}) $u\leq 0$ and clearly $u
\not \equiv 0$ (because $a_\tau \not \equiv 0$). Note that
$W^{s,p}_{\delta}\hookrightarrow
L^{\infty}_{\delta}\hookrightarrow L^{\infty}$ (the latter
embedding is true because $\delta<0)$. Let $m=\|u\|_{\infty}+1$;
so in particular $-m<\inf u<0$. Recall that we have freedom in
choosing the fixed number $\mu$ as small as we want. We claim
that if $\mu>0$ is sufficiently small, then $\tilde{\psi}_-:=
\frac{1}{m}\mu \, u$ satisfies (\ref{eqnglobal2}). Clearly
$\tilde{\psi}_-\leq 0$; also
\begin{align*}
u>-m \Longrightarrow \mu(u+m)>0 \Longrightarrow \mu
(\frac{u+m}{m})>0 \Longrightarrow \frac{1}{m}\mu u>-\mu
\Longrightarrow \tilde{\psi}_->-\mu.
\end{align*}
Moreover, for all $\varphi \in W^{s,p}_{\delta}$ we have
\begin{align*}
\tilde{H}(\tilde{\psi}_-, a_{W_{\varphi}}, a_\tau,
a_\rho)&=-\Delta_{\tilde{h}}\tilde{\psi}_-+a_\tau(\tilde{\psi}_-+\mu)^5-(1+\xi)^{-12}a_{W_{\varphi}}(\tilde{\psi}_-+\mu)^{-7}
\\
& \quad \quad -(1+\xi)^{-8}a_{\rho}(\tilde{\psi}_-+\mu)^{-3}\\
&\leq -\Delta_{\tilde{h}} (\frac{1}{m}\mu
u)+a_\tau(\frac{1}{m}\mu u+\mu)^5\\
&=-\frac{1}{m}\mu a_\tau +\mu^5a_\tau(\frac{1}{m}u+1)^5\\
&= \mu a_\tau \big[-\frac{1}{m}+\mu^4(\frac{1}{m}u+1)^5\big].
\end{align*}
Now note that $-m<u<m$ and so $0<1+\frac{1}{m}u<2$, therefore
\begin{equation*}
\tilde{H}(\tilde{\psi}_-, a_{W_{\varphi}}, a_\tau, a_\rho)\leq
\mu a_\tau \big[-\frac{1}{m}+32 \mu^4\big].
\end{equation*}
Thus if we choose $\mu$ so that $\mu^4<\frac{1}{32m}$, then
$\tilde{H}(\tilde{\psi}_-, a_{W_{\varphi}}, a_\tau, a_\rho)\leq
0$.
\end{proof}
\begin{remark}\lab{aebarriers}
The compatible global barrier constructions in ~\cite{24} and
~\cite{HNT07b} both make critical use of the fact that the
conformal factor $\phi$, which is the primary unknown in their
formulations, is positive. When the subsolution and supersolution
are both positive, then one can scale the subsolution to make it
smaller than the supersolution. In the formulation presented in
this paper, which is designed to allow very low regularity
assumptions on the data on AF manifolds, the primary unknown is a
shifted version of the conformal factor ($\psi$). $\psi$ can be
negative and so in particular the scaling argument cannot be
directly applied here. Due to the nonlinear nature of the
Hamiltonian constraint, this situation cannot be resolved simply
by finding compatible barriers for the original positive unknown
$\phi$ and then shifting those to obtain compatible barriers for
$\psi$.
\end{remark}
%%%%%%%%%%%%%%%%%%%%%%%%%%%%%%%%%%%%%%%%%%%%%%%%%%%%%%%%%%%%%%%%%%%%%%%%%%%%%%
\section{The Main Existence Result for Rough non-CMC Solutions}
   \label{sec:main}

We now establish existence of coupled non-CMC weak solutions for AF manifolds
by combining the results for the individual Hamiltonian and
momentum constraints developed in Sections~\ref{sec:ham} and~\ref{sec:mom},
the barrier constructions developed in Section~\ref{sec:subsuper},
together with the following topological fixed-point theorem
for the coupled system from~\cite{HNT07b}:
%The proof of
%existence of solution is based on the following corollary of the
%\textbf{Schauder theorem}.
\begin{theorem}[Coupled Schauder Theorem] \cite{HNT07b}
\label{thm:schauder}
Let $X$ and $Y$ be Banach spaces, and let $Z$ be an ordered Banach
space with compact embedding $X \hookrightarrow Z$. Let
$[\psi_-,\psi_+] \subset Z$ be a non-empty interval, and set
$U=[\psi_-,\psi_+]\cap \bar{B}_M \subset Z$ where $\bar{B}_M$ is a
closed ball of finite radius $M>0$ in $Z$. Assume $U$ is nonempty
and let $S:U \to \mathcal{R}(S) \subset Y$ and $T:U \times
\mathcal{R}(S) \to U \cap X$ be continuous maps. Then, there
exist $w\in\mathcal{R}(S)$ and $\psi \in U \cap X$ such that
\begin{equation*}
\psi=T(\psi,w)\quad\textrm{and}\quad w=S(\psi).
\end{equation*}
\end{theorem}
\begin{remark}
In \cite{HNT07b} the above theorem is stated with the extra assumption
that $\bar{B}_M$ is a ball of radius $M$ \textbf{about the
origin} but the same proof works even if $\bar{B}_M$ is not
centered at the origin.
\end{remark}

With all of the supporting results we need now in place,
we state and prove our main result.
\begin{theorem}
\label{thm:main}
Let $(M,h)$ be a $3$-dimensional AF Riemannian manifold of class
$W^{s,p}_{\delta}$ where $p\in (1,\infty)$, $s\in
(1+\dfrac{3}{p},\infty)$ and $-1< \delta<0$ are given. Suppose
$h$ admits no nontrivial conformal Killing field (see Remark
\ref{remconformalkilling}) and is in the positive Yamabe class.
Let $\beta \in (-1,\delta]$. Select $q$ and $e$ to satisfy:
\begin{align*}
&\frac{1}{q}\in
(0,1)\cap(0,\frac{s-1}{3})\cap[\frac{3-p}{3p},\frac{3+p}{3p}],\\
&e\in(1+\frac{3}{q},\infty)\cap[s-1,s]\cap[\frac{3}{q}+s-\frac{3}{p}-1,\frac{3}{q}+s-\frac{3}{p}].
\end{align*}
Let $q=p$ if $e=s\not \in \mathbb{N}_0$. Moreover if $s>2,\, s\not
\in \mathbb{N}_0$, assume $e<s$.

Assume that the data
satisfies:
\begin{itemize}
\item $\tau \in W^{e-1,q}_{\beta-1}$ if $e\geq 2$ and $\tau\in W^{1,z}_{\beta-1}$ otherwise, where $z=\dfrac{3q}{3+(2-e)q}$ \\
(note that if $e=2$, then $W^{e-1,q}_{\beta-1}=W^{1,z}_{\beta-1}$),
\item $\sigma \in W^{e-1,q}_{\beta-1}$,
\item $\rho \in W^{s-2,p}_{\beta-2}\cap
L^{\infty}_{2\beta-2}$, $\rho\geq 0$ ($\rho$ can be identically
zero),
\item $J\in \textbf{W}^{e-2,q}_{\beta-2}$.
\end{itemize}
Recall that we have freedom in choosing the positive constant
$\mu$ in equations (\ref{eqweak1}) and (\ref{eqweak2}). If $\mu$
is chosen to be sufficiently small and if
$\|\sigma\|_{L^{\infty}_{\beta-1}}$,
$\|\rho\|_{L^{\infty}_{2\beta-2}}$, and
$\|J\|_{\textbf{W}^{e-2,q}_{\beta-2}}$ are sufficiently small,
 then there exists $\psi\in W^{s,p}_{\delta}$ with $\psi>-\mu$ and
$W\in \textbf{W}^{e,q}_{\beta}$ solving (\ref{eqweak1}) and
(\ref{eqweak2}).
\end{theorem}
%%%%%%%%%%%%%%%%%%%%%%%%%%%%%%%%%%%%%%%
%% \begin{comment}
%% \begin{remark}
%% Because of the stated restrictions on $s$, $p$, $q$, $e$, and
%% $z$, it is possible to consider $\mathcal{A}_L$ as a map from
%% $W^{e,q}_{\beta}$ to $W^{e-2,q}_{\beta-2}$ and moreover we have
%% the following inclusions:
%% \begin{equation*}
%% W^{e,q}_{\beta}\hookrightarrow W^{1,2r}_{\beta},\quad
%% W^{1,z}_{\beta-1}\hookrightarrow
%% W^{e-1,q}_{\beta-1}\hookrightarrow L^{2r}_{\beta-1}.
%% \end{equation*}
%% Therefore by making the above assumptions, we are working in the
%% framework of the previously stated weak formulation.
%% \end{remark}
%% \end{comment}
%%%%%%%%%%%%%%%%%%%%%%%%%%%%%%%%%%%%%%%
\begin{remark}
As discussed in Appendix~\ref{app:bessel}, the assumptions ``$p=q$ if
$e=s\not \in \mathbb{N}_0$'' and ``$e<s$ if $s>2,\, s\not \in
\mathbb{N}_0$'' can be removed if we replace weighted
Sobolev-Slobodeckij spaces with weighted Bessel potential spaces.
\end{remark}
\begin{proof}{\bf (Theorem~\ref{thm:main})}
First we prove the claim for the case $s\leq2$ and then we extend the proof for $s>2$ by bootstrapping.

\medskip
\noindent
\textbf{\bf Case 1:} $s\leq2$

\smallskip
 Note that by assumption $e\leq s$, so $e$ is also less than or equal to $2$. Also since $2\geq s>1+\frac{3}{p}$, $p$ is larger than $3$.

  Since $s\leq 2$, it follows from Proposition
 \ref{propglobal1}, Remark \ref{remlambdaexistence} and Proposition \ref{propglobal2} that if $\mu$ is chosen to be sufficiently small, then
 for
$\|\sigma\|_{L^{\infty}_{\beta-1}}$,
$\|\rho\|_{L^{\infty}_{2\beta-2}}$, and
$\|J\|_{\textbf{W}^{e-2,q}_{\beta-2}}$ sufficiently small, there
exists a compatible pair of global subsolution and supersolution.
We fix such $\mu$ and assume that
$\|\sigma\|_{L^{\infty}_{\beta-1}}$,
$\|\rho\|_{L^{\infty}_{2\beta-2}}$, and
$\|J\|_{\textbf{W}^{e-2,q}_{\beta-2}}$ are sufficiently small
(according to Proposition \ref{propglobal1}).

\noindent
\emph{Step 1: The choice of function spaces.}
\begin{itemizeXM}
\item $X=W^{s,p}_{\delta}$, with $s$ and $p$ as given in the theorem
statement.
\item $Y=\textbf{W}^{e,q}_{\beta}$, with $e$, $q$ as given in the theorem
statement.
\item $Z=W^{\tilde{s},p}_{\tilde{\delta}}$, $\tilde{s}\in (1,1+\dfrac{3}{p})$ and $\tilde{\delta}>\delta$, so
that $X=W^{s,p}_{\delta}\hookrightarrow
W^{\tilde{s},p}_{\tilde{\delta}}=Z$ is compact. Note that
$\tilde{s}\in (1,1+\dfrac{3}{p})$ implies that $\tilde{s}\in
(\frac{3}{p},s)$ (because $p>3$ and $s>1+\frac{3}{p}$).
\item $U=[\psi_-,\psi_+]_{W^{\tilde{s},p}_{\tilde{\delta}}}\cap \bar{B}_M \subset W^{\tilde{s},p}_{\tilde{\delta}}=Z$,
with $\psi_-$ and $\psi_+$ compatible global barriers constructed
in the previous section and with sufficiently large $M$ to be
determined below.
%Note that the existence of compatible global
%barriers follows immediately from the assumptions of the theorem.
\end{itemizeXM}

\noindent
\emph{Step 2: Construction of the mapping $S$.} Using Lemma \ref{lempA1}, it can be easily checked that for any
$\psi \in [\psi_-,\psi_+]_{\tilde{s},p,\tilde{\delta}}$,
$\textbf{f}(\psi)=b_\tau(\psi+\mu)^6+b_J\in
\textbf{W}^{e-2,q}_{\beta-2}$. Therefore, since the metric admits
no nontrivial conformal Killing field, by Theorem \ref{thm1}, the
momentum constraint is uniquely solvable for any ``source'' $\psi\in
[\psi_-,\psi_+]_{\tilde{s},p,\tilde{\delta}} $ (it is easy to see
that the assumptions of Theorem \ref{thm1} are satisfied; see
Remark \ref{rem2.2}) . The ranges for the exponents ensure that
the momentum constraint solution map
\begin{equation*}
S:[\psi_-,\psi_+]_{\tilde{s},p,\tilde{\delta}} \rightarrow
\textbf{W}^{e,q}_{\beta}=Y,\quad
S(\psi)=-\mathcal{A}_L^{-1}\textbf{f}(\psi)
\end{equation*}
is continuous. Indeed, by Lemma \ref{lempA1}, $\psi\rightarrow
\textbf{f}(\psi)$ is a continuous map from
$W^{\tilde{s},p}_{\tilde{\delta}}$ to
$\textbf{W}^{e-2,q}_{\beta-2}$ and by Theorem \ref{thm1},
$\mathcal{A}_L^{-1}:\textbf{W}^{e-2,q}_{\beta-2}\rightarrow
\textbf{W}^{e,q}_{\beta}$ is continuous.

\noindent
\emph{Step 3: Construction of the mapping $T$.}
%Again assume the existence of global barriers $\psi_-$ and $\psi_+$.
%To deal with the non-trivial kernel that exists for $\Delta$ on AE
%manifolds, we introduce a shift $a_s\in W^{s-2,p}_{\delta-2}$ and
%we define the shifted operators:
Our construction of the mapping $T$ makes use of Lemmas
\ref{lem3}, and \ref{lem4} where one of the assumptions is that
$a_R\geq 0$. To satisfy this assumption, first we need to make a
conformal transformation. To this end, we proceed as follows: By
assumption $h$ belongs to the positive Yamabe class. In
particular, there exists $\xi\in W^{s,p}_{\delta}$, $\xi>-1$ such
that $R_{\tilde{h}}=0$ where $\tilde{h}=(1+\xi)^4 h$. Let
$\tilde{\psi}_+$ and $\tilde{\psi}_-$ be the functions that were
constructed in the proofs of Proposition \ref{propglobal1} and
Proposition \ref{propglobal2}. Also let
\begin{equation*}
\tilde{a}_\tau:=a_\tau,\,\,
\tilde{a}_{\rho}:=(1+\xi)^{-8}a_\rho,\,\,
\tilde{a}_W:=(1+\xi)^{-12}a_W,\,\,
\tilde{a}_R:=a_{R_{\tilde{h}}}=0.
\end{equation*}
Notice that the above notations agree with the ones that are
introduced in Appendix~\ref{app:covariance}.
Using Lemma \ref{lempA1} it is easy to
see that $\tilde{a}_\rho$, $\tilde{a}_W$ remain in
$W^{s-2,p}_{\beta-2}$. So we may use $(\tilde{h}, \tilde{a}_\tau,
\tilde{a}_\rho, \tilde{a}_W, \tilde{a}_R=0, \tilde{\psi}_+,
\tilde{\psi}_-)$ as data in Lemmas \ref{lem2}, \ref{lem3}, and
\ref{lem4}. That is, if we define
\begin{align*}
&\tilde{a}_s:=\tilde{a}_R+3\frac{(\mu+\tilde{\psi}_+)^2}{(\mu+\tilde{\psi}_-)^{6}}\tilde{a}_\rho+5(\mu+\tilde{\psi}_+)^4\tilde{a}_\tau+7\frac{(\mu+\tilde{\psi}_+)^6}{(\mu+\tilde{\psi}_-)^{14}}\tilde{a}_W,\\
& \tilde{A}_L^{shifted}:W^{s,p}_{\delta}\rightarrow
W^{s-2,p}_{\delta-2},\quad
\tilde{A}_L^{shifted}\psi=-\Delta_{\tilde{h}}\psi+\tilde{a}_s\psi,\\
&\tilde{\textbf{f}}_W^{shifted}(\psi)=\tilde{a}_{\tau}(\mu+\psi)^5+\tilde{a}_R(\mu+\psi)-\tilde{a}_{\rho}(\mu+\psi)^{-3}-\tilde{a}_W(\mu+\psi)^{-7}-\tilde{a}_s
\psi ,\\
& \tilde{T}^{shifted}:
[\tilde{\psi}_-,\tilde{\psi}_+]_{\tilde{s},p,\tilde{\delta}}\times
W^{s-2,p}_{\beta-2}\rightarrow W^{s,p}_{\delta},\quad
\tilde{T}^{shifted}(\psi,\tilde{a}_W)=-(\tilde{A}_L^{shifted})^{-1}\tilde{\textbf{f}}_W^{shifted}(\psi),
\end{align*}
then, according to the aforementioned lemmas,
$\tilde{T}^{shifted}$ is continuous with respect to both of its
arguments and it is invariant on
$[\tilde{\psi}_-,\tilde{\psi}_+]_{\tilde{s},p,\tilde{\delta}}$.
Notice that if we define the \textbf{scaled Hamiltonian
constraint} as in Appendix~\ref{app:covariance}, that is, if we let
\begin{equation*}
\tilde{H}(\psi, a_W, a_\tau, a_\rho)=-\Delta_{\tilde{h}}
\psi+\tilde{a}_R(\psi +\mu)+\tilde{a}_\tau(\psi+\mu)^5-\tilde{a}_W
(\psi+\mu)^{-7}-\tilde{a}_\rho(\psi+\mu)^{-3}
\end{equation*}
then $\tilde{\psi}_-$ and $\tilde{\psi}_+$ are subsolution and
supersolution of $\tilde{H}=0$ and moreover
\begin{equation*}
\tilde{H}(\psi, a_W, a_\tau, a_\rho)=0 \Longleftrightarrow
\tilde{T}^{shifted}(\psi,\tilde{a}_W)=\psi.
\end{equation*}
Now we define the mapping
$T:[\psi_-,\psi_+]_{\tilde{s},p,\tilde{\delta}}\times
\textbf{W}^{e,q}_{\beta}\rightarrow W^{s,p}_{\delta}$ as follows:
\begin{equation*}
T(\psi,W)=(\xi+1)\tilde{T}^{shifted}(\frac{\psi-\mu\,\xi}{\xi+1},(\xi+1)^{-12}a_W)+\mu\,\xi.
\end{equation*}
Here $\psi_+$ and $\psi_-$ are the supersolution and subsolution
that were constructed in the proofs of Proposition
\ref{propglobal1} and Proposition \ref{propglobal2}. Recall that
by our construction
\begin{equation*}
\tilde{\psi}_-=\frac{\psi_--\mu\,\xi}{\xi+1},\quad
\tilde{\psi}_+=\frac{\psi_+-\mu\,\xi}{\xi+1}
\end{equation*}
so for $\psi \in [\psi_-,\psi_+]_{\tilde{s},p,\tilde{\delta}}$,
we have $\tilde{\psi}_-\leq\frac{\psi-\mu\,\xi}{\xi+1}\leq
\tilde{\psi}_+$. In fact using Lemma \ref{lempA1} one can easily
show that $T$ is well-defined. That is,
$\frac{\psi-\mu\,\xi}{\xi+1}$ is in $[\tilde{\psi}_-,
\tilde{\psi}_+]_{\tilde{s},p,\tilde{\delta}}$ and
$(\xi+1)\tilde{T}^{shifted}(.,.)+\mu\,\xi$ is in
$W^{s,p}_{\delta}$. Continuity of $T$ follows from the continuity
of $\tilde{T}^{shifted}$ and Lemma \ref{lempA1}.

Considering the coupled Schauder theorem, in order to complete
the proof for the case $s\leq 2$, it is enough to prove the
following claim:

\emph{
\textbf{Claim}:
There exists $M>0$ such that if we set
$U=[\psi_-,\psi_+]_{\tilde{s},p,\tilde{\delta}}\cap
\bar{B}_M(\mu\,\xi)$, then $U$ is nonempty and
%and $T$ is invariant
%on $U$, that is,
\begin{equation*}
%\psi \in [\psi_-,\psi_+]_{\tilde{s},p,\tilde{\delta}}\cap
%\bar{B}_M (\mu\,\xi)
(\psi,W) \in U \times S(U)\Longrightarrow T(\psi,a_W)\in U,
%[\psi_-,\psi_+]_{\tilde{s},p,\tilde{\delta}}\cap\bar{B}_M
%(\mu\,\xi)
\end{equation*}
where
$\bar{B}_M (\mu\,\xi)$ is the ball of radius $M$ in
$W^{\tilde{s},p}_{\tilde{\delta}}$ centered at $\mu\,\xi \in
W^{s,p}_{\delta}\hookrightarrow W^{\tilde{s},p}_{\tilde{\delta}}$.
}

\medskip
\emph{\textbf{Proof of Claim.}}
First, as mentioned above, note that $T(\psi, a_W)$ certainly
belongs to $X=W^{s,p}_{\delta}$, so instead of $T(\psi,a_W)\in U$
on the right hand side we could write $T(\psi,a_W)\in U\cap X$.
We now prove that
if $\psi\in[\psi_-,\psi_+]_{\tilde{s},p,\tilde{\delta}}$, then for
all $a_W\in W^{s-2,p}_{\beta-2}$ (and so for all $W\in
\textbf{W}^{e,q}_{\beta}$), $T(\psi,W)\in
[\psi_-,\psi_+]_{\tilde{s},p,\tilde{\delta}}$:
\begin{equation*}
\psi\in
[\psi_-,\psi_+]_{\tilde{s},p,\tilde{\delta}}\Longrightarrow
\frac{\psi-\mu\,\xi}{\xi+1}\in
[\tilde{\psi}_-,\tilde{\psi}_+]_{\tilde{s},p,\tilde{\delta}}.
\end{equation*}
But we know that $\tilde{T}^{shifted}$ is invariant on
$[\tilde{\psi}_-,\tilde{\psi}_+]_{\tilde{s},p,\tilde{\delta}}$
and so
\begin{equation*}
\forall\, a_W\in W^{s-2,p}_{\beta-2} \quad
\tilde{T}^{shifted}(\frac{\psi-\mu\,\xi}{\xi+1},(1+\xi)^{-12}a_W)\in
[\tilde{\psi}_-,\tilde{\psi}_+]_{\tilde{s},p,\tilde{\delta}}.
\end{equation*}
Therefore for all $W\in \textbf{W}^{e,q}_{\beta}$
\begin{align*}
(1+\xi)\tilde{\psi}_-+\mu\,\xi &\leq (1+\xi)
\tilde{T}^{shifted}(\frac{\psi-\mu\,\xi}{\xi+1},(1+\xi)^{-12}a_W)+\mu\,\xi
\leq (1+\xi)\tilde{\psi}_++\mu\,\xi\\
\psi_-& \leq T(\psi,W) \leq \psi_+
\end{align*}
Thus $T(\psi,W)\in [\psi_-,\psi_+]_{\tilde{s},p,\tilde{\delta}}$
(note that as it was already mentioned $T(\psi, W) \in
W^{s,p}_{\delta}\hookrightarrow
W^{\tilde{s},p}_{\tilde{\delta}}$).

Now to complete the proof of the claim above, it is enough to
show that the following auxiliary claim holds true:

\emph{
\textbf{Auxiliary Claim}:
There exists $\hat{M}>0$ such that for
all $M \geq \hat{M}$ the following holds:
\begin{equation*}
\textrm{If } \psi \in
[\psi_-,\psi_+]_{\tilde{s},p,\tilde{\delta}}\cap \bar{B}_M
(\mu\,\xi)
\end{equation*}
\begin{equation}\lab{globalball}
\textrm{Then } \forall \, W\in
S([\psi_-,\psi_+]_{\tilde{s},p,\tilde{\delta}}),
\quad
T(\psi,a_W)\in \bar{B}_M (\mu\,\xi).
\end{equation}
}
\begin{remark}
We make two remarks before we continue.
\begin{enumerate}
\item In order to prove the main claim, it is enough to
prove the auxiliary claim for $W\in
S([\psi_-,\psi_+]_{\tilde{s},p,\tilde{\delta}}\cap  \bar{B}_M
(\mu\,\xi))$ not $W\in
S([\psi_-,\psi_+]_{\tilde{s},p,\tilde{\delta}})$.
So what we will prove here is slightly stronger than what we need.
\item Since we will prove (\ref{globalball}) is true for
\textbf{all} $M \geq \hat{M}$, we can certainly choose an $M$ such
that $[\psi_-,\psi_+]_{\tilde{s},p,\tilde{\delta}}\cap  \bar{B}_M
(\mu\,\xi))\neq \emptyset$.
\end{enumerate}
\end{remark}

\medskip
\emph{\textbf{Proof of Auxiliary Claim.}}
We will rely on two supporting results
(Lemma~\ref{globalbounded1} and~\ref{globalbounded2}),
which will be stated and proved following the completion
of our proof of the main result here.

To begin, let $t\in (\frac{3}{p},\tilde{s})\cap [1,1+\frac{3}{p})$
and let $\gamma\in(\tilde{\delta},0)$; also for all $\psi \in
[\psi_-,\psi_+]_{\tilde{s},p,\tilde{\delta}}$ let
$\tilde{\psi}:=\frac{\psi-\mu\,\xi}{\xi+1}$. It follows from Lemma
\ref{lem2} that there exists $K>0$ such that for all $\psi \in
[\psi_-,\psi_+]_{\tilde{s},p,\tilde{\delta}}$ and for all $W \in
S([\psi_-,\psi_+]_{\tilde{s},p,\tilde{\delta}})$
\begin{equation*}
\|\tilde{T}^{shifted}(\tilde{\psi},
\tilde{a}_W)\|_{\tilde{s},p,\tilde{\delta}}\leq
\tilde{\tilde{C}}\|\tilde{T}^{shifted}(\tilde{\psi},
\tilde{a}_W)\|_{s,p,\delta}\leq K
[1+\|\tilde{a}_W\|_{s-2,p,\delta-2}](1+\|\tilde{\psi}\|_{t,p,\gamma}).
\end{equation*}
 Now note that
$W^{\tilde{s},p}_{\tilde{\delta}}\hookrightarrow
W^{t,p}_{\gamma}$ is compact and $W^{t,p}_{\gamma}\hookrightarrow
L^p_{\gamma}$ is continuous. Therefore by Ehrling's lemma (Lemma
\ref{lemehrling}) for any $\epsilon>0$ there exists
$\tilde{C}(\epsilon)>0$ such that
\begin{equation*}
\|\tilde{\psi}\|_{t,p,\gamma}\leq \epsilon
\|\tilde{\psi}\|_{\tilde{s},p,\tilde{\delta}}+\tilde{C}(\epsilon)\|\tilde{\psi}\|_{L^p_{\gamma}}.
\end{equation*}
Since $-\mu<\tilde{\psi}_-\leq \tilde{\psi} \leq\tilde{\psi}_+$,
$\|\tilde{\psi}\|_{L^p_{\gamma}}$ is bounded uniformly with a
constant $P$ which we absorb into $\tilde{C}(\epsilon)$.
Making use of Lemma \ref{globalbounded2} below, we have
\begin{equation*}
\|\tilde{T}^{shifted}(\tilde{\psi},
\tilde{a}_W)\|_{\tilde{s},p,\tilde{\delta}}\leq K [1+C](1+\epsilon
\|\tilde{\psi}\|_{\tilde{s},p,\tilde{\delta}}+\tilde{C}(\epsilon))
\end{equation*}
Therefore we can write
$\forall\, \psi \in [\psi_-,\psi_+]_{\tilde{s},p,\tilde{\delta}}$
and
$\forall \, W\in S([\psi_-,\psi_+]_{\tilde{s},p,\tilde{\delta}})$,
\begin{align*}
\|
T(\psi,W)-\mu\,\xi\|_{\tilde{s},p,\tilde{\delta}}
&=\| (1+\xi)\tilde{T}^{shifted}(\tilde{\psi},
\tilde{a}_W)\|_{\tilde{s},p,\tilde{\delta}}
\\
&=\|
\tilde{T}^{shifted}(\tilde{\psi}, \tilde{a}_W)
\|_{\tilde{s},p,\tilde{\delta}}+\| \xi \tilde{T}^{shifted}(\tilde{\psi},
\tilde{a}_W)\|_{\tilde{s},p,\tilde{\delta}}\\
&\leq C_4 (\|
\xi\|_{s,p,\delta}+1)\|
\tilde{T}^{shifted}(\tilde{\psi}, \tilde{a}_W)
\|_{\tilde{s},p,\tilde{\delta}}
\\
& \quad \quad \textrm{(note that $W^{s,p}_{\delta}\times W^{s,p}_{\delta} \hookrightarrow W^{s,p}_{\delta}\hookrightarrow
W^{\tilde{s},p}_{\tilde{\delta}}$)}\\
&\leq C_4 (\| \xi\|_{s,p,\delta}+1)K
[1+C](1+\epsilon
\|\tilde{\psi}\|_{\tilde{s},p,\tilde{\delta}}+\tilde{C}(\epsilon)).
\end{align*}
Now let $A:=C_4 (\| \xi\|_{s,p,\delta}+1)K [1+C]$,
so for all $\psi \in
[\psi_-,\psi_+]_{\tilde{s},p,\tilde{\delta}}$ and for all $W\in
S([\psi_-,\psi_+]_{\tilde{s},p,\tilde{\delta}})$
\begin{align*}
%&\forall\, \psi \in
%[\psi_-,\psi_+]_{\tilde{s},p,\tilde{\delta}}\cap \bar{B}_M
%(\mu\,\xi)\quad  \forall \, W\in
%S([\psi_-,\psi_+]_{\tilde{s},p,\tilde{\delta}}) \\
\|
T(\psi,W)-\mu\,\xi\|_{\tilde{s},p,\tilde{\delta}}\leq
A(1+\epsilon
\|\frac{\psi-\mu\,\xi}{\xi+1}\|_{\tilde{s},p,\tilde{\delta}}+\tilde{C}(\epsilon))
\end{align*}
Using the argument in Lemma \ref{globalbounded1} below, one can show that
for $f\in W^{\tilde{s},p}_{\tilde{\delta}}$
\begin{equation*}
\|
\frac{1}{\xi+1}f\|_{\tilde{s},p,\tilde{\delta}}\leq C_5
(1+\| \frac{\xi}{\xi+1}\|_{s,p,\delta})\|
f\|_{\tilde{s},p,\tilde{\delta}},
\end{equation*}
so if we let $\alpha:=C_5(1+\|
\frac{\xi}{\xi+1}\|_{s,p,\delta})$, then
$\|\frac{\psi-\mu\,\xi}{\xi+1}\|_{\tilde{s},p,\tilde{\delta}}\leq
\alpha \|
\psi-\mu\,\xi\|_{\tilde{s},p,\tilde{\delta}}$
 and therefore
\begin{equation*}
\|
T(\psi,W)-\mu\,\xi\|_{\tilde{s},p,\tilde{\delta}}\leq
A(1+\epsilon \alpha \|
\psi-\mu\,\xi\|_{\tilde{s},p,\tilde{\delta}}+\tilde{C}(\epsilon)).
\end{equation*}
Let $\epsilon=\frac{1}{2\alpha A}$ and define
$\hat{M}:=2A+2A\tilde{C}(\epsilon)$. For all $M\geq \hat{M}$ we
have
\begin{align*}
\forall\, \psi &\in
[\psi_-,\psi_+]_{\tilde{s},p,\tilde{\delta}}\cap \bar{B}_M
(\mu\,\xi)\quad \forall \, W\in
S([\psi_-,\psi_+]_{\tilde{s},p,\tilde{\delta}})\\
&\|
T(\psi,W)-\mu\,\xi\|_{\tilde{s},p,\tilde{\delta}}\leq
A(1+\epsilon \alpha M+\tilde{C}(\epsilon)) \quad \textrm{(note
that $\|
\psi-\mu\,\xi\|_{\tilde{s},p,\tilde{\delta}}\leq M$)}\\
& \hspace{3cm}=A+(\epsilon \alpha A)M+ A\tilde{C}(\epsilon)\\
&\hspace{3cm}=A+\frac{1}{2}M+\frac{\hat{M}-2A}{2}\\
&\hspace{3cm}=\frac{1}{2}M+\frac{1}{2}\hat{M}\leq M.
\end{align*}
Therefore $T(\psi, W)\in \bar{B}_M (\mu\,\xi)$. This completes
the proof of the auxiliary claim.
Clearly the claim of the theorem now follows from the coupled
Schauder theorem.

\medskip
\noindent
\textbf{Case 2:} $s>2$

\smallskip
We say the 10-tuple $A=(s,p,e,q,\delta,\beta,\tau,\sigma,\rho,J)$
is \textbf{beautiful} if it satisfies the hypotheses of the
theorem, that is, if
\begin{align*}
&p\in (1,\infty),\quad s\in (1+\dfrac{3}{p},\infty),\quad
-1<\beta \leq \delta<0,\\
 &\frac{1}{q}\in
(0,1),\quad
e\in(1+\frac{3}{q},\infty)\cap[s-1,s]\cap[\frac{3}{q}+s-\frac{3}{p}-1,\frac{3}{q}+s-\frac{3}{p}]\\
& \textrm{$p=q$ if $e=s \not \in \mathbb{N}_0$},\quad
\textrm{$e<s$ if $s>2$ and $s\not \in \mathbb{N}_0$}
\end{align*}
and
\begin{itemize}
\item $\tau \in W^{e-1,q}_{\beta-1}$ if $e \geq 2$ and $\tau\in W^{1,z}_{\beta-1}$ otherwise, where $z=\dfrac{3q}{3+(2-e)q}$,
\item $\sigma \in W^{e-1,q}_{\beta-1}$,
\item $\rho \in W^{s-2,p}_{\beta-2}\cap
L^{\infty}_{2\beta-2}$, $\rho\geq 0$,
\item $J\in \textbf{W}^{e-2,q}_{\beta-2}$.
\end{itemize}
Note that the condition $\frac{1}{q}\in
\cap(0,\frac{s-1}{3})\cap[\frac{3-p}{3p},\frac{3+p}{3p}]$ in the
statement of the theorem was to ensure that the intersection
for the admissible intervals for $e$ is nonempty. Here since we
start by the assumption that $e$ exists, we do not need to
explicitly
state that condition.

We say that a 10-tuple
$\tilde{A}=(\tilde{s},\tilde{p},\tilde{e},\tilde{q},\tilde{\delta},\tilde{\beta},\tilde{\tau},\tilde{\sigma},\tilde{\rho},\tilde{J})$
is \textbf{faithful} to the 10-tuple
$A=(s,p,e,q,\delta,\beta,\tau,\sigma,\rho,J)$ if
\begin{align*}
&\tilde{\delta}=\delta+\frac{|\delta|}{2},\quad
\tilde{\beta}=\beta+\frac{|\delta|}{2},\quad
\tilde{\tau}=\tau,\quad \tilde{\sigma}=\sigma,\quad
\tilde{\rho}=\rho,\quad \tilde{J}=J,\\
&\tilde{e}=\max\{2,e-2\},\quad \tilde{s}=\max\{2,s-2\},\\
& \frac{1}{\tilde{p}}\leq \frac{1}{p},\quad
1<\tilde{e}-\frac{3}{\tilde{q}}\leq e-\frac{3}{q},\quad
1<\tilde{s}-\frac{3}{\tilde{p}}\leq s-\frac{3}{p}.
\end{align*}
We say that $\tilde{A}$ is \textbf{extremely faithful} to $A$ if
$\tilde{A}$ is both \textbf{beautiful} and \textbf{faithful} to $A$.
\begin{remark}\lab{boundcontrol}
Note that if $\tilde{A}$ is \textbf{faithful} to $A$, then
\begin{equation*}
L^{\infty}_{\beta-1}\hookrightarrow
L^{\infty}_{\tilde{\beta}-1},\quad
L^{\infty}_{2\beta-2}\hookrightarrow
L^{\infty}_{2\tilde{\beta}-2},\quad
\textbf{W}^{e-2,q}_{\beta-2}\hookrightarrow
\textbf{W}^{\tilde{e}-2,\tilde{q}}_{\tilde{\beta}-2}.
\end{equation*}
So, in particular,
$\|\cdot\|_{L^{\infty}_{\tilde{\beta}-1}}$,
$\|\cdot\|_{L^{\infty}_{2\tilde{\beta}-2}}$,
$\|\cdot\|_{\textbf{W}^{\tilde{e}-2,\tilde{q}}_{\tilde{\beta}-2}}$
can be controlled by $\|\cdot\|_{L^{\infty}_{\beta-1}}$,
$\|\cdot\|_{L^{\infty}_{2\beta-2}}$,
$\|\cdot\|_{\textbf{W}^{e-2,q}_{\beta-2}}$,
respectively.
\end{remark}
We now complete the proof of the theorem for $s>2$, under the condition
that the following two claims hold.
We will then proceed to prove both claims.

\smallskip
\emph{
\textbf{Claim 1:} Suppose the 10-tuple $\tilde{A}$ is \textbf{faithful} to the \textbf{beautiful} 10-tuple $A$.  If $(\psi, W)\in W^{\tilde{s},p}_{\tilde{\delta}}\times
\textbf{W}^{\tilde{e},\tilde{q}}_{\tilde{\beta}}$is a solution of
the constraint equations with data $(\tau, \sigma, \rho,J)$ (which
is the same as $(\tilde{\tau}, \tilde{\sigma},
\tilde{\rho},\tilde{J})$) , then $(\psi,W)\in
W^{s,p}_{\delta}\times \textbf{W}^{e,q}_{\beta}$.
}

\smallskip
\emph{
\textbf{Claim 2:} If $A$ is a \textbf{beautiful} 10-tuple with $s>2$, then
there exists a 10-tuple $\tilde{A}$ that is \textbf{extremely faithful} to
$A$.
}

\medskip
\emph{\textbf{Proof of the Theorem under Claims 1 and 2.}}
The argument to complete the proof in the case $s>2$ based on
these two claims holding is as follows.
Let $A$ denote the 10-tuple associated to the given data
in the statement of the theorem.
By \textbf{Claim 2}, there exists a finite chain
\begin{equation*}
A=A_0\rightarrow A_1=(s_1,p_1,...)\rightarrow
A_2=(s_2,p_2,...)\rightarrow ... \rightarrow A_m=(s_m.p_m,...)
\end{equation*}
of 10-tuples such that $s_m=2$ and each $A_i$ is \textbf{extremely faithful} to $A_{i-1}$. Now since $A_m$ is \textbf{beautiful} and $s_m=2$,
by what was proved in the previous case we can choose $\mu$ small
enough so that (\ref{eqweak1}) and (\ref{eqweak2}) have a solution
 $(\psi, W)\in
W^{s_m=2,p_m}_{\delta_m}\times \textbf{W}^{e_m,q_m}_{\beta_m}$
(note that according to Remark \ref{boundcontrol} by assuming
$\|
 \sigma
\|_{L^{\infty}_{\beta-1}}$,
$\| \rho\|_{L^{\infty}_{2\beta-2}}$, $\| J
\|_{\textbf{W}^{e-2,q}_{\beta-2}}$ are sufficiently small,
we can ensure that $\|
 \sigma
\|_{L^{\infty}_{\beta_m-1}}$,
$\| \rho\|_{L^{\infty}_{2\beta_m-2}}$, $\| J
\|_{\textbf{W}^{e_m-2,q_m}_{\beta_m-2}}$ are as small as needed).
By \textbf{Claim 1}, since each $A_i$ is \textbf{faithful} to $A_{i-1}$,
we can conclude that $(\psi,W)\in
W^{s,p}_{\delta}\times \textbf{W}^{e,q}_{\beta}$.
The main claim of the theorem in the case of $s>2$ now follows.

Therefore, in the case $s>2$ it is enough to prove
\textbf{Claim 1} and \textbf{Claim 2}, which we now proceed to do.
Before we begin, note that since in both claims $A$ is assumed to be
\textbf{beautiful}, we have $-1<\beta\leq
\delta<0$ and so clearly $-1<\tilde{\beta}\leq \tilde{\delta}<0$;
moreover
$\beta<\tilde{\beta}$ and $\delta<\tilde{\delta}$.

\medskip
\emph{\textbf{Proof of Claim 1.}}

\smallskip
\emph{Step 1:} $b_\tau (\psi+\mu)^6+b_J \in \textbf{W}^{e-2,q}_{\beta-2}$.\\
Note that $b_\tau, b_J \in \textbf{W}^{e-2,q}_{\beta-2}$ and
$\psi\in W^{\tilde{s},\tilde{p}}_{\tilde{\delta}}$. By Lemma
\ref{lempA1} in order to show that $b_\tau (\psi+\mu)^6 \in
\textbf{W}^{e-2,q}_{\beta-2}$ it is enough to prove the following:
\begin{equation*}
\textrm{(i)}\,\, e-2\in [-\tilde{s},\tilde{s}]\,\, (\textrm{$e-2
\in (-\tilde{s},\tilde{s})$ if $\tilde{s}\not \in \mathbb{N}_0$} )
,\quad\textrm{(ii)}\,\, e-2-\frac{3}{q}\in
[-3-\tilde{s}+\frac{3}{\tilde{p}},\tilde{s}-\frac{3}{\tilde{p}}].
\end{equation*}
For (ii) we have
\begin{align*}
& e-\frac{3}{q}>1 \Rightarrow
e-\frac{3}{q}-2>-1>-3-(\tilde{s}-\frac{3}{\tilde{p}})\quad
(\textrm{note that $\tilde{s}>\frac{3}{\tilde{p}}$}),\\
& e\leq s+\frac{3}{q}-\frac{3}{p}\leq
\tilde{s}+2+\frac{3}{q}-\frac{3}{\tilde{p}}\Rightarrow
e-2-\frac{3}{q}\leq \tilde{s}-\frac{3}{\tilde{p}} \quad
(\textrm{note $s\leq \tilde{s}+2$ and
$\frac{1}{\tilde{p}}\leq \frac{1}{p}$} ).
\end{align*}
In order to prove (i) we consider two cases:

\smallskip
\emph{Case 1:} $0<s-2\leq 2$. In this case $\tilde{s}=2$ and therefore
\begin{equation*}
e-2\in [-\tilde{s},\tilde{s}]\Leftrightarrow e-2\in
[-2,2]\Leftrightarrow e\in[0,4] \textrm{ (clearly true since $e\in[s-1,s]$)}.
\end{equation*}

\smallskip
\emph{Case 2:} $s-2>2$. In this case $\tilde{s}=s-2$. Therefore
\begin{equation*}
e-2\in [-\tilde{s},\tilde{s}]\Leftrightarrow e-2 \in
[-s+2,s-2]\Leftrightarrow e\in[4-s,s].
\end{equation*}
$e\leq s$ is true by assumption. Also by assumption $e\geq s-1$
and since $s>4$ we have $s-1>4-s$. It follows that $e>4-s$. Note
that if $\tilde{s}=s-2>0$ is not in $\mathbb{N}_0$, then $s\not
 \in \mathbb{N}_0$ and so since $A$ is \textbf{beautiful} and $s>2$ we can
conclude that $e<s$. That is, in this case we have $e-2\in
(-\tilde{s},\tilde{s})$ exactly as desired.

\medskip
\emph{Step 2:} $W\in \textbf{W}^{e,q}_{\beta}$.

\smallskip
By what was shown in the previous step we know that $\mathcal{A}_L
W =-(b_\tau (\psi+\mu)^6+b_J) \in \textbf{W}^{e-2,q}_{\beta-2}$.
It follows from Remark \ref{remb1} that $W\in
\textbf{W}^{e,q}_{\beta}$.

\medskip
\emph{Step 3:} $\psi \in W^{s,p}_{\delta}$.

\smallskip
Since $W\in \textbf{W}^{e,q}_{\beta}$ according to the argument
that we had in deriving \textbf{Weak Formulation~\ref{weakf2}} we have
$a_W\in W^{s-2,p}_{\delta-2}$. It follows that $A_L \psi \in
W^{s-2,p}_{\delta-2}$. So again by Remark \ref{remb1}, we can
conclude that $\psi\in W^{s,p}_{\delta}$.

\medskip
Therefore, we have shown that \textbf{Claim 1} holds.
We now proceed to \textbf{Claim 2}.

\medskip
\emph{\textbf{Proof of Claim 2.}}
We want to find a 10-tuple $\tilde{A}$
that is \textbf{extremely faithful} to $A$. Note that all the components of
$\tilde{A}$, except $\tilde{p}$ and $\tilde{q}$, are automatically
determined by $A$. So we need to find $\tilde{p}$ and $\tilde{q}$
so that $\tilde{A}$ becomes \textbf{extremely faithful} to $A$.
We must consider three cases:

\medskip
\emph{Case 1:} $0<s-2\leq 2$, $e-2\leq 2$ (so $\tilde{s}=\tilde{e}=2$)

\smallskip
Select $\tilde{p}$ and $\tilde{q}$ to satisfy
\begin{equation*}
\frac{1}{\tilde{p}}\in
[\frac{1}{p}-\frac{s-2}{3},\frac{1}{3})\cap (0,\frac{1}{p}),
\qquad
\qquad
\frac{1}{\tilde{q}}\in
[\frac{1}{q}-\frac{e-2}{3},\frac{1}{3})\cap
(\frac{1}{\tilde{p}},\infty).
\end{equation*}
Our claim is that the 10-tuple
$\tilde{A}=(\tilde{s}=2,\tilde{p},\tilde{e}=2,\tilde{q},\tilde{\delta}=\delta+\frac{|\delta|}{2},\tilde{\beta}=\beta+\frac{|\delta|}{2},\tilde{\tau}=\tau,\tilde{\sigma}=\sigma,\tilde{\rho}=\rho,\tilde{J}=J)$
is \textbf{extremely faithful} to $A$.

First note that it is possible to pick such $\tilde{p}$ and
$\tilde{q}$. Indeed,
\begin{align*}\displaystyle
 [\frac{1}{p}-\frac{s-2}{3},\frac{1}{3})\neq \varnothing,
  \hspace*{0.5cm}
     & \mbox{ since } s>1+\frac{3}{p}, \\
 [\frac{1}{p}-\frac{s-2}{3},\frac{1}{3})\cap (0,\frac{1}{p})\neq \varnothing,
  \hspace*{0.5cm}
     & \mbox{ since for } s>2,
       \mbox{ we have } \frac{1}{p}-\frac{s-2}{3}<\frac{1}{p}, \\
 [\frac{1}{q}-\frac{e-2}{3},\frac{1}{3})\neq \varnothing,
  \hspace*{0.5cm}
     & \mbox{ since } e>1+\frac{3}{q}, \\
 [\frac{1}{q}-\frac{e-2}{3},\frac{1}{3})
    \cap (\frac{1}{\tilde{p}},\infty)\neq \varnothing,
  \hspace*{0.5cm}
     & \mbox{ since } \frac{1}{\tilde{p}}<\frac{1}{3}.
\end{align*}
In order to show that $\tilde{A}$ is \textbf{extremely faithful} to $A$ we
need to show that 1) $\tilde{A}$ is \textbf{faithful} to
$A$ and 2) $\tilde{A}$ is \textbf{beautiful}.

1) $\tilde{A}$ is \textbf{faithful} to $A$:
\begin{align*}
  (i) & \textrm{ By definition of $\tilde{p}$ we have }
\frac{1}{\tilde{p}}\leq \frac{1}{p}.\\
  (ii) & \textrm{ Clearly } \tilde{e}=2=\max\{2,e-2\},\quad
\tilde{s}=2=\max\{2,s-2\}.\\
 (iii) & \ \frac{1}{p}-\frac{s-2}{3}\leq
\frac{1}{\tilde{p}}<\frac{1}{3}\Rightarrow
-1<\frac{-3}{\tilde{p}}\leq s-\frac{3}{p}-2
\\
& \quad \quad \quad \quad
\Rightarrow
1<2-\frac{3}{\tilde{p}}\leq s-\frac{3}{p} \Rightarrow
1<\tilde{s}-\frac{3}{\tilde{p}}\leq s-\frac{3}{p} \quad
(\tilde{s}=2).\\
(iv) & \textrm{ Similarly } \frac{1}{q}-\frac{e-2}{3}\leq
\frac{1}{\tilde{q}}<\frac{1}{3}\Rightarrow
1<\tilde{e}-\frac{3}{\tilde{q}}\leq e-\frac{3}{q}.
\end{align*}

2) $\tilde{A}$ is \textbf{beautiful}:
\begin{align*}
  (i) & \textrm{ Clearly } \tilde{p}, \tilde{q}\in (1,\infty).\,
\\
& \textrm{ In addition, by what was proved above, }
\tilde{s}>1+\frac{3}{\tilde{p}} \textrm{ and }
\tilde{e}>1+\frac{3}{\tilde{q}}.\\
  (ii) & \ \tilde{e}\in [\tilde{s}-1,\tilde{s}] \Leftrightarrow 2\in
[2-1,2]\quad \textrm{(which is clearly true)}.\\
  (iii) & \ \tilde{e}\in
[\frac{3}{\tilde{q}}+\tilde{s}-\frac{3}{\tilde{p}}-1,\frac{3}{\tilde{q}}+\tilde{s}-\frac{3}{\tilde{p}}]\Leftrightarrow
2\in
[\frac{3}{\tilde{q}}-\frac{3}{\tilde{p}}+1,\frac{3}{\tilde{q}}-\frac{3}{\tilde{p}}+2]
\\
& \quad
\Leftrightarrow 0\leq \frac{3}{\tilde{q}}-\frac{3}{\tilde{p}} \leq 1 \quad
(\tilde{s}=\tilde{e}=2)\\
& \quad \Leftrightarrow \frac{1}{\tilde{p}}\leq
\frac{1}{\tilde{q}} \textrm{ and } \frac{1}{\tilde{q}}\leq
\frac{1}{3}+\frac{1}{\tilde{p}} \quad
  \textrm{(since we know
$\frac{1}{3}>\frac{1}{\tilde{q}}>\frac{1}{\tilde{p}}$)}.
\end{align*}
Also since $\tilde{s}-\frac{3}{\tilde{p}}\leq s-\frac{3}{p}$,
$\tilde{e}-\frac{3}{\tilde{q}}\leq e-\frac{3}{q}$,
$\beta<\tilde{\beta}$ and $\delta< \tilde{\delta}$, it follows
from the embedding theorem that
\begin{align*}
&W^{s,p}_{\delta}\hookrightarrow
W^{\tilde{s},\tilde{p}}_{\tilde{\delta}},\quad
W^{s-2,p}_{\beta-2}\hookrightarrow
W^{\tilde{s},\tilde{p}}_{\tilde{\beta}-2},\\
&W^{e,q}_{\beta}\hookrightarrow
W^{\tilde{e},q}_{\tilde{\beta}},\quad
W^{e-1,q}_{\beta-1}\hookrightarrow
W^{\tilde{e}-1,q}_{\tilde{\beta}-1},\quad
W^{e-2,q}_{\beta-2}\hookrightarrow
W^{\tilde{e}-2,q}_{\tilde{\beta}-2}.
\end{align*}
Therefore $\tau$, $\sigma$, $\rho$ and $J$ are in the correct
spaces.

\medskip
\emph{Case 2:} $s-2>2$, $e-2\leq 2$
       (so $\tilde{s}=s-2$, $\tilde{e}=2$)

\smallskip
Select $\tilde{q}$ such that $\frac{1}{\tilde{q}}\in
[\frac{1}{q}-\frac{e-2}{3},\frac{1}{3})\cap
[\frac{1}{p}-\frac{2}{3},\frac{1}{p})$. Let
$\tilde{p}:=\tilde{q}$. Our claim is that the 10-tuple
$\tilde{A}=(\tilde{s}=s-2,\tilde{p},\tilde{e}=2,\tilde{q},\tilde{\delta}=\delta+\frac{|\delta|}{2},\tilde{\beta}=\beta+\frac{|\delta|}{2},\tilde{\tau}=\tau,\tilde{\sigma}=\sigma,\tilde{\rho}=\rho,\tilde{J}=J)$
is \textbf{extremely faithful} to $A$.

First note that it is possible to pick such $\tilde{q}$. Indeed,
$[\frac{1}{q}-\frac{e-2}{3},\frac{1}{3})\neq \varnothing$ because
$e>1+\frac{3}{q}$. For the intersection to be nonempty we need to
check $\frac{1}{p}-\frac{2}{3}<\frac{1}{3}$ and
$\frac{1}{q}-\frac{e-2}{3}<\frac{1}{p}$. The first inequality is
clearly true. The second inequality is also true because
\begin{align*}
&e>\frac{3}{q}-\frac{3}{p}+s-1>\frac{3}{q}-\frac{3}{p}+2
\quad \textrm{ (note that $s>4$)}\\
&\hspace{1cm} \Rightarrow \frac{3}{q}-(e-2)<\frac{3}{p}\Rightarrow
\frac{1}{q}-\frac{e-2}{3}<\frac{1}{p}.
\end{align*}
1) $\tilde{A}$ is \textbf{faithful} to $A$:
\begin{align*}
(i) & \ \tilde{p}=\tilde{q}, \textrm{ and }
\frac{1}{\tilde{q}}<\frac{1}{p} \Rightarrow
\frac{1}{\tilde{p}}\leq \frac{1}{p}.\\
(ii) & \ \frac{1}{q}-\frac{e-2}{3}\leq
\frac{1}{\tilde{q}}<\frac{1}{3}
\Rightarrow
1<\tilde{e}-\frac{3}{\tilde{q}}\leq e-\frac{3}{q}. \quad
(\tilde{e}=2)\\
(iii) & \ \tilde{s}=s-2>2 \Rightarrow
\tilde{s}-\frac{3}{\tilde{q}}>2-\frac{3}{\tilde{q}}>1
\Rightarrow \tilde{s}-\frac{3}{\tilde{p}}>1.
\quad \textrm{ (note
$\frac{1}{\tilde{q}}<\frac{1}{3}$ and $\tilde{q}=\tilde{p}$)}\\
(iv) & \ \frac{1}{\tilde{q}}\geq \frac{1}{p}-\frac{2}{3}
\Rightarrow \frac{3}{p}\leq 2+\frac{3}{\tilde{q}}
\\
& \quad \quad \Rightarrow
s-2-\frac{3}{\tilde{q}}\leq s-\frac{3}{p}
\Rightarrow\tilde{s}-\frac{3}{\tilde{p}}\leq
s-\frac{3}{p}.
\quad \textrm{ (note $\tilde{s}=s-2$ and $\tilde{q}=\tilde{p}$)}
\end{align*}

2) $\tilde{A}$ is \textbf{beautiful}:
\begin{align*}
(i) & \textrm{ Clearly } \tilde{p}, \tilde{q}\in (1,\infty).
\textrm{ By what was proved above }
\tilde{s}>1+\frac{3}{\tilde{p}} \textrm{ and }
\tilde{e}>1+\frac{3}{\tilde{q}}.\\
(ii) & \ \tilde{e}\in [\tilde{s}-1,\tilde{s}] \Leftrightarrow 2\in
[s-3,s-2]\Leftrightarrow 4\leq s \leq 5.
\\
& \quad \quad \textrm{ (by assumption $s>4$; also $s-1\leq e\leq 4$ and so $s\leq 5$)}.\\
(iii) & \ \tilde{e}\in
[\frac{3}{\tilde{q}}+\tilde{s}-\frac{3}{\tilde{p}}-1,\frac{3}{\tilde{q}}+\tilde{s}-\frac{3}{\tilde{p}}]\Leftrightarrow
2\in
[\frac{3}{\tilde{q}}+s-3-\frac{3}{\tilde{p}},\frac{3}{\tilde{q}}+s-2-\frac{3}{\tilde{p}}]\\
&\hspace{2cm} \Leftrightarrow 4\leq
s+\frac{3}{\tilde{q}}-\frac{3}{\tilde{p}} \leq 5 \Leftrightarrow
4\leq s \leq 5.
\\
& \quad \quad \textrm{ (which is true; note that
$\tilde{s}=s-2,\,\tilde{e}=2,\, \tilde{q}=\tilde{p}$})
\end{align*}
The proof that $\tau, \sigma, \rho$ and $J$ belong to
the correct spaces is exactly the same as Case 1.

\medskip
\emph{Case 3:} $s-2>2$, $e-2>2$ (so $\tilde{s}=s-2$, $\tilde{e}=e-2$).

\smallskip
Select $\tilde{q}$ to satisfy
\begin{equation*}
\frac{1}{\tilde{q}}\in
[\frac{1}{q}-\frac{2}{3},\frac{e}{3}-1)\cap (0,\frac{1}{q})\cap
(\frac{1}{q}-\frac{1}{p},\infty).
\end{equation*}
Define $\tilde{p}$ by
$\frac{1}{\tilde{p}}:=\frac{1}{\tilde{q}}-\frac{1}{q}+\frac{1}{p}$.
Our claim is that the 10-tuple
$\tilde{A}=(\tilde{s}=s-2,\tilde{p},\tilde{e}=e-2,\tilde{q},\tilde{\delta}=\delta+\frac{|\delta|}{2},\tilde{\beta}=\beta+\frac{|\delta|}{2},\tilde{\tau}=\tau,\tilde{\sigma}=\sigma,\tilde{\rho}=\rho,\tilde{J}=J)$
is \textbf{extremely faithful} to $A$.

First note that it is possible to pick such $\tilde{q}$. Indeed,
$[\frac{1}{q}-\frac{2}{3},\frac{e}{3}-1)\neq \varnothing$ because
$e>1+\frac{3}{q}$. In order to show that the intersection of the
three intervals is nonempty we consider two possibilities:
\begin{itemizeXXM}
\item Possibility 1: $\frac{1}{q}-\frac{1}{p}>0$. In this case
\begin{equation*}
(0,\frac{1}{q})\cap
(\frac{1}{q}-\frac{1}{p},\infty)=(\frac{1}{q}-\frac{1}{p},\frac{1}{q}),
\end{equation*}
and so it is enough to show that
\begin{equation*}
[\frac{1}{q}-\frac{2}{3},\frac{e}{3}-1)\cap
(\frac{1}{q}-\frac{1}{p},\frac{1}{q})\neq \varnothing.
\end{equation*}
This is true because
\begin{align*}
(i)  & \textrm{ Clearly } \frac{1}{q}-\frac{2}{3}<\frac{1}{q},\\
(ii) & \ e\geq
\frac{3}{q}-\frac{3}{p}+s-1>\frac{3}{q}-\frac{3}{p}+3 \Rightarrow
\frac{e}{3}-1>\frac{1}{q}-\frac{1}{p}.
\quad \textrm{ (note that $s>4$})
\end{align*}
\item Possibility 2: $\frac{1}{q}-\frac{1}{p}\leq 0$. In this case
\begin{equation*}
(0,\frac{1}{q})\cap
(\frac{1}{q}-\frac{1}{p},\infty)=(0,\frac{1}{q}),
\end{equation*}
and so it is enough to show that
\begin{equation*}
[\frac{1}{q}-\frac{2}{3},\frac{e}{3}-1)\cap (0,\frac{1}{q})\neq
\varnothing.
\end{equation*}
This is true because
\begin{equation*}
(i)  \textrm{ Clearly } \frac{1}{q}-\frac{2}{3}<\frac{1}{q},
\quad \textrm{ and } \quad
(ii) \ e>3 \Rightarrow \frac{e}{3}-1>0 .
\end{equation*}
\end{itemizeXXM}

1) $\tilde{A}$ is \textbf{faithful} to $A$:
\begin{align*}
(i) & \ \frac{1}{\tilde{q}}<\frac{1}{q}\Rightarrow
\frac{1}{\tilde{q}}-\frac{1}{q}+\frac{1}{p}<\frac{1}{p}\Rightarrow
\frac{1}{\tilde{p}}<\frac{1}{p}.\\
(ii) & \ \frac{1}{\tilde{q}}<\frac{e}{3}-1\Rightarrow
e-2-\frac{3}{\tilde{q}}>1 \Rightarrow
\tilde{e}-\frac{3}{\tilde{q}}>1. \quad (\tilde{e}=e-2)\\
(iii) & \ \frac{1}{q}-\frac{2}{3}\leq
\frac{1}{\tilde{q}}\Rightarrow \frac{3}{q}-2\leq
\frac{3}{\tilde{q}}\Rightarrow e-2-\frac{3}{\tilde{q}}\leq
e-\frac{3}{q} \Rightarrow \tilde{e}-\frac{3}{\tilde{q}}\leq
e-\frac{3}{q}.  \quad (\tilde{e}=e-2) \\
(iv) & \ 3+\frac{3}{\tilde{q}}<e<s+\frac{3}{q}-\frac{3}{p}\Rightarrow
3+\frac{3}{\tilde{q}}<s+\frac{3}{q}-\frac{3}{p}\Rightarrow
1<s-2-\frac{3}{\tilde{q}}+\frac{3}{q}-\frac{3}{p}
\\
& \quad \quad \Rightarrow
1<s-2-\frac{3}{\tilde{p}}
\Rightarrow1<\tilde{s}-\frac{3}{\tilde{p}}.
\quad \textrm{ (note that
$\frac{1}{\tilde{p}}:=\frac{1}{\tilde{q}}-\frac{1}{q}+\frac{1}{p}$
and $\tilde{s}=s-2$)}\\
(v) & \ \frac{1}{q}-\frac{2}{3}\leq \frac{1}{\tilde{q}}\Rightarrow
0\leq\frac{3}{\tilde{q}}-\frac{3}{q}+2\Rightarrow \frac{3}{p}\leq
\frac{3}{\tilde{q}}-\frac{3}{q}+\frac{3}{p}+2 \Rightarrow
\frac{3}{p}\leq \frac{3}{\tilde{p}}+2 \\& \hspace{1cm}\Rightarrow
s-2-\frac{3}{\tilde{p}}\leq
s-\frac{3}{p}\Rightarrow\tilde{s}-\frac{3}{\tilde{p}}\leq
s-\frac{3}{p}.
\end{align*}

2) $\tilde{A}$ is \textbf{beautiful}:
\begin{align*}
(i) & \textrm{ Clearly } \tilde{p}, \tilde{q}\in (1,\infty).
\textrm{ By what was proved above, }
\tilde{s}>1+\frac{3}{\tilde{p}} \textrm{ and }
\tilde{e}>1+\frac{3}{\tilde{q}}.\\
(ii) & \ \tilde{e}\in [\tilde{s}-1,\tilde{s}] \Leftrightarrow
e-2\in
[s-3,s-2]\Leftrightarrow e\in[s-1,s]. \textrm{ (which is clearly true)} \\
(iii) & \ \tilde{e}\in
[\frac{3}{\tilde{q}}+\tilde{s}-\frac{3}{\tilde{p}}-1,\frac{3}{\tilde{q}}+\tilde{s}-\frac{3}{\tilde{p}}]\Leftrightarrow
e-2\in
[\frac{3}{\tilde{q}}+s-3-\frac{3}{\tilde{p}},\frac{3}{\tilde{q}}+s-2-\frac{3}{\tilde{p}}]\\
&\hspace{1cm} \Leftrightarrow
e\in[\frac{3}{\tilde{q}}-\frac{3}{\tilde{p}}+s-1,s+\frac{3}{\tilde{q}}-\frac{3}{\tilde{p}}]\Leftrightarrow
e\in [\frac{3}{q}-\frac{3}{p}+s-1,s+\frac{3}{q}-\frac{3}{p}].
\\
& \quad \quad
\textrm{ (note that
$\frac{1}{\tilde{q}}-\frac{1}{\tilde{p}}=\frac{1}{q}-\frac{1}{p}$)}.
\end{align*}
The last inclusion is true because $A$ is \textbf{beautiful}. The proof of
the fact that $\tau, \sigma, \rho$ and $J$ belong to the correct
spaces is exactly the same as Case 1.

Note that $e\leq s$, so if $s-2\leq 2$ then $e-2\leq 2$ and
therefore the case where $s-2\leq 2$, $e-2>2$ does not happen.

This establishes \textbf{Claim 2}, and by earlier arguments the
main claim of the Theorem now follows.
\end{proof}

%%%%%%%%%%%%%%%%%%%%%%%%%%%%%%%%%%%%%%%%%%%%%%%%%%%%%%%%%%%%%%%%%%%%%%%%%%%%%%
\section{Two Auxiliary Results}

We now state and prove two auxiliary lemmas that were
used in the proof of Theorem~\ref{thm:main}.

\begin{lemma}\lab{globalbounded1}
Let $\chi\in W^{s,p}_{\delta}$, $\chi>-1$ and let $f\in
W^{s-2,p}_{\delta-2}$. Then $\frac{1}{1+\chi}f \in
W^{s-2,p}_{\delta-2}$ and
\begin{equation*}
\| \frac{1}{1+\chi}f \|_{s-2,p,\delta-2}\preceq (1+\| \frac{\chi}{\chi+1}\|_{s,p,\delta})\|
f\|_{s-2,p,\delta-2}.
\end{equation*}
In particular, for a fixed $\chi$, the mapping $f \mapsto
\frac{1}{1+\chi}f$ (from $W^{s-2,p}_{\delta-2} to
W^{s-2,p}_{\delta-2}$) sends bounded sets to bounded sets.
\end{lemma}
\begin{proof}{\bf (Lemma~\ref{globalbounded1})}
By Lemma \ref{lempA1} $\frac{1}{1+\chi}f \in
W^{s-2,p}_{\delta-2}$. Moreover
\begin{equation*}
\| \frac{1}{1+\chi}f \|_{s-2,p,\delta-2}=\|
(\frac{1}{1+\chi}-1+1)f\|_{s-2,p,\delta-2}=\|
\frac{-\chi}{\chi+1}f+f
\|_{s-2,p,\delta-2}.
\end{equation*}
It follows from Lemma \ref{lempA1} that $\frac{-\chi}{\chi+1}\in
W^{s,p}_{\delta}$. Also by Lemma \ref{lemA1}
$W^{s,p}_{\delta}\times W^{s-2,p}_{\delta-2}\rightarrow
W^{s-2,p}_{\delta-2}$. Thus
\begin{align*}
\|
\frac{-\chi}{\chi+1}f+f
\|_{s-2,p,\delta-2}
&\leq \|
\frac{-\chi}{\chi+1}f\|_{s-2,p,\delta-2}+\| f
\|_{s-2,p,\delta-2}
\\
& \preceq  \|
\frac{-\chi}{\chi+1}\|_{s,p,\delta}\|
f\|_{s-2,p,\delta-2}+ \|
f\|_{s-2,p,\delta-2}
\\
& =(1+\|
\frac{\chi}{\chi+1}\|_{s,p,\delta})\|
f\|_{s-2,p,\delta-2}.
\end{align*}
\end{proof}

\begin{lemma}\lab{globalbounded2}
There exists a constant $C$ independent of $W$ such that
\begin{equation*}
\forall\, W \in S([\psi_-,\psi_+]_{\tilde{s},p,\tilde{\delta}}),
\quad
\| \tilde{a}_W\|_{s-2,p,\delta-2}\leq C.
\end{equation*}
\end{lemma}
\begin{proof}{\bf (Lemma~\ref{globalbounded2})}
By Corollary \ref{coro2.8} if $W \in
S([\psi_-,\psi_+]_{\tilde{s},p,\tilde{\delta}})$, that is, if $W$
is the solution to the momentum constraint with some source $\psi
\in [\psi_-,\psi_+]_{\tilde{s},p,\tilde{\delta}}$, then
\begin{align*}
\| W\|_{e,q,\beta} &\leq C_1\big[(\mu+\|
\psi\|_{L^{\infty}_{\delta}})^6 \|
b_\tau\|_{L^z_{\beta-2}}+\|
b_J\|_{\textbf{W}^{e-2,q}_{\beta-2}}\big]\\
& \leq C_1\big[(\mu+\max\{\|
\psi_-\|_{L^{\infty}_{\delta}}, \|
\psi_+\|_{L^{\infty}_{\delta}}\})^6 \|
b_\tau\|_{L^z_{\beta-2}}+\|
b_J\|_{\textbf{W}^{e-2,q}_{\beta-2}}\big].
\end{align*}
Here we used the fact that $|\psi|\leq \max\{|\psi_+|,|\psi_-|\}$
and so $\| \psi\|_{L^{\infty}_{\delta}}\leq
\max\{\| \psi_-\|_{L^{\infty}_{\delta}}, \|
\psi_+\|_{L^{\infty}_{\delta}}\}$. Consequently there is a
constant $C_2$ such that for all $W\in
S([\psi_-,\psi_+]_{\tilde{s},p,\tilde{\delta}})$ we have
$\| W\|_{e,q,\beta}\leq C_2$.

Considering the restrictions on the exponents $s,p, \delta, e,q,
\beta$ and using our embedding theorem and multiplication lemma,
 it is easy to check $W^{s-2,p}_{\beta-2} \hookrightarrow
W^{s-2,p}_{\delta-2}$, $W^{e-1,q}_{2\beta-2}\hookrightarrow
W^{s-2,p}_{\beta-2}$, and $W^{e-1,q}_{\beta-1}\times
W^{e-1,q}_{\beta-1}\hookrightarrow W^{e-1,q}_{2\beta-2}$.
Therefore we can write
\begin{align*}
\| a_W\|_{s-2,p,\delta-2} &\preceq \|
a_W\|_{s-2,p,\beta-2}\preceq \|a_W\|_{e-1,q,2\beta-2}\\
& \preceq \| \sigma+
\mathcal{L}W\|^2_{e-1,q,\beta-1}\preceq (\| \sigma
\|_{e-1,q,\beta-1}+\|
\mathcal{L}W\|_{e-1,q,\beta-1})^2\\
&\preceq \| \sigma
\|_{e-1,q,\beta-1}^2+\|
\mathcal{L}W\|_{e-1,q,\beta-1}^2\preceq \| \sigma
\|_{e-1,q,\beta-1}^2+\|
W\|_{e,q,\beta}^2\\
&\leq \| \sigma
\|_{e-1,q,\beta-1}^2+C_2.
\end{align*}
Hence there is a constant $C_3$ such that for all $W\in
S([\psi_-,\psi_+]_{\tilde{s},p,\tilde{\delta}})$ we have
$\| a_W\|_{s-2,p,\delta-2}\leq C_3$. Now notice
that $\tilde{a}_W=(1+\xi)^{-12}a_W$, that is, $\tilde{a}_W$ is
obtained from $a_W$ by applying the mapping $f \mapsto
\frac{1}{1+\xi}f$ twelve times. But by Lemma \ref{globalbounded1}
the mapping $f \mapsto \frac{1}{1+\xi}f$ sends bounded sets in
$W^{s-2,p}_{\delta-2}$ to bounded sets in $W^{s-2,p}_{\delta-2}$.
Consequently there exists a constant $C$ such that
\begin{equation*}
\forall\, W \in S([\psi_-,\psi_+]_{\tilde{s},p,\tilde{\delta}}),
\quad
\| \tilde{a}_W\|_{s-2,p,\delta-2}\leq C.
\end{equation*}
\end{proof}

%%%%%%%%%%%%%%%%%%%%%%%%%%%%%%%%%%%%%%%%%%%%%%%%%%%%%%%%%%%%%%%%%%%%%%%%%%%%%%
\section*{Acknowledgments}
   \label{sec:ack}

The authors would like to thank J. Isenberg, D. Maxwell, and R. Mazzeo 
for helpful comments at various stages of the work.
The authors would also like to especially thank C. Meier for his explanation
of the barrier constructions for AF manifolds appearing in~\cite{24}.
This work developed as part of the joint FRG Project between 
M. Holst, J. Isenberg, D. Maxwell, and R. Mazzeo,
supported by NSF FRG Award~1262982.
MH was supported in part by NSF Awards~1262982, 1318480, and 1620366.
AB was supported by NSF Award~1262982.

%%%%%%%%%%%%%%%%%%%%%%%%%%%%%%%%%%%%%%%%%%%%%%%%%%%%%%%%%%%%%%%%%%%%%%%%%%%%%%
\appendix
\section{Weighted Sobolev Spaces}
   \label{app:spaces}

We first assemble some basic results we need for weighted Sobolev spaces.
We limit our selves to simply stating the results we need,
unless the proof of the result is either unavailable or difficult to find
in the form we need, in which case we include a concise proof.

Consider an open cover of $\mathbb{R}^n$ that consists of the
following sets:
\begin{equation*}
B_2,\quad B_4\setminus \bar{B_1}, \quad B_8\setminus
\bar{B_2},\,...,\,B_{2^{j+1}}\setminus \bar{B}_{2^{j-1}},
\end{equation*}
where $B_r$ is the open ball of radius $r$ centered at the
origin. For all $r$ let $S_r f(x):= f(rx)$. Consider the
following partition of unity subordinate to the above cover of
$\mathbb{R}^n$\cite{2}:
% A type of partition of unity subordinate to the
% above cover of $\mathbb{R}^n$ which is commonly used in
% Littlewood-Paley theory can be constructed as follows:
\begin{align*}
&\varphi_0=1\quad \textrm{on} \quad B_1,\quad \supp\varphi_0\subseteq B_2,\\
&\varphi(x)= \varphi_0(x)-\varphi_0(2x) \quad (\supp\varphi
\subseteq B_2, \quad \varphi=0 \,\, \textrm{on} \,\,
B_{\frac{1}{2}}),\\
& \forall j\geq 1 \quad \varphi_j= S_{2^{-j}}\varphi.
\end{align*}
One can easily check that $\sum_{j=0}^{\infty}\varphi_j (x)=1$.

For $s\in\mathbb{\mathbb{R}}$, $p\in (1,\infty)$, the weighted
Sobolev space $W^{s,p}_{\delta}(\mathbb{R}^n)$ is defined as
follows:
\begin{equation*}
W^{s,p}_{\delta}(\mathbb{R}^n)=\{u\in S'(\mathbb{R}^n):
\|u\|_{W^{s,p}_{\delta}(\mathbb{R}^n)}^p=\sum_{j=0}^{\infty}2^{-p\delta
j} \|S_{2^j}(\varphi_j u)\|_{W^{s,p}(\mathbb{R}^n)}^p<\infty \}.
\end{equation*}
Here $W^{s,p}(\mathbb{R}^n)$ is the Sobolev-Slobodeckij space
which is defined as follows:
\begin{itemizeX}
\item If $s=k\in \mathbb{N}_0$, $p\in[1,\infty]$,
\begin{equation*}
W^{k,p}(\mathbb{R}^n)=\{u\in L^p (\mathbb{R}^n):
\|u\|_{W^{k,p}(\mathbb{R}^n)}:=\sum_{|\nu|\leq
k}\|\partial^{\nu}u\|_p<\infty\}
\end{equation*}
\item If $s=\theta\in(0,1)$, $p\in[1,\infty)$,
\begin{equation*}
W^{\theta,p}(\mathbb{R}^n)=\{u\in L^p (\mathbb{R}^n):
 |u|_{W^{\theta,p}(\mathbb{R}^n)}:=\big(\int\int_{\mathbb{R}^n\times
\mathbb{R}^n}\frac{|u(x)-u(y)|^p}{|x-y|^{n+\theta p}}dx
dy\big)^{\frac{1}{p}} <\infty\}
\end{equation*}
\item If $s=\theta\in(0,1)$, $p=\infty$,
\begin{equation*}
W^{\theta,\infty}(\mathbb{R}^n)=\{u\in L^{\infty} (\mathbb{R}^n):
 |u|_{W^{\theta,\infty}(\mathbb{R}^n)}:=\esssup_{x,y \in
\mathbb{R}^n, x\neq y}\frac{|u(x)-u(y)|}{|x-y|^{\theta}} <\infty\}
\end{equation*}
\item If $s=k+\theta,\, k\in \mathbb{N}_0,\, \theta\in(0,1)$,
$p\in[1,\infty]$,
\begin{equation*}
W^{s,p}(\mathbb{R}^n)=\{u\in
W^{k,p}(\mathbb{R}^n):\|u\|_{W^{s,p}(\mathbb{R}^n)}:=\|u\|_{W^{k,p}(\mathbb{R}^n)}+\sum_{|\nu|=k}
|\partial^{\nu}u|_{W^{\theta,p}(\mathbb{R}^n)}<\infty\}
\end{equation*}
\item If $s<0$ and $p\in(1,\infty)$,
\begin{equation*}
W^{s,p}(\mathbb{R}^n)=(W^{-s,p'}(\mathbb{R}^n))^{*} \quad
(\frac{1}{p}+\frac{1}{p'}=1).
\end{equation*}
\end{itemizeX}
Alternatively, we could have defined $W^{s,p}(\mathbb{R}^n)$ as a
Bessel potential space, that is,
\begin{equation*}
W^{s,p}(\mathbb{R}^n)= \{u\in S'(\mathbb{R}^n):
\|u\|_{W^{s,p}(\mathbb{R}^n)}:=\|\mathcal{F}^{-1}(\langle\xi\rangle^s\mathcal{F}u)\|_{L^p}<\infty\},
\end{equation*}
where $\langle\xi\rangle:=(1+|\xi|^2)^{\frac{1}{2}}$.
It is a well known fact that for $k\in \mathbb{N}_0$ the above
definition of $W^{k,p}(\mathbb{R}^n)$ agrees with the first
definition \cite{31}. Also for $s\in \mathbb{R}$ and $p=2$ the two
definitions agree\cite{31}. It is customary to use $H^{s,p}$
instead of $W^{s,p}$ for unweighted Bessel potential spaces. We
denote the corresponding weighted spaces by $H^{s,p}_{\delta}$. In
this paper (except in Appendix E) we use the first definition. The
norm on $W^{k,p}_{\delta}(\mathbb{R}^n)$ is equivalent to the
following norm \cite{2,15}: (since the norms are equivalent we
use the same notation for both norms)
\begin{align*}
\|u\|_{W^{k,p}_{\delta}(\mathbb{R}^n)}=\sum_{|\beta|\leq k}
\|\langle x\rangle^{-\delta-\frac{n}{p}+|\beta|}\partial^{\beta}u\|_{L^p(\mathbb{R}^n)}.
\end{align*}
When $s=0$, we denote $W^{s,p}_{\delta}(\mathbb{R}^n)$ by
$L^p_{\delta}(\mathbb{R}^n)$. In particular we have
\begin{equation*}
\| u \|_{L^p_{\delta}(\mathbb{R}^n)}=\|
\langle x\rangle^{-\delta-\frac{n}{p}}u\|_{L^p(\mathbb{R}^n)}.
\end{equation*}
\begin{remark}\lab{remunweighted}
We take a moment to make the following three observations.
\begin{itemizeXXM}
\item Considering the above formula for the norm, it is obvious that
if $\delta\leq -\frac{n}{p}$ then
$\langle x\rangle^{-\delta-\frac{n}{p}+|\beta|}\geq 1$ and therefore
$\|u\|_{W^{k,p}(\mathbb{R}^n)}\leq
\|u\|_{W^{k,p}_{\delta}(\mathbb{R}^n)}$ and
$W^{k,p}_{\delta}(\mathbb{R}^n)\hookrightarrow
W^{k,p}(\mathbb{R}^n)$.
\item Note that if $u\in L^{p}_{\delta}(\mathbb{R}^n)$ and $v\in
L^{\infty}(\mathbb{R}^n)$, then
\begin{align*}
\| v u \|_{L^{p}_{\delta}(\mathbb{R}^n)}
& =\|
\langle x\rangle^{-\delta-\frac{n}{p}}v u\|_{L^p(\mathbb{R}^n)}
\\
& \leq \| v \|_{\infty}\| \langle x\rangle^{-\delta-\frac{n}{p}}
u\|_{L^p(\mathbb{R}^n)}
\\
& =\| v
\|_{\infty}\| u
\|_{L^{p}_{\delta}(\mathbb{R}^n)}.
\end{align*}
\item It is easy to show that for $p\in(1,\infty)$,
$\langle x\rangle^{\delta'}\in L^{p}_{\delta}(\mathbb{R}^n)$ for
every $\delta'<\delta$, but $\langle x\rangle^{\delta}\not \in
L^{p}_{\delta}(\mathbb{R}^n)$.
\end{itemizeXXM}
\end{remark}
\begin{remark}
Suppose $1<p<\infty$. Note that in the case of unweighted Sobolev
spaces, for $s<0$, $W^{s,p}(\mathbb{R}^n)$ is \textbf{defined} as
the dual of $W^{-s,p'}(\mathbb{R}^n)$ . In fact, since
$W^{s,p}(\mathbb{R}^n)$ is reflexive, we have
$(W^{s,p}(\mathbb{R}^n))^{*}=W^{-s,p'}(\mathbb{R}^n)$ for all
$s\in \mathbb{R}$. Contrary to the unweighted case, in case of
weighted Sobolev spaces our definition of
$W^{s,p}_{\delta}(\mathbb{R}^n)$ for $s<0$ is not based on
duality. Nevertheless, as it is stated in the next theorem,
$(W^{s,p}_{\delta}(\mathbb{R}^n))^{*}$ can be \textbf{identified}
with $W^{-s,p'}_{-n-\delta}(\mathbb{R}^n)$. This identification
can be done by defining a suitable bilinear form
$W^{-s,p'}_{-n-\delta}(\mathbb{R}^n) \times
W^{s,p}_{\delta}(\mathbb{R}^n)\rightarrow \mathbb{R}\,\,$
\cite{16}.
\end{remark}
\begin{remark}
In the literature, the \textbf{growth parameter} $\delta$ has been
incorporated in the definition of weighted spaces in more than one
way. Our convention for the growth parameter agrees with
Bartnik's convention \cite{5} and Maxwell's convention
\cite{2,dM06,dM05b}. The following items describe how our definition
corresponds with the other related definitions of weighted spaces
in the literature:
\begin{itemize}
\item For $s\in \mathbb{Z}$ our spaces
$W^{s,p}_{\delta}(\mathbb{R}^n)$ correspond with the spaces
$h^{s}_{p,ps-p\delta-n}(\mathbb{R}^n)$ in \cite{15,16} and
$H^{s,p}_{\delta}(\mathbb{R}^n)$ in \cite{2}.
\item For $s\not \in \mathbb{Z}$ our spaces
$W^{s,p}_{\delta}(\mathbb{R}^n)$ correspond with the spaces
$b^{s}_{p,p,ps-p\delta-n}(\mathbb{R}^n)$ in \cite{15,16} and
$W^{p}_{s,-\delta-\frac{n}{p}}(\mathbb{R}^n)$ in \cite{6}.
\item For $s\in \mathbb{R}$ and $p=2$ our spaces $W^{s,p}_{\delta}(\mathbb{R}^n)$
correspond with the spaces $H^{s}_{\delta}(\mathbb{R}^n)$ in
\cite{2,dM06}.
\end{itemize}
\end{remark}
The space $W^{s,p}_{loc}(\mathbb{R}^n)$ is defined as the set of
distributions $u\in D'(\mathbb{R}^n)$ for which $\chi u \in
W^{s,p}(\mathbb{R}^n)$ for all $\chi\in
C_c^{\infty}(\mathbb{R}^n)$. $W^{s,p}_{loc}(\mathbb{R}^n)$ is a
Frechet space with the topology defined by the seminorms
$p_{\chi}(u)=\|\chi u \|_{W^{s,p}(\mathbb{R}^n)}$ for $\chi \in
C_c^{\infty}(\mathbb{R}^n)$ \cite{9}. Also
$C^{\infty}(\mathbb{R}^n)$ is dense in
$W^{s,p}_{loc}(\mathbb{R}^n)$.
% \cite{10}.
\begin{theorem}\lab{thmA1}\cite{2,dM06,dM05b,5,6,15,16}
Let $p_1, p_2, p, q\in (1,\infty)$, $\delta, \delta_1, \delta_2,
\delta' \in \mathbb{R}$.
\begin{enumerate}
\item If $p\geq q$ and $\delta'< \delta$ then $L^P_{\delta'}(\mathbb{R}^n)\subseteq
L^q_{\delta}(\mathbb{R}^n)$ is continuous.
\item For $s \geq s'$ and $\delta \leq \delta'$ the inclusion $W^{s,p}_{\delta}(\mathbb{R}^n)\subseteq
W^{s',p}_{\delta'}(\mathbb{R}^n)$ is continuous.
%(If $s'<0$, then $p\neq 1$)
\item For $s > s'$ and $\delta
< \delta'$ the inclusion $W^{s,p}_{\delta}(\mathbb{R}^n)\subseteq
W^{s',p}_{\delta'}(\mathbb{R}^n)$ is compact.
%(If $s'<0$, then $p\neq 1$)
\item If  $0\leq sp < n$ then $W^{s,p}_{\delta}(\mathbb{R}^n)\subseteq
L^{r}_{\delta}(\mathbb{R}^n)$ is continuous for every $r$ with
$\frac{1}{p}-\frac{s}{n}\leq \frac{1}{r}\leq \frac{1}{p}$.
\item If $sp = n$ then $W^{s,p}_{\delta}(\mathbb{R}^n)\subseteq
L^{r}_{\delta}(\mathbb{R}^n)$ is continuous for every $r\geq p$.
\item If $sp > n$ then $W^{s,p}_{\delta}(\mathbb{R}^n)\subseteq
L^{r}_{\delta}(\mathbb{R}^n)$ is continuous for every $r\geq p$.
Moreover $W^{s,p}_{\delta}(\mathbb{R}^n)\subseteq
C^0_{\delta}(\mathbb{R}^n)$ is continuous where
$C^0_{\delta}(\mathbb{R}^n)$ is the set of continuous functions
$f:\mathbb{R}^n \rightarrow \mathbb{R}$ for which
$\|f\|_{C^0_{\delta}}:=\sup_{x\in
\mathbb{R}^n}(\langle x\rangle^{-\delta}|f|)<\infty$.
\item If $\frac{1}{r}=\frac{1}{p_1}+\frac{1}{p_2}< 1$, then
pointwise multiplication is a continuous bilinear map
$L^{p_1}_{\delta_1}(\mathbb{R}^n) \times
L^{p_2}_{\delta_2}(\mathbb{R}^n) \rightarrow
L^r_{\delta_1+\delta_2}(\mathbb{R}^n)$.
\item Pointwise multiplication is a continuous bilinear map $C^0_{\delta_1}(\mathbb{R}^n)\times L^{p}_{\delta_2}(\mathbb{R}^n)\rightarrow
L^{p}_{\delta_1+\delta_2}(\mathbb{R}^n)$.
\item For $s\in \mathbb{R}$ (and $p\in (1,\infty)$),
$W^{s,p}_{\delta}(\mathbb{R}^n)$is a reflexive space and
$(W^{s,p}_{\delta}(\mathbb{R}^n))^{*}=W^{-s,p'}_{-n-\delta}(\mathbb{R}^n)$.
\item \textbf{Real Interpolation}: Suppose $\theta \in (0,1)$. If
\begin{equation*}
s=(1-\theta)s_0+\theta s_1, \quad \quad
\frac{1}{p}=\frac{1-\theta}{p_0}+\frac{\theta}{p_1}, \quad \quad
\delta=(1-\theta)\delta_0+\theta \delta_1
\end{equation*}
then
$W^{s,p}_{\delta}(\mathbb{R}^n)=(W^{s_0,p_0}_{\delta_0}(\mathbb{R}^n),
W^{s_1,p_1}_{\delta_1}(\mathbb{R}^n))_{\theta,p}$ unless
$s_0,s_1\in \mathbb{R}$ with $s_0\neq s_1$ and $s\in \mathbb{Z}$.
In the case where $s_0$ and $s_1$ are not both positive and
exactly one of $s_0$ and $s_1$ is an integer, we additionally
assume that $p_0=p_1=p$.
\item \textbf{Complex Interpolation}: Suppose $\theta \in (0,1)$. If
\begin{equation*}
s=(1-\theta)s_0+\theta s_1, \quad \quad
\frac{1}{p}=\frac{1-\theta}{p_0}+\frac{\theta}{p_1}, \quad \quad
\delta=(1-\theta)\delta_0+\theta \delta_1
\end{equation*}
then
$W^{s,p}_{\delta}(\mathbb{R}^n)=[W^{s_0,p_0}_{\delta_0}(\mathbb{R}^n),
W^{s_1,p_1}_{\delta_1}(\mathbb{R}^n)]_{\theta}$ provided $s_0,
s_1, s\in \mathbb{Z}$ or $s_0,
s_1, s \not \in \mathbb{Z}$.

\textbf{Note: The above interpolation facts do not say anything
about the case where $s_0\, or \,s_1 \in \mathbb{R}\setminus
\mathbb{Z}$ and $s\in \mathbb{Z}$.}
\item $C_c^{\infty}(\mathbb{R}^n)$ is dense in
$W^{s,p}_{\delta}(\mathbb{R}^n)$ for all $s\in \mathbb{R}$.
\end{enumerate}
\end{theorem}
\begin{remark}\lab{reminfinitybound}
We define $L^{\infty}_{\delta}(\mathbb{R}^n)$ as follows: $ f\in
L^{\infty}_{\delta}(\mathbb{R}^n)\Leftrightarrow \langle x\rangle^{-\delta}f
\in L^{\infty}(\mathbb{R}^n)$. We equip this space with the norm
$\|
f\|_{L^{\infty}_{\delta}(\mathbb{R}^n)}:=\|
\langle x\rangle^{-\delta}f
\|_{L^{\infty}(\mathbb{R}^n)}$.  More generally, for all $k\in
\mathbb{N}_0$
\begin{align*}
& W^{k,\infty}_{\delta}(\mathbb{R}^n)
:=\{u\in L^{\infty}_{\delta}(\mathbb{R}^n)
  : \partial^{\alpha} u \in L^{\infty}_{\delta-|\alpha|}(\mathbb{R}^n)
    \quad \forall \,\, |\alpha|\leq k\},
\\
& \| u\|_{W^{k,\infty}_{\delta}(\mathbb{R}^n)}
=\sum_{|\alpha|\leq k}
\| \partial^{\alpha}u\|_{L^{\infty}_{\delta-|\alpha|}(\mathbb{R}^n)}.
\end{align*}
It is easy to show that $C^0_{\delta}(\mathbb{R}^n)$ is a subspace
of $L^{\infty}_{\delta}(\mathbb{R}^n)$, pointwise multiplication
is a continuous bilinear map
$L^{\infty}_{\delta_1}(\mathbb{R}^n)\times
L^{p}_{\delta_2}(\mathbb{R}^n)\rightarrow
L^{p}_{\delta_1+\delta_2}(\mathbb{R}^n)$ and the inclusion
$L^{\infty}_{\tilde{\delta}}(\mathbb{R}^n)\subseteq
L^{p}_{\delta}(\mathbb{R}^n)$ is continuous for
$\tilde{\delta}<\delta$ and $p\in (1,\infty)$ \cite{5}. Also if
$sp>n$, then the inclusions
$W^{s,p}_{\delta}(\mathbb{R}^n)\subseteq
C^0_{\delta}(\mathbb{R}^n)\subseteq L^{\infty}_{\delta}(\mathbb{R}^n)$ are continuous.

Note that if we let
$r:=\langle x\rangle=(1+|x|^2)^{\frac{1}{2}}$, then for $u \in
L^{\infty}_{\delta}(\mathbb{R}^n)$ we have
\begin{equation*}
\| u \|_{L^{\infty}_{\delta}(\mathbb{R}^n)}=
\esssup_{x\in \mathbb{R}^n} (r^{-\delta}|u|) \Longrightarrow
|u|\leq r^{\delta}\| u
\|_{L^{\infty}_{\delta}(\mathbb{R}^n)}\,\, \textrm{a.e}.
\end{equation*}
\end{remark}
\begin{definition}
Let $\Omega$ be an open subset of $\mathbb{R}^n$.
$W^{s,p}_{\delta}(\Omega)$ is defined as the restriction of
$W^{s,p}_{\delta}(\mathbb{R}^n)$ to $\Omega$ and is equipped with
the following norm:
\begin{equation*}
\|u\|_{W^{s,p}_{\delta}(\Omega)}=\inf_{v\in
W^{s,p}_{\delta}(\mathbb{R}^n),
v|_{\Omega}=u}\|v\|_{W^{s,p}_{\delta}(\mathbb{R}^n)}.
\end{equation*}
\end{definition}
%%%%%%%%%%%%%%%%%%%%%%%%%%%%%%%%%%%%%%%
%%%%%%%%%%%%%%%%%%%%%%%%%%%%%%%%%%%%%%%
When there is no ambiguity about the domain we may write
\begin{itemize}
\item $W^{s,p}$ instead of $W^{s,p}(\Omega)$,
\item $W^{s,p}_{\delta}$ instead of $W^{s,p}_{\delta}(\Omega)$,
\item $\| .\|_{W^{s,p}}$ or $\|\cdot\|_{s,p}$ instead of
$\|\cdot\|_{W^{s,p}(\Omega)}$,
\item $\| .\|_{W^{s,p}_{\delta}}$ or $\|\cdot\|_{s,p,\delta}$ instead of
$\|\cdot\|_{W^{s,p}_{\delta}(\Omega)}$.
\end{itemize}
%$\|\cdot\|_{s,p}$ instead of $\|\cdot\|_{W^{s,p}(\Omega)}$, and
%$\|\cdot\|_{s,p,\delta}$ instead of $\|\cdot\|_{W^{s,p}_{\delta}(\Omega)}$
% position of removed theorem and lemma
%%%%%%%%%%%%%%%%%%%%%%%%%%%%%%%%%%%%%%%
\begin{definition}\lab{defweightedsobolevae}
Let $(M,h)$ be an n-dimensional AF manifold of class
$W^{\alpha,\gamma}_{\rho}$.
In addition, let $\{(U_i,\phi_i)\}_{i=1}^{m}$ be
the collection of end charts. We can extend this set to an atlas
$\{(U_i,\phi_i)\}_{i=1}^{k}$ such that for $i>m$ the set
$\bar{U}_i$ is compact and $\phi_i(U_i)=B_1:=\{x\in \mathbb{R}^n:
|x|<1\}$. Let $\{\chi_i\}_{i=1}^k$ be a partition of unity
subordinate to the cover $\{U_i\}_{i=1}^k$. The weighted Sobolev
space $W^{s,p}_{\delta}(M)$ is the subset of $W^{s,p}_{loc}(M)$
consisting of functions $u$ that satisfy
%\begin{enumerate}
%\item $u\circ \phi_i^{-1}\in
%W^{s,p}_{\delta}(E_1)$ for all $1\leq i \leq m$
%\item
\begin{equation*} \|u\|_{W^{s,p}_{\delta}(M)}:=\sum_{i=1}^m
\|(\phi_i^{-1})^{*}(\chi_i
u)\|_{W^{s,p}_{\delta}(\mathbb{R}^n)}+\sum_{i=m+1}^k
\|(\phi_i^{-1})^{*}(\chi_i u)\|_{W^{s,p}(B_1)}<\infty
\end{equation*}
%\end{enumerate}
 The collection $\{(U_i,\phi_i)\}_{i=1}^{k}$ is called
an \textbf{AF atlas} for $M$.
\end{definition}
\begin{remark}
 The above definition of $W^{s,p}_{\delta}(M)$ does not depend
on the metric $h$ and its class and it is also independent of the
chosen partition of unity, but it is based on the specific charts
that were introduced in the definition of AF manifolds. This
definition is not necessarily coordinate independent (of course
see Remark \ref{remextraconditionaf}). Indeed, as for the case of
compact manifolds, one can easily show that different choices for
$\{U_i, \phi_i\}_{i=m+1}^{k}$ result in equivalent norms; but the
dependence of the norm on the end charts is more critical. In
what follows we always assume that one fixed AF atlas is given
and we just work with that fixed atlas.
\end{remark}
\begin{remark}
Let $\pi: E\rightarrow M$ be a smooth vector bundle over $M$.
Completely analogous to Definition \ref{defweightedsobolevae},
one can define the Sobolev space $W^{s,p}_{\delta}(E)$ of
sections of $E$ by using a finite trivializing cover of coordinate
charts and a partition of unity subordinate to the cover.
\end{remark}
\begin{remark}\lab{remgeodesic}
By using partition of unity arguments one can prove all the items
in Theorem \ref{thmA1} for AF manifolds (see below; also for item
9. there are several ways to construct an isomorphism between
 $(W^{s,p}_{\delta}(M))^{*}$ and $W^{-s,p'}_{-n-\delta}(M)$, see our discussion about duality
pairing in Appendix B). Of course note that for instance we have
$\|f\|_{C^0_{\delta}(M)}:=\sup_{x\in
M}([(1+|x|^2)^{\frac{1}{2}}]^{-\delta}|f|)$, where $|x|$ is the
geodesic distance from $x$ to a fixed point $O$ in the compact
core. As opposed to $\mathbb{R}^n$, in a general Riemannian
manifold $|x|^2$ is not smooth, so there is no advantage in using
$(1+|x|^2)^{\frac{1}{2}}$ instead of for example $1+|x|$. In the
literature the norms $\|f\|_{C^0_{\delta}(M)}:=\sup_{x\in
M}((1+|x|)^{-\delta}|f|)$ and
$\|f\|_{L^{\infty}_{\delta}(M)}=\|(1+|x|)^{-\delta}f\|_{\infty}$
have also been used for $C^{0}_{\delta}(M)$ and
$L^{\infty}_{\delta}(M)$, respectively. Clearly these norms are
equivalent to the original ones.
\end{remark}
\begin{remark}\lab{remextraconditionaf}
Item (3) in the definition of AF manifolds (Definition \ref{defAE})
guarantees that $L^{p}_{\delta}(M)$ is independent of the chosen
$AF$ atlas and in fact $\| u\|_{L^{p}_{\delta}(M)}$
agrees with the following norm that uses the natural volume form
of $M$ \cite{5} :
\begin{equation*}
\| u \|_{L^p_{\delta}(M)}=\|
\langle x\rangle^{-\delta-\frac{n}{p}}u\|_{L^p(M)}.\quad \big(\|
u\|_{L^p(M)}= (\int_{M}|u|^p dV_h)^{\frac{1}{p}}\big).
\end{equation*}
Of course it is not necessary to single out weighted Lebesgue
spaces and require their definition to be coordinate independent.
One may choose to treat the spaces $L^{p}_{\delta}(M)$ as general
$W^{s,p}_{\delta}(M)$ spaces are treated. This is the reason why
in some of the literature item (3) in Definition \ref{defAE} is not
considered as part of the definition.
\end{remark}
Here we just show two of the previously stated facts for weighted
spaces on $\mathbb{R}^n$ are also true for weighted spaces on AF
manifolds. The other items in Theorem \ref{thmA1} and Remark
\ref{reminfinitybound} can be proved for AF manifolds in a
similar way.
\begin{itemizeX}
\item \textbf{Continuous Embedding:} For $s\geq s'$ and $\delta \leq
\delta'$the inclusion $W^{s,p}_{\delta}(M)\subseteq
W^{s',p}_{\delta'}(M)$ is continuous:
\begin{align*}
\|u\|_{W^{s',p}_{\delta'}(M)}&=\sum_{i=1}^m
\|(\phi_i^{-1})^{*}(\chi_i
u)\|_{W^{s',p}_{\delta'}(\mathbb{R}^n)}+\sum_{i=m+1}^k
\|(\phi_i^{-1})^{*}(\chi_i u)\|_{W^{s',p}(B_1)}\\
& \preceq \sum_{i=1}^m \|(\phi_i^{-1})^{*}(\chi_i
u)\|_{W^{s,p}_{\delta}(\mathbb{R}^n)}+\sum_{i=m+1}^k
\|(\phi_i^{-1})^{*}(\chi_i u)\|_{W^{s,p}(B_1)}\\
&=\|u\|_{W^{s,p}_{\delta}(M)}.
\end{align*}
\item \textbf{Compact Embedding:} For $s>s'$ and $\delta<\delta'$
the inclusion $W^{s,p}_{\delta}(M)\subseteq W^{s',p}_{\delta'}(M)$
is compact:\\
Let $\{u_j\}$ be a bounded sequence in $W^{s,p}_{\delta}$:
$\|u_j\|_{W^{s,p}_{\delta}}\leq \tilde{M}$. We must prove that
there exists a subsequence of $\{u_j\}$ that is Cauchy in
$W^{s',p}_{\delta'}$ (recall that $W^{s',p}_{\delta'}$ is
complete).
\begin{equation*}
\tilde{M}\geq \|u_j\|_{W^{s,p}_{\delta}}=\sum_{i=1}^m
\|(\phi_i^{-1})^{*}(\chi_i
u_j)\|_{W^{s,p}_{\delta}(\mathbb{R}^n)}+\sum_{i=m+1}^k
\|(\phi_i^{-1})^{*}(\chi_i u_j)\|_{W^{s,p}(B_1)}.
\end{equation*}
Therefore
\begin{align*}
\begin{cases}
\forall\,\, 1\leq i \leq m \quad \forall j \quad
\|(\phi_i^{-1})^{*}(\chi_i
u_j)\|_{W^{s,p}_{\delta}(\mathbb{R}^n)}\leq \tilde{M} ,\\
\forall\,\, m+1\leq i \leq k \quad \forall j \quad
\|(\phi_i^{-1})^{*}(\chi_i u_j)\|_{W^{s,p}(B_1)}\leq \tilde{M}.
\end{cases}
\end{align*}
Since $W^{s,p}_{\delta}(\mathbb{R}^n)\hookrightarrow
W^{s',p}_{\delta'}(\mathbb{R}^n)$ and $W^{s,p}(B_1)\hookrightarrow
W^{s',p}(B_1)$ are compact (by Theorem \ref{thmA1} and
Rellich-Kondrachov theorem, respectively), we can conclude that
\begin{align*}
\begin{cases}
\forall\,\, 1\leq i \leq m, \exists \textrm{ a
subsequence of } \{(\phi_i^{-1})^{*}(\chi_i
u_j)\}_{j=1}^{\infty} \textrm{ that converges in }
W^{s',p}_{\delta'}(\mathbb{R}^n),\\
\forall\,\, m+1\leq i \leq k, \exists \textrm{ a
subsequence of } \{(\phi_i^{-1})^{*}(\chi_i
u_j)\}_{j=1}^{\infty} \textrm{ that converges in }
W^{s',p}(B_1).
\end{cases}
\end{align*}
In fact, by a diagonalization argument we can construct a
subsequence $\{v_j\}$ that converges in the corresponding space
for all $1\leq i \leq k$ (Start with $i=1$ and find a subsequence
that converges. Then for $i=2$ find a subsequence from the
preceding subsequence that converges and so on. At each step we
find a subsequence of the preceding subsequence). So
\begin{align*}
\begin{cases}
\forall \,\, 1\leq i \leq m \quad \{(\phi_i^{-1})^{*}(\chi_i
v_j)\}_{j=1}^{\infty} \quad \textrm{converges in
$W^{s',p}_{\delta'}(\mathbb{R}^n)$},\\
\forall m+1\leq i \leq k \quad \{(\phi_i^{-1})^{*}(\chi_i
v_j)\}_{j=1}^{\infty} \quad \textrm{converges in $W^{s',p}(B_1)$}.
\end{cases}
\end{align*}
We claim that $\{v_j\}$ is Cauchy in $W^{s',p}_{\delta'}(M)$. Let
$\epsilon>0$ be given. For each $1\leq i \leq m$, let $N_i$ be
such that if $l, \tilde{l}>N_i$ then
\begin{equation*}
\|(\phi_i^{-1})^{*}(\chi_i v_{\tilde{l}})-(\phi_i^{-1})^{*}(\chi_i
v_l)\|_{W^{s',p}_{\delta'}(\mathbb{R}^n)}<\frac{\epsilon}{k}.
\end{equation*}
Also for each $m+1\leq i \leq k$, let $N_i$ be such that if $l,
\tilde{l}>N_i$ then
\begin{equation*}
\|(\phi_i^{-1})^{*}(\chi_i v_{\tilde{l}})-(\phi_i^{-1})^{*}(\chi_i
v_l)\|_{W^{s',p}(B_1)}<\frac{\epsilon}{k}.
\end{equation*}
Now let $N=\max\{N_1,...,N_k\}$. Clearly for all $l, \tilde{l}>N$
we have
\begin{align*}
\|v_l-v_{\tilde{l}}\|_{W^{s',p}_{\delta'}(M)}
& =\sum_{i=1}^m \|(\phi_i^{-1})^{*}(\chi_i
(v_l-v_{\tilde{l}}))\|_{W^{s',p}_{\delta'}(\mathbb{R}^n)}
\\
& \quad \quad +\sum_{i=m+1}^k \|(\phi_i^{-1})^{*}(\chi_i (v_l-v_{\tilde{l}}))\|_{W^{s',p}(B_1)}
\\
& < k\frac{\epsilon}{k}=\epsilon.
\end{align*}
This proves that $\{v_j\}$ is Cauchy in $W^{s',p}_{\delta'}(M)$.
\end{itemizeX}
\begin{theorem}\cite{7,jZ77}\lab{thmembedunweight}
If $s_1-\frac{n}{p_1}\geq s_0-\frac{n}{p_0}$, $1 <p_1\leq p_0<
\infty$, $s_1\geq s_0\geq0$, then
$W^{s_1,p_1}(\Omega)\hookrightarrow W^{s_0,p_0}(\Omega)$ (that
is, $W^{s_1,p_1}(\Omega)\subseteq W^{s_0,p_0}(\Omega)$ and the
inclusion map is continuous).
\end{theorem}
\begin{remark}
Note that if $\Omega$ is a bounded domain, then the restriction
$p_1 \leq p_0$ can be removed. Indeed,  if $\Omega$ is bounded and
$p_1>p_0$ then $L^{p_1}\subseteq L^{p_0}$ and consequently if
$k\geq l \in \mathbb{N}_0$ then $W^{k,p_1}(\Omega)\hookrightarrow
W^{k,p_0}(\Omega)\hookrightarrow W^{l,p_0}(\Omega)$. The claim
can be proved by interpolation for the cases where $s_0$ or $s_1$
are not integers (the details are similar to the proof of Lemma
\ref{lemA11} below).
\end{remark}
\begin{lemma}\lab{lemA4}
Let $k\in \mathbb{N}_0$, $\delta \in \mathbb{R}$ and
$p\in(1,\infty)$. Then
\begin{equation*}
u\in W^{k,p}_{\delta}(\mathbb{R}^n) \Longleftrightarrow
\partial^{\alpha} u \in L^p_{\delta-|\alpha|}(\mathbb{R}^n) \quad \forall \,
|\alpha| \leq k .
\end{equation*}
\end{lemma}
\begin{proof}{\bf (Lemma~\ref{lemA4})}
The case $k=0$ is obvious. In general we have
\begin{align*}
u\in W^{k,p}_{\delta}&\Longleftrightarrow
\|u\|_{W^{k,p}_{\delta}}<\infty \Longleftrightarrow \forall\,
|\alpha|\leq k\quad
\|\langle x\rangle^{-\delta-\frac{n}{p}+|\alpha|}\partial^{\alpha}u\|_{L^p}<\infty\\
&\Longleftrightarrow \forall\, |\alpha|\leq k\quad
\|\langle x\rangle^{-(\delta-|\alpha|)-\frac{n}{p}}\partial^{\alpha}u\|_{L^p}<\infty\\
&  \Longleftrightarrow \forall \, |\alpha| \leq k \quad
\partial^{\alpha} u \in L^p_{\delta-|\alpha|}.
\end{align*}
\end{proof}
\begin{lemma}\lab{lemA11}
Let $s\in \mathbb{R}$, $p,q\in (1,\infty)$. If $p\geq q$ and
$\delta'<\delta$, then
$W^{s,p}_{\delta'}(\mathbb{R}^n)\hookrightarrow
W^{s,q}_{\delta}(\mathbb{R}^n)$.
\end{lemma}
\begin{proof}{\bf (Lemma~\ref{lemA11})}
We consider three cases:
\begin{itemizeX}
\item \textbf{Case 1:} $s=k\in \mathbb{N}_0$.\\
\begin{align*}
u\in W^{k,p}_{\delta'} &\Rightarrow \forall \, |\alpha| \leq k
\quad \partial^{\alpha} u \in
L^p_{\delta'-|\alpha|}\\
&\Rightarrow \forall \, |\alpha| \leq k \quad \partial^{\alpha} u
\in L^q_{\delta-|\alpha|} \quad (\textrm{by item1. of Theorem
\ref{thmA1}})\\
&\Rightarrow u\in W^{k,q}_{\delta}.
\end{align*}
In fact,
\begin{align*}
\| u\|_{k,q,\delta} &=\sum_{|\beta|\leq k}\|
\langle x\rangle^{-\delta-\frac{n}{p}+|\beta|}\partial^{\beta}u\|_{L^{q}(\mathbb{R}^n)}=\sum_{|\beta|\leq
k}\|
\langle x\rangle^{-(\delta-|\beta|)-\frac{n}{p}}\partial^{\beta}u\|_{L^{q}(\mathbb{R}^n)}\\
&=\sum_{|\beta|\leq k}\|
 \partial^{\beta}u\|_{L^{q}_{\delta-|\beta|}(\mathbb{R}^n)}\preceq \sum_{|\beta|\leq k}\|
 \partial^{\beta}u\|_{L^{p}_{\delta'-|\beta|}(\mathbb{R}^n)} \quad (L^{p}_{\delta'-|\beta|}\hookrightarrow L^{q}_{\delta-|\beta|})\\
 &=\sum_{|\beta|\leq k}\|
\langle x\rangle^{-\delta'-\frac{n}{p}+|\beta|}\partial^{\beta}u\|_{L^{p}(\mathbb{R}^n)}=\|
u\|_{k,p,\delta'}.
\end{align*}
\item \textbf{Case 2:} $s\geq 0$, $s\not \in \mathbb{N}_0$.

Let $k=\floor{s}$, $\theta=s-k$. By what was proved in the
previous case
\begin{equation*}
W^{k,p}_{\delta'}\hookrightarrow W^{k,q}_{\delta},\quad
W^{k+1,p}_{\delta'}\hookrightarrow W^{k+1,q}_{\delta}.
\end{equation*}
Since $s=(1-\theta) k+ \theta (k+1)$, the claim follows from real
interpolation.
\item \textbf{Case 3:} $s<0$.\\
By assumption $p\geq q$ and $\delta'< \delta$, therefore
\begin{equation*}
p'\leq q', \quad -n-\delta'>-n-\delta.
\end{equation*}
Here $p'$ and $q'$ are the conjugates of $p$ and $q$,
respectively. Thus by what was proved in the previous cases we
have
\begin{equation*}
W^{-s,q'}_{-n-\delta}\hookrightarrow W^{-s,p'}_{-n-\delta'}.
\end{equation*}
The result follows by taking the dual.
\end{itemizeX}
\end{proof}
%%%%%%%%%%%%%%%%%%%%%%%%%%%%%%%%%%%%%%%
%%%%%%%%%%%%%%%%%%%%%%%%%%%%%%%%%%%%%%%

\begin{lemma}\lab{coroA3}
Let the following assumptions hold:
\begin{enumerate}[(i)]
\item $1 < p \leq r < \infty$,
\item $t,s \in \mathbb{R}$ with $0\leq t \leq s$,
\item $s-\frac{n}{p}\geq t-\frac{n}{r}$.
\end{enumerate}
Then: For all $\delta \in \mathbb{R}$
$W^{s,p}_{\delta}\hookrightarrow W^{t,r}_{\delta}$.
\end{lemma}
\begin{proof}{\bf (Lemma~\ref{coroA3})}
%\textbf{Method 1:} by INTERPOLATION \textbf{WARINING:
%I have not proved this yet}\\
%\textbf{Method 2:}
%We use the following fact:\\
%Fact: Let $\{a_j\}_{j=1}^{\infty}$ be a sequence of nonnegative
%numbers. Then the function $f:(0,\infty)\rightarrow \mathbb{R}$
%defined by
%$f(\alpha)=\big(\sum_{j=1}^{\infty}a_j^{\alpha}\big)^{\frac{1}{\alpha}}$
%is a decreasing function.\\
In the proof we use the fact that if $1\leq \alpha\leq\beta$, then
$l^{\alpha}\hookrightarrow l^{\beta}$ ($l^{\alpha}$ denotes the
space of $\alpha$-power summable sequences); in fact for any
sequence $a=\{a_j\}$, $\|a\|_{l^\beta}\leq \|a\|_{l^\alpha}$.
From the assumption it follows that $W^{s,p}\hookrightarrow
W^{t,r}$ and so
\begin{align*}
\|u\|_{t,r,\delta}&=\big[\sum_{j=0}^{\infty}2^{-r\delta j}
\|S_{2^j}(\varphi_j u)\|_{t,r}^r\big]^{\frac{1}{r}}\preceq
\big[\sum_{j=0}^{\infty}2^{-r\delta j} \|S_{2^j}(\varphi_j
u)\|_{s,p}^r\big]^{\frac{1}{r}}
\\
&=\big[\sum_{j=0}^{\infty}(2^{-\delta
j} \|S_{2^j}(\varphi_j u)\|_{s,p})^r\big]^{\frac{1}{r}}
\leq \big[\sum_{j=0}^{\infty}(2^{-\delta j} \|S_{2^j}(\varphi_j
u)\|_{s,p})^p\big]^{\frac{1}{p}}
\\
& =\|u\|_{s,p,\delta} \quad
(\textrm{Note that $p\leq r$ and so $\|\cdot\|_{l^r}\leq
\|\cdot\|_{l^p}$} ).
\end{align*}
\end{proof}
\begin{theorem}[Embedding Theorem I]\lab{thmA3}
Let the following assumptions hold:
\begin{enumerate}[(i)]
\item $1 < p \leq r < \infty$,
\item $t,s \in \mathbb{R}$ with $ t \leq s$,
\item $s-\frac{n}{p}\geq t-\frac{n}{r}$.
\end{enumerate}
Then: If $\delta'\leq \delta$ then
%For all $\delta \in \mathbb{R}$
$W^{s,p}_{\delta'}\hookrightarrow W^{t,r}_{\delta}$.
\end{theorem}
\begin{proof}{\bf (Theorem~\ref{thmA3})}
Note that, since $\delta'\leq\delta$,
$W^{s,p}_{\delta'}\hookrightarrow W^{s,p}_{\delta}$, so we just
need to show that $W^{s,p}_{\delta}\hookrightarrow
W^{t,r}_{\delta}$. By Lemma
 \ref{coroA3} we know that the claim is true for the case $0\leq
t$. So we just need to consider the case where $t<0$.
\begin{itemizeX}

\medskip
\item \textbf{Case 1:} $t<0, s\leq 0$

\smallskip
It is enough to show that $(W^{t,r}_{\delta})^{*}\hookrightarrow
(W^{s,p}_{\delta})^{*}$, that is, we need to prove that
\begin{equation*}
W^{-t,r'}_{-n-\delta}\hookrightarrow W^{-s,p'}_{-n-\delta}.
\end{equation*}
Note that $-t$ and $-s$ are nonnegative so we just need to check
that the assumptions of Lemma \ref{coroA3} hold true:
\begin{align*}
&t\leq s \leq 0 \Rightarrow 0\leq -s\leq -t \\
&1< p \leq r \Rightarrow 1< r' \leq p'\\
&s-\frac{n}{p}\geq t-\frac{n}{r} \Rightarrow s+\frac{n}{r}-n\geq
t+\frac{n}{p}-n \Rightarrow s-\frac{n}{r'}\geq t-\frac{n}{p'}
\Rightarrow -t-\frac{n}{r'}\geq -s-\frac{n}{p'}.
\end{align*}
\medskip
\item \textbf{Case 2:} $t<0, s> 0$

\smallskip
In this case we will prove that there exists $q\geq 1$ such that
\begin{equation*}
W^{s,p}_{\delta}\hookrightarrow L^q_{\delta}\hookrightarrow
W^{t,r}_{\delta}
\end{equation*}
By what was proved previously, in order to make sure that the
above inclusions hold true it is enough to find $q$ such that
\begin{align*}
& t-\frac{n}{r}\leq 0-\frac{n}{q} \leq s-\frac{n}{p} \quad (\Leftrightarrow -\frac{s}{n}+\frac{1}{p}\leq \frac{1}{q}\leq -\frac{t}{n}+\frac{1}{r})\\
&p\leq q \leq r \quad (\Leftrightarrow \frac{1}{r}\leq
\frac{1}{q}\leq \frac{1}{p})
\end{align*}
Note that by assumption $-\frac{s}{n}+\frac{1}{p}\leq
-\frac{t}{n}+\frac{1}{r}$. If $-\frac{s}{n}+\frac{1}{p}=
-\frac{t}{n}+\frac{1}{r}$, then $q$ defined by
$\frac{1}{q}=-\frac{s}{n}+\frac{1}{p}(=
-\frac{t}{n}+\frac{1}{r})$ clearly satisfies the desired
conditions. So it remains to consider the case where
$-\frac{s}{n}+\frac{1}{p}< -\frac{t}{n}+\frac{1}{r}$. The
inequalities in the first line are satisfied if and only if
\begin{equation*}
\frac{1}{q}=-\frac{s}{n}+\frac{1}{p}+\sigma
(\frac{s-t}{n}+\frac{1}{r}-\frac{1}{p}).
\end{equation*}
for some $\sigma\in [0,1]$. The question is ``can we choose
$\sigma$ so that the above expression lies between $\frac{1}{r}$
and $\frac{1}{p}$?'' We want to find $\sigma \in [0,1]$ such that
\begin{equation*}
\frac{1}{r}\leq -\frac{s}{n}+\frac{1}{p}+\sigma
(\frac{s-t}{n}+\frac{1}{r}-\frac{1}{p})\leq \frac{1}{p}
\end{equation*}
That is we want to find $\sigma \in [0,1]$ such that
\begin{equation*}
\frac{\frac{1}{r}-\frac{1}{p}+\frac{s}{n}}{\frac{s-t}{n}+\frac{1}{r}-\frac{1}{p}}
\leq \sigma \leq
\frac{\frac{s}{n}}{\frac{s-t}{n}+\frac{1}{r}-\frac{1}{p}}.
\end{equation*}
Note that since $\frac{1}{r}\leq \frac{1}{p}$ clearly
\begin{equation*}
\frac{\frac{1}{r}-\frac{1}{p}+\frac{s}{n}}{\frac{s-t}{n}+\frac{1}{r}-\frac{1}{p}}
\leq \frac{\frac{s}{n}}{\frac{s-t}{n}+\frac{1}{r}-\frac{1}{p}}.
\end{equation*}
So it would be possible to find $\sigma$ if and only if
\begin{equation*}
\frac{\frac{s}{n}}{\frac{s-t}{n}+\frac{1}{r}-\frac{1}{p}} \geq 0
\quad \textrm{and} \quad
\frac{\frac{1}{r}-\frac{1}{p}+\frac{s}{n}}{\frac{s-t}{n}+\frac{1}{r}-\frac{1}{p}}\leq
1.
\end{equation*}
The first inequality is true because by assumption $s>0$ and
$s-\frac{n}{p}\geq t-\frac{n}{r}$. The second inequality is true
because by assumption $t<0$ and
\begin{equation*}
\frac{\frac{1}{r}-\frac{1}{p}+\frac{s}{n}}{\frac{s-t}{n}+\frac{1}{r}-\frac{1}{p}}\leq
1 \Leftrightarrow \frac{1}{r}-\frac{1}{p}+\frac{s}{n}\leq
\frac{s-t}{n}+\frac{1}{r}-\frac{1}{p}\Leftrightarrow
\frac{t}{n}\leq 0.
\end{equation*}
\end{itemizeX}
\end{proof}
\begin{theorem}[Embedding Theorem II]\lab{thmA4}
Let the following assumptions hold:
\begin{enumerate}[(i)]
\item $1 < p , r < \infty$,
\item $t,s \in \mathbb{R}$ with $ t \leq s$,
\item $s-\frac{n}{p}\geq t-\frac{n}{r}$,
\item $\delta'$ is \textbf{strictly} less than $\delta$.
\end{enumerate}
Then: %If $\delta'< \delta$ then
$W^{s,p}_{\delta'}\hookrightarrow W^{t,r}_{\delta}$. (Note that if
$p>r$, then the third assumption follows from the second
assumption.)
\end{theorem}
\begin{proof}{\bf (Theorem~\ref{thmA4})}
If $p\leq r$, then the claim follows from Theorem
\ref{thmA3}. Let's assume $p>r$. Then by Lemma \ref{lemA11} we
have $W^{s,p}_{\delta'}\hookrightarrow W^{s,r}_{\delta}$ and by
Theorem \ref{thmA1} we have $W^{s,r}_{\delta}\hookrightarrow
W^{t,r}_{\delta}$. Consequently $W^{s,p}_{\delta'}\hookrightarrow
W^{t,r}_{\delta}$.
\end{proof}
%%%%%%%%%%%%%%%%%%%%%%%%%%%%%%%%%%%%%%%
\begin{lemma}[Multiplication by bounded smooth
functions]\lab{lemA9} Let $\sigma \in \mathbb{R}$,
$q\in[1,\infty)$ (if $\sigma<0, q\neq 1$). Let $N=\lceil |\sigma|
\rceil$. If $f\in BC^{N}(\mathbb{R}^n) $ and $u\in
W^{\sigma,q}(\mathbb{R}^n)$, then  $fu \in
W^{\sigma,q}(\mathbb{R}^n)$ and moreover
$\|fu\|_{\sigma,q}\preceq \|u\|_{\sigma,q}$ (the implicit
constant depends on $f$ but it does not depend on $u$).
\end{lemma}
\begin{proof}{\bf (Lemma~\ref{lemA9})}
The proof consists of four steps:
\begin{itemizeX}
\item \textbf{Step 1:} $\sigma=k\in \mathbb{N}_0$.
The claim is proved in \cite{13}.

\item \textbf{Step 2:} $0<\sigma<1$.
The claim has been proved in \cite{12} for the case
where $\sigma \in (0,1)$, $f$ is Lipschitz continuous and $0\leq f
\leq 1$. With an obvious modification that proof also works for
the case where $f\in BC^1(\mathbb{R}^n)$.

\item \textbf{Step 3:} $1<\sigma\not \in \mathbb{N}$.
In this case we can proceed as follows: Let
$k=\floor{\sigma}$, $\theta=\sigma-k$.
\begin{align*}
\|fu\|_{\sigma,q}&=\|fu\|_{k,q}+\sum_{|\nu|=k}\|\partial^{\nu}(fu)\|_{\theta,q}\\
&\leq \|fu\|_{k,q}+\sum_{|\nu|=k}\sum_{\beta\leq
\nu}\|\partial^{\nu-\beta}f  \partial^{\beta}u\|_{\theta,q}\\
&\preceq \|u\|_{k,q}+\sum_{|\nu|=k}\sum_{\beta\leq \nu}\|
\partial^{\beta}u\|_{\theta,q} \quad (\textrm{by Step1 and
Step2})\\
&=\|u\|_{\sigma,q}+\sum_{|\nu|=k}\sum_{\beta<\nu}\|
\partial^{\beta}u\|_{\theta,q}\\
&\preceq \|u\|_{\sigma,q}+\sum_{|\nu|=k}\sum_{\beta<\nu}\|u\|_{\theta+|\beta|,q}\quad (\partial^{\beta}:W^{\theta+|\beta|,q}\rightarrow W^{\theta,q} \textrm{is continuous} )\\
&\preceq
\|u\|_{\sigma,q}+\sum_{|\nu|=k}\sum_{\beta<\nu}\|u\|_{\sigma,q}
\quad (\theta+|\beta|<\sigma \Rightarrow
W^{\sigma,q}\hookrightarrow W^{\theta+|\beta|,q})\\
&\preceq \|u\|_{\sigma,q}.
\end{align*}

\item \textbf{Step 4:} $\sigma<0$.
For this case we use a duality argument:
\begin{align*}
\|fu\|_{\sigma,q}&=\sup_{v\in
W^{-\sigma,q'}\setminus\{0\}}\frac{|\langle
fu,v\rangle|}{\|v\|_{-\sigma,q'}} =\sup_{v\in
W^{-\sigma,q'}\setminus\{0\}}\frac{|\langle
u,fv\rangle|}{\|v\|_{-\sigma,q'}}
\\
&\leq \sup_{v\in
W^{-\sigma,q'}\setminus\{0\}}\frac{\|u\|_{\sigma,q}\|fv\|_{-\sigma,q'}}{\|v\|_{-\sigma,q'}}
\preceq \sup_{v\in
W^{-\sigma,q'}\setminus\{0\}}\frac{\|u\|_{\sigma,q}\|v\|_{-\sigma,q'}}{\|v\|_{-\sigma,q'}}=\|u\|_{\sigma,q}.
\end{align*}
\end{itemizeX}
\end{proof}
\begin{lemma}\lab{lemA10}
Let $\sigma, \delta \in \mathbb{R}$, $q\in(1,\infty)$. Let
$N=\floor{ |\sigma|}+1$. Suppose $f\in C^{N}(\mathbb{R}^n) $ is
such that for all multi-indices $\nu$ with $|\nu|\leq N$
\begin{equation*}
|\partial^{\nu}f(x)|\leq b(\nu)|x|^{-|\nu|},
\end{equation*}
where $b(\nu)$ are appropriate numbers independent of $x$. If
$u\in W^{\sigma,q}_{\delta}(\mathbb{R}^n)$, then $fu \in
W^{\sigma,q}_{\delta}(\mathbb{R}^n)$ and moreover
$\|fu\|_{\sigma,q,\delta}\preceq \|u\|_{\sigma,q,\delta}$ where
the implicit constant depends on $b(\nu)$.
\end{lemma}
\begin{proof}{\bf (Lemma~\ref{lemA10})}
The case $\sigma\geq 0$ is a special case of Lemma 3
in \cite{15}. For the case $\sigma <0$ we may use a duality
argument exactly similar to the proof of Lemma \ref{lemA9}.
\end{proof}

%%%%%%%%%%%%%%%%%%%%%%%%%%%%%%%%%%%%%%%
Most of the claims of the following lemma are discussed in
\cite{jZ77} for $s\geq 0$. The argument in \cite{jZ77} in part is
based on a similar multiplication lemma for Besov spaces. An
entirely different approach to the proof which includes some
cases that are not considered in \cite{jZ77} can be found in
\cite{35}. By using a duality argument one can extend the proof to
negative values of $s$ \cite{HNT07b,34,35}.
%A more detailed discussion
%about the proof of this Lemma can be found in Appendix.
\begin{lemma}[Multiplication Lemma, Unweighted spaces]\lab{lemA13}
%Let $\Omega$ be a bounded domain with smooth boundary in
%$\mathbb{R}^n$.
Let $s_i \geq s$ with $s_1+s_2\geq 0$, and $1 < p, p_i < \infty$
($i=1,2$) be real numbers satisfying
\begin{align*}
& s_i-s\geq n(\dfrac{1}{p_i}-\dfrac{1}{p}),\quad \textrm{(if $s_i=s \not \in \mathbb{Z}$, then let $p_i\leq p$)}\\
& s_1+s_2-s>n(\dfrac{1}{p_1}+\dfrac{1}{p_2}-\dfrac{1}{p})\geq 0.
%\quad p\leq \max\{p_1,p_2\},
\end{align*}
%where the strictness of the inequalities can be interchanged if
%$s\in \mathbb{N}_0$.
In case $s<0$, in addition let
\begin{equation*}
s_1+s_2> n(\dfrac{1}{p_1}+\dfrac{1}{p_2}-1) \quad
\textrm{(equality is allowed if $min(s_1,s_2)<0$)}.
\end{equation*}
Also in case where $s_1+s_2=0$ and $min(s_1,s_2) \not \in
\mathbb{Z}$, in addition let $\frac{1}{p_1}+\frac{1}{p_2}\geq 1$.
Then the pointwise multiplication of functions extends uniquely
to a continuous bilinear map
\begin{equation*}
W^{s_1,p_1}(\mathbb{R}^n)\times
W^{s_2,p_2}(\mathbb{R}^n)\rightarrow W^{s,p}(\mathbb{R}^n).
\end{equation*}
\end{lemma}
\begin{remark}
Note that in case $s_i=s\not \in \mathbb{Z}$, the condition
$p_i\leq p$ together with $s_i-s\geq
n(\frac{1}{p_i}-\frac{1}{p})$ in fact implies that we must have
$p_i=p$.
\end{remark}
\begin{corollary}
%Let $M=\mathbb{R}^n$ or let $M$ be an $n$ dimensional smooth AF manifold.
Let $s_i \geq s$ with $s_1+s_2> 0$, and $2 \leq p< \infty$
($i=1,2$) be real numbers satisfying
\begin{equation*}
s_1+s_2-s>\frac{n}{p}.
\end{equation*}
Then the pointwise multiplication of functions extends uniquely
to a continuous bilinear map
\begin{equation*}
W^{s_1,p}(\mathbb{R}^n)\times W^{s_2,p}(\mathbb{R}^n)\rightarrow
W^{s,p}(\mathbb{R}^n).
\end{equation*}
\end{corollary}
\begin{corollary} As a direct consequence of the multiplication
lemma we have:
\begin{itemize}
\item If $p\in (1,\infty)$ and $s\in (\frac{n}{p},\infty)$, then
$W^{s,p}(\mathbb{R}^n)$ is a Banach algebra.
\item Let $p\in (1,\infty)$ and $s\in (\frac{n}{p},\infty)$. Suppose $q\in (1,\infty)$ and $\sigma\in[-s,s]$ satisfy $\sigma-\frac{n}{q}\in [-n-s+\frac{n}{p}, s-\frac{n}{p}]$; in
case $s \not \in \mathbb{N}_0$, assume $\sigma\neq -s$; in case
$s \not \in \mathbb{N}_0$, $q<p$, in addition assume $\sigma \neq
s$. Then the pointwise multiplication is bounded as a map
$W^{s,p}(\mathbb{R}^n)\times
W^{\sigma,q}(\mathbb{R}^n)\rightarrow W^{\sigma,q}(\mathbb{R}^n)$.
\end{itemize}
\end{corollary}
Note: In the statement of the second item of the above corollary,
the case $\sigma=-s \not \in \mathbb{Z}$ has been excluded.
However, it follows from the multiplication lemma that the claim
holds true even if $\sigma=-s \not \in \mathbb{Z}$ provided we
additionally assume $\frac{1}{p}+\frac{1}{q}\geq 1$. Of course, if
$\sigma=-s$, the assumption $\frac{1}{p}+\frac{1}{q}\geq 1$
together with $\sigma-\frac{n}{q} \in [-n-s+\frac{n}{p},
s-\frac{n}{p}]$ implies that $\frac{1}{p}+\frac{1}{q}= 1$.
\begin{lemma}[Multiplication Lemma, Weighted spaces]\lab{lemA1}
Assume that $s, s_1, s_2$ and that $1 < p, p_1, p_2 < \infty$ are real
numbers satisfying
\begin{enumerate}[(i)]
\item  $s_i \geq s$ \quad($i=1,2$) (if $s_i=s \not \in \mathbb{Z}$, then let $p_i\leq p$),
\item  $s_1+s_2\geq 0$ \quad (if $s_1+s_2=0$ and $min(s_1,s_2)\not \in \mathbb{Z}$, then let $\frac{1}{p_1}+\frac{1}{p_2}\geq 1$),
\item  $s_i-s\geq n(\frac{1}{p_i}-\frac{1}{p})$  \quad($i=1,2$),
\item  $s_1+s_2-s>n(\frac{1}{p_1}+\frac{1}{p_2}-\frac{1}{p})\geq
0$.
\end{enumerate}
%Strictness of the first inequalities in $(iii)$ and $(iv)$ can be
%interchanged if $s\in \mathbb{N}_0$.
In case $min(s_1,s_2)<0$, in
addition let
%$1<p, p_i<\infty$, and let
\begin{enumerate}[(v)]
\item $\quad s_1+s_2\geq n(\frac{1}{p_1}+\frac{1}{p_2}-1)$.
\end{enumerate}
In case $s<0$ and $min(s_1,s_2)\geq 0$, we assume the above
inequality is strict ($s_1+s_2>
n(\frac{1}{p_1}+\frac{1}{p_2}-1)$). Then for all $\delta_1 ,
\delta_2 \in \mathbb{R}$, the pointwise multiplication of
functions extends uniquely to a continuous bilinear map
\begin{equation*}
W^{s_1,p_1}_{\delta_1}(\mathbb{R}^n)\times
W^{s_2,p_2}_{\delta_2}(\mathbb{R}^n)\rightarrow
W^{s,p}_{\delta_1+\delta_2}(\mathbb{R}^n).
\end{equation*}
\end{lemma}
\begin{proof}{\bf (Lemma~\ref{lemA1})}
A proof for the case $p_1=p_2=p=2$ is given in \cite{2}. In what
follows we use Lemma \ref{lemA13} to extend that proof to our
general setting. In the proof we will make use of the following
facts:
\begin{itemizeX}
\item \textbf{Fact 1}:  If $f$ is a smooth function with compact
support and $u\in W^{t,q}$ then $fu\in W^{t,q}$ and
$\|fu\|_{W^{t,q}}\preceq \|u\|_{W^{t,q}}$ (this is in fact a
special case of Lemma \ref{lemA9}).
\item \textbf{Fact 2}: For all $j\geq 1$, $S_{2^j}\varphi_j=S_{2^j}
S_{2^{-j}}\varphi=\varphi$. So $S_{2^j}\varphi_j$ is zero if
$x\not \in B_2\setminus B_{\frac{1}{2}}$. Also for $j=0$,
$S_{2^j}\varphi_j=\varphi_0$ is zero if $x\not \in B_2$.
\item \textbf{Fact 3}:
%If $1\leq \alpha\leq\beta$, then
%$l^{\alpha}\hookrightarrow l^{\beta}$ ; in fact for any sequence
%$a=\{a_j\}$, $\|a\|_{l^\beta}\leq \|a\|_{l^\alpha}$.\\
 Let $\{a_j\}_{j=1}^{m}$ be positive numbers. Define $f:
(0,\infty)\rightarrow \mathbb{R}$ as follows:
\begin{equation*}
f(r)=(\sum_{j=1}^{m}a_j^r)^{\frac{1}{r}}.
\end{equation*}
$f(r)$ is a decreasing function. The reason is as follows:
Suppose $t\geq r>0.$ We want to show
$(\sum_{j=1}^{m}a_j^t)^{\frac{1}{t}}\leq
(\sum_{j=1}^{m}a_j^r)^{\frac{1}{r}}$. Since $g(x)=x^t$ is an
increasing function over $(0,\infty)$, it is enough to show
$(\sum_{j=1}^{m}a_j^t)\leq (\sum_{j=1}^{m}a_j^r)^{\frac{t}{r}}$.
Letting $b_j=a_j^r$, $\beta=\frac{t}{r}$, we want to prove
$(\sum_{j=1}^{m}b_j^{\beta})\leq (\sum_{j=1}^{m}b_j)^{\beta}$. To
this end we just need to show that
\begin{equation*}
\sum_{j=1}^{m}(\frac{b_j}{\sum_{j=1}^{m}b_j})^{\beta}\leq 1.
\end{equation*}
Set $e_j=\frac{b_j}{\sum_{j=1}^{m}b_j}$. Clearly
$\sum_{j=1}^{m}e_j = 1$. Since $\beta \geq 1$, $0\leq e_j \leq
1$, we have $e_j^{\beta}\leq e_j$. Therefore
\begin{equation*}
\sum_{j=1}^{m}e_j^{\beta}\leq \sum_{j=1}^{m}e_j=1.
\end{equation*}
%%%%%%%%%%%%%%%%%%%%%%%%%%%%%%%%%%%%%%%
%%%%%%%%%%%%%%%%%%%%%%%%%%%%%%%%%%%%%%%
\item \textbf{Fact 4}: For $a_k>0$ we have :$\sum_{k=1}^{m}a_k^p
\sim (\sum_{k=1}^{m}a_k)^p$
\\ (that is $\sum_{k=1}^{m}a_k^p \preceq
(\sum_{k=1}^{m}a_k)^p \preceq \sum_{k=1}^{m}a_k^p$).
\item \textbf{Fact 5}: $\|S_r u\|_{W^{s,p}}\leq
C(r,s,p,n)\|u\|_{W^{s,p}}$.
%I think for proof we may start with
%the case $s\in \mathbb{N}_0$ and then use interpolation for general $s$.
\end{itemizeX}
Now let's start proving the lemma. Suppose $u_i \in W^{s_i,p_i}_
{\delta_i}$. Let $\varphi_j=0$ for $j<0$. We have
\begin{equation*}
S_{2^j}(\varphi_j u_1 u_2)=S_{2^j}(\varphi_j)S_{2^j}u_1 S_{2^j}
u_2.
\end{equation*}
By \textbf{Fact 2}, for $j\geq 1$, $S_{2^j}\varphi_j$ is zero if
$x\not \in B_2\setminus B_{\frac{1}{2}}$. Also it is easy to see
that for $x \in B_2\setminus B_{\frac{1}{2}}$, $\varphi_k(2^j
x)=0$ if $k \not \in \{j-1, j, j+1\}$. Since for all $x$,
$\sum_{k=0}^{\infty}\varphi_k(2^j x)=1$, we can conclude that for
$x \in B_2\setminus B_{\frac{1}{2}}$
\begin{equation*}
\sum_{k=j-1}^{j+1} \varphi_k(2^j x)=1.
\end{equation*}
For $j=0$, $S_{2^j}\varphi_j$ is zero if $x\not \in B_2$; one can
easily check that if $j=0$, the above equality holds true for all
$x\in B_2$. Therefore for all $x$
\begin{equation*}
S_{2^j}(\varphi_j u_1
u_2)=S_{2^j}(\varphi_j)\sum_{k=j-1}^{j+1}S_{2^j}(\varphi_k u_1)
\sum_{l=j-1}^{j+1}S_{2^j}(\varphi_l u_2).
\end{equation*}
Now by \textbf{Fact 1} and \textbf{Fact 4} we have
\begin{equation*}
\|S_{2^j}(\varphi_j u_1 u_2)\|_{W^{s,p}}^p \preceq \sum_{k, l
=j-1}^{j+1}\|S_{2^j}(\varphi_k u_1)S_{2^j}(\varphi_l
u_2)\|_{W^{s,p}}^p,
\end{equation*}
and by the multiplication lemma for the corresponding unweighted
Sobolev spaces we get
%on bounded smooth domains we get (note that $\supp
%S_{2^j}\varphi_j \subseteq B_3$)
\begin{align*}
\|S_{2^j}(\varphi_j u_1 u_2)\|_{W^{s,p}}^p &\preceq \sum_{k, l
=j-1}^{j+1} \|S_{2^j}(\varphi_k u_1)\|_{W^{s_1,p_1}}^p
\|S_{2^j}(\varphi_l u_2)\|_{W^{s_2,p_2}}^p\\
& \preceq \sum_{k, l =j-1}^{j+1} \|S_{2^{j-k}}S_{2^k}(\varphi_k
u_1)\|_{W^{s_1,p_1}}^p \|S_{2^{j-l}}S_{2^l}(\varphi_l
u_2)\|_{W^{s_2,p_2}}^p.
\end{align*}
%%%%%%%%%%%%%%%%%%%%%%%%%%%%%%%%%%%%%%%
%%%%%%%%%%%%%%%%%%%%%%%%%%%%%%%%%%%%%%%
$S_{2^{j-k}}$ is one of $S_{2^{-1}}$, $S_{2^{0}}$, or
$S_{2^{1}}$. So, by \textbf{Fact 5}
\begin{align*}
\sum_{k=j-1}^{j+1} \|S_{2^{j-k}}S_{2^k}(\varphi_k
u_1)\|_{W^{s_1,p_1}}^p &\leq \sum_{k=j-1}^{j+1}
   (\|S_{2^{-1}}S_{2^k}(\varphi_k u_1)\|_{W^{s_1,p_1}}^p
+\|S_{2^{0}}S_{2^k}(\varphi_k u_1)\|_{W^{s_1,p_1}}^p
\\
& \quad \quad
+\|S_{2^{1}}S_{2^k}(\varphi_k u_1)\|_{W^{s_1,p_1}}^p)\\
& \preceq \sum_{k=j-1}^{j+1} \|S_{2^k}(\varphi_k
u_1)\|_{W^{s_1,p_1}}^p
\end{align*}
and the similar result is true for $\sum_{l=j-1}^{j+1}
\|S_{2^{j-l}}S_{2^l}(\varphi_l u_2)\|_{W^{s_2,p_2}}^p$.
Consequently
\begin{align*}
\|S_{2^j}(\varphi_j u_1 u_2)\|_{W^{s,p}}^p & \preceq \sum_{k, l
=j-1}^{j+1}\|S_{2^k}(\varphi_k u_1)\|_{W^{s_1,p_1}}^p
\|S_{2^l}(\varphi_l u_2)\|_{W^{s_2,p_2}}^p
\\
& \preceq (\sum_{k =j-1}^{j+1}\|S_{2^k}(\varphi_k
u_1)\|_{W^{s_1,p_1}})^p (\sum_{l =j-1}^{j+1}\|S_{2^l}(\varphi_l
u_2)\|_{W^{s_2,p_2}})^p
\end{align*}
Therefore
\begin{align*}
\sum_{j=0}^{\infty} & 2^{-p(\delta_1+\delta_2)j}
   \|S_{2^j}(\varphi_j u_1 u_2)\|_{W^{s,p}}^p
\\
& \preceq \sum_{j=0}^{\infty}\big[2^{-p\delta_1 j}(\sum_{k
=j-1}^{j+1}\|S_{2^k}(\varphi_k u_1)\|_{W^{s_1,p_1}})^p
\,2^{-p\delta_2 j}(\sum_{l =j-1}^{j+1}\|S_{2^l}(\varphi_l
u_2)\|_{W^{s_2,p_2}})^p\big] .
\end{align*}
Let
\begin{align*}
&a_j= 2^{-\delta_1 j }\sum_{k =j-1}^{j+1}\|S_{2^k}(\varphi_k
u_1)\|_{W^{s_1,p_1}}\\
&b_j=2^{-\delta_2 j }\sum_{l =j-1}^{j+1}\|S_{2^l}(\varphi_l
u_2)\|_{W^{s_2,p_2}}.
\end{align*}
So we have
\begin{align*}
\|u_1u_2\|_{W^{s,p}_{\delta_1+\delta_2}}=\big[\sum_{j=0}^{\infty}2^{-p(\delta_1+\delta_2)j}\|S_{2^j}(\varphi_j
u_1 u_2)\|_{W^{s,p}}^p\big]^{\frac{1}{p}}\preceq
\big[\sum_{j=0}^{\infty}(a_jb_j)^p\big]^{\frac{1}{p}}.
\end{align*}
Now let $r$ be such that
$\frac{1}{r}=\frac{1}{p_1}+\frac{1}{p_2}$. By assumption
$\frac{1}{p_1}+\frac{1}{p_2}-\frac{1}{p}\geq 0$ and so $r\leq p$.
Thus by \textbf{Fact 3} and Holder's inequality we get
\begin{align*}
\big[\sum_{j=0}^{\infty}(a_jb_j)^p\big]^{\frac{1}{p}} &\leq
\big[\sum_{j=0}^{\infty}(a_jb_j)^r\big]^{\frac{1}{r}}
\\
&\leq \big[\sum_{j=0}^{\infty}(a_j)^{p_1}\big]^{\frac{1}{p_1}}\big[\sum_{j=0}^{\infty}(b_j)^{p_2}\big]^{\frac{1}{p_2}}\\
&\preceq \big[\sum_{j=0}^{\infty}2^{-p_1 \delta_1 j }\sum_{k
=j-1}^{j+1}\|S_{2^k}(\varphi_k u_1)\|_{W^{s_1,p_1}}^{p_1}
\big]^{\frac{1}{p_1}}
\\
& \quad \quad \big[\sum_{j=0}^{\infty}2^{-p_2 \delta_2 j }\sum_{l
=j-1}^{j+1}\|S_{2^l}(\varphi_l
u_2)\|_{W^{s_2,p_2}}^{p_2}  \big]^{\frac{1}{p_2}}\\
&\preceq \big[\sum_{j=0}^{\infty}2^{-p_1 \delta_1 j}
\|S_{2^j}(\varphi_j u_1)\|_{W^{s_1,p_1}}^{p_1}
\big]^{\frac{1}{p_1}} \big[\sum_{j=0}^{\infty}2^{-p_2 \delta_2 j
}\|S_{2^j}(\varphi_j u_2)\|_{W^{s_2,p_2}}^{p_2}
\big]^{\frac{1}{p_2}}
\\
&
=\|u_1\|_{W^{s_1,p_1}_{\delta_1}}\|u_2\|_{W^{s_2,p_2}_{\delta_2}}.
\end{align*}
This proves $\|u_1u_2\|_{W^{s,p}_{\delta_1+\delta_2}}\preceq
\|u_1\|_{W^{s_1,p_1}_{\delta_1}}\|u_2\|_{W^{s_2,p_2}_{\delta_2}}$.
\end{proof}
\begin{remark}
By using partition of unity and charts one can show that the
above lemma also holds for AF manifolds.
\end{remark}
\begin{corollary} [The case where $p_1=p_2=p$]\lab{coroA2}
Assume $s \leq min\{s_1, s_2\}$, $s_1+s_2 > s + \frac{n}{p}$,
$s_1+s_2> 0$, $s_1+s_2 > n(\frac{2}{p}-1)$ and
$\delta_1+\delta_2\leq \delta$, then the multiplication
\begin{equation*}
W^{s_1,p}_{\delta_1}\times W^{s_2,p}_{\delta_2}\rightarrow
W^{s,p}_{\delta},
\end{equation*}
is continuous.
\end{corollary}
\begin{corollary}\lab{coroA1}
Let $p\in (1,\infty)$, $s\in (\frac{n}{p},\infty)$, and $\delta <
0$, then the space $W^{s,p}_{\delta}$ is an algebra.
\end{corollary}
\begin{lemma}\lab{lempA1}
Let the following assumptions hold:
\begin{itemize}
\item $f:\mathbb{R}\rightarrow \mathbb{R}$ is smooth,
\item $u\in W^{s,p}_{\rho}(\mathbb{R}^n)$, where $s>\frac{n}{p}$,
$\rho<0$, and $p\in (1,\infty)$,
\item $v\in W^{\sigma,q}_{\delta}(\mathbb{R}^n)$, where $\delta \in \mathbb{R}$, $q\in (1,\infty)$ and  $(i)$ $\sigma \in
[-s,s]$ ($\sigma \neq  -s$ if $s \not \in \mathbb{N}_0$;  $\sigma
\neq  s$ if $s \not \in \mathbb{N}_0$ and $q<p$), $(ii)$
$\sigma-\frac{n}{q} \in [-n-s+\frac{n}{p}, s-\frac{n}{p}]$.
\end{itemize}
Then: $f(u)v \in W^{\sigma,q}_{\delta}(\mathbb{R}^n)$ and
moreover the map taking $(u,v)$ to $f(u)v$ is continuous.
\end{lemma}
Note: The claim of the above lemma holds true even if $\sigma=-s
\not \in \mathbb{Z}$ provided we additionally assume
 $\frac{1}{p}+\frac{1}{q}= 1$.
\begin{proof}{\bf (Lemma~\ref{lempA1})}
A proof for the case $p=q=2$ is given in \cite{2}. Here we use the
multiplication lemma to extend that proof to our general setting.
In the proof we make use of the following facts:
\begin{itemizeX}
\item \textbf{Fact 1}:  If $\eta$ is a smooth function with compact
support, $f$ is as in the statement of lemma, and $u\in W^{t,q}$
with $tq>n$, then $\eta f(u)\in W^{t,q}$ and the map taking $u$ to
$\eta f(u)$ is continuous from $W^{t,q}$ to $W^{t,q}$.
\item \textbf{Fact 2}: For all $j\geq 1$, $S_{2^j}\varphi_j=S_{2^j}
S_{2^{-j}}\varphi=\varphi$. So $S_{2^j}\varphi_j$ is zero if
$x\not \in B_2\setminus B_{\frac{1}{2}}$. Also it is easy to see
that for $x \in B_2\setminus B_{\frac{1}{2}}$, $\varphi_k(2^j
x)=0$ if $k \not \in \{j-1, j, j+1\}$. Since for all $x$,
$\sum_{k=0}^{\infty}\varphi_k(2^j x)=1$, we can conclude that for
$x \in B_2\setminus B_{\frac{1}{2}}$
\begin{equation*}
\sum_{k=j-1}^{j+1} \varphi_k(2^j x)=1.
\end{equation*}
For $j=0$, $S_{2^j}\varphi_j$ is zero if $x\not \in B_2$; one can
easily check that if $j=0$, the above equality holds true for all
$x\in B_2$.
\item \textbf{Fact 3}: $\|S_r u\|_{W^{t,e}}\leq
C(r,t,e,n)\|u\|_{W^{t,e}}$.
\end{itemizeX}
We prove the lemma in six steps:

 \textbf{Step 1}:  Suppose $u$ and $v$ satisfy the hypotheses of
the lemma. Then considering \textbf{Fact 2} and the fact that
$W^{s,p}\times W^{\sigma, q}\hookrightarrow W^{\sigma,q}$
%on bounded domains (note that
%$\supp S_{2^j}\varphi_j \subset B_3$ and in fact $S_{2^j}\varphi_j$
%is zero if $x\not \in B_2\setminus B_{\frac{1}{2}}$, the precise
%argument is similar to the one given in Remark
%\ref{remcorrect1})
, we can write
\begin{align*}
\|f(u)v\|_{W^{\sigma,q}_{\delta}}^q &=\sum_{j=0}^{\infty}
2^{-q\delta j}\|S_{2^j}(\varphi_j f(u) v)\|_{W^{\sigma,q}}^q\\
&=\sum_{j=0}^{\infty} 2^{-q\delta j}
\|\sum_{k=j-1}^{j+1}(S_{2^j}\varphi_k)f(\sum_{i=j-1}^{j+1}S_{2^{j-i}}S_{2^i}(\varphi_i
u ))S_{2^j}(\varphi_j  v)\|_{W^{\sigma,q}}^q\\
&\preceq \sum_{j=0}^{\infty} 2^{-q\delta j}
\|\sum_{k=j-1}^{j+1}(S_{2^j}\varphi_k)f(\sum_{i=j-1}^{j+1}S_{2^{j-i}}S_{2^i}(\varphi_i
u ))\|_{W^{s,p}}^q \|S_{2^j}(\varphi_j  v)\|_{W^{\sigma,q}}^q.
\end{align*}
In the second line above we made use of the fact that for $j\geq
1$, $S_{2^j}\varphi_j$ is zero if $x\not \in B_2\setminus
B_{\frac{1}{2}}$ and the following equality holds over
$B_2\setminus B_{\frac{1}{2}}$
\begin{align*}
S_{2^j}f(u) & = f(S_{2^j}u)=f(u(2^j x))=f(\sum_{i=j-1}^{j+1}
\varphi_i(2^j x)u(2^j x)) =f(\sum_{i=j-1}^{j+1}S_{2^j} (\varphi_i
u))
\\
&=f(\sum_{i=j-1}^{j+1}S_{2^{j-i}}S_{2^i}(\varphi_i u )).
\end{align*}
If $j=0$, $S_{2^j}\varphi_j$ is zero if $x\not \in B_2$ and the
above equality holds over $B_2$.

 For all j define
\begin{equation*}
R_j u=\sum_{i=j-1}^{j+1}S_{2^{j-i}}S_{2^i}(\varphi_i u ).
\end{equation*}
So we have
\begin{equation*}
\|f(u)v\|_{W^{\sigma,q}_{\delta}}^q\preceq \sum_{j=0}^{\infty}
2^{-q\delta j} \|\sum_{k=j-1}^{j+1}(S_{2^j}\varphi_k)f(R_j u
)\|_{W^{s,p}}^q \|S_{2^j}(\varphi_j  v)\|_{W^{\sigma,q}}^q.
\end{equation*}
\textbf{Step 2}: Note that if $g\in W^{s,p}_{\rho}$ since $\rho<0$
we have $S_{2^i} (\varphi_i g)\rightarrow 0$ in $W^{s,p}$ as
$i\rightarrow \infty$. Indeed, $2^{-p\rho i}\geq 1$ and therefore
we may write
\begin{align*}
\|g\|^p_{W^{s,p}_{\rho}}< \infty
   & \Rightarrow \sum_{i=0}^{\infty}
2^{-p\rho i}\|S_{2^i}(\varphi_i g)\|_{W^{s,p}}^p < \infty
     \Rightarrow \sum_{i=0}^{\infty} \|S_{2^i}(\varphi_i
g)\|_{W^{s,p}}^p<\infty
\\
& \Rightarrow \lim_{i\rightarrow \infty}S_{2^i} (\varphi_i g) =0
\quad \textrm{in}\, W^{s,p}
\end{align*}
Moreover it follows from $2^{-p\rho i}\geq 1$  that if $g\in
W^{s,p}_{\rho}$
 with $\rho<0$ then it holds that $\|S_{2^i}(\varphi_i g)\|_{W^{s,p}}\leq \|g\|_{W^{s,p}_{\rho}}$ for all $i\geq 0$. Also
 we have
\begin{align*}
\|R_j
g-0\|_{W^{s,p}}&=\|\sum_{i=j-1}^{j+1}S_{2^{j-i}}S_{2^i}(\varphi_i
g )\|_{W^{s,p}}\leq
\sum_{i=j-1}^{j+1}\|S_{2^{j-i}}S_{2^i}(\varphi_i g )\|_{W^{s,p}}\\
&\leq \sum_{i=j-1}^{j+1}(\|S_{2^{-1}}S_{2^i}(\varphi_i g
)\|_{W^{s,p}}+\|S_{2^{0}}S_{2^i}(\varphi_i g )\|_{W^{s,p}}
\\
& \quad \quad +\|S_{2^{1}}S_{2^i}(\varphi_i g )\|_{W^{s,p}})\\
&\preceq \sum_{i=j-1}^{j+1}\|S_{2^i}(\varphi_i g
)\|_{W^{s,p}}\rightarrow 0.
\end{align*}
\textbf{Step 3}: Let $\eta_j:=
\sum_{k=j-1}^{j+1}(S_{2^j}\varphi_k)$. For $j>1$ we may write
\begin{align*}
\sum_{k=j-1}^{j+1}(S_{2^j}\varphi_k) & =
\sum_{k=j-1}^{j+1}S_{2^j}S_{2^{-k}}\varphi=\sum_{k=j-1}^{j+1}S_{2^j}S_{2^{-k}}S_2\varphi_1=\sum_{k-j=-1}^{k-j=1}S_{2^{j-k+1}}\varphi_1
\\
& =\sum_{i=0}^2 S_{2^i} \varphi_1=:\eta.
\end{align*}
That is for $j>1$, $\eta_j$ does not depend on $j$. Now, by
\textbf{Step 2}, we know that $R_j u\rightarrow 0$ in $W^{s,p}$.
So it follows from \textbf{Fact 1} that $\eta f(R_j u)\rightarrow
\eta f(0)$ in $W^{s,p}$. Consequently $\{\|\eta f(R_j
u)\|_{W^{s,p}}\}_{j=2}^{\infty}$ is a bounded sequence:
\begin{equation*}
\exists \, M_1 \quad \textrm{such that} \quad \forall \, j\geq 2
\quad \|\eta f(R_j u)\|_{W^{s,p}}<M_1.
\end{equation*}
Let
\begin{equation*}
M= \max\{M_1, \|\eta_1 f(R_1 u)\|_{W^{s,p}}, \|\eta_0 f(R_0
u)\|_{W^{s,p}}\}
\end{equation*}
(M is independent of $j$ but it may depend on $u$).

So by what was proved in \textbf{Step 1} we have
\begin{equation*}
\|f(u)v\|_{W^{\sigma,q}_{\delta}}^q\preceq \sum_{j=0}^{\infty}
2^{-q\delta j} M^q \|S_{2^j}(\varphi_j  v)\|_{W^{\sigma,q}}^q=M^q
\|v\|_{W^{\sigma,q}_{\delta}}^q.
\end{equation*}
This shows that $f(u)v$ is in $W^{\sigma,q}_{\delta}$. Now it
remains to prove the continuity.

\noindent \textbf{Step 4}: Let $(u_k,v_k)$ be a sequence in
$W^{s,p}_{\rho} \times W^{\sigma,q}_{\delta}$ that converges to
$(u,v)\in W^{s,p}_{\rho} \times W^{\sigma,q}_{\delta}$. We must
show that $f(u_k)v_k \rightarrow f(u)v$ in
$W^{\sigma,q}_{\delta}$. Note that
\begin{equation*}
f(u)v-f(u_k)v_k=f(u)(v-v_k)+(f(u)-f(u_k))v_k.
\end{equation*}
By what was proved in \textbf{Step 3}, we have
\begin{equation*}
\|f(u)(v-v_k)\|_{W^{\sigma,q}_{\delta}}\preceq
\|v-v_k\|_{W^{\sigma,q}_{\delta}}\rightarrow 0.
\end{equation*}
So it remains to show that
$\|(f(u)-f(u_k))v_k\|_{W^{\sigma,q}_{\delta}}\rightarrow 0$.

\noindent \textbf{Step 5}: By calculations similar to what was
done in \textbf{Step 1} we have
\begin{align*}
\|(f(u)-f(u_k))v_k\|_{W^{\sigma,q}_{\delta}}^q &\preceq
\sum_{j=0}^{\infty} 2^{-q\delta j} \|\eta_j (f(R_j u )-f(R_j u_k))
\|_{W^{s,p}}^q \|S_{2^j}(\varphi_j  v_k)\|_{W^{\sigma,q}}^q\\
&\preceq \|v_k\|_{W^{\sigma,q}_{\delta}}^q \sup_{j\geq 0}\|\eta_j
(f(R_j u )-f(R_j u_k)) \|_{W^{s,p}}^q.
\end{align*}
Note that $\{v_k\}$ is convergent and so $\{v_k\}$ is bounded in
$W^{\sigma,q}_{\delta}$. Thus it is enough to show that
$\sup_{j\geq 0}\|\eta_j (f(R_j u )-f(R_j u_k))
\|_{W^{s,p}}\rightarrow 0$ as $k\rightarrow \infty$.

\noindent \textbf{Step 6}: We need to show
\begin{equation*}
\forall\, \, \epsilon>0\,\, \exists \,\, N\quad s.t. \quad \forall
\, k\geq N \quad \sup_{j\geq 0}\|\eta_j (f(R_j u )-f(R_j u_k))
\|_{W^{s,p}}<\epsilon.
\end{equation*}
Let $\epsilon>0$ be given. Note that
\begin{equation}\lab{eqproof1}
\|\eta_j (f(R_j u )-f(R_j u_k)) \|_{W^{s,p}}\leq \|\eta_j (f(R_j u
)-f(0)) \|_{W^{s,p}}+\|\eta_j (f(0)-f(R_j u_k)) \|_{W^{s,p}}.
\end{equation}
%We will show that by taking $k$ large enough the supremum over
%$j$ of each of the right hand side terms will become as small as we want.
Let's start by considering the first term on RHS. By \textbf{Fact
1}, there exists $\alpha>0$ such that if $\|g\|_{W^{s,p}}<\alpha$
then $\|\eta_j (f(g )-f(0)) \|_{W^{s,p}}<\frac{\epsilon}{4}$.
Note that for $j>1$, $\eta_j$ does not depend on $j$ and so
$\alpha$ can be chosen independent of $j$. By \textbf{Step 2} we
know that $R_ju\rightarrow 0$ in $W^{s,p}$ and so there exists a
number $P\geq 2$ such that for $j\geq P$, $\|R_j
u\|_{W^{s,p}}<\frac{\alpha}{2}$. It follows that
\begin{equation*}
\forall \, j\geq P \quad \|\eta_j (f(R_j u )-f(0))
\|_{W^{s,p}}<\frac{\epsilon}{4}
\end{equation*}
So
\begin{equation}\lab{eqproof2}
\sup_{j\geq P}\|\eta_j (f(R_j u )-f(0)) \|_{W^{s,p}}\leq
\frac{\epsilon}{4}.
\end{equation}
Now we show that there exists $N_1$ such that if $k\geq N_1$ then
it holds that $\sup_{j\geq P} \|\eta (f(0)-f(R_j u_k))
\|_{W^{s,p}}\leq \frac{\epsilon}{4}$. (Note that since $P \geq 2$
we have $\eta_j=\eta$.)
\begin{itemizeX}
\item \textbf{Claim:} For all $j$, $R_j u_k\rightarrow R_j u$ in
$W^{s,p}$ uniformly with respect to $j$ as $k\rightarrow \infty$.
\item \textbf{Proof of the claim:} By what was stated in \textbf{Step 2}, since
we have that $\rho<0$, $\|S_{2^i}(\varphi_i(u_k-u))\|_{W^{s,p}}
\leq \|u_k-u\|_{W^{s,p}_{\rho}}$ for all $i$, and we have
\begin{align*}
%&\forall \, i \quad \|S_{2^i}(\varphi_i(u_k-u))\|_{W^{s,p}} \leq
%\|u_k-u\|_{W^{s,p}_{\rho}},\Rightarrow\\
%\end{equation*}
%Now by what was shown in \textbf{Step 2}
%\begin{equation*}
\|R_j (u_k-u)\|_{W^{s,p}} & \preceq
\sum_{i=j-1}^{j+1}\|S_{2^i}(\varphi_i(u_k-u))\|_{W^{s,p}}\leq
\sum_{i=j-1}^{j+1} \|u_k-u\|_{W^{s,p}_{\rho}}
\\
&=3\|u_k-u\|_{W^{s,p}_{\rho}}\rightarrow 0 \quad \textrm{uniform
in $j$ as} \quad k\rightarrow \infty
\end{align*}
\end{itemizeX}
Therefore
\begin{equation*}
\exists \, N_1\,\, s.t.\,\, \forall j \quad \forall k\geq N_1
\quad \|R_j (u_k-u)\|_{W^{s,p}}<\frac{\alpha}{2}.
\end{equation*}
In particular, for all $j\geq P$ and $k\geq N_1$ we have
\begin{equation*}
\|R_j u_k\|_{W^{s,p}}\leq \|R_j (u_k-u)\|_{W^{s,p}}+\|R_j u
\|_{W^{s,p}}<\frac{\alpha}{2}+\frac{\alpha}{2}=\alpha.
\end{equation*}
Consequently for all $j\geq P$ and $k\geq N_1$ we have
\begin{equation*}
\|\eta (f(0)-f(R_j u_k)) \|_{W^{s,p}}< \frac{\epsilon}{4} ,
\end{equation*}
which implies
\begin{equation}\lab{eqproof3}
\forall \,\, k\geq N_1 \quad \sup_{j\geq P} \|\eta (f(0)-f(R_j
u_k)) \|_{W^{s,p}}\leq \frac{\epsilon}{4}.
\end{equation}
From (\ref{eqproof1}), (\ref{eqproof2}), and  (\ref{eqproof3}) we
get
\begin{equation*}
\forall \,\, k\geq N_1 \quad \sup_{j\geq P}\|\eta_j (f(R_j u
)-f(R_j u_k)) \|_{W^{s,p}}\leq \frac{\epsilon}{2}.
\end{equation*}
Now note that by the claim that was proved above, we know that
$R_j u_k \rightarrow R_j u$ in $W^{s,p}$. So by \textbf{Fact 1},
$\|\eta_j (f(R_j u )-f(R_j u_k)) \|_{W^{s,p}}\rightarrow 0$ for
any fixed $j$ as $k \rightarrow \infty$. In particular for $0\leq
j \leq P-1$,
\begin{equation*}
\exists \, M_j\,\, s.t.\,\, \forall\, k\geq M_j\quad \|\eta_j
(f(R_j u )-f(R_j u_k)) \|_{W^{s,p}}< \frac{\epsilon}{2}.
\end{equation*}
So if we let $N=\max\{N_1,M_0,M_1,...,M_{P-1}\}$, then for all
$k\geq N$
\begin{align*}
&\sup_{j\geq P}\|\eta_j (f(R_j u )-f(R_j u_k)) \|_{W^{s,p}}\leq
\frac{\epsilon}{2}\\
&\sup_{0\leq j\leq P-1}\|\eta_j (f(R_j u )-f(R_j u_k))
\|_{W^{s,p}}\leq \frac{\epsilon}{2}.
\end{align*}
That is
\begin{equation*}
\forall\,\, k\geq N \quad \sup_{j\geq 0}\|\eta_j (f(R_j u )-f(R_j
u_k)) \|_{W^{s,p}}\leq \frac{\epsilon}{2}<\epsilon,
\end{equation*}
which is exactly what we wanted to prove.
\end{proof}
\begin{remark}
%Using extension operators one can easily show that Lemma
%\ref{lempA1} remains true if we replace $\mathbb{R}^n$ by an open
%subset of $\mathbb{R}^n$ with smooth boundary.
Obviously the above result also holds true if $f$ is only smooth
on an open interval containing the range of $u$. By using
partition of unity and charts one can show that the claim also
holds for AF manifolds (of any class).
\end{remark}
\begin{corollary}\lab{corpA1}
Suppose $f:\mathbb{R}\rightarrow \mathbb{R}$ is smooth and
$f(0)=0$. If $u\in W^{s,p}_{\rho}$ where $sp>n$, $\rho<0$ then
$f(u) \in W^{s,p}_{\rho}$ and the map taking $u$ to $f(u)$ is
continuous from $W^{s,p}_{\rho}$ to $W^{s,p}_{\rho}$.
\end{corollary}
\begin{proof}{\bf (Corollary~\ref{corpA1})}
$f(0)=0$, so by Taylor's theorem we have $f(x)=xF(x)$
where $F$ is smooth. Therefore by the previous lemma,
$f(u)=uF(u)\in W^{s,p}_{\rho}$ and moreover the map taking $u$ to
$f(u)=uF(u)$ is continuous from $W^{s,p}_{\rho}$ to
$W^{s,p}_{\rho}$.
\end{proof}
\begin{lemma}\lab{lemA8}
Let the following assumptions hold:
\begin{itemize}
\item $p\in(1,\infty)$, $s\in (\frac{n}{p},\infty)$, $\delta<0$
and $u\in W^{s,p}_{\delta}$,
\item $\nu\in \mathbb{R}$, $\sigma \in [-1,1]$, $\theta=\frac{1}{p}-\frac{s-1}{n}$, $\frac{1}{q}\in (\frac{1+\sigma}{2}\theta,
1-\frac{1-\sigma}{2}\theta)$ and $v\in W^{\sigma,q}_{\nu}$,
\item $f:[\inf u,\sup u]\rightarrow \mathbb{R}$ is a smooth
function. (Note that $W^{s,p}_{\delta}\hookrightarrow
C^{0}_{\delta}\hookrightarrow L^{\infty}$ and therefore $\inf u$
and $\sup u$ are finite numbers.)
\end{itemize}
Then:
\begin{equation*}
\|v f(u)\|_{\sigma,q,\nu}\preceq
\|v\|_{\sigma,q,\nu}(\|f(u)\|_{\infty}+\|f'(u)\|_{\infty}\|u\|_{s,p,\delta}).
\end{equation*}
\end{lemma}
%%%%%%%%%%%%%%%%%%%%%%%%%%%%%%%%%%%%%%%
%%%%%%%%%%%%%%%%%%%%%%%%%%%%%%%%%%%%%%%

\begin{proof}{\bf (Lemma~\ref{lemA8})}
First we prove the claim for the case $\sigma=1$.
We have
\begin{align*}
\|v f(u)\|_{1,q,\nu}&\preceq \|\langle x\rangle^{-\nu-\frac{n}{q}}v
f(u)\|_{L^q}+\|\langle x\rangle^{-\nu-\frac{n}{q}+1}\nabla (v f(u))\|_{L^q}\\
&\preceq \|\langle x\rangle^{-\nu-\frac{n}{q}}v
f(u)\|_{L^q}+\|\langle x\rangle^{-\nu-\frac{n}{q}+1}(\nabla v)
f(u)\|_{L^q}
\\
& \quad \quad +\|\langle x\rangle^{-\nu-\frac{n}{q}+1}v f'(u)\nabla u \|_{L^q}\\
&\preceq
\|\langle x\rangle^{-\nu-\frac{n}{q}}v\|_{L^q}\|f(u)\|_{L^{\infty}}+\|\langle x\rangle^{-\nu-\frac{n}{q}+1}\nabla
v\|_{L^q}\|f(u)\|_{L^{\infty}}
\\
& \quad \quad +\|\langle x\rangle^{-(\nu-1)-\frac{n}{q}}v
\nabla
u\|_{L^q}\|f'(u)\|_{L^{\infty}}\\
&(\textrm{note that $f$ is smooth on $[\inf u, \sup u]$ so}\,\,
f(u)\in L^{\infty}, \,\, f'(u)\in
L^{\infty})\\
&\preceq \|v\|_{1,q,\nu}\|f(u)\|_{L^{\infty}}+\|v\nabla
u\|_{L^q_{\nu-1}}\|f'(u)\|_{L^{\infty}}\\
&\preceq
\|v\|_{1,q,\nu}\|f(u)\|_{L^{\infty}}+\|v\|_{1,q,\nu}\|\nabla
u\|_{s-1,p,\delta-1}\|f'(u)\|_{L^{\infty}}
\\
& \quad \quad
(\textrm{$\frac{1}{q}\geq \theta$ so}\,\,W^{1,q}_{\nu}\times
W^{s-1,p}_{\delta-1}\hookrightarrow
L^q_{\delta+\nu-1}\hookrightarrow L^q_{\nu-1})\\
&\preceq \|v\|_{1,q,\nu}\|f(u)\|_{L^{\infty}}+\|v\|_{1,q,\nu}\|u\|_{s,p,\delta}\|f'(u)\|_{L^{\infty}}\\
&=\|v\|_{1,q,\nu}(\|f(u)\|_{L^{\infty}}+\|f'(u)\|_{L^{\infty}}\|u\|_{s,p,\delta}).
\end{align*}
Now we prove the case $\sigma=-1$ by a duality argument. Note that
\begin{equation*}
\|v f(u)\|_{-1,q,\nu}=\sup_{\eta\in C_c^{\infty}}
\frac{|\langle vf(u),\eta\rangle_{W^{-1,q}_{\nu}\times W^{1,q'}_{-n-\nu}
}|}{\|\eta\|_{1,q',-n-\nu}}.
\end{equation*}
We have
\begin{align*}
\frac{|\langle vf(u),\eta\rangle_{W^{-1,q}_{\nu}\times W^{1,q'}_{-n-\nu}
}|}{\|\eta\|_{1,q',-n-\nu}}
&= \frac{|\langle v,f(u)\eta\rangle_{W^{-1,q}_{\nu}\times W^{1,q'}_{-n-\nu}
}|}{\|\eta\|_{1,q',-n-\nu}}
\\
&\leq \frac{\|v\|_{-1,q,\nu}\|f(u)\eta\|_{
1,q',-n-\nu}}{\|\eta\|_{1,q',-n-\nu}}
\end{align*}
By assumption $\frac{1}{q}< 1-\theta$, so $\frac{1}{q'}>\theta$
and thus we can apply what was proved for the case $\sigma=1$ to
$\|f(u)\eta\|_{ 1,q',-n-\nu}$:
\begin{align*}
\frac{\|v\|_{-1,q,\nu}\|f(u)\eta\|_{1,q',-n-\nu}}{\|\eta\|_{1,q',-n-\nu}}&\preceq
\frac{\|v\|_{-1,q,\nu}[\|\eta\|_{1,q',-n-\nu}(\|f(u)\|_{L^{\infty}}+\|f'(u)\|_{L^{\infty}}\|u\|_{s,p,\delta})]}{\|\eta\|_{1,q',-n-\nu}}\\
&= \|v\|_{-1,q,\nu}
(\|f(u)\|_{L^{\infty}}+\|f'(u)\|_{L^{\infty}}\|u\|_{s,p,\delta}).
\end{align*}
Therefore
\begin{equation*}
\|v f(u)\|_{-1,q,\nu}\preceq \|v\|_{-1,q,\nu}
(\|f(u)\|_{L^{\infty}}+\|f'(u)\|_{L^{\infty}}\|u\|_{s,p,\delta}).
\end{equation*}
Now we prove the case where $\sigma \in (-1,1)$ by interpolation.
According to what was proved we have
\begin{align}\lab{eqimpor}
&\|v f(u)\|_{1,q_1,\nu} \preceq \|v\|_{1,q_1,\nu}(\|f(u)\|_{L^{\infty}}+\|f'(u)\|_{L^{\infty}}\|u\|_{s,p,\delta}),\\
 &\|v f(u)\|_{-1,q_2,\nu}\preceq \|v\|_{-1,q_2,\nu}
(\|f(u)\|_{L^{\infty}}+\|f'(u)\|_{L^{\infty}}\|u\|_{s,p,\delta}),\lab{eqimpor2}
\end{align}
where $q_1$ and $q_2$ are any two numbers that satisfy $\theta<
\frac{1}{q_1}<1$ and $0<\frac{1}{q_2}<1-\theta$. Let
$t=\frac{1-\sigma}{2}$. Clearly $t\in (0,1)$. Also note that if
we set $\frac{1}{q}=\frac{1-t}{q_1}+\frac{t}{q_2}$ then
\begin{align*}
&\frac{1}{q_1}>\theta,\,\, \frac{1}{q_2}>0 \Rightarrow
\frac{1}{q}>(1-t)\theta=\frac{1+\sigma}{2}\theta. \\
& \frac{1}{q_2}<1-\delta,\,\, \frac{1}{q_1}<1 \Rightarrow
\frac{1}{q}<1-t\theta=1-\frac{1-\sigma}{2}\theta.
\end{align*}
So by choosing appropriate $q_1$ and $q_2$ we can get any $q$
with the property that $\frac{1}{q}\in (\frac{1+\sigma}{2}\theta,
1-\frac{1-\sigma}{2}\theta)$. This implies if $\frac{1}{q}\in
(\frac{1+\sigma}{2}\theta, 1-\frac{1-\sigma}{2}\theta)$ then we
may find $q_1$ and $q_2$ for which inequalities \ref{eqimpor},
\ref{eqimpor2} hold true and moreover
\begin{align*}
&(W^{1,q_1}_{\nu}, W^{-1,q_2}_{\nu})_{t,q}=W^{\sigma,q}_{\nu} \quad \textrm{if $\sigma\neq 0$} \quad (\textrm{real interpolation})\\
&[W^{1,q_1}_{\nu}, W^{-1,q_2}_{\nu}]_t=W^{\sigma,q}_{\nu} \quad
\textrm{if $\sigma= 0$}\quad (\textrm{complex interpolation}).
\end{align*}
So by interpolation we get
\begin{equation*}
\|v f(u)\|_{\sigma,q,\nu}\preceq
\|v\|_{\sigma,q,\nu}(\|f(u)\|_{\infty}+\|f'(u)\|_{\infty}\|u\|_{s,p,\delta}).
\end{equation*}
\end{proof}

%%%%%%%%%%%%%%%%%%%%%%%%%%%%%%%%%%%%%%%%%%%%%%%%%%%%%%%%%%%%%%%%%%%%%%%%%%%%%%
\section{Differential Operators in Weighted Spaces}
   \label{app:operators}

We now assemble some results we need for differential operators in
Weighted spaces.
Again, we limit our selves to simply stating the results we need,
unless the proof of the result is either unavailable or difficult to find
in the form we need, in which case we include a concise proof.

 Let $M$ be an $n$-dimensional AF manifold and let $E$ be a
smooth vector bundle over $M$ with fiber dimension $k$. Consider
the linear differential operator $A:\Gamma(E)\to \Gamma(E)$ of
order $m$ where $\Gamma(E)$ denotes the space of smooth sections
of $E$. By definition, we know that in any local coordinates
(trivializing $E$) $A$ can be written as $A=\sum_{|\nu|\leq
m}a_{\nu}\partial^{\nu}$ where $a_\nu$ is a $\mathbb{R}^{k\times
k}$ valued function.
\begin{definition}\lab{defellipticoperator}
Let $\alpha\in \mathbb{R}$, $\gamma\in (1,\infty)$, and $\rho<0$.
\begin{itemize}
\item We say $A$ belongs to the class
$D^{\alpha,\gamma}_{m}(E)$ if and only if $a_{\nu}\in
W^{\alpha-m+|\nu|,\gamma}$ for $|\nu|\leq m$.
\item We say $A$ belongs to the class
$D^{\alpha,\gamma}_{m,\rho}(E)$ if and only if $a_{\nu}\in
W^{\alpha-m+|\nu|,\gamma}_{\rho-m+|\nu|}$ for $|\nu|< m$ and
there are constants $a_{\nu}^{\infty}$ such that
$a_{\nu}^{\infty}-a_{\nu}\in W^{\alpha, \gamma}_{\rho}$ for all
$|\nu|=m$. We call
$A_{\infty}=\sum_{|\nu|=m}a_{\nu}^{\infty}\partial_\nu$ the
principal part of $A$ at infinity.
\item If $\alpha \gamma >n$, then the highest order coefficients
of $A\in D^{\alpha,\gamma}_{m}(E)$ are continuous and so it makes
sense to talk about their pointwise values. We say $A$ is
elliptic if for each $x$, the constant coefficient operator
$\sum_{|\nu|=m}a_\nu (x)\partial^{\nu}$ is elliptic.
\item If $\alpha \gamma >n$, then the highest order coefficients
of $A\in D^{\alpha,\gamma}_{m,\rho}(E)$ are continuous and so it
makes sense to talk about their pointwise values. We say $A$ is
elliptic if $A_{\infty}$ is elliptic and moreover for each $x$,
the constant coefficient operator $\sum_{|\nu|=m}a_\nu
(x)\partial^{\nu}$ is elliptic.
\end{itemize}
\end{definition}
%By using the multiplication lemma, we can prove the following
%theorem:
\begin{theorem}\lab{thmB1} %\cite{1,2}
If $\delta\in \mathbb{R}$, $\rho<0$ and if $A\in
D^{\alpha,\gamma}_{m,\rho}(E)$ then $A$ can be viewed as a bounded
linear map
\begin{equation*}
A: W^{s,q}_{\delta}(E)\to W^{\sigma,q}_{\delta-m}(E),
\end{equation*}
provided %$\gamma,q\in(1,\infty)$, $s\geq m-\alpha$,
\begin{align*}
  (i) & \ \gamma,q\in(1,\infty), \\
 (ii) & \ s\geq m-\alpha \quad (\textrm{ let $\frac{1}{q}+\frac{1}{\gamma}\geq 1$ if $s=m-\alpha \not \in \mathbb{Z}$}) ,\\
(iii) & \ \sigma\leq \min(s,\alpha)-m \quad (\textrm{let $\gamma \leq q$ if $\alpha-m=\sigma \not \in \mathbb{Z}$})\\
 (iv) & \ \sigma<s-m+\alpha-\dfrac{n}{\gamma}, \\
  (v) & \ \sigma-\dfrac{n}{q}\leq \alpha-\dfrac{n}{\gamma}-m, \\
 (vi) & \ s-n/q> m-n-\alpha+n/\gamma.
\end{align*}
If moreover $A_{\infty}=0$, then $A$ is a continuous map
\begin{equation*}
A: W^{s,q}_{\delta}(E)\to W^{\sigma,q}_{\delta-m+\rho}(E)
\end{equation*}
\end{theorem}
\begin{proof}{\bf (Theorem~\ref{thmB1})}
First let's consider the case where $A_{\infty}\neq
0$. The goal is to find sufficient conditions to make sure that
$A=\sum_{|\nu|\leq m}a_{\nu}\partial^{\nu}$ is a continuous
operator from $W^{s,q}_{\delta}\rightarrow W^{\sigma, q}_{\beta}$.
Clearly this will be true provided
\begin{enumerate}
\item For all $|\nu|<m$
\begin{equation*}
W^{\alpha-m+|\nu|,\gamma}_{\rho-m+|\nu|}\times
W^{s-|\nu|,q}_{\delta-|\nu|}\hookrightarrow W^{\sigma, q}_{\beta},
\quad (\textrm{note that}\,\,  a_{\nu}\in
W^{\alpha-m+|\nu|,\gamma}_{\rho-m+|\nu|},\,\,\partial^{\nu}u \in
W^{s-|\nu|,q}_{\delta-|\nu|})
\end{equation*}
It follows from the multiplication lemma and previously mentioned
embedding theorems that the above embedding holds true provided
(the numbering of the items corresponds to the numbering of the
assumptions in multiplication lemma)
\begin{align*}
 (ii) & \ s\geq m-\alpha,\quad (\textrm{
$\frac{1}{q}+\frac{1}{\gamma}\geq 1$ if $s=m-\alpha \not \in \mathbb{Z}$}) \\
  (i) & \ \sigma \leq \alpha-m \quad(\textrm{$\gamma \leq q$ if
         $\alpha-m=\sigma \not \in \mathbb{Z}$}), \\
(i), (iii) & \ \sigma \leq s-(m-1),\\
 (iv) & \ \sigma< s-m+\alpha-\frac{n}{\gamma}, \\
(iii) & \ \sigma-\frac{n}{q}\leq \alpha-\frac{n}{\gamma}-m,\\
  (v) & \ s-\frac{n}{q}> m-n-\alpha+\frac{n}{\gamma},
\end{align*}
and of course we need $(\rho-m+|\nu|)+(\delta-|\nu|)$ to be less
than or equal to  $\beta$, that is, $\rho-m+\delta\leq \beta$.
\item For $|\nu|=m$
\begin{align*}
&W^{\alpha,\gamma}_{\rho}\times
W^{s-m,q}_{\delta-m}\hookrightarrow W^{\sigma, q}_{\beta},\\
&W^{s-m, q}_{\delta-m}\hookrightarrow W^{\sigma, q}_{\beta}.
\end{align*}
Note that,
$a_{\nu}\partial^{\nu}=(a_{\nu}-a_{\nu}^{\infty})\partial^{\nu}+a_{\nu}^{\infty}\partial^{\nu}$.
$a_{\nu}^{\infty}$ is constant and $(a_{\nu}-a_{\nu}^{\infty})\in
W^{\alpha,\gamma}_{\rho}$, so it should be clear why we need the
above embeddings to be true. By using the multiplication lemma it
turns out that the only extra assumption that we need for the
first embedding to be true is that $\sigma\leq s-m$ and then the
only extra assumption that we need for the second embedding to be
true is that $\beta\geq \delta-m$.
\end{enumerate}
To complete the proof we just need to note that if $A_{\infty}=0$
then we do not need to have the embedding $W^{s-m,
q}_{\delta-m}\hookrightarrow W^{\sigma, q}_{\beta}$ and so
$\beta$ can be any number larger than or equal to
$\delta-m+\rho$.
\end{proof}
%\begin{example}
\begin{remark} In the above proof we implicitly assumed that the
following statement is true:
%\begin{itemize}
%\item If $u\in W^{s,q}_{\delta}(E)$, then the trivialization of
%$u$in any local chart also belongs to $W^{s,q}_{\delta}$.
%\item
If $A: \Gamma(E)\rightarrow \Gamma(E)$ is a partial
differential operator whose representation in each local chart is
continuous from $W^{s,q}_{\delta}$ to $W^{\sigma,q}_{\beta}$,
then $A$ is a continuous operator from $W^{s,q}_{\delta}(E)$ to
$W^{\sigma,q}_{\beta}(E)$.
%\end{itemize}
\end{remark}
\textbf{Example:} If the metric of an asymptotically flat manifold
is of class $W^{\alpha,\gamma}_{\rho}$ with $\alpha\gamma>n$ and
$\rho<0$, then the Laplacian and conformal Laplacian are elliptic
operators in class $D^{\alpha,\gamma}_{2,\rho}(M\times
\mathbb{R})$; vector Laplacian
is an elliptic operator in the class $D^{\alpha,\gamma}_{2,\rho}(TM)$.
%\end{example}

\textbf{Duality Pairing.} Let $\bar{h}$ denote the Euclidean
metric on $\mathbb{R}^n$. Let $\sigma, \delta\in \mathbb{R}$ and
$q\in (1,\infty)$. We denote the duality pairing $W^{-\sigma,
q'}_{-n-\delta}(\mathbb{R}^n) \times
W^{\sigma,q}_{\delta}(\mathbb{R}^n) \rightarrow \mathbb{R}$ by
$\langle \cdot, \cdot \rangle_{W^{-\sigma, q'}_{-n-\delta} \times
W^{\sigma,q}_{\delta}}$ or just $\langle \cdot,\cdot
\rangle_{(\mathbb{R}^n,\bar{h})}$ if the spaces are clear from
the context. Clearly the duality pairing is a continuous bilinear
map. The restriction of this map to
$C_c^{\infty}(\mathbb{R}^n)\times C_c^{\infty}(\mathbb{R}^n)$ is
the $L^2$ inner product:
\begin{equation*}
\forall \,\, u,v\in C_c^{\infty}(\mathbb{R}^n) \quad \langle
u,v\rangle_{(\mathbb{R}^n,\bar{h})}=\int_{\mathbb{R}^n} uv dx.
\end{equation*}

 Now suppose $(M,h)$ is an $n$-dimensional AF manifold of class
$W^{\alpha,\gamma}_{\rho}$ where $\rho<0$ and $\gamma\in
(1,\infty)$. Our claim is that $(W^{\sigma,q}_{\delta}(M))^{*}$
can be identified with $W^{-\sigma, q'}_{-n-\delta}(M)$. This
identification can be done in at least two ways which we describe
below:
\begin{itemizeX}
\item \textbf{First Method:} By using the corresponding AF atlas and the subordinate
partition of unity that was used in the Definition
\ref{defweightedsobolevae} one can construct a smooth metric
$\hat{h}$ such that $(M,\hat{h})$ is of class
$W^{\alpha,\gamma}_{\rho}$. Recall that our definition of Sobolev
spaces on $M$ is independent of the underlying metric. The
bilinear map $\langle \cdot,\cdot
\rangle_{(M,\hat{h})}:C_c^{\infty}(M)\times C_c^{\infty}(M)
\rightarrow \mathbb{R}$ which is defined by
\begin{equation*}
\langle u,v\rangle_{(M,\hat{h})}=\int_{M} uv dV_{\hat{h}}
\end{equation*}
can be uniquely extended to a continuous bilinear form
\begin{equation*}
\langle \cdot, \cdot
\rangle_{(M,\hat{h})}:W^{-\sigma,q'}_{-n-\delta}(M) \times
W^{\sigma,q}_{\delta}(M) \rightarrow \mathbb{R}.
\end{equation*}
The above bilinear map induces a topological isomorphism
$(W^{\sigma,q}_{\delta}(M))^{*}=W^{-\sigma, q'}_{-n-\delta}(M)$;
if $u$, $v$ are smooth and $v$ has compact support in $U_j$
(domain of a coordinate chart in the AF atlas), then
\begin{equation*}
\langle u,v\rangle_{(M,\hat{h})} =\langle u\circ \phi_j^{-1},
\sqrt{\det \hat{h}}\,v\circ
\phi_j^{-1}\rangle_{(\mathbb{R}^n,\bar{h})}.
\end{equation*}
Note that in the above, $u\circ \phi_j^{-1}$ represents any
extension of $u\circ \phi_j^{-1}$ from $W^{-\sigma,
q'}_{-n-\delta}(\phi_j(U_j)) $ to $W^{-\sigma,
q'}_{-n-\delta}(\mathbb{R}^n) $. Also $v\circ \phi_j^{-1}$
represents the extension of $v\circ \phi_j^{-1}\in W^{\sigma,
q}_{\delta}(\phi_j(U_j)) $ by zero. Since $v$ has compact
support, we know that $\sqrt{\det \hat{h}}\,v\circ \phi_j^{-1}\in
W^{\sigma, q}_{\delta}(\mathbb{R}^n)$.

Similarly there exists a continuous bilinear form $\langle
\cdot,\cdot\rangle_{(M,\hat{h})}:W^{-\sigma, q'}_{-n-\delta}(TM)
\times W^{\sigma,q}_{\delta}(TM) \rightarrow \mathbb{R}$ whose
restriction to $C_c^{\infty}(TM)\times C_c^{\infty}(TM)$ is
\begin{equation*}
(Y,X)\mapsto \int_{M} \hat{h}(Y,X) dV_{\hat{h}}.
\end{equation*}
 This map induces an isomorphism $(W^{\sigma,q}_{\delta}(TM))^{*}=W^{-\sigma, q'}_{-n-\delta}(TM)$; if $X\in W^{\sigma,q}_{\delta}(TM)$, $Y\in W^{-\sigma, q'}_{-n-\delta}(TM)$
 are smooth and $X$ has compact support in $U_j$ then
\begin{equation*}
\langle Y,X\rangle_{(M,\hat{h})}=\sum_{l,p} \langle Y_l\circ
\phi_j^{-1}, \sqrt{\det \hat{h}}\,\,\hat{h}^{lp}\,X_p\circ
\phi_j^{-1}\rangle_{(\mathbb{R}^n,\bar{h})}.
\end{equation*}
 The disadvantage of this method is that the
restriction of the bilinear form that was constructed above to
$C_c^{\infty}$ is $\int_{M} uv dV_{\hat{h}}$ instead of $\int_{M}
uv dV_{h}$. We prefer to construct the isomorphism using the
rough metric instead of $\hat{h}$. It turns out that this can be
done for a limited range of $\sigma$ and $q$.
\item \textbf{Second Method:} Suppose $\alpha \gamma>n$. Then
there exists a continuous function $f$ such that $dV_h=f
dV_{\hat{h}}$ and $f-\varsigma\in W^{\alpha,\gamma}_{\rho}$ for
some constant $\varsigma>0$ \cite{2,6}. Formally we can write
\begin{equation*}
\langle u,v\rangle_{(M,h)}=\int_M uv dV_h=\int_M uv f
dV_{\hat{h}}=\int_M u fv dV_{\hat{h}}=\langle
u,fv\rangle_{(M,\hat{h})}.
\end{equation*}
This motivates the following definition:
\begin{equation*}
\forall u\in W^{-\sigma,q'}_{-n-\delta}\,\,\, \forall \, v\in
W^{\sigma,q}_{\delta}\quad \langle u,v\rangle_{(M,h)} :=\langle
u,fv\rangle_{(M,\hat{h})}.
\end{equation*}
Of course for the above definition to make sense we need to make
sure that $fv\in W^{\sigma,q}_{\delta}$. Note that $f-\varsigma\in
W^{\alpha,\gamma}_{\rho}$ and so by Lemma \ref{lempA1} this holds
provided
\begin{align*}
&\sigma\in [-\alpha,\alpha]\,\,\, (\textrm{$\sigma \neq -\alpha $
if $\alpha \not \in \mathbb{N}_0$; $\sigma \neq \alpha $ if
$\alpha \not \in \mathbb{N}_0$ and $q<\gamma$})\\
%,\quad
&\sigma-\frac{n}{q}\in
[-n-\alpha+\frac{n}{\gamma},\alpha-\frac{n}{\gamma}].
\end{align*}
It is easy to see that (since $\alpha \gamma>n$) if $\sigma\in
[0,\alpha]$ and $q=\gamma$
  then the above conditions hold true. Clearly the restriction of
$\langle \cdot,\cdot \rangle_{(M,h)}$ to $C_c^{\infty}\times
C_c^{\infty} $ is given by $\langle u,v\rangle_{(M,h)}=\int_M uv
dV_h$. This shows that this bilinear form does not depend on the
choice of $\hat{h}$. The above pairing makes sense even if $u\in
W^{-\sigma,q'}_{loc}$ and $v\in W^{\sigma,q}_{loc}$ provided at
least one of $u$ or $v$ has compact support.

Similarly for vector fields $X$ and $Y$ formally we may write
\begin{align*}
\langle Y,X\rangle_{(M,h)}&=\int_M h(Y,X) dV_h=\int_M h_{bc}X^c
Y^b
fdV_{\hat{h}}\\
&=\int_M \hat{h}_{ad}(f \hat{h}^{ab}h_{bc}X^c) Y^d dV_{\hat{h}}
\quad \quad(Y^b=\delta_{d}^bY^d= \hat{h}_{ad}\hat{h}^{ab}Y^d)\\
&=\int_M \hat{h}(Y,X_{*})dV_{\hat{h}} =\langle
Y,X_{*}\rangle_{(M,\hat{h})}\quad (X_{*}^a:=f
\hat{h}^{ab}h_{bc}X^c).
\end{align*}
This motivates the following definition:
\begin{equation*}
\forall\,\, Y\in
\textbf{W}^{-\sigma,q'}_{-n-\delta}\,\,\,\forall\, X\in
\textbf{W}^{\sigma,q}_{\delta}\quad \langle Y,X\rangle_{(M,h)}:=
\langle Y,X_{*}\rangle_{(M,\hat{h})},\quad
(\textbf{W}^{\sigma,q}_{\delta}:=W^{\sigma,q}_{\delta}(TM))
\end{equation*}
where $X_{*}^a:=f \hat{h}^{ab}h_{bc}X^c$. Again one can check
that the above definition makes sense provided $\sigma\in
[-\alpha,\alpha]$ $\,\,$ ($\sigma \neq -\alpha $ if $\alpha \not
\in \mathbb{N}_0$; $\sigma \neq \alpha $ if $\alpha \not \in
\mathbb{N}_0$ and $q<\gamma$), $\,\,\,\sigma-\frac{n}{q}\in
[-n-\alpha+\frac{n}{\gamma},\alpha-\frac{n}{\gamma}]$. As an
example, if $n=3$ and $\alpha>1$ (and of course
$\alpha>\frac{3}{\gamma}$), then the duality pairing of
$W^{-1,2}_{-3-\delta}$ and $W^{1,2}_{\delta}$ is well-defined:
\begin{equation*}
1\in (-\alpha,\alpha),\quad 1-\frac{3}{2}\in
[-3-\alpha+\frac{3}{\gamma},\alpha-\frac{3}{\gamma}] \quad
(\textrm{because $\alpha>\frac{3}{\gamma}$}).
\end{equation*}
\end{itemizeX}
\begin{remark}\textbf{Order on $W^{-\sigma,\gamma}_{\delta}(M)$ for $\sigma \in (-\infty,
\alpha ]$} \lab{remorder}

As before suppose $(M,h)$ is an $n$-dimensional AF manifold of
class $W^{\alpha,\gamma}_{\rho}$ where $\rho<0$, $\gamma\in
(1,\infty)$, and $\alpha \gamma>n$.
\begin{itemizeX}
\item If $\sigma \leq 0$, then $W^{-\sigma,q}_{\delta}\hookrightarrow
L^{q}_{\delta}$ and so the elements of $W^{-\sigma,q}_{\delta}$
are ordinary functions (or more precisely, equivalence classes of
ordinary functions). In this case we define an order on
$W^{-\sigma,q}_{\delta}$ as follows: the functions $u, v \in
W^{-\sigma,q}_{\delta}$ satisfy $u\geq v$ if and only if
$u(x)-v(x)\geq 0$ for almost all $x\in M$.
\item If $\sigma \in (0,\alpha]$ ($\sigma \neq \alpha $ if $\alpha \not
\in \mathbb{N}_0$), then it is easy to check that the duality
pairing $\langle \cdot,\cdot\rangle_{(M,h)}:W^{-\sigma,
\gamma}_{\delta}(M) \times W^{\sigma,\gamma'}_{-n-\delta}(M)
\rightarrow \mathbb{R}$ is well-defined. We define an order on
$W^{-\sigma,\gamma}_{\delta}$ as follows: the functions $u, v \in
W^{-\sigma,\gamma}_{\delta}$ satisfy $u\geq v$ if and only if
$\langle u-v,\xi\rangle_{(M,h)}\geq 0$ for all $\xi \in
C_c^{\infty}(M)$ with $\xi\geq 0$. Notice that if $u$ and $v$ are
ordinary functions in $W^{-\sigma,\gamma}_{\delta}(M)$, then it
follows from the definition that $u\geq v$ if and only if
$u(x)\geq v(x)\,  a.e.$.

According to the above items, if $\alpha\geq 1$ we have a
well-defined order on $W^{\alpha-2,\gamma}_{\delta}(M)$ and in
particular if $u$ is an ordinary function in
$W^{\alpha-2,\gamma}_{\delta}(M)$, then $u\geq 0$ if and only if
$u(x)\geq 0$ for almost all $x$.
\end{itemizeX}
\end{remark}
%\begin{comment}
In what follows, we state and prove Lemma \ref{lemb1}, Lemma
\ref{lemb2}, Proposition \ref{propb1}, and Lemma \ref{lemb3} for
$\mathbb{R}^n$. The proofs can be easily extended to AF manifolds.
\begin{lemma}\lab{lemb1}
Let the following assumptions hold:
\begin{itemize}
\item $A\in D^{\alpha, \gamma}_{m}$ where $\gamma \in (1,\infty)$ and $\alpha-\frac{n}{\gamma}>\max\{0,\frac{m-n}{2}\}$; $A$ is elliptic.
%\item $q\in (1,\infty)$.
%$\alpha-\frac{n}{\gamma}>\max\{0,\frac{m-n}{2}\}$.
\item $q\in (1,\infty)$, $s\in(m-\alpha,\alpha]$ \quad (\textrm{if $s=\alpha \not \in \mathbb{N}_0$, then let $q \in [\gamma, \infty)$}).
\item $s-\frac{n}{q}\in
(m-n-\alpha+\frac{n}{\gamma},\alpha-\frac{n}{\gamma}]$.
\end{itemize}
Then: %$A: W^{s,q}\rightarrow W^{s-m,q}$ is continuous. Moreover
If $U$ and $V$ are bounded open sets with $U\subset\subset V$,
then there exists $\tilde{s}<s$ such that for all $u\in W^{s,q}$
\begin{equation}\lab{eqestimate111}
\|u\|_{s,q, U}\preceq \|Au\|_{s-m,q, V}+\|u\|_{\tilde{s},q, V}.
\end{equation}
\end{lemma}
Note: The assumptions in the statement of the lemma are to ensure
that $A\in D^{\alpha, \gamma}_{m}$ sends elements of $W^{s,q}$ to
elements of $W^{s-m,q}$. In fact the conditions in Theorem
\ref{thmB1} work for unweighted spaces too and the restrictions
in the statement of the above lemma agree with the conditions in
Theorem \ref{thmB1}. The assumption
$\alpha-\frac{n}{\gamma}>\frac{m-n}{2}$ is to ensure that the
interval $(m-n-\alpha+\frac{n}{\gamma},\alpha-\frac{n}{\gamma}]$
is nonempty.
\begin{proof}{\bf (Lemma~\ref{lemb1})}
The proof of the interior regularity lemma in
\cite{HNT07b} (Lemma A.25), with obvious changes, goes through for the
above claim as well. The approach of the proof is similar to our
proof for Proposition \ref{propb1}. Since the claim is about
unweighted Sobolev spaces we do not repeat that argument here.
\end{proof}
%%%%%%%%%%%%%%%%%%%%%%%%%%%%%%%%%%%%%%%
\begin{lemma}\lab{lemb2}
Suppose $A$ is a constant coefficient elliptic operator that has
only derivatives of order $m$ with $m < n$ on $\mathbb{R}^n$.
%(clearly $A\in D^{\alpha ,\gamma}_{m,\rho}$ for all possible
%$\alpha, \gamma$, and $\rho<0$ because one can take $A_{\infty}=A$
%and so $A_{\infty}-A=0$; of course for our definition of
%ellipticity to work, $\alpha \gamma$ must be larger than $n$).
Then for $s\in \mathbb{R}$, $p\in (1,\infty)$, and $\delta \in
(m-n,0)$, $A: W^{s,p}_{\delta}\rightarrow
W^{s-m,p}_{\delta-m}$ is an isomorphism.
\end{lemma}
\begin{proof}{\bf (Lemma~\ref{lemb2})}
We closely follow and extend the proof that is given
for the special $p=2$ in \cite{2} [Lemma 4.8]. Let
$A_{s,p,\delta}$ denote the operator $A$ acting on
$W^{s,p}_{\delta}$. We consider three cases $s\geq m$, $s\in
(-\infty,0]$, and $s\in (0,m)$.
\begin{itemizeX}
\item \textbf{Case 1: $s \geq m$.} \\
For $s\in \mathbb{N}$ and $s\geq m$, the claim follows from the
argument in \cite{8}. If $s \not \in \mathbb{N}$, let $k=[s]$,
$\theta= s-k$. We know that $A_{k,p,\delta}$ and
$A_{k+1,p,\delta}$ have inverses and in fact
\begin{align*}
&A_{k,p,\delta}^{-1}: W^{k-m,p}_{\delta-m}\rightarrow
W^{k,p}_{\delta},\\
&A_{k+1,p,\delta}^{-1}: W^{k+1-m,p}_{\delta-m}\rightarrow
W^{k+1,p}_{\delta},
\end{align*}
are continuous maps. Note that
\begin{align*}
& W^{k+1,p}_{\delta}\hookrightarrow W^{k,p}_{\delta}, \quad
W^{k+1-m,p}_{\delta-m}\hookrightarrow W^{k-m,p}_{\delta-m},\\
&(W^{k,p}_{\delta},
W^{k+1,p}_{\delta})_{\theta,p}=W^{s,p}_{\delta}, \quad
(W^{k-m,p}_{\delta-m},
W^{k+1-m,p}_{\delta-m})_{\theta,p}=W^{s-m,p}_{\delta-m}.
\end{align*}
 So by
interpolation we get a continuous operator $T:
W^{s-m,p}_{\delta-m}\rightarrow W^{s,p}_{\delta}$ which must be
the restriction of $A_{k,p,\delta}^{-1}$ to
$W^{s-m,p}_{\delta-m}$. Now for all $u\in W^{s,p}_{\delta}$ we
have
\begin{align*}
u\in W^{s,p}_{\delta}\hookrightarrow W^{k,p}_{\delta} \Rightarrow
A_{s,p,\delta} u= A_{k,p,\delta} u \Rightarrow
T(A_{s,p,\delta}u) & =T(A_{k,p,\delta} u)
\\
& =A_{k,p,\delta}^{-1}(A_{k,p,\delta} u)=u.
\end{align*}
Similarly $A_{s,p,\delta}T u=u$. It follows that $T=
A_{s,p,\delta}^{-1}$.
\item \textbf{Case 2: $s\leq 0$.} \\
We want to show that $A_{s,p,\delta}: W^{s,p}_{\delta}\rightarrow
W^{s-m,p}_{\delta-m}$ is an isomorphism. We note that since
$A_{s,p,\delta}$ is a homogeneous constant coefficient elliptic
operator, its adjoint $(A_{s,p,\delta})^{*}:
W^{-s+m,p'}_{-\delta-n+m}\rightarrow W^{-s,p'}_{-\delta-n}$ is
also a homogeneous constant coefficient elliptic operator. So by
what was proved in the previous case we know that if $-s+m\geq m$
and $-\delta-n+m \in (m-n,0)$ (which are true because by
assumption $s\leq 0$ and $\delta \in (m-n,0)$) then
$(A_{s,p,\delta})^{*}$ is an isomorphism. Now for $u\in
W^{s-m,p}_{\delta-m}$ define the distribution $Tu$ by
\begin{align*}
\langle Tu, \varphi \rangle
& =\langle u,
  ((A_{s,p,\delta})^{*})^{-1}\varphi\rangle_{W^{s-m,p}_{\delta-m}\times
(W^{s-m,p}_{\delta-m}) ^{*}}
\\
& \quad \quad (\textrm{note
that}\,\,((A_{s,p,\delta})^{*})^{-1}:
W^{-s,p'}_{-\delta-n}\rightarrow(W^{s-m,p}_{\delta-m} )^{*}),
\end{align*}
for all $\varphi \in C_c^{\infty}$. We claim that $T$ is the
inverse of $A_{s,p,\delta}$. To this end first we show that $T$
sends $W^{s-m,p}_{\delta-m}$ to $W^{s,p}_{\delta}$ and then we
show that the composition of $T$ and $A_{s,p,\delta}$ is the
identity map.

Suppose $u\in W^{s-m,p}_{\delta-m}$. In order to prove that $Tu
\in W^{s,p}_{\delta}$ it is enough to show that
\begin{equation*}
\|Tu\|_{s,p,\delta}=\sup_{\varphi \in C_c^{\infty}}
 \frac{|\langle Tu, \varphi\rangle|}
      {\|\varphi\|_{-s,p',-\delta-n}}<\infty \quad (\textrm{we
are interpreting}\,\,  W^{s,p}_{\delta}\,\, \textrm{as}\,\,
(W^{-s,p'}_{-\delta-n})^{*})
\end{equation*}
We have
\begin{align*}
|\langle Tu, \varphi\rangle|&
\leq \|u\|_{s-m,p,\delta-m}\|((A_{s,p,\delta})^{*})^{-1}\varphi\|_{(W^{s-m,p}_{\delta-m})
^{*}}
\\
& \leq
\|u\|_{s-m,p,\delta-m}\|((A_{s,p,\delta})^{*})^{-1}\|_{op}\|\varphi\|_{-s,p',-\delta-n}
\end{align*}
Therefore
\begin{equation*}
\|Tu\|_{s,p,\delta} \leq
\|u\|_{s-m,p,\delta-m}\|((A_{s,p,\delta})^{*})^{-1}\|_{op}<
\infty.
\end{equation*}
This implies that $T$ sends $W^{s-m,p}_{\delta-m}$ to
$W^{s,p}_{\delta}$. Now note that for all $u\in
W^{s,p}_{\delta}$, $\varphi\in C_c^{\infty}$
\begin{equation*}
\langle T A_{s,p,\delta} u, \varphi\rangle
= \langle A_{s,p,\delta} u, ((A_{s,p,\delta})^{*})^{-1}\varphi\rangle
=\langle u, (A_{s,p,\delta})^{*}((A_{s,p,\delta})^{*})^{-1}\varphi\rangle
=\langle u,\varphi\rangle.
\end{equation*}
This means $T A_{s,p,\delta} u= u$. Similarly
$A_{s,p,\delta}Tu=u$.
\item \textbf{Case 3: $s\in (0,m)$.} \\
By what was proved in the previous cases we know that
$A_{0,p,\delta}$ and $A_{m,p,\delta}$ are invertible. In fact
\begin{align*}
&A_{0,p,\delta}^{-1}: W^{-m,p}_{\delta-m}\rightarrow
W^{0,p}_{\delta},\\
&A_{m,p,\delta}^{-1}: W^{0,p}_{\delta-m}\rightarrow
W^{m,p}_{\delta},
\end{align*}
are continuous maps. Let $\theta=\frac{s}{m}$. Note that
\begin{align*}
& W^{m,p}_{\delta}\hookrightarrow W^{0,p}_{\delta}, \quad
W^{0,p}_{\delta-m}\hookrightarrow W^{-m,p}_{\delta-m},\\
&(W^{0,p}_{\delta}, W^{m,p}_{\delta})_{\theta,p}=W^{s,p}_{\delta},
\quad (W^{-m,p}_{\delta-m},
W^{0,p}_{\delta-m})_{\theta,p}=W^{s-m,p}_{\delta-m}\quad \textrm{if $s\not \in \mathbb{N} $},\quad \textrm{(real interpolation)}\\
& [W^{0,p}_{\delta}, W^{m,p}_{\delta}]_{\theta}=W^{s,p}_{\delta},
\quad [W^{-m,p}_{\delta-m},
W^{0,p}_{\delta-m}]_{\theta}=W^{s-m,p}_{\delta-m} \quad
\textrm{if $s \in \mathbb{N} $}.\quad \textrm{(complex
interpolation)}
\end{align*}
 So by
interpolation we get a continuous operator $T:
W^{s-m,p}_{\delta-m}\rightarrow W^{s,p}_{\delta}$ which must be
the restriction of $A_{0,p,\delta}^{-1}$ to
$W^{s-m,p}_{\delta-m}$. Now for all $u\in W^{s,p}_{\delta}$ we
have
\begin{equation*}
u\in W^{s,p}_{\delta}\hookrightarrow W^{0,p}_{\delta} \Rightarrow
A_{s,p,\delta} u= A_{0,p,\delta} u \Rightarrow
T(A_{s,p,\delta}u)=T(A_{0,p,\delta}
u)=A_{0,p,\delta}^{-1}(A_{0,p,\delta} u)=u.
\end{equation*}
Similarly $A_{s,p,\delta}T u=u$. It follows that $T=
A_{s,p,\delta}^{-1}$.
\end{itemizeX}
\end{proof}
\begin{proposition}\lab{propb1}
Let the following assumptions hold:
\begin{itemize}
\item $A\in D^{\alpha, \gamma}_{m,\rho}$ where $\gamma \in (1,\infty)$,  $\alpha>\frac{n}{\gamma}$, $\rho<0$, and $m<n$; $A$ is elliptic.
%\item $q\in (1,\infty)$, $m< n$,
%$\alpha>\frac{n}{\gamma}$.%>\max\{0,\frac{m-n}{2}\}$.
%\textbf{byassumption $m<n$ so $\frac{m-n}{2}<0$, delete the max
%part}
\item $q\in (1,\infty)$, $s\in(m-\alpha,\alpha]$ \quad (\textrm{if $s=\alpha \not \in \mathbb{N}_0$, then let $q \in [\gamma, \infty)$}).
\item $s-\frac{n}{q}\in
(m-n-\alpha+\frac{n}{\gamma},\alpha-\frac{n}{\gamma}]$.
\item $\delta \in (m-n,0)$.
\end{itemize}
In particular note that if the elliptic operator $A\in D^{\alpha,
\gamma}_{m,\rho}$, is given, then $s:=\alpha$ and $q:=\gamma$
satisfy the desired
conditions.

Then: If $t<s$ and $\delta'>\delta$, then for every $u\in W^{s,q}_{\delta}$ we have
\begin{equation*}
\|u\|_{s,q,\delta}\preceq \|Au\|_{s-m, q,
\delta-m}+\|u\|_{t,q,\delta'}
\end{equation*}
Moreover $A: W^{s,q}_{\delta}\rightarrow W^{s-m,q}_{\delta-m}$ is
semi-Fredholm.
\end{proposition}
\begin{proof}{\bf (Proposition~\ref{propb1})}
The approach of the proof is standard. Here we will closely
follow the proof that is given for the case $q=2$ in \cite{2}
[Lemma 4.9]. In the proof we use the following facts (for these
facts $s, \delta\in \mathbb{R}$ and $p\in (1,\infty)$):
\begin{itemizeX}
\item \textbf{Fact 1:}(see Lemma \ref{lemA9}) If $f\in C_c^{\infty}(\mathbb{R}^n)$ and $u\in
W^{s,p}(\mathbb{R}^n)$, then $fu\in W^{s,p}(\mathbb{R}^n)$ and
moreover $\|fu\|_{s,p}\preceq \|u\|_{s,p}$ (the implicit constant
may depend on $f$ but is independent of $u$).
\item \textbf{Fact 2:}(see Lemma \ref{lemA10})
%If $f\in C_c^{\infty}(\mathbb{R}^n)$ and $u\in
%W^{s,p}_{\delta}(\mathbb{R}^n)$, then $fu\in
%W^{s,p}_{\delta}(\mathbb{R}^n)$ and moreover
%$\|fu\|_{s,p,\delta}\preceq \|u\|_{s,p,\delta}$ (the implicit
%constant may depend on $f$ but is independent of $u$).
%\item \textbf{Fact 3:}
Let $\chi \in C_c^{\infty}(\mathbb{R}^n)$ be a cutoff function
equal to $1$ on $B_1$ and equal to $0$ outside of $B_2$. Let
$\tilde{\chi}(x)=1-\chi(x)$ and for all $\epsilon>0$ define
$\chi_{\epsilon}(x)=\chi(\frac{x}{\epsilon}),
\tilde{\chi}_{\epsilon}(x)=\tilde{\chi}(\frac{x}{\epsilon})$.
Then we have $\|\chi_{\epsilon}u\|_{s,p,\delta}\preceq
\|u\|_{s,p,\delta}$ and
$\|\tilde{\chi}_{\epsilon}u\|_{s,p,\delta}\preceq
\|u\|_{s,p,\delta}$.
%\begin{itemizeX}
%\item $\|\chi_{\epsilon}u\|_{s,p,\delta}\preceq
%\|u\|_{s,p,\delta}$. This is a direct consequence of \textbf{Fact 2} (the
%implicit constant may depend on $\epsilon$ but is independent of
%$u$).
%\item $\|\tilde{\chi}_{\epsilon}u\|_{s,p,\delta}\preceq
%\|u\|_{s,p,\delta}$. Note that $\tilde{\chi}_{\epsilon}$ does not
%have compact support. $\chi$ can be defined in such a way that
%the implicit constant in this inequality becomes independent of
%$\epsilon$. %(\cite{6}).
%\end{itemizeX}
\item \textbf{Fact 3:} Let $u\in
W^{s,p}_{\delta}(\mathbb{R}^n)$.
%with $s\geq 0$
Also let $\Omega$ be an open bounded subset of $\mathbb{R}^n$.
Then
\begin{itemizeXX}
\item $u\in W^{s,p}(\Omega)$ and $\|u\|_{W^{s,p}(\Omega)}\preceq
\|u\|_{W^{s,p}_{\delta}(\mathbb{R}^n)}$.
\item If $\supp u \subseteq \Omega$, then $u\in W^{s,p}(\mathbb{R}^n)$ and
$\|u\|_{W^{s,p}(\Omega)}\simeq
\|u\|_{W^{s,p}_{\delta}(\mathbb{R}^n)}$.
\end{itemizeXX}
If $s\in \mathbb{N}_0$, the above items follow from the fact that
 weights are of the form $\langle x\rangle^a$ and so they attain their
maximum and minimum on any compact subset of $\mathbb{R}^n$. If $s$ is not an
integer, they can be proved by interpolation.
\item \textbf{Fact 4:} Suppose $f\in W^{s,p}_{\delta}$ with
$\delta>0$ and $f$ vanishes in a neighborhood of the origin. Then
\begin{equation*}
\lim_{i\rightarrow \infty}\|S_{2^{-i}}f\|_{s,p,\delta}=0.
\end{equation*}
The reason is as follows: by assumption there exists $l\in
\mathbb{N}$ such that $f=0$ on $B_{2^{-l}}$. So if $\hat{l}\in
\mathbb{Z}$ and $\hat{l}<-l-1$ then $S_{2^{\hat{l}}}f=0$ on $B_2$.
Indeed,
\begin{equation*}
x\in B_2 \,\, \Rightarrow \,\,|2^{\hat{l}} x|<2^{\hat{l}+1}<2^{-l}
\,\,\Rightarrow\,\, f(2^{\hat{l}}x)=0\,\, \Rightarrow \,\,
S_{2^{\hat{l}}}f(x)=0.
\end{equation*}
So for $i>l+2$ we can write
\begin{align*}
\|S_{2^{-i}}f\|_{s,p,\delta}^p&=\sum_{j=0}^{\infty} 2^{-p\delta
j}\|S_{2^j}(\phi_j S_{2^{-i}}f)\|_{W^{s,p}(\mathbb{R}^n)}^{p}
\\
&=\sum_{j=0}^{\infty}
2^{-p\delta j}\|\phi S_{2^{j-i}}f\|_{W^{s,p}(B_2)}^{p}
\quad \quad (S_{2^j}\phi_j=\phi, \,\, \supp \phi \subseteq B_2)\\
&=\sum_{j=i-l-1}^{\infty} 2^{-p\delta j}\|\phi
S_{2^{j-i}}f\|_{W^{s,p}(B_2)}^{p}\quad (S_{2^{j-i}}f=0\,\,
\textrm{on}\,\, B_2\,\, \textrm{if}\,\, j-i<-l-1)\\
&=\sum_{\hat{j}=1}^{\infty} 2^{-p\delta( \hat{j}+i-l-2)}\|\phi
S_{2^{\hat{j}-l-2}}f\|_{W^{s,p}(B_2)}^{p}\quad
(\hat{j}:=j-(i-l)+2)\\
%&=2^{-p\delta(i-l-2)}\sum_{\hat{j}=1}^{\infty}2^{-p\delta
%\hat{j}} \|S_{2^{\hat{j}}}\phi_{\hat{j}} S_{2^{\hat{j}}}
%S_{2^{-l-2}}
%f\|_{W^{s,p}(\mathbb{R}^n)}^{p}\\
&=2^{-p\delta(i-l-2)}\sum_{\hat{j}=1}^{\infty}2^{-p\delta
\hat{j}} \|S_{2^{\hat{j}}}(\phi_{\hat{j}} S_{2^{-l-2}}
f)\|_{W^{s,p}(\mathbb{R}^n)}^{p}\\
&=2^{-p\delta(i-l-2)} \|S_{2^{-l-2}}f\|_{s,p,\delta}^p\preceq
2^{-p\delta(i-l-2)} \|f\|_{s,p,\delta}^p.
\end{align*}
It follows that $\lim_{i\rightarrow
\infty}\|S_{2^{-i}}f\|_{s,p,\delta}=0$.
\item \textbf{Fact 5:} [\textbf{Equivalence Lemma}]\cite{17}
Let $E_1$ be a Banach space, $E_2$, $E_3$ normed spaces, and let
$A\in L(E_1, E_2)$, $B\in L(E_1, E_3)$ such that one has:
\begin{itemizeXX}
\item $\| u \|_{1}  \lesssim \| A u
\|_{2}+ \| B u\|_{3}.$
\item $B$ is compact.
\end{itemizeXX}
Then  $ker A$ is finite dimensional and the range of $A$ is
closed, i.e., $A$ is semi-Fredholm.
\end{itemizeX}
Now let's start the proof. Let $A=A_{\infty}+R$ where
$A_{\infty}$ is the principal part of $A$ at infinity. Let $r$ be
a fixed dyadic integer to be selected later and let
$u_r=\tilde{\chi}_r u$. By Lemma \ref{lemb2} we have
\begin{equation*}
\|u_r\|_{s,q,\delta}\preceq \|A_{\infty}u_r\|_{s-m,q,\delta-m}\leq
\|Au_r\|_{s-m,q,\delta-m}+\|R u_r\|_{s-m,q,\delta-m}.
\end{equation*}
The implicit constant in the above inequality does not depend on
$r$. Now $R\in D^{\alpha, \gamma}_{m,\rho}$ has vanishing
principal part at infinity. Therefore by Theorem {\ref{thmB1}} we
can consider $R$ as a continuous operator from
$W^{s,q}_{\delta-\rho}$ to $W^{s-m,q}_{(\delta-\rho
-m)+\rho}=W^{s-m,q}_{\delta -m}$. Also since $\rho<0$ we have
$u_r\in W^{s,q}_{\delta}\hookrightarrow W^{s,q}_{\delta-\rho}$.
Consequently
\begin{equation*}
\|R u_r\|_{s-m,q,\delta-m} \preceq
\|R\|_{op}\|u_r\|_{s,q,\delta-\rho}=\|R\|_{op}\|\tilde{\chi}_{\frac{r}{2}}u_r\|_{s,q,\delta-\rho}\quad
(\textrm{note that}\,\, \tilde{\chi}_{\frac{r}{2}}u_r=u_r)
\end{equation*}
Now it is easy to check that $W^{\alpha,q}_{-\rho}\times
W^{s,q}_{\delta}\hookrightarrow W^{s,q}_{\delta-\rho}$. Therefore
\begin{equation*}
\|R u_r\|_{s-m,q,\delta-m} \preceq
\|R\|_{op}\|\tilde{\chi}_{\frac{r}{2}}\|_{\alpha,q,-\rho}\|u_r\|_{s,q,\delta}.
\end{equation*}
%By \textbf{Fact 3}, $\|u_r\|_{s,q,\delta}\preceq \|u\|_{s,q,\delta}$ and
%the implicit constant does not depend on $r$. Also by \textbf{Fact 5},
By \textbf{Fact 4}, $\lim_{i\rightarrow
\infty}\|\tilde{\chi}_{2^i}\|_{\alpha,q,-\rho}\rightarrow 0$.
Thus we can choose the fixed dyadic number $r$ large enough so that
\begin{equation*}
\|R\|_{op}\|\tilde{\chi}_{\frac{r}{2}}\|_{\alpha,q,-\rho}<\frac{1}{2},
\end{equation*}
and so we get
\begin{equation*}
\|u_r\|_{s,q,\delta}\preceq
\|Au_r\|_{s-m,q,\delta-m}+\frac{1}{2}\|u_r\|_{s,q,\delta}.
\end{equation*}
Hence
\begin{equation*}
\|u_r\|_{s,q,\delta}\preceq \|Au_r\|_{s-m,q,\delta-m}\leq
\|\tilde{\chi}_r
Au\|_{s-m,q,\delta-m}+\|[A,\tilde{\chi}_r]u\|_{s-m,q,\delta-m}.
\end{equation*}
By \textbf{Fact 2}, $\|\tilde{\chi}_r Au\|_{s-m,q,\delta-m}\preceq
\|Au\|_{s-m,q,\delta-m}$. Also one can easily show that
$[A,\tilde{\chi}_r]u$ has support in $B_{2r}$ and so by \textbf{Fact 3},
$\|[A,\tilde{\chi}_r]u\|_{s-m,q,\delta-m}\simeq
\|[A,\tilde{\chi}_r]u\|_{W^{s-m,q}(B_{2r})}$. On the bounded
domain $B_{2r}$, $[A,\tilde{\chi}_r]\in D^{\alpha,
\gamma}_{m-1}$, so
$[A,\tilde{\chi}_r]:W^{s-1,q}(B_{2r})\rightarrow
W^{s-m,q}(B_{2r})$ is continuous. Consequently
\begin{equation*}
\|[A,\tilde{\chi}_r]u\|_{W^{s-m,q}(B_{2r})}\preceq
\|u\|_{W^{s-1,q}(B_{2r})}\preceq \|u\|_{W^{s,q}(B_{2r})}.
\end{equation*}
Thus
\begin{equation*}
\|u_r\|_{s,q,\delta}\preceq
\|Au\|_{s-m,q,\delta-m}+\|u\|_{W^{s,q}(B_{2r})}.
\end{equation*}
Now we can write
\begin{align*}
\|u\|_{s,q,\delta}&=\|u_r+\chi_r u\|_{s,q,\delta}\leq
\|u_r\|_{s,q,\delta}+\|\chi_r u\|_{s,q,\delta}\\
&\preceq \|u_r\|_{s,q,\delta}+\|\chi_r u\|_{W^{s,q}(B_{2r})}\quad
(\chi_r u \,\, \textrm{has support in}\,\, B_{2r},\,\, \textbf{Fact 3})\\
&\preceq \|u_r\|_{s,q,\delta}+\| u\|_{W^{s,q}(B_{2r})}
\quad (\textbf{Fact 1})\\
&\preceq \|Au\|_{s-m,q,\delta-m}+\|u\|_{W^{s,q}(B_{2r})}.
\end{align*}
From interior regularity estimate for elliptic operators on
unweighted Sobolev spaces [Lemma \ref{lemb1}] we know there
exists $\tilde{s}<s$ such that
\begin{equation}\lab{eqlemb1}
\|u\|_{W^{s,q}(B_{2r})}\preceq
\|Au\|_{W^{s-m,q}(B_{3r})}+\|u\|_{W^{\tilde{s},q}(B_{3r})},
\end{equation}
and by \textbf{Fact 3}
\begin{equation*}
\|Au\|_{W^{s-m,q}(B_{3r})} \preceq \|Au\|_{s-m,q,\delta-m}.
\end{equation*}
It follows that
\begin{equation*}
\|u\|_{s,q,\delta}\preceq
\|Au\|_{s-m,q,\delta-m}+\|u\|_{W^{\tilde{s},q}(B_{3r})}.
\end{equation*}
But by \textbf{Fact 3}, for any $\delta'\in \mathbb{R}$ we have
$\|u\|_{W^{\tilde{s},q}(B_{3r})}\preceq
\|u\|_{\tilde{s},q,\delta'}$. This implies
\begin{equation}\lab{eqestimate555}
\|u\|_{s,q,\delta}\preceq
\|Au\|_{s-m,q,\delta-m}+\|u\|_{\tilde{s},q,\delta'}.
\end{equation}
Now, if $t<s$ then either $t\geq \tilde{s}$ or $t<\tilde{s}$. If
$t\geq \tilde{s}$ then $W^{t,q}_{\delta'}\hookrightarrow
W^{\tilde{s},q}_{\delta'}$ and so $\|u\|_{\tilde{s},q, \delta'}
\preceq \|u\|_{t,q, \delta'}$. If $t<\tilde{s}$, then for
$\delta'>\delta$ we have $W^{s,q}_{\delta}\hookrightarrow
W^{\tilde{s},q}_{\delta'}\hookrightarrow W^{t,q}_{\delta'}$ where
the first embedding is compact and the second is continuous.
Therefore, by Ehrling's lemma, for all $\epsilon>0$ there exists
$C(\epsilon)$ such that
\begin{equation*}
\|u\|_{\tilde{s},q, \delta'}\leq \epsilon \|u\|_{s,q,
\delta}+C(\epsilon)\|u\|_{t,q, \delta'}
\end{equation*}
In particular the above inequality holds for
$\epsilon=\frac{1}{2}$. Combining this with (\ref{eqestimate555})
we can conclude that for all $t<s$ and $\delta'>\delta$
\begin{equation*}
\|u\|_{s,q,\delta}\preceq
\|Au\|_{s-m,q,\delta-m}+\|u\|_{t,q,\delta'}.
\end{equation*}
It remains to show that $A: W^{s,q}_{\delta}\rightarrow
W^{s-m,q}_{\delta-m}$ is semi-Fredholm. Pick any $\delta'$
strictly larger than $\delta$. By assumption $s>m-\alpha$, so we
have $\|u\|_{s,q,\delta}\preceq
\|Au\|_{s-m,q,\delta-m}+\|u\|_{m-\alpha,q,\delta'}$. Also
$W^{s,q}_{\delta}\hookrightarrow W^{m-\alpha,q}_{\delta'}$ is
compact. Hence by the estimate that was proved above and \textbf{Fact 5},
$A: W^{s,q}_{\delta}\rightarrow W^{s-m,q}_{\delta-m}$ is
semi-Fredholm.
\end{proof}
%%%%%%%%%%%%%%%%%%%%%%%%%%%%%%%%%%%%%%%
%%%%%%%%%%%%%%%%%%%%%%%%%%%%%%%%%%%%%%%
\begin{remark}\lab{remb1}
The proof of Proposition \ref{propb1} in fact shows that if $u\in
W^{t,q}_{\delta'}$ for some $t<s$ and $\delta'>\delta$ and if $Au
\in W^{s-m,q}_{\delta-m}$ then $u\in W^{s,q}_{\delta}$.
\end{remark}
\begin{lemma}\lab{lemb3}
Let the following assumptions hold:
\begin{itemize}
\item $A\in D^{\alpha, \gamma}_{m,\rho}$, $\gamma\in (1,\infty)$, $\alpha>\frac{n}{\gamma}$, $\rho<0$, and $m<n$. $A$ is elliptic.
%\item $\gamma\in (1,\infty)$, $\rho<0$, $m<n$,
%$\alpha>\frac{n}{\gamma}$.
\item $e\in(m-\alpha,\alpha]$ \quad (\textrm{if $e=\alpha \not \in \mathbb{N}_0$, then let $q \in [\gamma, \infty)$}).
\item $e-\frac{n}{q}\in
(m-n-\alpha+\frac{n}{\gamma},\alpha-\frac{n}{\gamma}]$.
%\item $\rho<0$, $\delta \in (m-n,0)$ and $\delta' \in \mathbb{R}$.
\end{itemize}
Then: If $u\in W^{e,q}_{\beta}$ for some $\beta<0$ satisfies
$Au=0$, then $u\in W^{e,q}_{\beta'}$ for all $\beta'\in (m-n,0)$.
\end{lemma}
\begin{proof}{\bf (Lemma~\ref{lemb3})}
Following the proof of \cite{dM06} [Lemma 3.8], let
$A=A_{\infty}+R$ where $A_{\infty}$ is the principal part of $A$
at infinity. $R$ has vanishing principal part at infinity and
therefore by Theorem \ref{thmB1}, $Ru \in
W^{e-m,q}_{\beta-m+\rho}$. Consequently $A_{\infty}u=-Ru \in
W^{e-m,q}_{\beta-m+\rho}$.
%But by
%Lemma (\ref{lemb2}),
Now we may consider two cases:
\begin{itemizeX}
\item If $\beta+\rho\leq m-n$, then $\beta-m+\rho\leq -n$ and so $W^{e-m,q}_{\beta-m+\rho}\hookrightarrow
W^{e-m,q}_{\eta}$ for all $\eta\geq -n$. Consequently $A_{\infty}u
\in W^{e-m,q}_{\eta}$ for all $\eta \geq -n$. Since $A_{\infty}:
W^{e,q}_{\beta'}\rightarrow W^{e-m,q}_{\beta'-m}$ is an
isomorphism for all $\beta' \in (m-n,0)$ we conclude that $u \in
W^{e,q}_{\beta'}$ for all $\beta'\in (m-n,0)$.
\item If $\beta +\rho>m-n$ then $A_{\infty}: W^{e,q}_{\beta+\rho}\rightarrow
W^{e-m,q}_{\beta-m+\rho}$ is an isomorphism. and therefore $u\in
W^{e,q}_{\beta+\rho}$ which implies $u\in W^{e,q}_{\beta'}$ for
all $\beta' \in (\beta+\rho,0)$
\end{itemizeX}
Combining the above observations, we can conclude that $u\in
W^{e,q}_{\beta'}$ for every $\beta'\in
(\max(m-n,\beta+\rho),0)$.

Now clearly for some $N\in \mathbb{N}$ we have $\beta+N\rho< m-n$
and therefore by iteration we get $u\in W^{e,q}_{\beta'}$ for every
$\beta'\in (m-n,0)$.
\end{proof}
\begin{remark}
In contrast to the notation that was introduced in the main text,
in this Appendix and in particular in the statement of lemmas
\ref{lemb8} and \ref{lemb4} we do not use the notation $A_L$ for
the Laplace operator when it acts on weighted Sobolev spaces.
\end{remark}
\begin{lemma}[Maximum principle]\lab{lemb8}
Suppose $(M^n, h)$ is an AF manifold of class $W^{s,p}_{\delta}$
where $s \in (\frac{n}{p},\infty)\cap [1,\infty)$,$p\in
(1,\infty)$, and $\delta<0$. Also suppose $a\in
W^{s-2,p}_{\eta-2}$, $\eta\in \mathbb{R}$, $\eta<0$.
%\begin{itemize}
Suppose that $a \geq 0$.
\begin{itemize}
\item(a) If $u\in W^{s,p}_{\rho}$ for some $\rho<0$ satisfies
\begin{equation*}
-\Delta_h u+ a u \leq 0
\end{equation*}
then $u \leq 0$. In particular, if $-\Delta_h u+ a u = 0$, then
applying this result to $u$ and $-u$ shows that $u=0$.
\item(b) Suppose that $u\in W^{s,p}_{\rho}$ is nonpositive and satisfies
\begin{equation*}
-\Delta_h u \leq 0 .
\end{equation*}
If $u(x)=0$ at some point $x\in M$, then $u$ vanishes identically.
\end{itemize}
\end{lemma}
\begin{proof}{\bf (Lemma~\ref{lemb8})}
For (a), we combine the proof that is given in
\cite{HNT07b} for the case of closed manifolds and the proof that is
given in \cite{2} for the case where $p=2$. Fix $\epsilon>0$. By
assumptions $u\in W^{s,p}_{\rho}\hookrightarrow C^{0}_{\rho}$ and
therefore $u$ goes to zero at infinity. Therefore if we let
$v=(u-\epsilon)^{+}:= \max\{u-\epsilon, 0\}$, then $v$ is compactly
supported. Note that if $f\in W^{1,q}_{loc}$ then $f^{+}\in
W^{1,q}_{loc}$  \cite{HNT07b} and so we have
\begin{equation*}
u\in W^{s,p}_{\rho}\hookrightarrow W^{1,n}_{\rho}\Rightarrow u\in
W^{1,n}_{loc}\Rightarrow u-\epsilon \in  W^{1,n}_{loc}
\Rightarrow v \in W^{1,n}_{loc}\Rightarrow v\in W^{1,n},
\end{equation*}
since $v$ has compact support.
Now $u\in W^{s,p}_{loc}$ and so $u\in W^{s,p}$ in the support of $v$.
By the multiplication lemma $W^{s,p}\times W^{1,n}\hookrightarrow
W^{1,n}$, therefore $uv$ is a nonnegative, compactly supported
element of $W^{1,n}$. Since $W^{1,n}\hookrightarrow
(W^{s-2.p})^{*}$ and $a\in W^{s-2,p}_{\eta-2}\subseteq
W^{s-2,p}_{loc}$, we can apply $a$ to $uv$ and since $a\geq 0$
and $uv\geq 0$ we have $\langle a,uv\rangle_{(M,h)}\geq 0$. Hence
\begin{equation*}
0\geq -\langle a,uv \rangle \geq \langle -\Delta_h
u,v\rangle=\langle \grad u, \grad v\rangle =\langle \grad v, \grad
v\rangle.
%\quad (\grad (u-\epsilon)=\grad u).
\end{equation*}
It follows that $v$ is constant with compact support which means
 $v\equiv 0$. Note that $v\equiv 0$ if and only if $u-\epsilon \leq 0$. So $u\leq
\epsilon$ for all $\epsilon
>0$. This
shows $u\leq 0$.

For (b), the proof is based on the weak Harnack inequality. The
exact same proof as the one that is given in \cite{HNT07b} for closed
manifolds, works for the above setting as well.
\end{proof}
\begin{lemma}[Elliptic estimate for Laplacian]\lab{lemb4}
Suppose $(M, h)$ is an $n$-dimensional ($n>2$) AF manifold of
class $W^{\alpha,\gamma}_{\rho}$, $\alpha\geq 1$,
$\alpha>\frac{n}{\gamma}$ ,$\rho<0$. If $\alpha>1$, let
$\sigma\in(2-\alpha,\alpha]$ be such that
$(\sigma-\frac{n}{\gamma})+(\alpha-\frac{n}{\gamma})>2-n$. If
$\alpha=1$, let $\sigma=1$. Then
\begin{enumerate}
\item $-\Delta_h \in D^{\alpha,\gamma}_{2,\rho}$.
\item For all $\delta\in \mathbb{R}$, $-\Delta_h:W^{\sigma,\gamma}_{\delta}\rightarrow W^{\sigma-2,\gamma}_{\delta-2}$ is a continuous elliptic operator.
\item For all $\delta\in (2-n,0)$, $-\Delta_h : W^{\sigma,\gamma}_{\delta} \rightarrow W^{\sigma-2,\gamma}_{\delta-2}$
is semi-Fredholm and satisfies the following elliptic estimate:
\begin{equation*}
\| u \|_{W^{\sigma,\gamma}_{\delta}}\,  \lesssim \, \| -\Delta_h u
\|_{W^{\sigma-2,\gamma}_{\delta-2}}+ \|
u\|_{W^{2-\sigma,\gamma}_{\delta'}}.
\end{equation*}
where $\delta'$ can be any real number larger than $\delta$.
\item For all $\delta \in (2-n,0)$, $-\Delta_h:W^{\alpha,\gamma}_{\delta}\rightarrow
W^{\alpha-2,\gamma}_{\delta-2}$ is an isomorphism. In particular
\begin{equation*}
\| u \|_{W^{\alpha,\gamma}_{\delta}}\,  \lesssim \, \| -\Delta_h u
\|_{W^{\alpha-2,\gamma}_{\delta-2}}
\end{equation*}
\end{enumerate}
\end{lemma}
\begin{proof}{\bf (Lemma~\ref{lemb4})}
Item 1 is a direct consequence of the multiplication
lemma and the expression of Laplacian in local coordinates. Item
2 is a direct consequence of item 1 and Theorem \ref{thmB1}. Item
3 is a direct consequence of item 1 and Proposition \ref{propb1}.
For the last item we can proceed as follows:

By item 3, $-\Delta_h:W^{\alpha,p}_{\delta}\rightarrow
W^{\alpha-2,p}_{\delta-2}$ is semi-Fredholm. On the other hand,
Laplacian of the rough metric can be approximated by the
Laplacian of smooth metrics and it is well known that Laplacian
of a smooth metric is Fredholm of index zero. Therefore since the
index of a semi-Fredholm map is locally constant, it follows that
$-\Delta_h$ is Fredholm with index zero. Now maximum principle
(Lemma \ref{lemb8}) implies that the kernel of
$-\Delta_h:W^{\alpha,p}_{\delta}\rightarrow
W^{\alpha-2,p}_{\delta-2}$ is trivial. An injective operator of
index zero is surjective as well. Consequently
$-\Delta_h:W^{\alpha,p}_{\delta}\rightarrow
W^{\alpha-2,p}_{\delta-2}$ is a continuous bijective operator.
Therefore by the open mapping theorem,
$-\Delta_h:W^{\alpha,p}_{\delta}\rightarrow
W^{\alpha-2,p}_{\delta-2}$ is an isomorphism of Banach spaces. In
particular the inverse is continuous and so $\| u
\|_{W^{\alpha,\gamma}_{\delta}}\,  \lesssim \, \|
-\Delta_h u
\|_{W^{\alpha-2,\gamma}_{\delta-2}}$.
\end{proof}
Compact perturbations of Fredholm operators of index $0$ remain
Fredholm of index $0$. The following lemma comes handy in
identifying a useful collection of compact operators.
\begin{lemma}\lab{lemb5}
Let the following assumptions hold:
\begin{itemize}
\item $\eta \in \mathbb{R}$, $\delta \in (-\infty, 0)$.
\item $p \in (1,\infty)$, $\alpha \in (\frac{n}{p},\infty)\cap (1,\infty)$.
\item $\sigma \in (2-\alpha,\alpha]\cap
(\frac{2n}{p}-n+2-\alpha,\infty)$.
\item $a(x)\in W^{\alpha-2,p}_{\eta-2}$.
\end{itemize}
Then: For all $\nu> \delta+\eta-2$, the operator $K:
W^{\sigma,p}_{\delta}(\mathbb{R}^n)\rightarrow
W^{\sigma-2,p}_{\nu}(\mathbb{R}^n)$ defined by $K(\psi)=a\psi$ is
compact. (In particular, we can set $\nu=\eta-2$ and for $n \geq
3$ we can set $\sigma=\alpha$.)
\end{lemma}
\begin{proof}{\bf (Lemma~\ref{lemb5})}
$W^{\sigma,p}_{\delta}$ is a reflexive Banach
space; so in order to show that $K$ is a compact operator, we just
need to prove that it is completely continuous, that is, we need
to show if $\psi_n\rightarrow \psi$ weakly in
$W^{\sigma,p}_{\delta}$, then $K\psi_n\rightarrow K\psi$ strongly
in $W^{\sigma-2,p}_{\nu}$. Let
\begin{align*}
&\beta = \min\{\alpha-\frac{n}{p}, \sigma-(2-\alpha),1,
\sigma-n(\frac{2}{p}-1)-2+\alpha\}.\\
&\theta = \sigma-\frac{1}{2}\beta.\\
& \delta'= \delta + \frac{1}{2}[\nu- (\delta+\eta-2)].
\end{align*}
\begin{itemizeX}
\item \textbf{\bf Step 1:} It follows from the assumptions that $\beta>0$ and so
$\theta<\sigma$. Also clearly $\delta'>\delta$.
Thus we can conclude that $W^{\sigma,p}_{\delta}\hookrightarrow
W^{\theta,p}_{\delta'}$ is compact. Therefore $\psi_n\rightarrow
\psi$ strongly in $W^{\theta,p}_{\delta'}$.
\item \textbf{\bf Step 2:} Now we prove that
\begin{equation*}
W^{\alpha-2,p}_{\eta-2}\times
W^{\theta,p}_{\delta'}\hookrightarrow W^{\sigma-2,p}_{\nu}.
\end{equation*}
According to the multiplication lemma we need to check the
following conditions
%\begin{itemizeX}
%\item $(\eta-2)+\delta'\leq \nu$ This is true by the definition of
%$\delta'$.
%\end{itemizeX}
\begin{enumerate}[(i)]
\item $\alpha-2\geq \sigma-2$. True because $\alpha \geq
\sigma$.\\
%\end{enumerate}
%\begin{enumerate}[(i)]
%\item
$\theta \geq \sigma-2$. True because $\beta \leq 1$ and so
\begin{equation*}
\theta=\sigma-\frac{1}{2}\beta \geq \sigma-\frac{1}{2}\geq
\sigma-2.
\end{equation*}
\item $\alpha-2+\theta > 0$. True because $\beta \leq
\sigma-(2-\alpha)$ and so
\begin{equation*}
\theta=\sigma-\frac{1}{2}\beta > \sigma-\beta \geq \sigma
-(\sigma-(2-\alpha))=2-\alpha.
\end{equation*}
\item $(\alpha-2)-(\sigma-2)\geq 0$. True because $\alpha \geq
\sigma$.
\end{enumerate}
%\begin{enumerate}[(iii)]
%\item
\hspace{0.8cm} $\theta-(\sigma-2)\geq 0$. This one was shown
above.
%\end{enumerate}
\begin{enumerate}[(iv)]
\item
$(\alpha-2)+\theta-(\sigma-2)>n(\frac{1}{p}+\frac{1}{p}-\frac{1}{p})$.
That is, we need to show $\theta > \sigma- (\alpha-\frac{n}{p})$.
This is true because $\beta \leq \alpha-\frac{n}{p}$ and so
\begin{equation*}
\theta=\sigma-\frac{1}{2}\beta \geq \sigma-
\frac{1}{2}(\alpha-\frac{n}{p})>\sigma- (\alpha-\frac{n}{p}).
\end{equation*}
\end{enumerate}
\begin{enumerate}[(v)]
\item $(\alpha-2)+\theta> n(\frac{1}{p}+\frac{1}{p}-1)$. This
is true because $\beta \leq \sigma-n(\frac{2}{p}-1)-2+\alpha$ and
so
\begin{equation*}
\theta=\sigma-\frac{1}{2}\beta > \sigma-\beta \geq \sigma
-[\sigma-n(\frac{2}{p}-1)-2+\alpha]=n(\frac{2}{p}-1)+2-\alpha.
\end{equation*}
\end{enumerate}
The numbering of the above items agrees with the numbering of the
conditions in the multiplication lemma. Also note that
$(\eta-2)+\delta'\leq \nu$ by the definition of $\delta'$.
\item \textbf{\bf Step 3:} By what was proved in \textbf{Step 2} we have
\begin{equation*}
\|a(\psi_n-\psi)\|_{W^{\sigma-2,p}_{\nu}}\preceq
\|a\|_{W^{\alpha-2,p}_{\eta-2}}\|\psi_n-\psi\|_{W^{\theta,p}_{\delta'}}.
\end{equation*}
But by \textbf{Step 1}, the right hand side goes to zero, which means
$a\psi_n\rightarrow a\psi$ strongly in $W^{\sigma-2,p}_{\nu}$.
\end{itemizeX}
\end{proof}
%%%%%%%%%%%%%%%%%%%%%%%%%%%%%%%%%%%%%%%
%%%%%%%%%%%%%%%%%%%%%%%%%%%%%%%%%%%%%%%
%%%%%%%%%%%%%%%%%%%%%%%%%%%%%%%%%%%%%%%
%% \begin{comment}
%% \textbf{Remark:} It might seem that we could have just assumed
%% $a\in W^{s-2,p}_{loc}$. But that does not work because we want
%% $-\Delta_h u+ au$ to be a meaningful expression, for instance as
%% an element of $W^{s-2,p}_{\rho-2}$.
%% \end{comment}
%%%%%%%%%%%%%%%%%%%%%%%%%%%%%%%%%%%%%%%
\begin{lemma}[Ehrling's lemma]\lab{lemehrling}\cite{18}
Let $X$, $Y$ and $Z$ be Banach spaces. Assume that $X$ is
compactly embedded in $Y$ and $Y$ is continuously embedded in $Z$.
Then for every $\epsilon >0$ there exists a constant $c(\epsilon)$
such that
\begin{equation*}
\| x \|_Y \leq \epsilon \| x \|_X +
c(\epsilon)\| x \|_Z
\end{equation*}
\end{lemma}

%%%%%%%%%%%%%%%%%%%%%%%%%%%%%%%%%%%%%%%%%%%%%%%%%%%%%%%%%%%%%%%%%%%%%%%%%%%%%%
\section{Artificial Conformal Covariance of the Hamiltonian Constraint}
   \label{app:covariance}

Here we develop several results we need involving properties of the
Hamiltonian constraint under a conformal change.
We closely follow the argument in \cite{HNT07b} for closed manifolds.

Let $(M, h)$ be a 3-dimensional AF manifold of class
$W^{s,p}_{\delta}$ where $p\in(1,\infty)$,  $s \in
(\frac{3}{p},\infty)\cap[1,\infty)$, and $\delta<0$. Suppose
$\beta<0$. For $\psi\in W^{s,p}_\delta$ and $a_\tau, a_\rho, a_W
\in W^{s-2,p}_{\beta-2}$, let
\begin{equation*}
H(\psi, a_W, a_\tau, a_\rho):= -\Delta_h \psi+a_{R_h}(\psi
+\mu)+a_\tau(\psi+\mu)^5-a_W (\psi+\mu)^{-7}-a_\rho(\psi+\mu)^{-3}
\end{equation*}
where $\mu$ is a fixed positive constant, $a_{R_h}=\frac{R_h}{8}$,
and $R_h\in W^{s-2,p}_{\delta-2}$ is the scalar curvature of the
metric $h$. Note that the Hamiltonian constraint can be
represented by the equation $H=0$.

Now let $\tilde{h}=(\xi+1)^4 h$ where $\xi \in W^{s,p}_{\delta}$
is a fixed function with $\xi>-1$. According to the discussion
right after Definition \ref{defAE} we know that $(M, \tilde{h})$
is also AF of class $W^{s,p}_{\delta}$. Define
\begin{equation*}
\tilde{H}(\psi, a_W, a_\tau, a_\rho):=-\Delta_{\tilde{h}}
\psi+a_{R_{\tilde{h}}}(\psi +\mu)+a_\tau(\psi+\mu)^5-\tilde{a}_W
(\psi+\mu)^{-7}-\tilde{a}_\rho(\psi+\mu)^{-3}
\end{equation*}
where $\tilde{a}_W:= (\xi+1)^{-12}a_W$ and
$\tilde{a}_\rho:=(\xi+1)^{-8}a_\rho$. Note that it follows from
Lemma \ref{lempA1} that $\tilde{a}_W$ and $\tilde{a}_\rho$ are in
$W^{s-2,p}_{\beta-2}$.
\begin{proposition}\label{prop:X2}
For all $\psi \in W^{s,p}_{\delta}$
\begin{equation*}
\tilde{H}(\psi, a_W, a_\tau,
a_\rho)=(\xi+1)^{-5}H((\xi+1)\psi+\mu\,
 \xi, a_W, a_\tau, a_\rho).
\end{equation*}
\end{proposition}
\begin{proof}{\bf (Proposition~\ref{prop:X2})}
Let $\theta= \xi+1$. Then we have
\begin{align*}
&R_{\tilde{h}}=(-8 \Delta_h \theta+R_h \theta)\theta^{-5},\\
& \Delta_h(\theta \psi+\mu (\theta-1))=\Delta_h(\theta
\psi)+\mu\Delta_h\theta=(\Delta_h\theta)\psi+\theta\Delta_h
\psi+2\langle \grad \psi,\grad \theta\rangle_h+\mu\Delta_h\theta,\\
&
\Delta_{\tilde{h}}\psi=\theta^{-4}\Delta_h\psi+2\theta^{-5}
\langle \grad\psi,\grad \theta\rangle_h.
\end{align*}
Therefore we can write
\begin{align*}
(\xi+1)^{-5}& H((\xi+1)\psi+\mu\,
 \xi, a_W, a_\tau,
 a_\rho)=\theta^{-5}H(\theta\psi+\mu\theta-\mu, a_W, a_\tau, a_\rho)\\
 &=\theta^{-5}[-\Delta_h
 (\theta\psi+\mu\theta-\mu)+\frac{1}{8}R_h(\theta \psi+\theta
 \mu)+a_\tau (\theta \psi+\theta \mu)^5
\\
& \quad \quad
- a_W(\theta\psi+\theta\mu)^{-7}-a_\rho(\theta\psi+\theta\mu)^{-3}]\\
 &=\theta^{-5}[(-\Delta_h\theta)\psi-\theta\Delta_h
\psi-2
\langle \grad \psi,\grad
\theta\rangle_h-\mu\Delta_h\theta+\frac{1}{8}R_h\theta(\psi+\mu)
\\
& \quad \quad +a_\tau\theta^{5}(\psi+\mu)^5-a_W\theta^{-7}(\psi+\mu)^{-7}-a_\rho\theta^{-3}(\psi+\mu)^{-3}]\\
&=\big[-\theta^{-4}\Delta_h \psi-2\theta^{-5}
\langle \grad \psi,\grad \theta\rangle_h\big]+\big[-\theta^{-5}(\Delta_h
\theta)\psi-\mu\theta^{-5}\Delta_h\theta
\\
& \quad \quad +\frac{1}{8}R_h\theta^{-4}(\psi+\mu)\big]
+a_\tau(\psi+\mu)^5-a_W\theta^{-12}(\psi+\mu)^{-7}-a_\rho\theta^{-8}(\psi+\mu)^{-3}\\
&=-\Delta_{\tilde{h}}\psi+\frac{1}{8}R_{\tilde{h}}(\psi+\mu)+a_\tau(\psi+\mu)^5-\tilde{a}_W(\psi+\mu)^{-7}-\tilde{a}_\rho(\psi+\mu)^{-3}
\\
&=\tilde{H}(\psi, a_W, a_\tau, a_\rho).
\end{align*}
\end{proof}
We have the following important corollary:
\begin{corollary}\lab{corocovariance}
Assume the above setting. Then we have
\begin{align*}
\tilde{H}(\tilde{\psi}, a_W, a_\tau, a_\rho)=0 \Longleftrightarrow
H((\xi+1)\tilde{\psi}+\mu\,\xi, a_W, a_\tau, a_\rho)=0, \\
\tilde{H}(\tilde{\psi}, a_W, a_\tau, a_\rho)\geq 0
\Longleftrightarrow
H((\xi+1)\tilde{\psi}+\mu\,\xi, a_W, a_\tau, a_\rho)\geq 0, \\
\tilde{H}(\tilde{\psi}, a_W, a_\tau, a_\rho)\leq 0
\Longleftrightarrow
H((\xi+1)\tilde{\psi}+\mu\,\xi, a_W, a_\tau, a_\rho)\leq 0. \\
\end{align*}
In particular, if $\tilde{\psi}_+$ and $\tilde{\psi}_-$ are sub
and supersolutions for the equation $\tilde{H}=0$, then
$\psi_+:=(\xi+1)\tilde{\psi}_++\mu\,\xi$ and
$\psi_-:=(\xi+1)\tilde{\psi}_-+\mu\,\xi$ are sub and
supersolutions for the equation $H=0$.
\end{corollary}

%%%%%%%%%%%%%%%%%%%%%%%%%%%%%%%%%%%%%%%%%%%%%%%%%%%%%%%%%%%%%%%%%%%%%%%%%%%%%%
\section{Metrics in the Positive Yamabe Class}
   \label{app:posyam}

Here we collect some facts regarding the Yamabe invariant
in the case of AF manifolds.

Let $(M,h)$ be a 3-dimensional AF manifold of class
$W^{s,p}_{\delta}$ where $p\in(1,\infty)$,  $s \in
(\frac{3}{p},\infty)\cap(1,\infty)$, and $-1<\delta<0$. We define
the Yamabe invariant as follows: \cite{dM06,6}
\begin{equation*}
\lambda_h=\inf_{f\in C_c^{\infty}(M),f\not \equiv 0} \frac{\int_M
8|\grad f|^2 dV_h +\langle R_h,f^2\rangle_{(M,h)}}{\|f\|_{L^6}^2}
%8|\gradf|^2 dV_h+\langle R_h,f^2\rangle_{(M,h)}}
\end{equation*}
We say $h$ is in the positive Yamabe class if and only if
$\lambda_h>0$. Contrary to what we have for closed manifolds(e.g
\cite{HNT07b}), as
it is discussed in \cite{dM06} and \cite{6} we have

\textbf{$\lambda_h>0$ if and only if there exists a conformal
factor $\eta>0$ such that $\eta-1\in W^{s,p}_{\delta}$ and $(M,
\eta^{4}h)$ is scalar flat.}

It is interesting to notice that if $\lambda_h>0$, then $h$ is
also conformal to a metric with \textbf{continuous positive}
scalar curvature.
\begin{proposition}\label{prop:X3}
Let $(M,h)$ be a 3-dimensional AF manifold of class
$W^{s,p}_{\delta}$ where $p\in(1,\infty)$,  $s \in
(\frac{3}{p},\infty)\cap(1,\infty)$, and $-1<\delta<0$. If $h$
belongs to the positive Yamabe class, then there exist $\chi \in
W^{s,p}_{\delta}$ such that if we set $\hat{h}=(1+\chi)^4 h$, then
$R_{\hat{h}}$ is continuous and positive.
\end{proposition}
\begin{proof}{\bf (Proposition~\ref{prop:X3})}
If $h$ is in the positive Yamabe class, then there
exists $\eta \in W^{s,p}_{\delta}$, $\eta>-1$ such that
$R_{\tilde{h}}=0$ where $\tilde{h}=(1+\eta)^4 h$. Let $f$ be a
smooth positive function in $W^{s-2,p}_{\delta-2}$. By Lemma
\ref{lemb4} there exists a unique function $v\in
W^{s,p}_{\delta}$ such that $-8\Delta_{\tilde{h}}v=f$. By the
maximum principle (Lemma \ref{lemb8}) v is positive. Now define
$\hat{h}=(1+v)^4 \tilde{h}$. We have
\begin{equation*}
R_{\hat{h}}=\big(-8\Delta_{\tilde{h}}v+R_{\tilde{h}}(1+v)\big)(1+v)^{-5}=8f(1+v)^{-5}.
\end{equation*}
Since $f$ and $v$ are both continuous and positive we can
conclude that $R_{\hat{h}}$ is continuous and positive. If we set
$\chi=v+\eta+\eta v$, then $\chi \in W^{s,p}_{\delta}$ and
\begin{equation*}
\hat{h}=(1+v)^4(1+\eta)^4 h=(1+\chi)^4 h.
\end{equation*}
Note that since $v>0$ and $\eta>-1$ we have $\chi>-1$.
\end{proof}

%%%%%%%%%%%%%%%%%%%%%%%%%%%%%%%%%%%%%%%%%%%%%%%%%%%%%%%%%%%%%%%%%%%%%%%%%%%%%%
\section{Analysis of the LCBY Equations in Bessel Potential Spaces}
   \label{app:bessel}
As it was pointed out in Appendix A, Sobolev-Slobodeckij spaces
are not the only option that we have if we wish to work with
noninteger order Sobolev spaces. Another option is to consider
the Bessel potential spaces $H^{s,p}(\mathbb{R}^n)$ and then
define the weighted spaces based on Bessel potential spaces.
Bessel potential spaces agree with Sobolev spaces
$W^{s,p}(\mathbb{R}^n)$ when $s$ is an integer and therefore they
can be considered as an extension of integer order Sobolev
spaces. There are two main advantages in working with Bessel
potential spaces (and the corresponding weighted versions) in
comparison with Sobolev-Slobodeckij spaces: First, Bessel
potential spaces have better interpolation properties; second we
have a better (stronger) multiplication lemma for Bessel
potential spaces.
\begin{theorem} [Complex Interpolation]\cite{36}
 Suppose $\theta \in (0,1)$,
 $0\leq s_0,s_1<\infty$, and
 $1<p_0,p_1<\infty$. If
\begin{equation*}
s=(1-\theta)s_0+\theta s_1, \quad \quad
\frac{1}{p}=\frac{1-\theta}{p_0}+\frac{\theta}{p_1},
\end{equation*}
then $H^{s,p}(\mathbb{R}^n)=[H^{s_0,p_0}(\mathbb{R}^n),
H^{s_1,p_1}(\mathbb{R}^n)]_{\theta}$.
\end{theorem}
\begin{lemma}\lab{thmemult}
Let $s_i \geq s$ with $s_1+s_2\geq 0$, and $1 < p, p_i < \infty$
($i=1,2$) be real numbers satisfying
\begin{equation*}
s_i-s\geq n(\dfrac{1}{p_i}-\dfrac{1}{p}),\quad
s_1+s_2-s>n(\dfrac{1}{p_1}+\dfrac{1}{p_2}-\dfrac{1}{p})\geq 0,
%\quad p\leq \max\{p_1,p_2\}.
\end{equation*}
%where the strictness of the inequalities can be interchanged if
%$s\in \mathbb{N}_0$.
In case $s<0$ let
\begin{equation*}
s_1+s_2> n(\dfrac{1}{p_1}+\dfrac{1}{p_2}-1) \quad
\textrm{(equality is allowed if $min(s_1,s_2)<0$)}.
\end{equation*}
Then the pointwise multiplication of functions extends uniquely
to a continuous bilinear map
\begin{equation*}
H^{s_1,p_1}(\mathbb{R}^n)\times
H^{s_2,p_2}(\mathbb{R}^n)\rightarrow H^{s,p}(\mathbb{R}^n).
\end{equation*}
\end{lemma}
We will prove the above lemma later in this Appendix.
\begin{remark}
We make the following observations.
\begin{itemizeX}
\item Note that in the above multiplication lemma there is no
restriction on the values of $p_1$ and $p_2$ with respect to $p$.
That is, it is allowed for $p_1$ or $p_2$ to be greater than $p$.
\item Note that contrary to what we had for
Sobolev-Slobodeckij spaces, the complex interpolation works
regardless of whether exponents are integer or noninteger. This
feature is crucial because complex interpolation works much
better for interpolation of bilinear forms. This is one of the
reasons that we have a stronger multiplication lemma for Bessel
potential spaces.
\end{itemizeX}
\end{remark}
Let us denote the weighted spaces based on $H^{s,p}$ by
$H^{s,p}_{\delta}$ (rather than $W^{s,p}_{\delta}$). Our spaces
$H^{s,p}_{\delta}(\mathbb{R}^n)$ correspond with the spaces
$h^{s}_{p,ps-p\delta-n}(\mathbb{R}^n)$ in \cite{15,16}.
\begin{theorem}[Complex Interpolation, Weighted Spaces]
\cite{15,16} Suppose $\theta \in (0,1)$. If
\begin{equation*}
s=(1-\theta)s_0+\theta s_1, \quad \quad
\frac{1}{p}=\frac{1-\theta}{p_0}+\frac{\theta}{p_1}, \quad \quad
\delta=(1-\theta)\delta_0+\theta \delta_1
\end{equation*}
then
$H^{s,p}_{\delta}(\mathbb{R}^n)=[H^{s_0,p_0}_{\delta_0}(\mathbb{R}^n),
H^{s_1,p_1}_{\delta_1}(\mathbb{R}^n)]_{\theta}$.
\end{theorem}
The corresponding weighted version of the multiplication lemma
can be proved using the exact same argument as the one that we
used for weighted Sobolev-Slobodeckij spaces.
\begin{lemma}[Multiplication Lemma, Weighted Bessel potential spaces] Assume $s,
s_1, s_2$ and $1 < p, p_1, p_2 < \infty$ are real numbers
satisfying
\begin{enumerate}[(i)]
\item  $s_i \geq s$ \quad($i=1,2$),
\item  $s_1+s_2\geq 0$,
\item  $s_i-s\geq n(\dfrac{1}{p_i}-\dfrac{1}{p})$  \quad($i=1,2$),
\item  $s_1+s_2-s>n(\dfrac{1}{p_1}+\dfrac{1}{p_2}-\dfrac{1}{p})\geq
0$.
\end{enumerate}
%Strictness of the first inequalities in $(iii)$ and $(iv)$ can be
%interchanged if $s\in \mathbb{N}_0$.
In case $s<0$, in addition let
%$1<p, p_i<\infty$, and let
\begin{enumerate}[(v)]
\item $\quad s_1+s_2> n(\dfrac{1}{p_1}+\dfrac{1}{p_2}-1)$ \quad
\textrm{(equality is allowed if $min(s_1,s_2)<0$)}.
\end{enumerate}
Then for all $\delta_1 , \delta_2 \in \mathbb{R}$, the pointwise
multiplication of functions extends uniquely to a continuous
bilinear map
\begin{equation*}
H^{s_1,p_1}_{\delta_1}(\mathbb{R}^n)\times
H^{s_2,p_2}_{\delta_2}(\mathbb{R}^n)\rightarrow
H^{s,p}_{\delta_1+\delta_2}(\mathbb{R}^n).
\end{equation*}
\end{lemma}
Again notice that $p_1$ and $p_2$ do NOT need to be less than or
equal to $p$. This extra degree of freedom that we have for
multiplication in Bessel potential spaces allows us to remove the
restrictions of the type ``$p=q$ if $e=s \not \in \mathbb{N}_0$''
in all the statements of the main text. Consequently we will have
a stronger existence theorem as follows:
\begin{theorem}
Let $(M,h)$ be a $3$-dimensional AF Riemannian manifold of class
$H^{s,p}_{\delta}$ where $p\in (1,\infty)$, $s\in
(1+\dfrac{3}{p},\infty)$ and $-1< \delta<0$ are given. Suppose
$h$ admits no nontrivial conformal Killing field and is in the
positive Yamabe class. Let $\beta \in (-1,\delta]$. Select $q$
and $e$ to satisfy:
\begin{align*}
&\frac{1}{q}\in
(0,1)\cap(0,\frac{s-1}{3})\cap[\frac{3-p}{3p},\frac{3+p}{3p}],\\
&e\in(1+\frac{3}{q},\infty)\cap[s-1,s]\cap[\frac{3}{q}+s-\frac{3}{p}-1,\frac{3}{q}+s-\frac{3}{p}].
\end{align*}
Assume that the data satisfies:
\begin{itemize}
\item $\tau \in H^{e-1,q}_{\beta-1}$ if $e\geq 2$ and $\tau\in H^{1,z}_{\beta-1}$ otherwise, where $z=\dfrac{3q}{3+(2-e)q}$ (note that if $e=2$, then $H^{e-1,q}_{\beta-1}=H^{1,z}_{\beta-1}$),
\item $\sigma \in H^{e-1,q}_{\beta-1}$,
\item $\rho \in H^{s-2,p}_{\beta-2}\cap
L^{\infty}_{2\beta-2}$, $\rho\geq 0$ ($\rho$ can be identically
zero),
\item $J\in \textbf{H}^{e-2,q}_{\beta-2}$.
\end{itemize}
If $\mu$ is chosen to be sufficiently small and if
$\|\sigma\|_{L^{\infty}_{\beta-1}}$,
$\|\rho\|_{L^{\infty}_{2\beta-2}}$, and
$\|J\|_{\textbf{H}^{e-2,q}_{\beta-2}}$ are sufficiently small,
 then there exists $\psi\in H^{s,p}_{\delta}$ with $\psi>-\mu$ and
$W\in \textbf{H}^{e,q}_{\beta}$ solving (\ref{eqweak1}) and
(\ref{eqweak2}).
\end{theorem}

%%%%%%%%%%%%%%%%%%%%%%%%%%%%%%%%%%%%%%%%%%%%%%%%%%%%%%%%%%%%%%%%%%%%%%%%%%%%%%
%\section{Multiplication in Sobolev-Slobodeckij and Bessel Potential Spaces}
% \label{app:mult}
Our plan for the remainder of this appendix is to discuss the
proof of the stronger version of multiplication lemma that was
stated for Bessel potential spaces.
%\begin{comment}
In our proof we will make use of some of the well-known results
for pointwise multiplication in Triebel-Lizorkin spaces that can
be found in \cite{37}. Just for the purpose of completeness we
quickly review the definition of Besov spaces and
Triebel-Lizorkin spaces and their relations to the
Sobolev-Slobodeckij spaces and Bessel potential spaces.
\begin{definition}
Consider the partition of unity $\{\varphi_j\}$ that was
introduced in the beginning of Appendix A.
\begin{itemizeX}
\item For $s\in \mathbb{R}$, $1\leq p<\infty$, and $1 \leq q<\infty$ (or
$p=q=\infty$) define the Triebel-Lizorkin space
$F^{s}_{p,q}(\mathbb{R}^n)$ as follows
\begin{equation*}
F^{s}_{p,q}(\mathbb{R}^n)=\{u\in S'(\mathbb{R}^n): \|
u\|_{F^{s}_{p,q}(\mathbb{R}^n)}= \big |\big| \|
2^{sj}\mathcal{F}^{-1}(\varphi_j \mathcal{F}u)\|_{l^q}
\big|\big|_{L^p(\mathbb{R}^n)}<\infty\}
\end{equation*}
\item For $s\in \mathbb{R}$, $1\leq p<\infty$, and $1 \leq
q<\infty$ define the Besov space $B^s_{p,q}(\mathbb{R}^n)$ as
follows
\begin{equation*}
B^{s}_{p,q}(\mathbb{R}^n)=\{u\in S'(\mathbb{R}^n): \|
u\|_{B^{s}_{p,q}(\mathbb{R}^n)}= \big |\big| \|
2^{sj}\mathcal{F}^{-1}(\varphi_j
\mathcal{F}u)\|_{L^p(\mathbb{R}^n)}
\big|\big|_{l^q}<\infty\}
\end{equation*}
\end{itemizeX}
\end{definition}
We have the following relations \cite{36,StHo2011a}:
\begin{align*}
L^p &= F^0_{p,2},\quad 1<p<\infty, \\
B^{s}_{p,p} &= F^{s}_{p,p}, \quad s\in \mathbb{R}, \quad 1<p<\infty, \\
H^{s,p} &= F^{s}_{p,2}, \quad s\in \mathbb{R}, \quad 1<p<\infty, \\
W^{k,p} &= H^{k,p}=F^{k}_{p,2}, \quad k\in \mathbb{Z}, \quad 1<p<\infty, \\
W^{s,p} & =B^s_{p,p}=F^s_{p,p},
   \quad s\in \mathbb{R}\setminus \mathbb{Z}, \quad 1<p<\infty, \\
\mbox{ If }   k &\in \mathbb{N},
\mbox{ then } B^{k}_{p,p}\hookrightarrow W^{k,p}
\mbox{ for }  1\leq p \leq 2
\mbox{ and }  W^{k,p}\hookrightarrow B^k_{p,p}
\mbox{ for }  p\geq 2.
\end{align*}
%\end{comment}
With these definitions and notation at our disposal, we now give
an abbreviated proof of the key multiplication Lemma~\ref{thmemult} that
was stated earlier.
%The new complete proof for all cases may be found in~\cite{35}.
\begin{proof}{\bf (Lemma~\ref{thmemult})}
We prove the lemma for the
case $s\geq 0$. The case $s<0$ can be proved by using a duality
argument that can be found in \cite{35}. We may consider three
cases:
\begin{itemizeX}
\item \textbf{Case 1:} \textbf{$p_1, p_2\leq p$:} This case is proved in \cite{35}. The
proof in \cite{35} is based on complex interpolation.
\item \textbf{Case 2:} \textbf{$p\leq \min\{p_1,p_2\}$:} In what follows we will discuss the proof of this
case. For now let's assume the lemma holds true in this case.
\item \textbf{Case 3:} \textbf{$p_1>p,\, p_2\leq p$ or $p_2>p,\, p_1\leq p$:} Here we
prove the case where $p_1>p,\, p_2\leq p$. The proof of the other
case is completely analogous. We have
\begin{equation*}
H^{s_1,p_1}\times H^{s_2,p_2}\hookrightarrow H^{s_1,p_1}\times
H^{s_2-\frac{n}{p_2}+\frac{n}{p},p}\hookrightarrow H^{s,p}.
\end{equation*}
Note that by assumption $s_2-\frac{n}{p_2}\geq s-\frac{n}{p}$ and
so $s_2-\frac{n}{p_2}+\frac{n}{p}\geq s \geq 0$. The first
embedding is true because $H^{s_2,p_2}\hookrightarrow
H^{s_2-\frac{n}{p_2}+\frac{n}{p},p}$ (one can easily check that
the conditions of Theorem \ref{thmembedunweight}, which is also
valid for Bessel potential spaces, are satisfied). Also as a
direct consequence of the claim of \textbf{Case 2}, the second embedding
holds true (note that $p\leq \min\{p, p_1\}$).
\end{itemizeX}
So it remains to prove the claim of \textbf{Case 2}, that is the case
where $p\leq \min\{p_1,p_2\}$. Of course if both $p_1$ and $p_2$
are equal to $p$, then the claim follows from case 1; so we may
assume at least one of $p_1$ or $p_2$ is strictly larger than
$p$. To prove \textbf{Case 2} we proceed as follows:
\begin{itemizeX}
\item \textbf{Step 1:} In this step we consider the case where
$s_1=s_2=s$. Note that by assumption
$\frac{1}{p_1}+\frac{1}{p_2}-\frac{1}{p}\geq 0$. If
$\frac{1}{p_1}+\frac{1}{p_2}=\frac{1}{p}$, then let $k=\floor{s}$.
We have (\cite{jZ77})
\begin{align*}
& H^{k+1,p_1}\times H^{k+1,p_2}\hookrightarrow H^{k+1,p},\\
& H^{k,p_1}\times H^{k,p_2}\hookrightarrow H^{k,p},
\end{align*}
so by complex interpolation we get
\begin{equation*}
H^{s,p_1}\times H^{s,p_2}\hookrightarrow H^{s,p}.
\end{equation*}
As a direct consequence of Theorem 2 in page 239 of \cite{37},
the above embedding remains valid if
$\frac{1}{p_1}+\frac{1}{p_2}>\frac{1}{p}$ and $p_1, p_2>p$. What
if $p_2=p$ or $p_1=p$? Here we consider the case where $p_2=p$
(and so $p_1>p$). The proof of the other case is completely
analogous. Note that by assumption $s>
n(\frac{1}{p_1}+\frac{1}{p_2}-\frac{1}{p})=\frac{n}{p_1}$; under
this assumption we need to prove the following:
\begin{equation*}
H^{s,p_1}\times H^{s,p}\hookrightarrow H^{s,p}.
\end{equation*}
If $s\neq \frac{n}{p}$, the above embedding follows from Theorem
1 in page 176 and Theorem 2 in page 177 of \cite{37}. Now if
$s=\frac{n}{p}$, we set $\epsilon=\frac{n}{p}-\frac{n}{p_1}$ and
then since the claim is true for $s\neq \frac{n}{p}$ we have
\begin{align*}
&H^{\frac{n}{p}-\frac{\epsilon}{2},p_1}\times
H^{\frac{n}{p}-\frac{\epsilon}{2},p}\hookrightarrow
H^{\frac{n}{p}-\frac{\epsilon}{2},p}, \\
& H^{\frac{n}{p}+\frac{\epsilon}{2},p_1}\times
H^{\frac{n}{p}+\frac{\epsilon}{2},p}\hookrightarrow
H^{\frac{n}{p}+\frac{\epsilon}{2},p},
\end{align*}
so the result follows from complex interpolation.
\item \textbf{Step 2:} Let $t_1,t_2\in
[0,\frac{n}{p_1}+\frac{n}{p_2}-\frac{n}{p}]$ and suppose
$\epsilon>0$ is such that
$t_1+t_2-(\frac{n}{p_1}+\frac{n}{p_2}-\frac{n}{p})-\epsilon\geq
0$. Then as a direct consequence of the Corollary that is stated
in page 189 of \cite{37} we have
\begin{equation*}
H^{t_1,p_1}\times H^{t_2,p_2}\hookrightarrow
H^{t_1+t_2-(\frac{n}{p_1}+\frac{n}{p_2}-\frac{n}{p})-\epsilon,p}
\end{equation*}
\item \textbf{Step 3:} Note that by \textbf{Step 1}, if
$b>\frac{n}{p_1}+\frac{n}{p_2}-\frac{n}{p}$, then
\begin{equation*}
H^{b,p_1}\times H^{b,p_2}\hookrightarrow H^{b,p}.
\end{equation*}
Also if we let $\frac{1}{r}=\frac{1}{p}-\frac{1}{p_2}$, then
$H^{b,p_1}\hookrightarrow L^r$ and so by Holder's inequality
\begin{equation*}
H^{b,p_1}\times H^{0,p_2}\hookrightarrow H^{0,p}.
\end{equation*}
By complex interpolation we get
\begin{equation*}
\forall t\in [0,b] \quad H^{b,p_1}\times H^{t,p_2}\hookrightarrow
H^{t,p}.
\end{equation*}
Therefore
\begin{equation*}
\forall \epsilon>0 \,\, \forall t\in
[0,\frac{n}{p_1}+\frac{n}{p_2}-\frac{n}{p}]\quad
H^{\frac{n}{p_1}+\frac{n}{p_2}-\frac{n}{p}+\epsilon,p_1}\times
H^{t,p_2}\hookrightarrow H^{t,p}.
\end{equation*}
\item \textbf{Step 4:} In this step we prove the claim of \textbf{Case 2}
in its general form. Without loss of generality we may assume
$s_1=\max\{s_1,s_2\}$, so $s_2\in [0,s_1]$. If
$s_1>\frac{n}{p_1}+\frac{n}{p_2}-\frac{n}{p}$, then by what was
proved in \textbf{Step 3} we have
\begin{equation*}
H^{s_1,p_1}\times H^{s_2,p_2}\hookrightarrow
H^{s_2,p}\hookrightarrow H^{s,p}.
\end{equation*}
In case $s_1\leq \frac{n}{p_1}+\frac{n}{p_2}-\frac{n}{p}$ (that
is, if $s_1,s_2\in [0,\frac{n}{p_1}+\frac{n}{p_2}-\frac{n}{p}] $),
choose $\epsilon>0$ such that
$s_1+s_2-(\frac{n}{p_1}+\frac{n}{p_2}-\frac{n}{p})-\epsilon>s\geq
0 $. Then by \textbf{Step 2} we have:
\begin{equation*}
H^{s_1,p_1}\times H^{s_2,p_2}\hookrightarrow
H^{s_1+s_2-(\frac{n}{p_1}+\frac{n}{p_2}-\frac{n}{p})-\epsilon,p}\hookrightarrow
H^{s,p}.
\end{equation*}
\end{itemizeX}
\end{proof}

%%%%%%%%%%%%%%%%%%%%%%%%%%%%%%%%%%%%%%%%%%%%%%%%%%%%%%%%%%%%%%%%%%%%%%%%%%%%%%
\section{An Alternative Weak Formulation of the LCBY Equations}
   \label{app:weak2}

In Section~\ref{sec:weak} we described a setting where the constraint
equations make sense with rough data. Here we describe a second
framework in which rough data is allowed. Recall that according to our preliminary discussion in Section~\ref{sec:weak},
we have already imposed the following restrictions:
\begin{equation*}
p\in (1,\infty),\quad s\in(\frac{3}{p},\infty)\cap[1,\infty),\quad
\delta<0.
\end{equation*}

\noindent
\textbf{Framework~2:}\\
In this framework we seek $W$ in $\textbf{W}^{1,2r}_{\beta}$
where $r\geq 1$ and $\beta<0$. For the momentum constraint to
make sense we need to ensure that
\begin{enumerate}
\item it is possible to extend the operator $-\Delta_L:
\textbf{C}^{\infty}\rightarrow \textbf{C}^{\infty}$ to an
operator $\mathcal{A}_L: \textbf{W}^{1,2r}_{\beta}\rightarrow
\textbf{W}^{-1,2r}_{\beta-2}$.
\item $b_\tau (\psi+\mu)^6+b_J \in \textbf{W}^{-1,2r}_{\beta-2}$.
\end{enumerate}
Note that $\Delta_L$ belongs to the class $D^{s,p}_{2,\delta}$.
Therefore by Theorem \ref{thmB1} we just need to check the
following conditions (numbering corresponds
to conditions in Theorem \ref{thmB1}):
\begin{center}
\begin{tabular}{lll}
    (i) & $2r\in (1,\infty)$,      & (since $r\geq 1$) \\
   (ii) & $1> 2-s$,                & (enough to assume $s> 1$) \\
  (iii) & $-1< \min(1,s)-2$,       & (enough to assume $s> 1$) \\
   (iv) & $-1< 1-2+s-\frac{3}{p}$, & (since $\frac{3}{p}<s$)  \\
    (v) & $-1-\frac{3}{2r}\leq s-\frac{3}{p}-2$,
                                   & (so need to assume
             $1-\frac{3}{2r}\leq s-\frac{3}{p}$) \\
   (vi) & $1-\frac{3}{2r}> 2-3-s+\frac{3}{p}$.
%\Leftrightarrow
%1-\frac{3}{2r}\geq -1-(s-\frac{3}{p})
& (since $r\geq 1$ and so $1-\frac{3}{2r}\geq \frac{-1}{2}>-1-(s-\frac{3}{p})$)
\end{tabular}
\end{center}
So the only extra assumptions that we need to make is that
$1-\frac{3}{2r}\leq s-\frac{3}{p}$ and $s>1$. Also in order to
ensure that the second condition holds true it is enough to assume
$\tau$ is given in $L^{2r}_{\beta-1}$ and $J$ is given in
$\textbf{W}^{-1,2r}_{\beta-2}$. Indeed, note that $\tau \in
L^{2r}_{\beta-1}$ implies $b_{\tau}\in
\textbf{W}^{-1,2r}_{\beta-2}$. Since $\psi \in W^{s,p}_{\delta}$,
it follows from Lemma \ref{lempA1} that $b_\tau (\psi+\mu)^6 \in
\textbf{W}^{-1,2r}_{\beta-2}$; Lemma \ref{lempA1} can be applied
because clearly $2r \in (1,\infty)$ and moreover
\begin{center}
\begin{tabular}{lll}
 (i) & $-1 \in (-s,s)$, & (since $s> 1$) \\
%& 2r\in (1,\infty)\quad (\textrm{checked above})\\
 (ii) & $-1-\frac{3}{2r}\leq s-\frac{3}{p}$,
       & (since we assumed $1-\frac{3}{2r}\leq s-\frac{3}{p}$) \\
      & $-3-s+\frac{3}{p}\leq -1-\frac{3}{2r}$.
%\Leftrightarrow
%s-\frac{3}{p}+2\geq \frac{3}{2r}
       & (the same as item (vi) above)
\end{tabular}
\end{center}
The numbering of the above items corresponds to the numbering of
the conditions in Lemma \ref{lempA1}.

 Now let's consider the
Hamiltonian constraint. Note that $W \in
\textbf{W}^{1,2r}_{\beta}$ and so $\mathcal{L}W\in
L^{2r}_{\beta-1}$. So for $a_W$ to be well defined it is enough
to assume $\sigma \in L^{2r}_{\beta-1}$. Exactly similar to our
 discussion for weak formulation 1, for Hamiltonian constraint to
make sense it is enough to ensure that
\begin{equation*}
f(\psi,W)=a_R(\psi+\mu)+a_\tau(\psi+\mu)^5-a_W(\psi+\mu)^{-7}-a_\rho(\psi+\mu)^{-3}\in
W^{s-2,p}_{\eta-2},
\end{equation*}
where $\eta=\max\{\delta,\beta\}$.
%\begin{itemizeX}
One way to guarantee that the above statement holds true is to
ensure that
\begin{equation*}
a_\tau, a_\rho, a_W\in W^{s-2,p}_{\beta-2}, \quad a_R \in
W^{s-2,p}_{\delta-2},
\end{equation*}
We claim that for the above statements to be true it is enough to
make the following extra assumptions:
\begin{equation*}
 s\leq 2, \quad 1-\frac{3}{2r}\geq \frac{1}{2}(s-\frac{3}{p}), \quad
\rho\in W^{s-2,p}_{\beta-2}.
\end{equation*}
The details are as follows:
\begin{enumerate}
\item $a_\tau=\frac{1}{12} \tau^2$.\\
We want to ensure $a_\tau \in W^{s-2,p}_{\beta-2}$. Note that
$\tau \in L^{2r}_{\beta-1}$, so $\tau^2\in L^r_{2\beta-2}$. Thus
we want to have $L^r_{2\beta-2}\hookrightarrow
W^{s-2,p}_{\beta-2}$. We will see that this embedding becomes
true provided $s\leq 2$ and $1-\frac{3}{2r}\geq
\frac{1}{2}(s-\frac{3}{p})$.

We just need to check that the
assumptions of Theorem \ref{thmA4} are satisfied
(numbering corresponds to the assumptions in Theorem \ref{thmA4})
\begin{align*}
%&2\beta-2<\beta-2 \quad (\beta<0)\\
%& p\leq r \,\,(\textrm{because}\quad s\leq 2)\\
(ii) & \  0\geq s-2 \quad (\textrm{equivalent to $s\leq 2$}),\\
(iii) & \  0-\frac{3}{r}\geq s-2-\frac{3}{p}\quad
(\textrm{equivalent to $1-\frac{3}{2r}\geq
\frac{1}{2}(s-\frac{3}{p})$})
,\\
(iv) & \  2\beta-2<\beta-2 \quad (\textrm{true because
$\beta<0$}).
\end{align*}
\item $a_R=\frac{R}{8}$.\\
We want to ensure $a_R\in W^{s-2,p}_{\delta-2}$. Note that $h$ is
an AF metric of class $W^{s,p}_{\delta}$ and $R$ involves the
second derivatives of $h$, so $R\in W^{s-2,p}_{\delta-2}$. We do
not need to impose any extra restrictions for this one.
\item $a_\rho=\kappa\rho/4$.\\
Clearly $a_\rho\in W^{s-2,p}_{\beta-2}$ iff $\rho\in
W^{s-2,p}_{\beta-2}$.
\item
$a_W=[\sigma_{ab}+(\mathcal{L}W)_{ab}][\sigma^{ab}+(\mathcal{L}W)^{ab}]/8$.\\
We want to ensure that $a_W\in W^{s-2,p}_{\beta-2}$. Note that
$\mathcal{L}W, \sigma \in L^{2r}_{\beta-1}$ and as discussed
above, $L^r_{2\beta-2}\hookrightarrow W^{s-2,p}_{\beta-2}$. So
$a_W=\frac{1}{8}|\sigma+ \mathcal{L}W|^2\in
L^r_{2\beta-2}\hookrightarrow W^{s-2,p}_{\beta-2}$.
\end{enumerate}
\begin{remark}
According to the above discussion we need $r\geq 1$ satisfy
\begin{equation*}
\frac{1}{2}(s-\frac{3}{p})\leq 1-\frac{3}{2r}\leq s-\frac{3}{p}.
\end{equation*}
In particular, if we choose $r$ such that
$\frac{1}{2}(s-\frac{3}{p})= 1-\frac{3}{2r}$, that is, if we set
$r=\frac{3p}{3+(2-s)p}$, then clearly $r$ satisfies the above
inequalities and moreover since  $s\leq 2$, we have $r\geq 1$.
\end{remark}
\begin{weakf}\lab{weakf1}
Let $(M,h)$ be a $3D$ AF Riemannian manifold of class
$W^{s,p}_{\delta}$ where $p\in (1,\infty)$, $\beta, \delta<0$ and
$s\in (\frac{3}{p},\infty)\cap (1,2]$. Let
$r=\frac{3p}{3+(2-s)p}$. Fix source functions:
\begin{equation*}
\tau \in L^{2r}_{\beta-1},\quad \sigma \in L^{2r}_{\beta-1},\quad
\rho \in W^{s-2,p}_{\beta-2} (\rho\geq 0),\quad J\in
\textbf{W}^{-1,2r}_{\beta-2}.
\end{equation*}
Let $\eta=\max\{\beta, \delta\}$. Define
$f:W^{s,p}_{\delta}\times \textbf{W}^{1,2r}_{\beta}\rightarrow W^{s-2,p}_{\eta-2}$ and $\textbf{f}:W^{s,p}_{\delta}\rightarrow \textbf{W}^{-1,2r}_{\beta-2}$ as
\begin{align*}
\quad f(\psi,W) &=a_R(\psi+\mu)+a_\tau(\psi+\mu)^5-a_W(\psi+\mu)^{-7}-a_\rho(\psi+\mu)^{-3},\\
\textbf{f}(\psi) &=b_\tau (\psi+\mu)^6+b_J.
\end{align*}
Find $(\psi, W)\in W^{s,p}_{\delta}\times
\textbf{W}^{1,2r}_{\beta}$ such that
\begin{align}\lab{eqv2weak1}
A_L \psi+f(\psi,W)=0, \\
\mathcal{A}_L W+\textbf{f}(\psi)=0.\lab{eqv2weak2}
\end{align}
\end{weakf}
\begin{remark}
Consider \textbf{Weak Formulation~\ref{weakf2}}. In the case where $s\leq
2$ and $\frac{1}{q}\geq \frac{2-d}{6}$ where $d=s-\frac{3}{p}$,
this formulation becomes a special case of
\textbf{Weak Formulation~\ref{weakf1}}.
Indeed, we just need to check that in this case
$W^{e,q}_{\beta}\hookrightarrow W^{1,2r}_{\beta}$. By Theorem
 \ref{thmA3} we need to make sure that the followings hold true:
\begin{align*}
  (i) & \ q\leq 2r,
     \quad (\textrm{true because $\frac{1}{q}\geq
           \frac{2-d}{6}= \frac{3+(2-s)p}{6p}= \frac{1}{2r}$}) \\
 (ii) & \ e\geq 1, \quad (\textrm{true because $e>1+\frac{3}{q}$})\\
(iii) & \ e-\frac{3}{q}\geq 1-\frac{3}{2r}.\\
\end{align*}
(The numbering of the above items agrees with the
numbering of the assumptions in Theorem \ref{thmA3}.)
 The third condition is true because
 \begin{equation*}
e-\frac{3}{q}\geq 1-\frac{3}{2r}\Leftrightarrow e\geq
1+\frac{3}{q}-\frac{3+(2-s)p}{2p}\Leftrightarrow e \geq
\frac{3}{q}+\frac{d}{2},
 \end{equation*}
 and
\begin{itemizeX}
\item if $d>2$, then $d-1>\frac{d}{2}$ and so $e\geq
\frac{3}{q}+d-1>\frac{3}{q}+\frac{d}{2}$,
\item if $d\leq 2$, then $1\geq \frac{d}{2}$ and so $e>1+\frac{3}{q}\geq
\frac{3}{q}+\frac{d}{2}$.
\end{itemizeX}
\end{remark}

\bibliographystyle{abbrv}
\bibliography{refs}
\end{document}